\documentclass[11pt, a4paper, twoside, openright]{report}

\usepackage{amsthm,mathrsfs,mathptmx,cite,textcase,ifthen,cite}
\usepackage{amssymb}
\usepackage{amsmath}
\usepackage{amsfonts,csquotes}
\usepackage{lscape}
\usepackage{latexsym,mathrsfs}
\usepackage{graphicx}
\usepackage{graphics}
\usepackage{fancyhdr}
\usepackage{epsfig}
\usepackage{setspace}
\usepackage[T1]{fontenc}
\usepackage{color}
\usepackage{epstopdf}
\usepackage{url,anysize}
\usepackage{parskip}
\usepackage{tikz}
\usetikzlibrary{shapes, arrows}

\usepackage[square,comma,numbers,sort&compress]{natbib}

\newcommand{\properpagestyle}{\pagestyle{myheadings}\markboth{}{}\markright{}}

\makeatletter
\def\cleardoublepage{\clearpage\if@twoside\ifodd\c@page\else
\hbox{}
\vspace*{\fill}
\begin{center}

\end{center}
\vspace{\fill}
\thispagestyle{empty}
\newpage
\if@twocolumn\hbox{}\newpage\fi\fi\fi}
\makeatother

\usepackage[outer=.85in, top=.85in, bottom=.85in, inner=1.2in, ]{geometry}

\usepackage[compact]{titlesec}
\titlespacing{\section}{0pt}{12pt}{12pt}
\titlespacing{\subsection}{0pt}{10pt}{10pt}

\usepackage[]{nomencl}
\def\execute{%
\begingroup
\catcode`\%=12
\catcode`\\=12
\executeaux}
\def\executeaux#1{\immediate\write18{#1}\endgroup}
\execute{makeindex Main.nlo -s nomencl.ist -o Main.nls}
\makenomenclature
\renewcommand{\nomname}{List of symbols}

\renewcommand\contentsname{Contents}

\theoremstyle{plain}
\newtheorem{theorem}{Theorem}
\newtheorem{lemma}{Lemma}

\theoremstyle{definition}
\newtheorem{definition}{Definition}

\theoremstyle{remark}
\newtheorem*{remark}{Remark}

\theoremstyle{definition}
\newtheorem{result}{Result}
\newtheorem{example}{Example}
\newtheorem*{example*}{Example}
\newtheorem{corollary}{Corollary}
\numberwithin{equation}{section}
\numberwithin{section}{chapter}
\numberwithin{figure}{chapter}
\numberwithin{table}{chapter}
\numberwithin{theorem}{chapter}
\numberwithin{result}{chapter}
\numberwithin{corollary}{chapter}
\numberwithin{lemma}{chapter}
\numberwithin{definition}{chapter}
\numberwithin{example}{chapter}
\newcommand{\ThesisTitle}{Characterization of Entanglement in Higher Dimensional Bipartite as well as Multipartite Quantum System and its Application
}
\newcommand{\Student}{Shruti Aggarwal}
\newcommand{\Enrollment}{2K19/PHD/AM/04}
\newcommand{\Supervisor}{Dr.\ Satyabrata Adhikari}

\newcommand{\Institute}{Delhi Technological University}

\newcommand{\SubmissionDate}{April, 2024}
\renewcommand{\bibname}{References}

\begin{document}

\newcommand{\TitlePage}{\thispagestyle{empty} \fontfamily{phv}\selectfont
\begin{center} \Large
\textbf{\MakeTextUppercase{\ThesisTitle}}
\end{center}
\vspace{0.8 cm}
\begin{center}

\emph{A thesis submitted to}\\ 
\textbf{\large\MakeTextUppercase{\Institute}}\\ \ \\
\emph{in partial fulfillment of the requirements for the award of degree of}\\ \ \\
  \textbf{\large DOCTOR OF PHILOSOPHY }\\ 
\emph{in}\\ 
\textbf{\large MATHEMATICS } \\ \ \\ \ \\
\emph{By}\\
  \textbf{\Large\MakeTextUppercase \Student}\\ \ \\
  \vspace{0.5cm}
\emph{Under the Supervision of}  \\
\large \textbf{ \Supervisor} \\
 \end{center}
\vspace{1.5 cm}
\begin{figure}[h]
  \centering
  \includegraphics[width=3.4 cm]{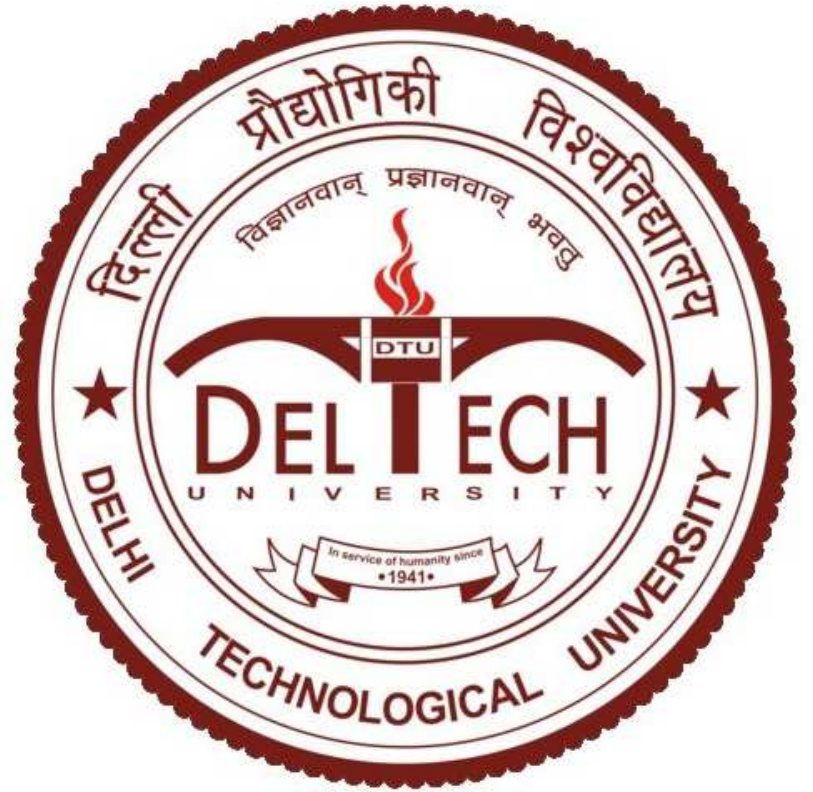}\\
\end{figure}

\begin{center}
DEPARTMENT OF APPLIED MATHEMATICS \\
DELHI TECHNOLOGICAL UNIVERSITY\\
(Formerly Delhi College of Engineering)\\
\small BAWANA ROAD, NEW DELHI-110 042, INDIA.
\end{center}
\vspace{1.5cm}
\begin{flushleft}
\textbf{\SubmissionDate\hfill Enroll. No. : \Enrollment}\\
\end{flushleft}
}
\TitlePage
\newpage

\thispagestyle{empty}
\
\newpage

\setcounter{page}{1}

%
%

%
\newpage
\addcontentsline{toc}{chapter}{Title}
\thispagestyle{empty}
\text{}
    \vspace{8cm}
    \begin{center}
   \large\bf $\copyright$  DELHI TECHNOLOGICAL UNIVERSITY, DELHI, 2024\\ ALL RIGHTS RESERVED.
   \end{center}

\newpage
\
\thispagestyle{empty}
\newpage
\pagenumbering{roman}
\setcounter{page}{1}
\addcontentsline{toc}{chapter}{Declaration}

\begin{center}
	{ \bf \large  DECLARATION\\ \ \\}
\end{center}

\noindent I, the undersigned, hereby declare that the research work reported in this thesis entitled ``\textbf{{\ThesisTitle}}'' for the award of the degree of \emph{Doctor of Philosophy in Mathematics} has been carried out by me under the supervision of \ifthenelse{\isundefined{\SupervisorTwo}}{\emph{\Supervisor}}{{\emph{\Supervisor }} and \emph{\SupervisorTwo}}, Department of Applied Mathematics, Delhi Technological University, Delhi, India.\\

\indent The research work embodied in this thesis, except where otherwise indicated, is my original research. This thesis has not been submitted by me earlier in part or full to any other University or Institute for the award of any degree or diploma. This thesis does not contain other person's data, graphs, or other information unless specifically acknowledged. \\ \ \\
\\ \ \\
\begin{flushright}
	{\bf Date : 02.04.2024  \hfill{\bf (\Student)}\\
	}
{\bf \hfill Enroll. No. : \Enrollment}
\end{flushright}

\newpage
\thispagestyle{empty}
\
\newpage
\addcontentsline{toc}{chapter}{Certificate}
\begin{center}{ \bf \large
		{CERTIFICATE}\\ \ \\}
\end{center}

\indent On the basis of a declaration submitted by {\bf Ms. \Student}, Ph.D. Research Scholar, I hereby certify that the thesis entitled {\bf``\ThesisTitle''} submitted to the Department of Applied Mathematics, Delhi Technological University, Delhi, India for the award of the degree of \emph{Doctor of Philosophy in Mathematics}, is a record of bonafide research work carried out by her under my supervision.

I have read this thesis and, in my opinion, it is fully adequate in scope and quality as a thesis for the degree of Doctor of Philosophy.

To the best of my knowledge, the work reported in this thesis is original and has not been submitted to any other Institution or University in any form for the award of any Degree or Diploma.
\\ \ \\
\vspace{2cm}
\begin{table}[h]
\begin{tabular}{l r}
	{\bf (\Supervisor)}  &  \hspace{5cm} {\bf (Prof. R. Srivastava)}\\
	Supervisor and Assistant Professor  & \hspace{4cm}  Professor and Head\\
    Department of Applied Mathematics & \hspace{2cm} Department of Applied Mathematics\\
	Delhi Technological University & \hspace{2cm} Delhi Technological University \\
	New Delhi, India & \hspace{2cm} New Delhi, India. \\
\end{tabular}
\end{table}

%
%
%
%

\newpage
\thispagestyle{empty}
\
\newpage
\addcontentsline{toc}{chapter}{Acknowledgements}
\begin{center}
	{\bf \large  ACKNOWLEDGEMENTS}
\end{center}

{\textsf{I express my deep sense of gratitude to my research supervisor \Supervisor, Department of Applied Mathematics, Delhi Technological University (DTU), New Delhi, India.
His unwavering enthusiasm for research kept me constantly engaged with my research.}} 

{\textsf{I am extremely thankful to Prof. R. Srivastava, Professor and Head; Prof. S. Sivaprasad Kumar, former HOD, Department of Applied Mathematics, DTU, for providing me a constant support and the necessary facilities in the department. }}

{\textsf{I also wish to extend my gratitude to SRC members: Prof. Archan S. Majumdar, Senior Professor, S. N. Bose National Centre for Basic Sciences, Kolkata; Prof. Tabish Qureshi, Professor and Hony Director, Centre for Theoretical Physics, Jamia Millia Islamia, New Delhi for their insightful comments and encouragement, but also for the questions which incented me to widen my research from various perspectives. My sincere thanks go to Prof. Anjana Gupta, DTU, and Prof. R. K. Sinha, Department of Applied Physics, DTU, for their valuable guidance from the beginning of my Ph.D. research work.}}\\ 
{\textsf{I am grateful to Prof. Sangita Kansal, DRC Chairman; Prof. Naokant Deo, former DRC Chairman, and the esteemed faculty of the Department of Applied Mathematics, DTU, for their support and encouragement. Special thanks to Prof. R Srivastava, DTU, who has offered invaluable advice and unwavering guidance that will benefit me throughout my life.}}

{\textsf{My acknowledgements would not be complete without thanking the academic and technical support of the DTU, and its staff, particularly the Academic-PG section. The library facilities of the University have been indispensable. I also extend my thanks to the administrative staff in the Department of Applied Mathematics for their support and assistance.}}

{\textsf{Lastly, I would be remiss in not mentioning my family and friends. I thank my spouse for being the pillar of strength for me in this process. I'm extremely grateful to my parents and siblings, their belief in me has kept my spirits and motivation high during this entire journey. }}

{\textsf{My appreciation also extends to Dr. Anu Kumari, and Ms. Anuma Garg for valuable discussions on the subject matter and for helping me develop my ideas. I also place on record, my sense of gratitude to one and all who helped me directly or indirectly in this journey.}}

{\textsf{I would like to acknowledge the Council of Scientific and Industrial Research (CSIR), Government of India, for providing fellowship (JRF and SRF) that made my PhD work possible.}}

{\textsf{Finally, I thank GOD, the Almighty, for letting me through all the difficulties and for showering blessings throughout my journey to complete the research successfully.}}\\

\noindent {\bf Date: 02.04.2024} \hfill {\bf (\MakeTextUppercase\Student)}

\newpage
\
\newpage
\text{}
\vspace{6cm}
\begin{center}
	{\Large{\bf{ \emph{
				Dedicated to my parents}}}}\\
			\emph{For their endless love, support and encouragement}
	\vspace{1cm}
\end{center}


{\singlespacing\tableofcontents}

\addcontentsline{toc}{chapter}{Preface}
\chapter*{Preface}
\indent 
In recent years considerable progress has been made towards developing a general theory of quantum entanglement. In particular, the criteria to decide whether a given quantum state is entangled are of high theoretical and practical interest. This problem is additionally complicated by the existence of bound entanglement, which are weak entangled states and hard to detect. In this thesis, we have worked on the characterization of bipartite and tripartite entanglement. We have established a few separability criteria that successfully detect negative partial transpose entangled states (NPTES) as well as positive partial transpose entangled states (PPTES). 
In this thesis, we propose some theoretical ideas to realize these entanglement detection criteria experimentally. \\
We have taken an analytical approach to construct a witness operator. To achieve this, we first constructed a linear map using the combination of partial transposition and realignment operation. Then we find some conditions on the parameters of the map for which the map represents a positive but not completely positive map. We then construct an entanglement witness operator, which is based on the Choi matrix. Finally, we prove its efficiency by detecting several bipartite bound entangled states that were previously undetected by some well-known separability criteria and find that our witness operator detects bound entangled states that are undetected by these criteria.
It is known that the witness operator is useful in the detection as well as the quantification of entangled states. This motivated us for the construction of a family of witness operators that can detect many mixed entangled states. We have shown the significance of our constructed witness operator by detecting many bound entangled states and then, we use the expectation value of the witness operator to estimate the lower bound of the concurrence of those bound entangled states.\\ 
Detection of entanglement through partial knowledge of the quantum state is a challenge to implement efficiently. Here we propose a separability criterion for detecting bipartite entanglement in arbitrary dimensional quantum states using partial moments of the realigned density matrix. Our approach enables the detection of both distillable and bound entangled states through a common framework. We illustrate the efficiency of our method through examples of states belonging to both the above categories, which are not detectable using other schemes relying on partial state information. \\
Realignment criteria is a powerful tool for the detection of entangled states in bipartite and multipartite quantum systems. Since the matrix corresponding to the realignment map is indefinite, the experimental implementation of the map is an obscure task. We have approximated the realignment map to a positive map using the method of structural physical approximation (SPA) and then we have shown that the SPA of the realignment map (SPA-R) is completely positive. The positivity of the constructed map is characterized using moments that can be physically measured. Next, we develop a separability criterion based on our SPA-R map in the form of an inequality and have shown that the developed criterion not only detects NPTES but also PPTES. \\
We have addressed the problem of experimental realization of the realignment operation. We first show that the realignment matrix can be expressed in terms of the partial transposition operation along with a permutation matrix which can be implemented via the SWAP operator acting on the density matrix. We have shown that the first moment of the realignment matrix can be expressed as the expectation value of a SWAP operator which indicates the possibility of its measurement. We defined a new matrix realignment operation for three-qubit states and have devised an entanglement criterion that can detect three-qubit genuine entangled states.\\
Chapter 1 is introductory in nature. Chapters 2 -- 6 are based on the research work published in the form of research papers in reputed refereed journals. Finally, we conclude with future scope and references. Each chapter begins with a brief outline of the work carried out in that chapter. \\\\\\
\noindent {\bf Date : 02.04.2024} \hfill {\bf (\MakeTextUppercase\Student)}
\\ \noindent {\bf Place : DTU, New Delhi, India.}
\newpage
\
\newpage

\addcontentsline{toc}{chapter}{List of Symbols}
\chapter*{List of Symbols}
\thispagestyle{empty}

\begin{tabular}{cp{0.6\textwidth}}
	$V$ & Vector Space \\
	$F$ & Field of real or complex numbers\\
	$\mathbb{N}$ & Set of natural numbers\\
	$\mathbb{R}^n$ & $n$-dimensional Euclidean space\\
	$\mathbb{C}^n$ & Complex vector space of dimension $n$ \\
	$\mathcal{H}$ & Hilbert space\\
	$\mathbb{R^{+}}$ & Set of positive real numbers \\
	$*$ & Complex conjugate \\
	$\dagger$ & Conjugate transpose operation \\
	$||.||$ & norm \\
	$|.\rangle$ & ket vector\\
	$\langle.|$ & bra vector\\
	$\langle .|.\rangle$ & inner product\\
	$I_n$ & Identity matrix of order $n$ \\
	$O_n$ & Null matrix of order $n$ \\
		$\lambda_i(A)$ & $i$th eigenvalue of a matrix $A$\\
	$\sigma_i(A)$ & $i$th singular value of a matrix $A$\\
	$L(\mathcal{H})$ & Set of all linear operators defined on $\mathcal{H}$\\
	$B(\mathcal{H})$ & Set of all bounded linear operators defined on $\mathcal{H}$\\
	$M_n(\mathbb{C})$ & Set of all $n \times n$ matrices with complex entries\\
	$M_{m,n}{(\mathbb{C})}$ & Set of all $m \times n$ matrices with complex entries\\
	$Tr[.]$ & Trace of a matrix\\
	$\sigma_x$ & Pauli $X$-matrix \\
		$\sigma_y$ & Pauli $Y$-matrix \\
			$\sigma_z$ & Pauli $Z$-matrix \\
		$\rho$ & Density matrix\\
		$\rho^{T_B}$ & Partially transposed matrix with respect to subsystem $B$\\
		$R(\rho)$ & Realignment matrix\\
		$\widetilde{R}(\rho)$ & Structural physical approximation of realignment matrix\\
\end{tabular}

\addcontentsline{toc}{chapter}{List of figures}
{\singlespacing\listoffigures}

\addcontentsline{toc}{chapter}{List of tables}
{\singlespacing\listoftables}
\properpagestyle
\pagenumbering{arabic}
%


\chapter{General Introduction}\label{ch1}
\vspace{1.5cm}
	{\small \emph{"As far as the laws of mathematics refer to reality, they are not certain; and as far as they are certain, they do not refer to reality."\\
			-Albert Einstein}}
\vspace{1.5cm}
\noindent\hrule

\noindent\emph{The general introduction gives an overview of the basic principles of quantum mechanics and quantum information theory.  In this initial chapter, we summarize many useful concepts and facts, some of which provide a foundation for the material in the rest of the thesis. 
We also include some additional useful components of the elementary topics of quantum mechanics and linear algebra.
The reader may use this chapter as a review before beginning the main part of the thesis; subsequently, it can serve as a convenient reference for notation and definitions that are encountered in the later chapters. Finally, the chapter contains the motivation and plan of the work carried out in the thesis.}
\noindent\hrulefill
\newpage


\section{Mathematical Background}\label{mb}
The basic principles of quantum mechanics have been formulated in concrete mathematical terms. Let us begin with a discussion of the concept of states, state vectors, and their associated mathematical structure. In this section, we will discuss the mathematical formalism, introducing the concept of linear operators and their association with a physical observable. Further, we will recapitulate some basic concepts of linear algebra and a few important results \cite{niel,horn}.

\subsection{Hilbert Space} \label{HS}
\noindent To start with, let us first review some fundamental concepts about vector spaces, which demand the introduction of some basic properties such as a spanning set, linearly independent, dependent vectors, and the basis of a vector space. Next, we discuss some operations on vector spaces like inner product and norm that lead us to the description of a Hilbert space.\\
In mathematics and physics, a vector space (or a linear space) over the field $F$ is a non-empty set that is closed under vector addition and scalar multiplication. The elements of a vector space are called vectors. The $n$-dimensional Euclidean space $\mathbb{R}^n = \{(a_1, a_2, . , ., ., a_n);\; a_i \in \mathbb{R}\}$ is a basic example of a vector space. 
In quantum mechanics, a vector is denoted by $|v\rangle$, where $v$ is the label for the vector. The vector $|v\rangle$ is called a\textbf{ ket }vector.  Corresponding to the set of set of ket vectors, there exists another set of vectors that are dual to the ket vectors. These are called \textbf{bra }vectors and are denoted by the symbol $\langle.|$. For instance, the bra vector corresponding to the ket vector $| v\rangle$ is $\langle v |$. 
{\definition (\textit{Spanning set}) Let $V$ be a vector space. Let $S = \{|v_i\rangle \;|\; |v_i\rangle \in V; \; i = 1, 2, . . ., n\}$ be a non-empty set consisting of vectors in $V$. Then $S$ is called a spanning set of $V$ if any vector $|v\rangle \in V$ can be expressed as a linear combination of $|v_i\rangle$, $i=1, 2, . . .,n$, i.e.,
 $|v\rangle = \sum_{i=1}^{n} a_i |v_i\rangle$ where $a_i \in F$. }
\begin{example*}
	In the vector space $\mathbb{C}^2$, $\{|v_1\rangle, |v_2\rangle\}$ forms a spanning set, where
	\begin{eqnarray}
		|v_1\rangle =
		\begin{pmatrix}
			1\\
			0
		\end{pmatrix}; \quad \quad	|v_2\rangle =
		\begin{pmatrix}
			0\\
			1
		\end{pmatrix}
	\end{eqnarray} since any vector $|v\rangle = \begin{pmatrix}
		a\\
		b
	\end{pmatrix} \in \mathbb{C}^2$ can be expressed as $|v\rangle = a|v_1\rangle + b|v_2\rangle$.
\end{example*} 
{\remark A vector space may have more than one spanning set.}
{\definition (\textit{Linearly independent and dependent vectors})
A set of vectors $\{|v_1\rangle, |v_2\rangle, ., ., ., |v_n\rangle\}$ is linearly independent if there exist scalars $a_i \in F$, $i=1,2,...,n$ such that 
\begin{eqnarray}
	a_1|v_1\rangle + a_2|v_2\rangle + . . . + a_n|v_n\rangle = 0 \implies a_i = 0, \; \forall\; i= 1,2,...,n
\end{eqnarray}
Otherwise, the set is linearly dependent.
\definition(\textit{Basis and dimension}) A linearly independent set of vectors which spans the vector space $V$ is called a basis for $V$. The number of elements in the basis is defined to be the dimension of $V$.\\
A standard basis is conventionally known as \textit{computational basis} in quantum mechanics. The computational basis for two-dimensional complex vector space may be denoted as $\{|0\rangle, |1\rangle\}$, where the basis vectors may be expressed in terms of the column vectors are given as
\begin{eqnarray}
	|0\rangle = \begin{pmatrix}
		0 \\1
	\end{pmatrix};\quad
	|1\rangle = \begin{pmatrix}
	1 \\0
\end{pmatrix};
\end{eqnarray}} 
\subsubsection{Inner Product space}
Let $V$ be a vector space. The inner product between two vectors $|v_1\rangle$ and $|v_2\rangle$ in $V$ generates a complex number and is denoted by $\langle v_1| v_2 \rangle $ in quantum mechanics.  Mathematically, an inner product is a function $(.,.): V \times V \longrightarrow \mathbb{C}$ such that it satisfies the following properties:
\begin{enumerate}
	\item[(i)] \textit{Linearity} in the first argument, i.e.,  $\langle \sum_i\alpha_i v_1 |  v_2 \rangle  = \sum_i\alpha_i \langle v_1 |  v_2 \rangle  $\\
	(Conjugate linearity in the second argument, i.e., $\langle  v_1 |\sum_i\alpha_i  v_2 \rangle  = \sum_i\alpha_i^* \langle v_1 |  v_2 \rangle  $)
	\item[(ii)] \textit{Conjugate symmetry}: $\langle v_1| v_2 \rangle $ =  $\langle v_2| v_1 \rangle^* $
	\item[(iii)]  \textit{Positivity:} $\langle v | v \rangle \geq 0 $ and  $\langle v | v \rangle = 0$ if and only if $v = 0$
\end{enumerate}
where $|v\rangle,|v_1\rangle,|v_2\rangle \in V$; $\alpha_i \in F$, a field of complex numbers; and $*$ represents the complex conjuagte.\\
For example, the inner product of the vectors $|x\rangle = \begin{pmatrix}
	x_1 \\x_2
\end{pmatrix}\in \mathbb{C}^2$ and $|y\rangle = \begin{pmatrix}
y_1 \\ y_2
\end{pmatrix} \in \mathbb{C}^2$ is given as
\begin{eqnarray}
	\langle x | y \rangle  = \begin{pmatrix}
		x_1^* & x_2^* 
	\end{pmatrix} \begin{pmatrix}
	y_1 \\ y_2 \\
\end{pmatrix} = \sum_{i=1}^2 x_i^* y_i
\end{eqnarray}
{\definition (\textit{Inner product space}) A vector space equipped with an inner product is called an inner product space.}

\subsubsection{Normed Vector Space}
 Every inner product space is a\textit{ normed vector space} where norm of a vector $|v\rangle$ is defined as 
\begin{eqnarray}
	||v|| = \sqrt{ \langle v | v \rangle}
\end{eqnarray}
The vector $|v\rangle$ is called a \textit{unit} vector or \textit{normalized} vector if $||v|| =1$. \\ 
Let $V$ be a vector space over the field of complex numbers $\mathbb{C}$. The norm can be defined as a mapping $||.||: V \rightarrow \mathbb{R^+}$ that must satisfy the following properties. For any $v, w \in V$ and $\alpha \in \mathbb{C}$,
 \begin{enumerate}
 	\item[(i)]\textit{ Positivity:} $||v|| \geq 0$, $||v|| = 0$ if and only if $v = 0$
 	\item [(ii)] \textit{Scalar multiplication:} $||\alpha v|| = |\alpha| ||v||$
 	\item  [(iii)] \textit{Triangle inequality:} $||v + w|| \leq ||v|| + ||w||$
 \end{enumerate}

{\definition(\textit{Orthogonal vectors}) Two vectors $|v\rangle$ and $|w\rangle$ are orthogonal if their inner product is zero, i.e., $\langle v|w\rangle = 0$.
\definition(\textit{Orthonormal vectors}) A set of vectors $\{|v_1\rangle, |v_2\rangle, . . .\}$ is orthonormal if each vector in the set is a unit vector, and distinct vectors in the set
	are orthogonal, i.e.,
	\begin{eqnarray}
		\langle v_i|v_j\rangle = \delta_{ij} = \begin{cases}
			1 & i=j \\0 & i \neq j
		\end{cases}
	\end{eqnarray}
\remark There is a useful method, known as the Gram-Schmidt procedure, which can be used to produce an orthonormal basis set for any vector space $V$.}\\
We are now in a position to define Hilbert space and study its properties.
 {\definition (\textit{Hilbert space}) A complete inner product space is called Hilbert space.}\\
 The concept of Hilbert space was introduced by the German mathematician David Hilbert in an abstract sense but later, it played a significant role in quantum mechanics. 
 A Hilbert space $\mathcal{H}$ is complete if every cauchy sequence in $\mathcal{H}$ converges in $\mathcal{H}$. In a finite dimensional vector space, Hilbert space is same as the inner product space while infinite dimensional Hilbert space satisfy additional restrictions that are beyond the inner product spaces. In the present thesis, finite dimensional Hilbert spaces are considered, unless otherwise stated.\\

\subsubsection{Some properties of Hilbert spaces:}
\begin{enumerate}
	\item[(i)] \textit{Parallelogram law}: For any two vectors $|v\rangle$ and $|w\rangle$ in the Hilbert space $\mathcal{H}$, 
	\begin{eqnarray}
		||v + w||^2 + ||v - w||^2 = 2(||v||^2 + ||w||^2)
	\end{eqnarray}
\item[(ii)] If the vectors $|v\rangle$ and $|w\rangle$ in the Hilbert space $\mathcal{H}$ are orthogonal, then
\begin{eqnarray}
	||v + w||^2 = ||v||^2 + ||w||^2
\end{eqnarray}
\item[(iii)] Every orthonormal set in a Hilbert space is linearly independent.
\end{enumerate}
\subsubsection{Cauchy-Schwarz inequality} 
An important geometric fact about Hilbert spaces may be discussed through the \textit{Cauchy-Schwarz inequality}. It may be stated in the following way: For any two vectors $|v\rangle$ and $|w\rangle$ in the Hilbert space $\mathcal{H}$,
\begin{eqnarray}
	|\langle v|w\rangle|^2 \leq \langle v|v\rangle \langle w|w\rangle \equiv ||v|| .||w|| \label{cauchy-sch}
\end{eqnarray}
Equality in (\ref{cauchy-sch}) occurs if and only if $|v\rangle$ and $|w\rangle$ are linearly dependent, i.e., $|v\rangle = c |w\rangle$, for some scalar $c$.\\
Geometrically, the Cauchy-Schwarz Inequality tells us that the magnitude of the inner product of two vectors is bounded by the product of their magnitudes.

\subsubsection{Tensor Product}
Let $\mathcal{H}_A$ and $\mathcal{H}_B$ be two Hilbert spaces of dimension $m$ and $n$, respectively.  Then $\mathcal{H}_A \otimes \mathcal{H}_B$ is an $mn$ dimensional Hilbert space. The elements of $\mathcal{H}_A \otimes \mathcal{H}_B$ are linear combinations
of \textit{tensor products} of elements  of $\mathcal{H}_A$ and $\mathcal{H}_B$. For $|v_i\rangle \in \mathcal{H}_A$,  $|w_i\rangle \in \mathcal{H}_B$ and $a_i \in F$, we have $\sum_i a_i|v_i\rangle \otimes |w_i\rangle \in \mathcal{H}_A \otimes \mathcal{H}_B$. \\
If $|i_A\rangle$ and $|j_B\rangle$ are orthonormal basis for  $\mathcal{H}_A$ and $\mathcal{H}_B$, then  $|i_A\rangle \otimes |j_B\rangle$ is a basis for $\mathcal{H}_A \otimes \mathcal{H}_B$. The vector $|v\rangle \otimes |w\rangle$ can be abbreviated as $|vw\rangle$.  For example, if $\mathcal{H}$ is a two-dimensional space with basis vectors $|0\rangle = \begin{pmatrix}
	1 \\ 0
\end{pmatrix}$ and $|1\rangle = \begin{pmatrix}
	0\\ 1
\end{pmatrix}$, then $\{|00\rangle, |01\rangle, |10\rangle, |11\rangle\}$ forms the basis of $\mathcal{H} \otimes \mathcal{H}$, where the basis vectors can be described as follows.
\begin{eqnarray}
	|00\rangle = |0\rangle \otimes |0\rangle = \begin{pmatrix}
		1 \\ 0 
	\end{pmatrix} \otimes \begin{pmatrix}
		1 \\ 0
	\end{pmatrix} = \begin{pmatrix}
		1 \\ 0 \\0 \\ 0 
	\end{pmatrix};\;\;\quad
	|01\rangle = |0\rangle \otimes |1\rangle = \begin{pmatrix}
		1 \\ 0 
	\end{pmatrix} \otimes \begin{pmatrix}
		0 \\ 1
	\end{pmatrix} = \begin{pmatrix}
		0\\ 1 \\0 \\ 0 
	\end{pmatrix};\nonumber\\
	|10\rangle = |1\rangle \otimes |0\rangle = \begin{pmatrix}
		0 \\ 1 
	\end{pmatrix} \otimes \begin{pmatrix}
		1 \\ 0
	\end{pmatrix} = \begin{pmatrix}
		0 \\ 0 \\1 \\ 0 
	\end{pmatrix};\;\;\quad
	|11\rangle = |1\rangle \otimes |1\rangle = \begin{pmatrix}
		0\\ 1 
	\end{pmatrix} \otimes \begin{pmatrix}
		0 \\ 1
	\end{pmatrix} = \begin{pmatrix}
		0 \\ 0 \\0 \\ 1
	\end{pmatrix}
\end{eqnarray}

\textit{Properties of Tensor Product}: Let $V$ and $W$ be two vector spaces. Then the following properties hold for tensor product.
\begin{itemize}
	\item [(i)] For $|v\rangle \in V$,  $|w\rangle \in W$ and for any scalar $c$
	\begin{eqnarray}
		c(|v\rangle \otimes |w\rangle) = 	(c|v\rangle) \otimes |w\rangle = |v\rangle \otimes (c|w\rangle)
	\end{eqnarray}
	\item [(ii)]For $|v_1\rangle,|v_2\rangle \in V$ and $|w\rangle \in W$ 
	\begin{eqnarray}
		(|v_1\rangle + |v_2\rangle) \otimes |w\rangle = 	|v_1\rangle  \otimes |w\rangle +  |v_2\rangle \otimes |w\rangle
	\end{eqnarray}
	\item [(iii)]For $|v\rangle \in V$ and $|w_1\rangle,|w_2\rangle  \in W$ 
	\begin{eqnarray}
		|v\rangle \otimes	(|w_1\rangle + |w_2\rangle)  = 	|v\rangle  \otimes |w_1\rangle +  |v\rangle \otimes |w_2\rangle
	\end{eqnarray}
\end{itemize}

\subsection{Linear Operator}
\noindent Operators play a fundamental role in quantum mechanics since they provide the mathematical language required to extract information about the physical properties of a quantum system. The required mathematical language can be expressed in terms of the eigenvalues of the operator. In this subsection, we discuss linear operators, and bounded operators and also describe the matrix representation of operators.
{\definition \textit{(Linear Operator):} A linear operator between two vector spaces $V$ and $W$ is defined as a function $T: V \longrightarrow W$ that preserves linearity, i.e.,
\begin{eqnarray}
	T\left(\sum_i a_i |v_i\rangle\right) = \sum_i a_i T(|v_i \rangle)
\end{eqnarray}}
Consider two vectors $|v\rangle$ and $|w\rangle$ in the two-dimensional vector space $\mathbb{C}^2$ as
\begin{eqnarray}
	|v\rangle = \begin{pmatrix}
		x_1 \\ x_2
	\end{pmatrix} \quad \text{and} \quad
	|w\rangle = \begin{pmatrix}
	y_1 \\ y_2
\end{pmatrix}
\end{eqnarray}
such that $T|v\rangle = |w\rangle$. This means that the operator $T$ when applied to the vector $|v\rangle$ changes it to another vector $|w\rangle$ of the same vector space. Physically that means, a physical system can be changed to another by the application of an operator in the Hilbert space of the system. Identity operator $I$ defined by $I|v\rangle = |v\rangle$ and zero operator $O$ defined by $O|v\rangle = 0$ are two important linear operators. The zero operator maps all the vectors of the vector space to the zero vector.\\
\textit{Composition of linear operators:} Let $T_1: U \longrightarrow V$ and $T_2: V \longrightarrow W$ be two linear operators. Then the composition of $T_2$ with $T_1$ is defined as $T_2 T_1 (|v\rangle) = T_2 ( T_1 |v\rangle) $ for all $|v\rangle \in U$.

\subsubsection{Matrix Representation of an Operator} 
Consider a linear operator $T:V \longrightarrow W$ between two vector spaces $V$ and $W$. Suppose $\{|v_1\rangle, |v_2\rangle, . . , |v_m\rangle\}$ and $\{|w_1\rangle, |w_2\rangle, . . , |w_n\rangle\}$ form basis for $V$ and $W$, respectively. Then for each $|v_j\rangle$ where $j = 1$ to $m$, there exist complex numbers $a_{1j}, a_{2j}, . . ., a_{nj}$ such that
\begin{eqnarray}
	T(|v_j\rangle) = \sum_i a_{ij}  |w_j\rangle
\end{eqnarray}
The matrix whose entries are the values $a_{ij}$ is called the matrix representation of the operator $T$. Hence, we use the terms "operator" and "matrix" interchangeably throughtout this thesis.
\subsubsection{Hilbert Schmidt inner product on operators}
The set $L(\mathcal{H})$ of all linear operators on a Hilbert space $\mathcal{H}$ forms a vector space. If $A, B \in L(\mathcal{H})$, then the inner product on $L(\mathcal{H}) \times L(\mathcal{H})$ may be defined by 
\begin{eqnarray}
	\langle A | B \rangle = Tr[A^\dagger B] \label{HS-ip}
\end{eqnarray}
The inner product defined in (\ref{HS-ip}) is known as \textit{Hilbert-Schmidt} or \textit{trace} inner product.
\subsubsection{Bounded operator} 
Let $\mathcal{H}_1$ and $\mathcal{H}_2$ be two Hilbert spaces. An operator $T: \mathcal{H}_1 \longrightarrow \mathcal{H}_2$ is bounded if it maps bounded subsets of $\mathcal{H}_1$ to bounded subsets of $\mathcal{H}_2$. Any linear operator defined on a finite dimensional normed linear space is bounded. 
\subsubsection{Tensor product of linear operators}
 Let $A$ and $B$ be linear operators on $V$ and $W$, respectively. Then $A \otimes B$ is a well defined linear operator on $V \otimes W$.  For $|v_i\rangle \in V$,  $|w_i\rangle \in W$,
\begin{eqnarray}
	(A \otimes B)(\sum_i a_i|v_i\rangle \otimes |w_i\rangle) = \sum_i a_i A  |v_i\rangle \otimes B |w_i\rangle
\end{eqnarray}
 Furthermore, for any four operators $A$, $B$, $C$, and $D$, we have
\begin{eqnarray}
	(A \otimes B) (C \otimes D) = AC \otimes BD
\end{eqnarray}

\subsection{Basic Concepts of Linear Algebra}
\noindent Linear algebra is the study of linear operations on vector spaces. A good understanding of quantum mechanics will be possible if one has a solid grasp of elementary linear algebra. 
Let us review some basic concepts from linear algebra which are used in the study of quantum mechanics.
\subsubsection{Matrices}
 A matrix is a rectangular array of elements arranged in rows and columns.
An $m\times n$ matrix with entries as complex numbers can be expressed as $A = [a_{ij}] \in M_{m,n}(\mathbb{C})$ where $1\leq i\leq m$ and $1\leq j\leq n$. Here $m$ denotes the number of rows and $n$ denotes the number of columns. We now discuss a few elementary operations defined on matrices.  
\begin{itemize}
	\item [(i)] \textbf{Trace:} Trace of a matrix is defined as the sum of its diagonal elements. If $A$ is a square matrix of order $n$, i.e., $A = [a_{ij}] \in M_n(\mathbb{C})$, then trace of $A$ is given as $Tr[A] = \sum_{i=1}^n a_{ii}$.
Some useful properties of trace are given below:
\begin{enumerate}
		\item[(a)] \textit{Linearity property:} If $A$ and $B$ are two linear operators,
	\begin{eqnarray}
		Tr[A+B] = Tr[A] + Tr[B] \;\; \text{and} \;\; Tr[cA] = cTr[A] \;\; \text{for any scalar $c$}
	\end{eqnarray}
	\item[(b)] \textit{Cyclic property:} If $A$ and $B$ are two linear operators then
	\begin{eqnarray}
		Tr[AB] = Tr[BA]
	\end{eqnarray}
	\item[(c)] \textit{Invariance under unitary transformation:} For any matrix $A$ and unitary transformation $U$, we have
	\begin{eqnarray}
		Tr[U A U^{\dagger}] = Tr[A]
	\end{eqnarray} 
	\item[(d)] \textit{Expectation value of the operator:} For any operator $A$ and the unit vector $|u\rangle$, we have
	\begin{eqnarray}
		Tr[A |u \rangle \langle u|] = \sum_i \langle i |A |u \rangle \langle u|i\rangle = \langle u|A|u\rangle
	\end{eqnarray}
	where $|i\rangle$ denotes the orthonormal basis.
	\item [(e)] \textit{Trace of tensor product:} For any two operators $A$ and $B$, we have
	\begin{eqnarray}
		Tr[A \otimes B] = Tr[A] Tr[B]
	\end{eqnarray}
\end{enumerate}
\item [(ii)] \textbf{Transpose of a matrix:} Transpose of a matrix is obtained after interchanging rows and columns of the matrix. If $A = [a_{ij}] \in M_{m,n}(\mathbb{C})$, the transpose of $A$, denoted by $A^T \in M_{n,m}(\mathbb{C})$ whose $(i, j)^{th}$ entry is $a_{ji}$, i.e., $A^T = [a_{ji}]$.
\item[(iii)] \textbf{Adjoint of a matrix:} Adjoint or the conjugate transpose of a matrix $A \in M_{m,n}(\mathbb{C})$, is denoted by $A^\dagger  \in M_{n,m}(\mathbb{C})$ and it is defined as $A^\dagger = [a_{ji}^*]$.
	\item[(iv)] \textbf{Rank of a matrix:} Rank of a matrix $A \in M_{m,n}(\mathbb{C})$ is defined as the number of linearly independent rows (or columns) of $A$.
\end{itemize}

\subsubsection{Eigenvalues and Eigenvectors} Let $A \in M_n (\mathbb{C})$. A scalar $\lambda \in \mathbb{C}$ such that $A|v\rangle = \lambda |v\rangle$ for some non-zero vector $|v\rangle \in \mathbb{C}^n$, is an \textit{eigenvalue} of $A$. The non-zero vector $|v\rangle$ corresponding to the eigenvalue $\lambda$ is referred to as an \textit{eigenvector} of the matrix $A$.
{\definition\textit{Spectrum of a matrix:} The set of all eigenvalues of the matrix $A$ is called the spectrum of $A$ and is denoted by $\sigma(A)$.
\definition \textit{Determinant:} Determinant of a matrix is defined as the product of its eigenvalues. For any matrix $A \in M_n (\mathbb{C})$ with eigenvalues $\lambda_{1}, \lambda_2, . . .,\lambda_n$, its determinant is given as $det(A)=\prod_{i=1}^n \lambda_i$
\definition\textit{Singular matrix:} If atleast one eigenvalue of a matrix $A$ is zero, then $A$ is said to be a singular matrix. Otherwise the matrix $A$ is said to be non-singular.
\remark Trace of a matrix can also be defined as the sum of its eigenvalues.}
\subsubsection{Characteristic Equation}
The characteristic polynomial of a matrix $A \in M_n (\mathbb{C})$ is defined as
\begin{eqnarray}
p_n (\lambda) =	det(A - \lambda I) = \lambda^n - Tr[A] \lambda^{n-1} + . . . (-1)^n det(A) \label{charpoly}
\end{eqnarray}
 The equation $p_n (\lambda) = 0$ is referred to as the characteristic equation of $A$. The solutions of the characteristic equation give the eigenvalues of $A$. The set of vectors corresponding to the eigenvalue $\lambda_i$ ($i=1, 2,....,n$) forms an \textit{eigenspace}, which is also a subspace of the vector space on which $A$ acts.

\subsubsection{Diagonalization}
Any matrix $A \in M_n(\mathbb{C})$ is said to be diagonalizable if it is similar to a diagonal matrix, i.e, if there exists an invertible matrix 
$P$ and a diagonal matrix 
$D$ such that $P^{-1}AP=D$ (or equivalently
$A=PDP^{-1}$).
The eigenvectors corresponding to the eigenvalues of the matrix $A$ forms the column vectors of the matrix $P$.
{\theorem $A \in M_n(\mathbb{C})$ is diagonalizable if it has $n$ distinct eigenvalues.}
{\remark Having distinct eigenvalues is sufficient for diagonalizability, but of course, it is not
	necessary. A square matrix with non-distinct eigenvalues may or may not be diagonalizable}.
 \subsubsection{Singular value decomposition}
  Let $A \in M_{m,n}{(\mathbb{C})}$. Then there exist unitary matrices $U \in M_m(\mathbb{C})$, $V \in M_n(\mathbb{C})$ and real numbers $\sigma_1 \geq \sigma_{2} \geq \sigma_{3} \geq . . . \geq \sigma_p$ (where $p = min\{m,n\}$), such that
\begin{eqnarray}
A = U \Sigma V	 \label{svd}
\end{eqnarray}
where $\Sigma =[\Sigma]_{ij} \in M_{m,n}{(\mathbb{C})}$ has entries $\sigma_j$ for $i=j$, and zero for $i \neq j$. \\
The numbers $\sigma_1, \sigma_2, . . ., \sigma_p$ are called the \textit{singular values} of $A$. 
Alternatively,  the singular value of $A$ can also be defined as the square root of the eigenvalue of the positive semi definite matrix $A^\dagger A$ (or equivalently of $A A^\dagger$).
If $\lambda_i$ $i=1,2,...,n$ denote the eigenvalue of $A^\dagger A$ (or $AA^\dagger$), then the singular value corresponding to the eigenvalue $\lambda_i$ is given by
\begin{eqnarray}
	\sigma_i = \sqrt{\lambda_i}, \;\; i = 1, 2,....,n
\end{eqnarray}
{\remark Let $A$ be a $m \times n$ matrix with complex entries. Then $rank(A)$ is equal to the number of non-zero singular values of $A$. }
\subsubsection{A few types of operators and their properties}
\noindent There are many operators in the literature of operator theory that may play a vital role in quantum mechanics. In this thesis, we present some of these, such as the Hermitian operator, unitary operator, positive and positive semi-definite operator.
\begin{enumerate}
	\item[(i)] \textbf{Hermitian operator:}
	An operator $A$ is said to be Hermitian or self-adjoint if it satisfies: $A^\dagger=A$. The eigenvalues of a Hermitian matrix are real. Thus, every observable in quantum mechanics is represented by a Hermitian matrix.
	\item[(ii)]	\textbf{Unitary operator:}
	An operator $U$ is said to be unitary if $U^\dagger U = UU^\dagger= I$. 
	Unitary operators are important
	because they preserve inner products between vectors, i.e., for any two vectors $|v\rangle$ and $|w\rangle$ in a vector space $V$,
	\begin{eqnarray}
		(U|v\rangle, U|w\rangle) = \langle v |U^\dagger U| w\rangle = \langle v |I| w\rangle =\langle v | w\rangle
	\end{eqnarray}
	All eigenvalues of a unitary matrix have modulus one, i.e., eigenvalues of a unitary matrix 
	can be written in the form $e^{i \theta}$ for some real $\theta$. Unitary matrices are important because they take part in the time evolution of the quantum state.
	{\remark If $\{|v_i\rangle\}$ and $\{|w_i\rangle\}$ for $i=1,2,...n$ are any two orthonormal bases of a vector space $V$, then the operator $U$ defined by $U= \sum_{i=1}^n |v_i\rangle \langle w_i|$ is a unitary operator.}
\item[(iii)]	\textbf{Positive operator:}
	An operator $A$ is defined to be a positive operator if for any vector $|v\rangle \in V$, $\langle v| A|v\rangle$ is a real non-negative number.
	$A$ is \textit{positive definite} if  $\langle v| A|v\rangle$ is strictly greater than zero for all $|v\rangle \neq 0$. Equivalently, a matrix $A$ is positive definite if its all eigenvalues are strictly positive.\\
	 The matrix $A$ is \textit{positive semi-definite} if its eigenvalues are non-negative. A quantum state is represented by the positive semi-definite operator. 
\end{enumerate}

\subsubsection{Special Matrices}
We now discuss about some particular matrices such as Pauli matrices, Gell-Mann matrices and the generalized Gell-Mann matrices, which may be useful in quantum information theory.
\begin{enumerate}
\item[(i)] \textbf{Pauli Matrices:}
Pauli matrices are useful in the study of quantum computation and quantum information. They are defined as follows.
\begin{eqnarray}
	\sigma_0  = I = \begin{pmatrix}
		1 &0 \\
		0 & 1 
	\end{pmatrix};\;\;\;\;\quad
 \sigma_x = \begin{pmatrix}
		0 &1 \nonumber\\
		1 & 0 
	\end{pmatrix}\\ \sigma_y \begin{pmatrix}
		0 & -i \\
		i & 0 
	\end{pmatrix};\;\;\quad
	 \sigma_z = \begin{pmatrix}
		1 &0 \\
		0 & -1 
	\end{pmatrix} \label{pauli}
\end{eqnarray}
Properties of Pauli matrices:
\begin{enumerate}
	\item[(a)] Pauli matrices are Hermitian and unitary.
	\item[(b)]  $Tr[\sigma_x] = Tr[\sigma_y] = Tr[\sigma_z]=0$
	\item[(c)] $\sigma_x^2 = \sigma_y^2 =\sigma_z^2 =I$
	\item [(d)] $\sigma_x \sigma_y = i \sigma_z;\; \sigma_y \sigma_z = i \sigma_x;\; \sigma_y \sigma_x = -i \sigma_z$
	\item[(e)] The Pauli matrices $\sigma_x, \sigma_y, \sigma_z$ form a basis of the two-dimensional space $M_2(\mathbb{C})$.
\end{enumerate}
\item[(ii)] \textbf{Gell-Mann Matrices:} \label{GMmatrices}
The eight Gell-Mann (GM) matrices in a three-dimensional space are defined as follows.
\begin{enumerate}
	\item[(a)] Three symmetric GM matrices
	\begin{eqnarray} \Lambda_1=
	\begin{pmatrix}
		0 & 1 & 0 \\
		1 & 0 & 0 \\
		0 & 0 & 0
	\end{pmatrix};\quad \Lambda_2 = \begin{pmatrix}
	0 & 0 & 1\\
	0 & 0 & 0\\
	1 & 0 & 0
\end{pmatrix};\quad  \Lambda_3 =  \begin{pmatrix}
0 & 0 & 0 \\
0 & 0 & 1 \\
0 & 1 & 0
\end{pmatrix}
\end{eqnarray}
	\item [(b)] Three anti-symmetric GM matrices
	\begin{eqnarray}  \Lambda_4 = 
		\begin{pmatrix}
			0 & -i & 0 \\
			i & 0 & 0 \\
			0 & 0 & 0
		\end{pmatrix};\quad  \Lambda_5 = \begin{pmatrix}
			0 & 0 & -i\\
			0 & 0 & 0\\
			i & 0 & 0
		\end{pmatrix};\quad  \Lambda_6 =  \begin{pmatrix}
			0 & 0 & 0 \\
			0 & 0 & -i \\
			0 & i & 0
		\end{pmatrix}
	\end{eqnarray}
	\item [(c)] Two diagonal GM matrices
	\begin{eqnarray} \Lambda_7 = 
		\begin{pmatrix}
			1 & 0 & 0 \\
			0 & -1 & 0 \\
			0 & 0 & 0
		\end{pmatrix};\quad  \Lambda_8 =  \frac{1}{\sqrt{3}}\begin{pmatrix}
			1 & 0 & 0\\
			0 & 1 & 0\\
			0 & 0 & -2
		\end{pmatrix}
	\end{eqnarray}
\end{enumerate}
GM matrices possess the property Hermitianity and unitarity. The GM matrices $\{\Lambda_i,\; i=1,2,...,8\}$ satisfies $Tr[\Lambda_i] = 0$ and $Tr[\Lambda_i \Lambda_j] = 2 \delta_{ij}$ where $i=1,2,...,8$.
\item[(iii)] \textbf{Generalized Gell-Mann matrices (GGM):} 
The generalized Gell-Mann matrices (GGM) are higher-dimensional extensions of the Pauli matrices (for qubits) and the Gell-Mann matrices (for qutrits). 
There are total $d^2 - 1$ GGM for a $d$-dimensional space.
All GGM are Hermitian and traceless. They are orthogonal and form a basis known as the generalized Gell-Mann matrix
Basis (GGB) \cite{bert2008}. The eight generalized Gell-Mann (GGM) matrices in a four-dimensional space are defined as follows.
\begin{enumerate}
	\item[(a)] Six symmetric GGM matrices
	\begin{eqnarray}&&G_1=
		\begin{pmatrix}
			0 & 1 & 0 &0\\
			1 & 0 & 0 & 0\\
			0 & 0 & 0 & 0\\
			0 & 0 & 0 & 0
		\end{pmatrix};\quad G_2=\begin{pmatrix}
			0 & 0 & 1 & 0\\
			0 & 0 & 0 & 0\\
			1 & 0 & 0 & 0 \\
				0 & 0 & 0 & 0
		\end{pmatrix};\quad G_3=\begin{pmatrix}
			0 & 0 & 0 & 1\\
			0 & 0 & 0 & 0\\
			0 & 0 & 0 & 0\\
				1 & 0 & 0 & 0\\
		\end{pmatrix}\nonumber\\
&&G_4=
	\begin{pmatrix}
		0 & 0 & 0 & 0\\
		0 & 0 & 1 & 0\\
		0 & 1 & 0 & 0\\
		0 & 0 & 0 & 0\\
	\end{pmatrix};\quad
	G_5=
\begin{pmatrix}
	0 & 0 & 0 & 0\\
	0 & 0 & 0 & 1\\
	0 & 0 & 0 & 0\\
	0 & 1 & 0 & 0\\
\end{pmatrix};\quad	
	G_6=
\begin{pmatrix}
	0 & 0 & 0 & 0\\
	0 & 0 & 0 & 0\\
	0 & 0 & 0 & 1\\
	0 & 0 & 1 & 0\\
\end{pmatrix} \nonumber
\end{eqnarray}

	\item [(b)] Six anti-symmetric GGM matrices
	\begin{eqnarray}G_7=
		\begin{pmatrix}
			0 & -i & 0 &0\\
			i & 0 & 0 & 0\\
			0 & 0 & 0 & 0\\
			0 & 0 & 0 & 0
		\end{pmatrix};\quad G_8=\begin{pmatrix}
			0 & 0 & -i & 0\\
			0 & 0 & 0 & 0\\
			i & 0 & 0 & 0 \\
			0 & 0 & 0 & 0
		\end{pmatrix};\quad G_9=\begin{pmatrix}
			0 & 0 & 0 & -i\\
			0 & 0 & 0 & 0\\
			0 & 0 & 0 & 0\\
			i & 0 & 0 & 0\\
		\end{pmatrix}&&\nonumber\\
		G_{10}=
		\begin{pmatrix}
			0 & 0 & 0 & 0\\
			0 & 0 & -i & 0\\
			0 & i & 0 & 0\\
			0 & 0 & 0 & 0\\
		\end{pmatrix};\quad
		G_{11}=
		\begin{pmatrix}
			0 & 0 & 0 & 0\\
			0 & 0 & 0 & -i\\
			0 & 0 & 0 & 0\\
			0 & i & 0 & 0\\
		\end{pmatrix};\quad	
		G_{12}=
		\begin{pmatrix}
			0 & 0 & 0 & 0\\
			0 & 0 & 0 & 0\\
			0 & 0 & 0 & -i\\
			0 & 0 & i & 0\\
		\end{pmatrix}&& \nonumber
	\end{eqnarray}
	\item [(c)] Three diagonal GGM matrices
	\begin{eqnarray} G_{13}=
		\begin{pmatrix}
			1 & 0 & 0 & 0\\
			0 & -1 & 0 & 0\\
			0 & 0 & 0 & 0 \\
			0 & 0 & 0 & 0
		\end{pmatrix}\quad G_{14}=\frac{1}{\sqrt{3}}\begin{pmatrix}
			1 & 0 & 0&0\\
			0 & 1 & 0 &0\\
			0 & 0 & -2 &0\\
			0 & 0 & 0 & 0
		\end{pmatrix}\quad G_{15}=\frac{1}{\sqrt{6}}\begin{pmatrix}
		1 & 0 & 0&0\\
		0 & 1 & 0 &0\\
		0 & 0 & 1 &0\\
		0 & 0 & 0 & -3
	\end{pmatrix} \nonumber
	\end{eqnarray}
\end{enumerate}
\end{enumerate}

\subsubsection{Partial Transpose} \label{pptnptmatrix}
Let $A=[A_{i,j}]_{i,j=1}^{n} \in M_n(M_k (\mathbb{C}))$ be an $n \times n$ matrix written in the block matrix form where each block $A_{i,j}$ represents a $k \times k$ matrix with complex entries. The partial transpose of the matrix $A$, denoted by $A^{\tau}$, is given by
\begin{eqnarray}
	A^{\tau}=[{A^{T}_{i,j}}]_{i,j=1}^{n} \label{atau}
\end{eqnarray}
where $T$ is the usual matrix transpose.\\
\textbf{PPT matrix:} A matrix $A \in M_n(\mathbb{C})$ has positive partial transpose (PPT), if its partial transposed matrix $A^\tau$ has no negative eigenvalues, i.e., it forms a positive semidefinite matrix. Such a matrix is called positive partial transpose (PPT) matrix. \\
\textbf{NPT matrix:} If the partial transposed matrix $A^\tau$ has atleast one negative eigenvalue, then it is called a negative partial transpose (NPT) matrix.
\begin{example*}
	Consider the matrix $\rho_f \in M_4(\mathbb{C})$, $0 \leq f \leq 1$,  and its partially transposed matrix $\rho_f^{\tau}$
\begin{eqnarray}
	\rho_f = \begin{pmatrix}
		\frac{1+2f}{6} & 0 & 0 & \frac{4f - 1}{6}\\
		0& \frac{1-f}{3} & 0 & 0\\
		0 & 0 & \frac{1-f}{3} & 0\\
		 \frac{4f - 1}{6} &0 & 0 & \frac{1+2f}{6}
	\end{pmatrix};\quad 
\rho_f^{\tau} = \begin{pmatrix}
	\frac{1+2f}{6} & 0 & 0 & 0\\
	0& \frac{1-f}{3} &  \frac{4f - 1}{6}& 0\\
	0 & \frac{4f - 1}{6} & \frac{1-f}{3} & 0\\
	0 &0 & 0 & \frac{1+2f}{6}
\end{pmatrix}
\end{eqnarray}
The minimum eigenvalue of  $\rho_f^{\tau}$ given by
$\lambda_{min} (\rho_f^\tau) = \frac{1-2f}{2} > 0$ when $0 < f \leq \frac{1}{2}$.
Hence, $\rho_f$ is a PPT matrix when $0 < f \leq \frac{1}{2}$ and it represents an NPT matrix when $\frac{1}{2} < f \leq 1$.
\end{example*}
\subsection{A few important results in Linear Algebra} \label{sec-laresults}
We discuss here few important definitions and results from linear algebra that will be needed in the later chapters of the present thesis.
\subsubsection{A few definitions}
{\definition \label{comm} (\textit{Commutator:}) The commutator between two operators $A$ and $B$ is defined as
\begin{eqnarray}
	[A,B] = AB - BA \label{commutator}
\end{eqnarray}
If $AB = BA$, i.e., $[A,B] =0$, we say $A$ commutes with $B$. Two operators $A$ and $B$ are \textit{compatible} if they commute, i.e., $[A,B] =0$ . If  $[A,B] \neq 0$, then the operators are \textit{incompatible} \cite{niel, snbiswas}.
\definition (\textit{Anti-commutator:}) The anti-commutator of two operators $A$ and $B$ is defined as
\begin{eqnarray}
	\{A,B\} = AB + BA
\end{eqnarray}
If $AB = - BA$, i.e., $\{A,B\}=0$, we say $A$ anti-commutes with $B$ \cite{niel}. 
\definition \label{posmap} (\textit{Positive map:}) A linear map $\phi : A\rightarrow B$ is called positive map if it maps positive elements of $A$ to positive elements of $B$ \cite{stormer}.
\definition (\textit{$k$-positive map:}) Let $\phi : A\rightarrow B$ be a linear map, and let $k \in \mathbb{N}$, the set of natural numbers. Then $\phi$ is \textit{k}-positive if $\phi \otimes i_k : A \otimes M_k \rightarrow B \otimes M_k$ is positive, $i_k$ denotes the
identity map on $M_k$.
\definition \label{cp map} (\textit{Completely Positive map:}) A linear map $\phi$ is said to be completely positive if $\phi$ is $k$-positive $\forall$ $k \in \mathbb{N}$ \cite{stormer}.
\definition \label{choimatrix}(\textit{Choi matrix:}) Let $\Phi: M_{d_1}(\mathbb{C}) \rightarrow M_{d_2}(\mathbb{C})$ be a linear map. The choi matrix for $\Phi$ is defined by the operator
\begin{equation}
	C_{\Phi} = \sum_{i,j}^{d_1,d_2} e_{ij} \otimes \Phi(e_{ij}) \label{choi}
\end{equation}
where $\{e_{ij}\}_{i,j=1}^{n}$ denotes the matrix units \cite{stormer}.
\definition \label{cjiso}(\textit{Choi-Jamiolkowski isomorphism:})
The correspondence between the map $\Phi$ and the choi matrix $C_{\Phi}$ given by the linear, one-one and onto map $\Phi\rightarrow C_{\Phi}$, which is called Choi-Jamiolkowski isomorphism \cite{jamio}.
\definition \label{tracenorm}(\textit{Trace norm:}) The trace norm $||A||_1$ of an $m \times n$ matrix $A$ is defined as the  sum of singular values of $A$ \cite{horn}.
\begin{eqnarray}
	||A||_1 = Tr[\sqrt{A^\dagger A}] = \sum_{i=1}^{n} \sigma_{i}(A)
	\label{trnorm}
\end{eqnarray}
where $\sigma_{i}(A),i=1,2,...n$ denotes the singular values of $A$.
\definition \label{frobnorm} (\textit{Frobenius norm:}) The norm $||A||_2$ is called Hilbert-Schmidt norm or Frobenius norm of $A$ and is equal to the sum of the squares of the singular values of $A$ \cite{horn}.
\begin{eqnarray}
	||A||_2^{2} = Tr[A^\dagger A] = \sum_{i=1}^{n} \sigma_{i}^{2}(A)
	\label{frob}
\end{eqnarray}}
\subsubsection{A few results}
{\result \label{res-amgmineq}  For any list of $n$ non-negative real numbers $x_1, x_2, . . . , x_n$, the arithmetic mean of $x_i$'s is greater than or equal to their geometric mean, i.e.,
	\begin{eqnarray}
		\frac{1}{n}{\sum_{i=1}^n x_i} \geq (\prod_{i=1}^n x_i)^{1/n} \label{eq-amgm}
	\end{eqnarray}
The inequality (\ref{eq-amgm}) is called AM-GM inequality where equality holds if and only if $x_1 = x_2 = ... x_n$.
\result \label{res-trnorm} For any matrix $A \in M_n(\mathbb{C})$, we have \cite{zou}
\begin{eqnarray}
	|Tr[A]| \leq ||A||_1 
	\label{modnorm1}
\end{eqnarray}
\result \label{res-weyl} (\textit{Weyl's Inequality \cite{horn}:}) Let $A$ and $B$ be Hermitian matrices in $M_n(\mathbb{C})$. Then the following inequality holds
\begin{eqnarray}
	\lambda_{min}(A) + \lambda_{min}(B) \leq \lambda_{min}(A+B) \label{weyl}
\end{eqnarray}
where $\lambda_{min}(.)$ denotes the minimum eigenvalue of the corresponding matrix.
\result \label{res-lmintr} For any two $n \times n$ Hermitian matrices $A$ and $B$, the following inequality holds \cite{jb}
\begin{eqnarray}
	\lambda_{min}(A) Tr[B] \leq Tr[AB] \leq \lambda_{max}(A) Tr[B] \label{jb}
\end{eqnarray}
\result \label{res-n1n2} For any matrix $A \in M_n(\mathbb{C})$ with rank $k$, we have \cite{zou}
\begin{eqnarray}
	||A||_1 \leq \sqrt{k} ||A||_2  \label{normineq}
\end{eqnarray}
{\result \label{res-yang} For any two positive semi-definite matrices $A$ and $B$ in $M_n(\mathbb{C})$, the following matrix inequality holds \cite{yang},
\begin{eqnarray}
	(Tr[AB]^q) \leq (Tr[A])^q (Tr[B])^q
	\label{ineq2}
\end{eqnarray} 
where $q$ is a positive integer.}
\result Let $A \in M_n(\mathbb{C})$ be any matrix with real eigenvalues and $\lambda_{min}^{lb}[A]$ denotes the lower bound of the minimum eigenvalue of $A$. Then \cite{wolko}
\begin{eqnarray}
	\lambda_{min}^{lb}[A] \leq \lambda_{min}[A] \
\end{eqnarray}
where the lower bound is given by 
\begin{eqnarray}
	\lambda_{min}^{lb}[A] =	\frac{Tr[A]}{n} - \sqrt{(n-1)\left(\frac{Tr[A^2]}{n}- (\frac{Tr[A]}{n})^2\right)} \label{lb}
\end{eqnarray}
\result Let $A$ be a complex $n \times n$ matrix with real eigenvalues $\lambda_i$ such that $\lambda_1 \geq \lambda_2 . \; . \; . \geq \lambda_n$. If $T_k$ denotes the $k^{th}$ order moment of $A$, i.e., $T_k = Tr[A^k]$ then, one 
has \cite{gupta} 
\begin{eqnarray}
	\lambda_1 \geq f(T_1, T_2, T_3) := \frac{T_1}{n} + \frac{b+ \sqrt{b^2 + 4 a^3}}{2a} \label{maxlb}
\end{eqnarray}
where
$$ a= \frac{T_2}{n} - \left(\frac{T_1}{n}\right)^2\quad \text{and}\quad
 b= \frac{1}{n^3} (n^2 T_3 -3 n T_1 T_2 + 2 T_1^3) $$.

\result  Let $A \in M_n(\mathbb{C})$ be a positive semi-definite matrix with eigenvalues $\lambda_1 \geq \lambda_2 . \; . \; . \geq \lambda_n$ and $T_k = Tr[A^k]$. Then the upper bound of the largest eigenvalue of $A$ is given as
\begin{eqnarray}
	\lambda_1 \leq g(T_1, T_2, T_3) \label{gtb}
\end{eqnarray}
where the upper bound $g(T_1, T_2, T_3)$ is the largest root of the following cubic equation \cite{sharma}:
\begin{eqnarray}
	T_1 x^3 -2 T_2 x^2 + T_3 x + T_2^2 - T_1 T_3=0  \label{cubiceq}
\end{eqnarray}
The explicit form of $g(T_1,T_2,T_3)$ is given by
\begin{eqnarray}
	g(T_1,T_2,T_3):=
	\frac{1}{6T_1} (4 T_2 + \frac{2\times 2^{1/3} r}{(p+\sqrt{q})^{1/3}} + 2^{2/3} (p+\sqrt{q})^{1/3}) \nonumber\\ \label{gt}
\end{eqnarray}
where 	$p = -27 T_1^2 T_2^2 +16 T_2^3 + 27 T_1^3 T_3 - 18 T_1 T_2 T_3;\;\; r=  4T_2^2 -3 T_1 T_3;\;\; q= p^2 - 4 r^3$

\result \label{res-choicp} Let $\Phi: M_{d_1}(\mathbb{C}) \rightarrow M_{d_2}(\mathbb{C})$ be a linear map and $C_{\Phi}$ be the choi matrix of $\Phi$.  Then the following conditions are equivalent  \cite{stormer}.
\begin{itemize}
	\item[1.] $C_{\Phi}$ is positive semi-definite.
	\item[2.] $\Phi$ is completely positive.
	\item[3.] $\Phi(A)= \sum_{i=1}^{k} V^{*}_i A V_i$ with $V_i:M_{d_2}(\mathbb{C}) \rightarrow M_{d_1}(\mathbb{C})$, a linear map and $k \leq min\{d_1, d_2\}$
\end{itemize} \label{thmchoi}

\result \label{rescp} Consider a map $\Phi: M_n(\mathbb{C}) \longrightarrow M_m(\mathbb{C})$. Let $A \in M_n$ and $B \in M_m$ be Hermitian matrices such that $\Phi(A)=B$. Then the map $\Phi$ is completely positive iff there exist non-negative real numbers $\gamma_1$ and $\gamma_2$ such that the following conditions hold \cite{poon}:
\begin{eqnarray}
	\lambda_{min}[B]&\geq& \gamma_1 \lambda_{min}[A]\\
	\lambda_{max}[B]&\leq& \gamma_2 \lambda_{max}[A]
\end{eqnarray}}

\section{Quantum Mechanics}\label{QM}
 Quantum mechanics is the the complete description of the world known to us. It is also the basis for an understanding of quantum information theory \cite{niel,mathews,srinivas,snbiswas}.
 This section provides a necessary background knowledge of quantum mechanics.
\subsection{Postulates of quantum mechanics}
We now discuss the fundamental postulates of qauntum meachanics. These postulates provide a connection between the physical
world and the mathematical formalism of quantum mechanics \cite{niel,mathews,srinivas}.\\
\textbf{Postulate 1.} (State space)
Any isolated physical system is associated with a Hilbert space known as the \textit{state space} of the system. The state of the system is completely described by a vector called the
\textit{state vector } belonging to a Hilbert space.\\
The simplest quantum mechanical system is called \textbf{qubit}. A qubit has two-dimensional state space. Let $\{|0\rangle, |1\rangle\}$ be an orthonormal basis for that state space. Then, an arbitrary state of a qubit is expressed as 
\begin{eqnarray}
|\phi\rangle =	a |0\rangle + b|1\rangle
\end{eqnarray}
where $a$ and $b$ are complex numbers with the normalization condition $|a|^2 +|b|^2 = 1$.\\
\noindent\textbf{Postulate 2.} (Evolution)
The evolution of a \textit{closed} quantum system is described by a unitary transformation. The state $|\psi_1(t_1)\rangle$ of a system at time $t_1$ is related to the state  $|\psi_2(t_2)\rangle$  of the system at time $t_2$ by a unitary operator $U(t_1, t_2)$ which depends only on the times $t_1$ and $t_2$.
\begin{eqnarray}
	|\psi_1(t_1)\rangle = U(t_1,t_2)|\psi_2(t_2)\rangle
\end{eqnarray}
Closed quantum systems do not suffer any unwanted interactions with the outside world. Although fascinating conclusions can be drawn about the information processing
tasks which may be accomplished in principle in such ideal systems, there are no perfectly closed systems in the real world. Real systems suffer from unwanted interactions with the
outside world. These undesirable interactions unfold as \textbf{noise} in quantum information processing systems.\\ 
\noindent\textbf{Postulate 3.} (Measurement)
Quantum measurements are described by a collection of
measurement operators denoted by $\{M_i\}$ with $\sum_i M_i^{\dagger} M_i = I$. These are operators acting on the state space of the system being measured. The index $i$ refers to the measurement outcomes that
may occur in the experiment. If the state of the quantum system is $|\psi\rangle$ immediately before the measurement then the probability that the outcome $i$  occurs is given by
\begin{eqnarray}
	p_i = \langle \psi |M_i^{\dagger}M_i| \psi \rangle
\end{eqnarray}
and the post-measurement state is given by $\frac{M_i|\psi\rangle}{\sqrt{p_i}}$.\\
For example, let us consider the measurement of a qubit in computational basis. Let $|\psi\rangle = a |0\rangle + b|1\rangle$ be the state being measured and $M_0 = |0\rangle \langle 0|$, $M_1 = |1\rangle \langle 1|$ be the measurement operators. Then the probability of obtaining measurement
outcome $"0"$ is $p_0 = \langle \psi |M_0^{\dagger}M_0| \psi \rangle = \langle \psi |M_0| \psi \rangle = |a|^2$. Similarly, the probability of obtaining the measurement outcome $"1"$ is $p_1 = |b|^2$. The post-measurement states in the two cases is therefore, $\frac{M_0 |\psi\rangle}{\sqrt{p_0}} = \frac{a |0\rangle}{|a|}$ and $\frac{M_1 |\psi\rangle}{\sqrt{p_1}} = \frac{b |1\rangle}{|b|}$.\\
The projective or von Neumann measurements and POVM measurements are special cases of Postulate 3 \cite{niel, akpan}.\\
\noindent\textbf{Postulate 4.} (Composite system) The state space of a composite physical system is the tensor product
of the state spaces of the component physical systems. Consider $n$ quantum systems and $|\psi_i\rangle$ represents the state of $i$th system, then the joint state of the total system is  $|\psi_1\rangle \otimes |\psi_2\rangle \otimes . . . \otimes |\psi_n\rangle$.
\subsection{Observables and their measurements}
In quantum mechanics, observables are dynamic variables like energy, position, momentum, angular momentum, etc. that can be measured \cite{snbiswas,mathews, srinivas}. In classical mechanics, these dynamical quantities are merely numbers, whereas, in quantum mechanics, these are represented by operators, in fact by Hermitian operators. The reason for choosing the Hermitian operators is due to the fact that Hermitian operators have real eigenvalues, and one needs a real number for the observation while making a measurement. Let $A$ be a dynamical system representing a certain dynamical quantity and consider the following eigenvalue equation
\begin{eqnarray}
	A |\psi \rangle = \lambda |\psi\rangle
\end{eqnarray}
If a physical system is in a state $|\psi\rangle$, then the measurement of the dynamical quantity will give a  real eigenvalue $\lambda$.
By Copenhagen interpretation of quantum mechanics, before the measurement of the dynamical variable $A$, the system is assumed to be in a superposition of eigenstates of $A$, where an arbitrary state is given by $|\alpha \rangle = \sum_i c_i |\alpha_i \rangle$. The state $|\alpha\rangle$ is normalized as follows.
\begin{eqnarray}
	\langle \alpha|\alpha\rangle = \sum_i \sum_j c_i^* c_j \langle\alpha_i|\alpha_j\rangle = \sum_i \sum_j c_i^* c_j \delta_{ij} = \sum_i |c_i|^2 = 1
\end{eqnarray}
 When a measurement is done, the system collapses to one of the eigenstates $|\alpha_i\rangle$ of the observable $A$. The probability that $|\alpha\rangle$ will go into the $i$th state $|\alpha_i\rangle$ is $|c_i|^2 = |\langle \alpha_i|\alpha\rangle|^2$.
This interpretation gives a physical meaning of superposition.

\subsection{Expectation value of an observable}
Expectation value of a Hermitian operator $A$ with respect to an arbitrary state $|\alpha\rangle$ is denoted by $\langle A \rangle$, and is defined by $\langle \alpha |A| \alpha\rangle$.  If $|\alpha_i\rangle$'s are the eigenstates of $A$ corresponding to the eigenvalues $\lambda_i$'s, $i = 1, 2, . . .$, then we have $A|\alpha_i\rangle = \lambda_i |\alpha_i \rangle$. Since $\langle \alpha_i | \alpha_j\rangle = \delta_{ij}$, we have
\begin{eqnarray}
	\langle \alpha_i |A| \alpha_j\rangle = \lambda_j \langle \alpha_i | \alpha_j\rangle = \lambda_j \delta_{ij} \label{step1.2}
\end{eqnarray}
 Using the completeness relation $\sum_i |\alpha_i\rangle \langle \alpha_i| = 1$, the expectation value $\langle A \rangle$ can be rewritten as
\begin{eqnarray}
	\langle A \rangle = \langle \alpha |A| \alpha\rangle &=& \sum_i \sum_j \langle \alpha|\alpha_i\rangle \langle \alpha_i |A| \alpha_j\rangle \langle \alpha_j|\alpha \rangle \nonumber\\ &=& \sum_i \sum_j \lambda_j \langle \alpha|\alpha_i\rangle \langle \alpha_j|\alpha \rangle \delta_{ij}\nonumber\\ &=& \sum_i \lambda_i \langle \alpha|\alpha_i \rangle \langle \alpha_i|\alpha \rangle = \sum_i \lambda_i  |\langle \alpha_i|\alpha \rangle|^2
\end{eqnarray}
where the third step follows from (\ref{step1.2}) and $|\langle \alpha_j|\alpha \rangle|^2$ is the probability that the state $|\alpha\rangle$ collapses in one of the eigenstates $|\alpha_i\rangle$.
Hence, expectation value of a dynamical variable is the average measured values of the dynamical variables. \\
If we perform a large number of experiments in which the state $|\alpha\rangle$ is prepared and the observable $A$ is measured, then the standard deviation $\bigtriangleup A$ of the observed values is determined by the formula 
\begin{eqnarray}
	\bigtriangleup A = \sqrt{\langle A^2 \rangle - \langle A \rangle^2} \label{sd}
\end{eqnarray}
\subsection{Uncertainty Principle}
One of the distinct features of quantum mechanics is quantum uncertainty, which means that not all properties of a system can be determined simultaneously \cite{srinivas}. The \textit{Heisenberg uncertainty principle} states that there is a limit to the precision with which a pair of complementary variables, such as a particle's position and momentum, can be measured simultaneously. The more precisely the position of a particle is determined, the less precisely its momentum can be known, and vice versa.\\
Let $A$ and $B$ be two Hermitian operators that are incompatible, i.e., $AB \neq BA$ (refer to def-(\ref{comm})). By the Heisenberg uncertainty principle,  the operators $A$ and $B$ are not \textit{simultaneously measurable}. So, we can find an error in our measurement. If the difference be described by the commutator of two incompatible operators, i.e., $[A,B] =	AB - BA = iC$, where $C$ is another observable, then the following inequality is satisfied when the average is taken over any arbitrary state \cite{snbiswas}.
\begin{eqnarray}
	\bigtriangleup A \bigtriangleup B \geq \frac{\langle C \rangle}{2}
\end{eqnarray}
This is precisely the mathematical statement of the uncertainty principle and is known as uncertainty relation.

\subsection{Tomography}
Measurement of a quantum state typically changes the state being measured, due to which many copies of the probe state are required to gain complete knowledge of the given state. Quantum state tomography is the process by which a quantum state is reconstructed using measurements on an ensemble of identical quantum states. The process of quantum state tomography typically involves the following steps:
\begin{enumerate}
	\item \textit{Preparation}: The first step is to prepare multiple identical copies of the quantum system in the state of interest. This is usually done by applying specific quantum operations or manipulations to the system.
	\item \textit{ Measurement in different bases}: Once the quantum system is prepared, it is measured multiple times in different bases. A basis is a set of orthogonal quantum states in which the measurement is performed. For example, in the case of a qubit (a two-level quantum system), one might choose the computational basis ($|0\rangle$ and $|1\rangle$) or the Hadamard basis ($|+\rangle$ and $|-\rangle$) for measurement.
	\item\textit{ Data collection}: The measurement outcomes are recorded for each copy of the quantum system on each chosen basis.
	\item \textit{Reconstruction}: Using the collected measurement data, various mathematical algorithms and techniques are employed to reconstruct or estimate the quantum state. The process involves solving an inverse problem to determine the quantum state that best fits the measurement results.
\end{enumerate}
The accuracy and reliability of quantum state tomography depend on factors such as the number of measurements performed, the quality of measurements, and the efficiency of the reconstruction algorithms used. Due to the complexity of quantum systems and the need for extensive resources, state tomography can be challenging and time-consuming, especially for larger quantum systems.
Despite its challenges, quantum state tomography plays a vital role in quantum information processing, quantum computing, and quantum technology. It enables researchers to characterize and understand quantum systems, which is crucial for developing and testing quantum algorithms and protocols.
\subsection{Shadow tomography}
Shadow tomography is a systematic framework that can help us benchmark the property of an unknown quantum system with only partial knowledge \cite{huang2020}.  Shadow tomography allows us to get specific
information about a quantum state using fewer copies of the state than full-scale tomography. Hence, shadow tomography is much more resource-saving than full tomography and it does not need joint operations among multiple copies to estimate non-linear functions. Shadow tomography consists of two phases: quantum measurement and data post-processing.
Suppose we are given some copies of an unknown quantum state $|\psi_i\rangle$, and we are interested in the
expectation value of an observable with respect to $|\psi_i\rangle$ within some error bounds. In particular,
if we would like to estimate $\langle \psi_i|A|\psi_i\rangle \pm \epsilon$ where $\epsilon$ is the error, then this can be done by taking $k$ copies of the state $|\psi_i\rangle$. It turns out that shadow tomography allows us to estimate observables $\{A_j\}$ where $j$ can assume a large value, using a
very small number of copies $k$ of the unknown state: given some number of copies $|\psi\rangle^{\otimes k}$ of an unknown state, we perform shadow tomography on
it to extract a piece of data known as a classical shadow $S$. This shadow $S$ is going to be a much
more efficient representation of $|\psi\rangle$ that allows us to estimate $\langle \psi|A_j |\psi \rangle \pm \epsilon$ for the observables that
we are interested in (but we will not be able to estimate all observables). The shadow tomography
procedure may depend on the observables {$A_j$}. The purpose of this procedure is to minimize the following
\begin{enumerate}
	\itemsep0em 
	\item The number of copies of $|\psi\rangle$ needed.
	\item  The complexity of producing the classical shadow $S$ from copies of $|\psi\rangle$.
	\item  The complexity of estimating the observables $\langle \psi|A_j|\psi \rangle$ given the classical shadow $S$.
\end{enumerate}
It is interesting and non-trivial to get any one of these items smaller than a polynomial in the dimension of $|\psi\rangle$ (which is exponential in the number of qubits) and the number of observables, which would be a thrift over standard tomography. Aaronson \cite{aaronson2018} gave a shadow tomography procedure where the required number of copies of $|\psi\rangle$ can be polylogarithmic in the dimension and in the number of
observables. Later, Huang et al. \cite{huang2020} presented another shadow tomography procedure that is more time-efficient than Aaronson's procedure. Shadow tomography has been found to be useful in entanglement detection. It has been shown that moments can be estimated by performing local random measurements on the quantum state, followed by post-processing using the classical shadows framework \cite{elben}.

\section{Basics of Quantum Information}
In this section, we discuss some basic concepts of quantum information theory such as quantum states, geometry of quantum states, representation of quantum states by density operator and reduced density operator. These concepts are especially useful in the study of quantum computation and quantum information.
\subsection{Quantum state}
Quantum state may be defined as the linear combination of classical states. In physics, the linear combination of states may be viewed as the superposition of states, which is one of the fundamental laws of quantum mechanics.  \\
\textbf{Polarization of light:} Consider a beam of light moving along the $z$-direction with a propogation vector, $\vec{k}$. Let $\hat{e}_x$ and $\hat{e}_y$ be two fundamental transverse polarization vectors of light along $x$ and $y$ direction, respectively. The polarization vector for this beam of light lies in the plane of $\hat{e}_x$ and $\hat{e}_y$. A physical system, for example, a photon with polarization along $\hat{e}_x$ or $\hat{e}_y$ denotes a state of a photon.  If a photon is in a state of arbitrary polarization, it is partly in a state of polarization along the vector $\hat{e}_x$ and partly in a state of polarization along the vector $\hat{e}_y$ \cite{snbiswas}. Let $|\phi\rangle$ describe the state of a photon with arbitrary polarization making an angle $\theta$ with $x$-axis. If $|\phi_1\rangle$ and $|\phi_2\rangle$ denote the state of a photon with polarization direction along $\hat{e}_x$ and $\hat{e}_y$, respectively, then the principle of superposition of states gives
\begin{eqnarray}
	|\phi\rangle = cos \theta |\phi_1 \rangle + sin \theta |\phi_2\rangle
\end{eqnarray}
where $cos^2 \theta$ and $sin^2 \theta$ denote the probabilities that the arbitrary state of a photon may appear along the vector $\hat{e}_x$ and $\hat{e}_y$, respectively.\\
	\begin{figure}[h!]
		\begin{center}
		\includegraphics[width=0.5\textwidth]{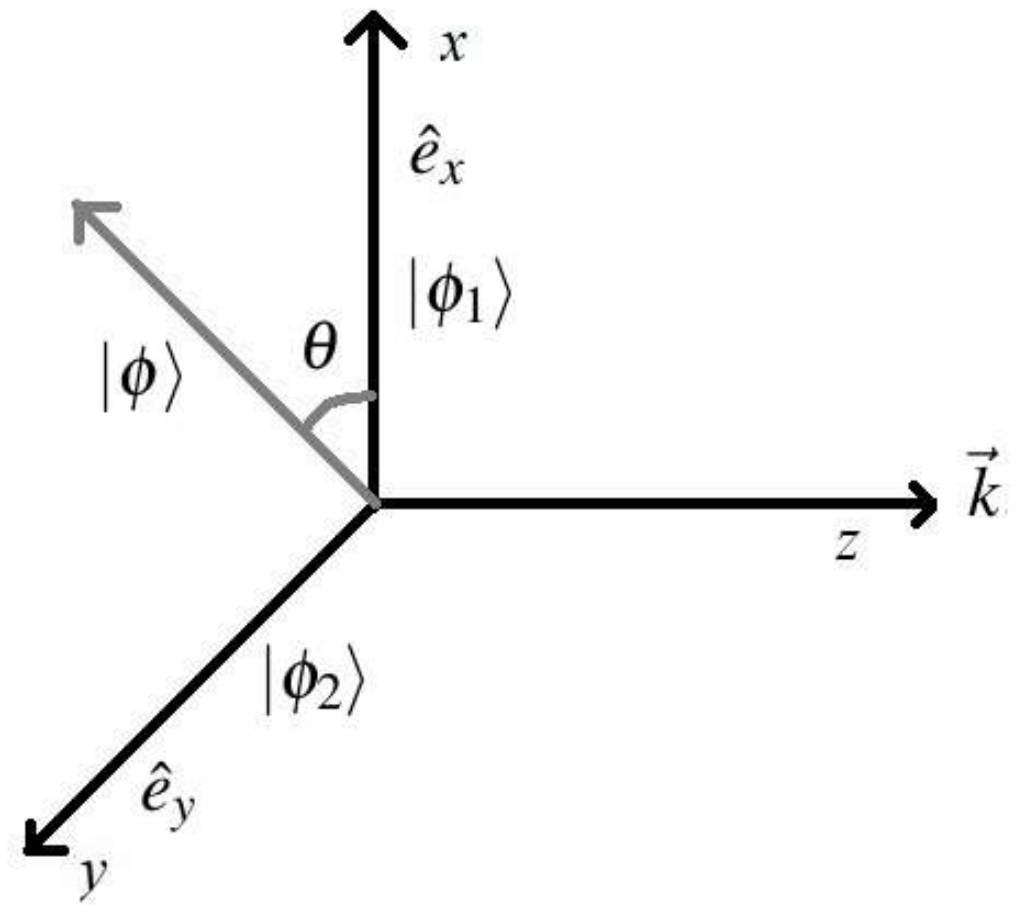}
		\caption{Light beam propogating along $z$-axis with propogation vector $\vec{k}$, light polarization vector $|\phi\rangle$ lies in $xy$ plane, photons polarized in an arbitrary direction making an angle $\theta$ with $x$-axis. }
	\end{center}
	\end{figure}
The vector $|\phi\rangle$ is called a state vector which is referred to as a \textit{ket} vector.
The quantum superposition occurs because of the lack of knowledge of the exact state. For example, when a beam of light consisting of one photon splits into two beams, it is not known in which of the split beams the photon is, except for the probability of finding it in either beam \cite{snbiswas}.  
{\remark If a ket vector $|\phi \rangle$ represents a physical system, then $c|\phi \rangle$ also represents the same physical system, where $c$ is a complex number and may be termed as the global phase factor.}
\subsection{Pure and Mixed Bipartite state}
 A bipartite quantum system described by the Hilbert space $\mathcal{H}$ is a composite system of two individual systems described by vectors in the tensor product of the two Hilbert spaces $\mathcal{H}_{A}$ and $\mathcal{H}_{B}$, i.e.,  $\mathcal{H} = \mathcal{H}_{A} \otimes \mathcal{H}_{B}$.\\
\textbf{Pure bipartite state:} Let $|a_i\rangle$ and $|b_j\rangle$ be the basis states of the Hilbert spaces $\mathcal{H}_{A}$ and $\mathcal{H}_{B}$ with dimensions $d_A$ and $d_B$, respectively. A pure state $|\psi\rangle \in \mathcal{H}$ can be written as
\begin{equation}
	|\psi\rangle = \sum_{i,j = 1}^{d_{A},d_{B}} c_{ij} |a_{i}\rangle \otimes |b_{j}\rangle \in \mathcal{H}_{A} \otimes \mathcal{H}_{B} \label{pure}
\end{equation}
 with a complex matrix $C = [c_{ij}]$ of order $d_{A} d_{B}$.\\ 
\textbf{Mixed bipartite state:} A convex mixture of the different pure bipartite states in an \textit{ensemble} is known to be a mixed bipartite state.
{\remark Pure and mixed states may be defined in a similar way for multipartite systems also.}
\subsection{Density operator}
Quantum mechanics has been formulated alternately using the language of operators or matrices. This alternate formulation is mathematically equivalent to the state vector approach. In general, the state of a quantum system is not known completely. If a quantum system is in one of the states $|\psi_i\rangle$ with probability $p_i$, then $\{p_i, |\psi_i \rangle\}$ is called an \textit{ensemble of pure states}. \\
The density operator for a quantum system is described as
\begin{eqnarray}
	\rho = \sum_i p_i |\psi_i\rangle \langle \psi_i | \;\; \text{where} \;\; 0 \leq p_i \leq 1 \;\; \text{and} \;\; \sum_i p_i =1
\end{eqnarray}
The density operator satisfies the following conditions:
\begin{enumerate}
	\item[(i)] (\textit{Hermitianity condition}) $\rho$ must be Hermitian.
	\item[(ii)] (\textit{Trace condition}) $\rho$ has trace equal to one, i.e., $Tr[\rho]=1$.
	\item[(iiI)] (\textit{Positive semi-definite condition}) $\rho$ must be a positive semi-definite operator, i.e., $\rho \geq 0$.
\end{enumerate}
{\remark: Let $\rho$ be a density operator. Then $Tr[\rho^2] \leq 1$ with equality if and only if $\rho$ is a pure state.}
\subsection{Reduced system}
Let us consider a composite system consisting of two subsystems $A$ and $B$, respectively. If we would like to gain some knowledge of the individual system from the composite system, then we have to perform some operations on the composite system. Partial trace may be considered as one such operation.\\
\textbf{Partial trace:} Let $A$ and $B$ be two physical systems. If $|a_1\rangle$, $|a_2\rangle$ be any two vectors in the state space of $A$; and  $|b_1\rangle$, $|b_2\rangle$ be any two vectors in the state space of $B$; then \textit{partial trace }over the system $A$ and partial trace over the system $B$ are defined as
\begin{eqnarray}
	Tr_A [|a_1\rangle \langle a_2] \otimes |b_1\rangle \langle b_2| ] &=& |b_1\rangle \langle b_2| 	Tr [|a_1\rangle \langle a_2|] \\
	Tr_B [|a_1\rangle \langle a_2] \otimes |b_1\rangle \langle b_2| ] &=& |a_1\rangle \langle a_2| 	Tr [|b_1\rangle \langle b_2|] \label{partracedef}
\end{eqnarray}
\textbf{Reduced density operator:} Suppose a bipartite quantum state is described by the density operator $\rho_{AB}$ where $A$ and $B$ are subsystems of the composite quantum system. Then the reduced density operators for the system $A$ and system $B$ are described as
\begin{eqnarray}
	\rho_A = Tr_B[\rho_{AB}] \label{reducedop}; \quad 	\rho_B = Tr_A[\rho_{AB}]
\end{eqnarray}
where $Tr_A$ and $Tr_B$ denotes the partial trace over the subsystem $A$ and $B$, respectively.

\subsection{Geometry of quantum states as Bloch sphere}
A single qubit in the state $|\psi\rangle = a |0\rangle + b|1\rangle$, $|a|^2 +|b|^2 =1$ can be visualized as a point $(\theta, \phi)$
on the unit three-dimensional sphere, where $a = cos(\theta/2), b = e^{i\phi} sin(\theta/2)$. This sphere is called the Bloch sphere and the vector $(cos \phi sin \theta, sin \phi sin \theta, cos \theta)$ is called the Bloch vector.
\begin{figure}[h!]
	\begin{center}
		\includegraphics[width=0.4\textwidth]{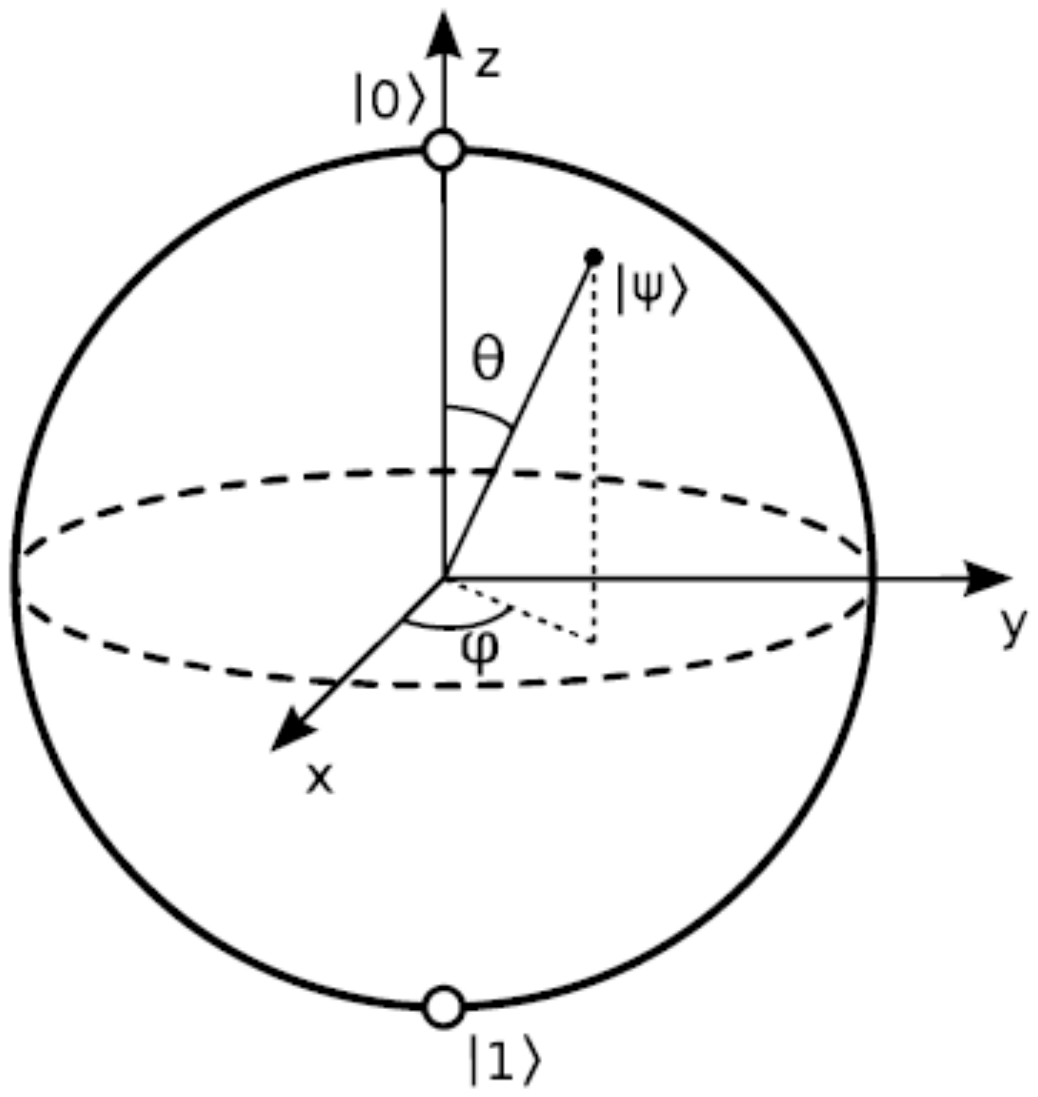}
		\caption{Bloch sphere representation of a qubit}
	\end{center}
\end{figure}
Bloch sphere provides a useful means of visualizing a quantum state. Any point on the Bloch sphere, and inside the ball corresponds to a physical state. The pure states lie on the sphere and the mixed ones lie inside the sphere \cite{niel}. \\
In general, a single qubit may be expressed as
\begin{eqnarray}
	\rho = \frac{I + \vec{r} .\vec{\sigma}}{2}
\end{eqnarray}
where $\vec{r} = (r_x,r_y,r_z)$ is the Bloch vector which is a real three-dimensional vector such that $||\vec{r}|| \leq 1$ and $\vec{\sigma} = (\sigma_{x},\sigma_{y},\sigma_{z})$. Here, $\sigma_i$ are the Pauli matrices defined in (\ref{pauli}). $\rho$ represents a pure state if and only if $||\vec{r}|| = 1$. The center of the Bloch sphere is represented by the maximally mixed state, i.e., at the center, the quantum state described by the density operator $\rho = I/2$. 
{\remark For the higher dimensional systems, the geometric character of the Bloch space turns out to be quite complicated and is still of great interest.}

\section{Quantum Entanglement}
In 1935, Einstein, Podolsky, and Rosen (EPR) found that quantum mechanics lacks a very important property known as the element of reality and locality \cite{epr}. Thus, they concluded that quantum mechanics is an incomplete theory. To support their statement, they considered the wave function of two physical quantities jointly and then revealed that the knowledge of one physical quantity is not sufficient to gain knowledge of the other physical quantity. But interestingly there exists a composite physical system consisting of two particles in which if we measure one system then instantaneously it affects the other system. Einstein called this feature, "spooky action at a distance". On the other hand, Schrodinger termed it as \textit{ entanglement} \cite{schrod}. Almost thirty years later, experiments were performed by John Bell that invalidated the EPR argument and proved the existence of the phenomenon known as entanglement \cite{jsbell}. Entanglement is a purely quantum mechanical feature that has no classical analogue. It enables particles to affect each other instantaneously across any distance. As time passed, a new theory of information evolved, that followed the laws of quantum mechanics. This new theory is called quantum information theory, which uses entanglement as a resource, enabling tasks like quantum cryptography \cite{bennett3}, quantum teleportation \cite{bennett1}, quantum superdense coding \cite{bennett5} or measurement based quantum computation \cite{raussen}. The potential offered by quantum entanglement to computing, security and communication makes it a topic of vital interest to researchers all across the globe. 
\subsection{Definition} \label{sec-entdef}
Let $\rho_{AB} \in B(\mathcal{H}_{A} \otimes \mathcal{H}_{B})$ be a density matrix for a bipartite system $\mathcal{H}_{A} \otimes \mathcal{H}_{B}$.\\ 
\textbf{Product state:} $\rho_{AB}^P$ is a product state if there exist states $\rho^{A}  \in B(\mathbb{H}_{A})$ and $\rho^{B} \in B(\mathbb{H}_{B})$ such that
\begin{equation}
	\rho_{AB}^P = \rho^{A} \otimes \rho^{B}
\end{equation}
\textbf{Separable state:} The state $\rho_{AB}^{sep} $ is called separable, if there are convex weights $p_{i}$ and product states $\rho_{i}^{A} \otimes \rho_{i}^{B}$ such that
\begin{equation}
	\rho_{AB}^{sep} = \sum_{i} p_{i} \rho_{i}^{A} \otimes \rho_{i}^{B} \;\; \text{where} \;\; 0 \leq p_i \leq 1; \;\; \sum_i p_i =1 \label{eq-sepstate}
\end{equation}
holds. Otherwise the state is \textbf{entangled}.

\subsection{Bell's inequality}
Einstein together with Podolsky and Rosen recognized that quantum theory allows a particular type of correlation (later known as entanglement) between two distant parts of a system \cite{epr}. They
argued that if such correlations allow the prediction of the result of a measurement
on one part of a system by looking at the very distant part, then, in a complete
and local theory, the predicted quantity has to have a definite value even before the measurement. However, according to them, this value could not be obtained from
quantum mechanics. The presence of entanglement led them in this way to the conclusion that quantum theory is an apparently incomplete theory.
For the following three decades, the debate about the EPR dilemma was philosophical in nature. This
situation, however, changed dramatically in the year 1964, when John Bell showed that the matter could be decided by an experiment \cite{jsbell}. He derived correlation inequalities, which can be violated in quantum mechanics, but have to be satisfied within every model that is local and complete. In a particular sense, the local and complete models are known as local hidden variable models (LHVM).
Experiments showing the first reliable violations of Bell's inequalities thereby confirming quantum mechanics and demonstrating the presence of entanglement in nature were then started in the early eighties \cite{aspect1981}. We begin the discussion of Bell inequalities studying the general assumptions required for their derivation.\\ 
(i) \textit{Assumption of realism}: The assumption that the physical properties  have definite values which exist independent of observation. \\
	(ii) \textit{Assumption of locality:} The assumption that the result of the measurement performed on one system does not influence the result of the measurement on second system.\\
These two assumptions together are known as the assumptions of \textit{local realism}. Bell inequalities showed that at least one of these assumptions is not correct. \\
Consider a correlation experiment in which the variables $A_1$, $A_2$ are measured on one subsystem of the composite system and $B_1$, $B_2$ on the other subsystem, and that the subsystems are spatially separated. Then
\begin{eqnarray}
|	\langle A_1 B_1 \rangle + \langle A_1 B_2 \rangle + \langle A_2 B_1 \rangle - \langle A_2 B_2 \rangle| \leq 2 \label{chsh}
\end{eqnarray}
where $	\langle A_i B_j \rangle $ is the expectation value of the  joint random variable $A_i B_j$.
This is known as Clauser, Horne, Shimony and Holt (CHSH) inequality which must be obeyed for any quantum state $\rho_{AB}$ admitting a LHVM \cite{clauser1969}. However, it was shown that the above inequality is violated with maximum violation  of $2\sqrt{2}$. The value $2\sqrt{2}$ is achieved for the following two-qubit states known as Bell states or EPR pairs \cite{aditi2001}.
\begin{eqnarray}
	|\phi^+ \rangle =	\frac{|00\rangle + |11\rangle}{\sqrt{2}};\;\;\;
	|\phi^- \rangle =	\frac{|00\rangle - |11\rangle}{\sqrt{2}} \nonumber\\
	|\psi^+ \rangle =	\frac{|10\rangle + |01\rangle}{\sqrt{2}};\;\;\;
	|\psi^- \rangle =	\frac{|01\rangle - |10\rangle}{\sqrt{2}} \label{bellstates}
\end{eqnarray}
 The violation of the inequality (\ref{chsh}) validates the existence of non-locality in the given bipartite state \cite{paterek2004, paterek2011, ravi2020}.

\subsection{Distillation}
Suppose two separated observers, Alice and Bob, would like to perform a quantum information task but do not share a maximally entangled state. Instead, they are supplied with as many mixed states $\rho_{AB}$ as they want. Let us assume that Alice and Bob are supplied with $k$ copies of $\rho_{AB}$. We now pose a question: Can they use the $k$ copies of the states described by $\rho_{AB}^{\otimes k}$ to establish the singlet states between them by LOCC? What is the cost of this transformation?
\begin{eqnarray}
	\underbrace{\rho_{AB} \otimes \rho_{AB} \otimes . . . .  \otimes \rho_{AB} }_{k \;copies} \xrightarrow{LOCC} |\psi ^- \rangle = \frac{1}{\sqrt2}(|01\rangle - |10\rangle)?
\end{eqnarray}
The distillable entanglement answers these questions and determines how
many singlet pairs can be extracted (or distilled) out of $k$ pairs of the state $\rho_{AB}$ using LOCC, in the limit of $k \rightarrow \infty$. The idea of entanglement distillation is for Alice and Bob to convert some
large number of copies of a state $\rho_{AB}$ into as many copies of the Bell state as possible using local operations and classical communication, requiring
not that they succeed exactly, but only with high fidelity. There are two qualitatively two types of entangled states classified on the basis of distillation: distillable entangled states are called free entangled states and undistillable entangled states are called bound entangled states (BES). Bound entangled states require entanglement for their creation, but the entanglement can then not be distilled again. If the density matrix representing a bipartite state is PPT, then the state is undistillable.  Further, there is evidence for the existence of undistillable NPT entangled states that have been shown in \cite{divincenzo2000}.
\subsubsection{Bound entangled states or PPT entangled states} \label{sec-bes}
The existence of bound entanglement stands among the most intriguing features of the entanglement theory. In fact, it represents a kind of irreversibility of the process of formation of entangled states \cite{terhal2003}. The essential quantum property that led to the observation of bound entanglement was the existence of entanglement with positive partial transposition (PPT) property, i.e., the density matrix representing the entangled state has positive partial transpose. Originally considered useless for quantum information processing, bound entangled states were later established as a valid resource in the contexts of quantum key distribution \cite{khoro}, entanglement activation \cite{pmrhoro}, metrology \cite{gt}, non-locality \cite{tn}, and their non-distillability has been linked to irreversibility in thermodynamics \cite{fgsl}.\\
The following two states $\rho_a$ and $\rho_{a,b,c}$ are examples of BES in $3 \otimes 3$ and $2\otimes 2\otimes 2$ dimensional system, respectively.
{\example
Let us consider a family of $3 \otimes 3$ BES described by the density operator $\rho_{a}$, which is given by \cite{horodeckia}
\begin{equation}
	\rho_a = \frac{1}{8a + 1}
	\begin{pmatrix}
		a&0&0&0&a&0&0&0&a\\
		0&a&0&0&0&0&0&0&0\\
		0&0&a&0&0&0&0&0&0\\
		0&0&0&a&0&0&0&0&0\\
		a&0&0&0&a&0&0&0&a\\
		0&0&0&0&0&a&0&0&0\\
		0&0&0&0&0&0&\frac{1}{2}(1 + a)&0&\frac{1}{2}\sqrt{1 - a^2}\\
		0&0&0&0&0&0&0&a&0\\
		a&0&0&0&a&0&\frac{1}{2}\sqrt{1 - a^2}&0&\frac{1}{2}(1 + a)\\
	\end{pmatrix} \label{astate-matrix}
\end{equation}
where $0\leq a \leq 1$. The state $\rho_a$ is separable when $a = 0$ and 1.\\
It is known that the density matrix $\rho_a$ represents a family of BES for $0 < a < 1$ \cite{horodeckia}.
\example Let us consider another family of BES  $\rho_{a,b,c}$ in $2\otimes 2\otimes 2$ dimensional system \cite{acin},
\begin{eqnarray}
	\rho_{a,b,c}= \frac{1}{N}
	\begin{pmatrix}
		1 & 0& 0&0&0&0&0&1 \\
		0 & a&0 &0&0&0&0&0 \\
		0 & 0& b&0&0&0&0&0\\
		0 & 0& 0&c&0&0&0&0\\
		0 & 0& 0&0&1/c&0&0&0\\
		0 & 0& 0&0&0&1/b&0&0\\
		0 & 0& 0&0&0&0&1/a&0\\
		1 & 0& 0&0&0&0&0&1
	\end{pmatrix}
\end{eqnarray} with $a,b,c > 0$ and $N= 2 + a+b+c+\frac{1}{a}+\frac{1}{b}+\frac{1}{c}$}.
\subsection{Detection of Entanglement}
Although entangled states are used in various quantum information processing tasks, the practical use of an entangled resource is restricted to the successful experimental realization of the resource. In real experimental setups, it is always a challenge to create and detect entangled states. Hence, successful generation, quantification and detection of entanglement are essential features in any quantum information processing protocol. Many authors contributed to a long list of various separability criteria and detection methods \cite{guhnerev,horodecki9}. However, a completely satisfactory solution to these problems has not been found despite many efforts in the past decades. To solve the separability problem, researchers have provided several criteria for the detection of entanglement. While it is relatively easy to detect and quantify the pure entangled states \cite{bennetch}, one of the most important open questions of quantum information theory is to determine whether a given mixed bipartite or multipartite state is entangled. This problem is additionally complicated by the existence of bound entangled states \cite{mprhoro}, which are weak entangled states and hard to detect.
Since entanglement detection problem is a NP hard problem so all entangled states cannot be detected by just one criterion and thus numerous criterion has been developed for the detection of entanglement.\\
\subsection{Classification of Entanglement}
Let us now start a discussion on the classification of multipartite systems. It may be observed that either with the increase in the dimension of a quantum system or with the number of qubits, the geometric structure of state space becomes highly complicated. This makes the detection of entanglement in multipartite quantum systems, a resource-consuming task. Hence, it is important to find implementable and efficient methods to classify multipartite entanglement. Acin et al. introduced a classification of the whole space of mixed three-qubit states into different entanglement classes \cite{acin}. Mixed three-qubit states are classified by generalizing the classification of three-qubit pure states.  Two states are called SLOCC (stochastic local operation and classical communication) equivalent if they can be obtained from each other under SLOCC, otherwise, they are SLOCC inequivalent \cite{dutta1}. Three-qubit states are classified into six SLOCC inequivalent classes: fully separable, three biseparable, and two genuinely entangled (GHZ and W) states \cite{vidal00}. 
A density matrix $\rho_{ABC}$ on the Hilbert space $\mathcal{H}_2\otimes \mathcal{H}_2 \otimes \mathcal{H}_2$ ($\mathcal{H}_2$ denotes the Hilbert space of dimension 2) is fully separable if it can be written as a convex combination of product states as follows:
\begin{eqnarray}
	\rho_{ABC} = \sum_i p_i \rho_i^{A} \otimes \rho_i^{B} \otimes \rho_i^{C} \;\text{with}\; p_i \geq 0 \; \text{and}\;\sum_i p_i = 1 \nonumber
\end{eqnarray}
If a state is not fully separable then it contains some entanglement. Then it may be classified as a biseparable or genuinely entangled state.  The three classes of biseparable states contain only bipartite entanglement between any two of the qubits. For example, the states in class $A|BC$ possess entanglement between the qubits $B$ and $C$ and are separable with respect to the qubit $A$. On the other hand, genuine three-qubit entanglement means that there exists entanglement between all three qubits. 
\subsection{Quantification of Entanglement}
Apart from the detection of entanglement, another major issue in quantum information theory is to quantify the amount of entanglement. Various entanglement measures exist in the literature such as concurrence \cite{wootters,wootters1,wootters2}, negativity \cite{vidal}, relative entropy of entanglement \cite{plenio}, geometric measure of entanglement \cite{wei} that can quantify the amount of entanglement in a two-qubit as well as higher dimensional bipartite pure and mixed state. A "good" measure of entanglement is any function $E: B(\mathcal{H}) \rightarrow [0,1]$, such that it satisfies the following postulates \cite{plenio}:
\begin{enumerate}
	\item[(i)] The measure of entanglement for any separable state $\rho^{sep}$ should be zero, i.e., $E{(\rho^{sep})}=0$, where $\rho^{sep} \in B(\mathcal{H})$.
	\item[(ii)] The amount of entanglement in any entangled state $\rho^{ent} \in B(\mathcal{H})$ should be unaffected for any local unitary transformation of the form  $U_{A}\otimes U_{B}$.
	\item[(iii)] Local operations and classical communication (LOCC) cannot increase the
	expected entanglement.
	\item[(iv)] For any $ \rho_i  \in B(\mathcal{H})$ and $\rho = \sum_i p_i \rho_i$, $E(\rho)$ satisfies convexity property, i.e., 
	\begin{eqnarray*}
		E(\sum_i p_i \rho_i) = 	\sum_i p_i E(\rho_i) \;\; \text{where} \;\;  0 \leq p_i \leq 1;\;\; \sum_i p_i =1
	\end{eqnarray*}
\end{enumerate}
\subsubsection{Concurrence}
 A very popular measure for the quantification of bipartite quantum entanglement is the \textit{concurrence}.  For the two-qubit case, an elegant formula for the concurrence was derived analytically by Wootters \cite{wootters}.\\
 \textbf{For pure state $|\psi_{AB}\rangle$ :} 
 $C(|\psi_{AB} \rangle)=\sqrt{2(1-Tr{(\rho_{A}^2)})}$, where $\rho_{A} = Tr_B (|\psi_{AB}\rangle \langle \psi_{AB}|)$ \\
 \textbf{For mixed state $\rho_{AB}$}:
 $C(\rho_{AB})= max(0,\sqrt{\lambda_{1}}-\sqrt{\lambda_{2}}-\sqrt{\lambda_{3}}-\sqrt{\lambda_{4}})$,\\
 where	
 $\lambda_{i}$ are the eigenvalues of $\rho_{AB}\widetilde{\rho_{AB}}$ arranged in descending order and $\widetilde{\rho_{AB}}$=$(\sigma_{y}\otimes\sigma_{y})\rho_{AB}(\sigma_{y}\otimes\sigma_{y})$; $\sigma_{y} = - i|0\rangle\langle1| + i|1\rangle\langle0|$.\\
  Due to the extremizations involved in the calculation, only a few explicit analytic formulae of concurrence for the higher dimensional bipartite systems have been found. The closed formula of concurrence for some special symmetric states has been obtained \cite{terhal2000,rungta2003}. In particular, some progress in the form of practical algorithms to find the lower bounds of the concurrence for qubit-qudit systems has been made \cite{chenliang}. Some results related to the lower bound of concurrence have been obtained using numerical optimization in arbitrary dimensional bipartite systems \cite{mintertf}. An
  optimized bound generally involves numerical optimization
  over a large number of free parameters which becomes computationally unmanageable in higher dimensional systems. Also, concurrence cannot reliably detect arbitrary entangled states even if one applies all known optimization methods \cite{mintertf}. To improve this situation, an analytic lower bound on concurrence for any dimensional mixed bipartite quantum states has been presented, which has further been shown to be exact for some special classes of states and detect many bound entangled states \cite{kchen}. Witness operator is also used in the quantification of entanglement \cite{brandao} and in the estimation of the lower bound of the concurrence of the entangled states in $d_{1}\otimes d_{2} (d_{1}\leq d_{2})$ dimensional systems \cite{mintert2}.

\section{A Few Method for Detection of Entanglement} \label{sec-entdet}
In this section, we will present some criteria for the detection of bipartite entanglement and discuss a few of them in
detail.
\subsection{Using the definition of entanglement}
 Pure entangled states may be detected using the definition of entanglement given in section \ref{sec-entdef}. The problem arises when we consider the detection of mixed entangled states. In this case, it is practically hard to apply the definition. This is because the representation of a quantum state is not unique. 
		Thus, we have to check for the given quantum state whether it can be expressed in the form (\ref{eq-sepstate}) for every possible basis. It is also known that one basis can be obtained from the other just by using a unitary transformation. Due to this, it is hard to conclude whether the given mixed state is entangled or not just by applying the definition of entanglement given in (\ref{eq-sepstate}).
\subsection{Schmidt decomposition}
The Schmidt decomposition is of central importance in the characterization and quantification of entanglement associated with pure states \cite{ekert1995}.\\
Let $|\psi\rangle_{AB}$ be a bipartite pure state such that $|i_A\rangle$ and $|i_B\rangle$ form an orthonormal basis of subsystems $A$ and $B$. By Schmidt decomposition, $\rho_{AB}$ can be expressed as
\begin{eqnarray}
	|\psi\rangle_{AB} = \sum_i \lambda_i |i_A\rangle |i_B\rangle \label{scpure}
\end{eqnarray}
where $\lambda_i$ are non-negative real numbers satisfying $\sum_i \lambda_i^2 = 1$ known as \textit{Schmidt coefficients} \cite{niel}. The number of non-zero coefficients $\lambda_i$ is called \textit{Schmidt rank}. 
{\theorem A pure state is entangled if and only if Schmidt rank is greater than one.}\\
\textbf{Schmidt decomposition for any pure or mixed bipartite state:} In general, the Schmidt decomposition of any pure or mixed bipartite state can be expressed in terms of the orthonormal basis of the operator space. To define the decomposition, let us consider a bipartite state (pure or mixed) described by the density operator $\rho_{AB}$ with $\{A_i\}$ and $\{B_i\}$ forming orthonormal bases of the operator spaces for subsystems $A$ and $B$ with respect to the Hilbert-Schmidt inner product, i.e., $Tr[A^{\dagger}_i A_j ] =Tr[B^{\dagger}_i B_j] = \delta_{ij}$. Then the Schmidt decomposition of $\rho_{AB}$ is defined as
\begin{eqnarray}
	\rho_{AB} = \sum_{i=1}^{r} \lambda_i\; A_i \otimes B_i \label{scmixed}
\end{eqnarray}
where the Schmidt coefficients  $\lambda_i$ are non-negative real numbers and $r$ denotes the Schmidt rank satisfying $1 \leq r\leq min\{dim(A), dim(B)\}$ \cite{peresbook}.\\
We will discuss the separability criteria using Schmidt decomposition for mixed states in the later section.

\subsection{PPT criteria} \label{ptmethod}
Partial transposition of a bipartite state is described as the transpose operation with respect to one subsystem. Consider a state $\rho_{AB}$ described in the following form:
\begin{eqnarray}
	\rho_{AB} = \sum_{ijkl} \rho_{ijkl} |i \rangle \langle j| \otimes |k \rangle \langle l| 
\end{eqnarray}
The partial transpose of $\rho_{AB}$ with respect to the subsystem $A$ is described as
\begin{eqnarray}
		\rho_{AB}^{T_A} = \sum_{ijkl} \rho_{ijkl} |j \rangle \langle i| \otimes |k \rangle \langle l| 
\end{eqnarray}
Similarly, we can define $\rho_{AB}^{T_B}$ by exchanging $k$ and $l$ instead of $i$ and $j$. \\
For example, the partial transpose $\rho_{AB}^{T_B}$ of the $2\otimes2$ state $\rho_{AB}$ w.r.t. the subsystem $B$ is given by
$$\rho_{AB}=\left[\begin{array}{cc|cc}
	\rho_{11} & \textcolor{red}{\rho_{12}} & \rho_{13} & \textcolor{red}{\rho_{14}}\\
	\textcolor{blue}{\rho_{12}^*} & \rho_{22} & \textcolor{blue}{\rho_{23}} & \rho_{24}\\ \hline
	\rho_{13}^* & \textcolor{red}{\rho_{23}^*} & \rho_{33} & \textcolor{red}{\rho_{34}}\\
	\textcolor{blue}{\rho_{14}^*} & \rho_{24}^* & \textcolor{blue}{\rho_{34}^*} & \rho_{44}\\
\end{array}\right] \;\rightarrow\; \rho_{AB}^{T_B}= \left[
\begin{array}{cc|cc}
	\rho_{11} & \textcolor{blue}{\rho_{12}^*} & \rho_{13} &\textcolor{blue}{\rho_{23}}\\
	\textcolor{red}{\rho_{12}} & \rho_{22} & \textcolor{red}{\rho_{14}} & \rho_{24}\\ \hline
	\rho_{13}^* & \textcolor{blue}{\rho_{14}^*} & \rho_{33} & \textcolor{blue}{\rho_{34}^*}\\
	\textcolor{red}{\rho_{23}^*} & \rho_{24}^* &\textcolor{red}{\rho_{34}} & \rho_{44}\\
\end{array}\right]$$
{\definition (\textit{PPT states:}) A state represented by the density matrix $\rho_{AB}$ has positive partial transpose if its partial transposition has no negative eigenvalues, i.e., it forms a positive semidefinite matrix. Such states are called PPT states. 
	\remark If a state is not PPT, then it is called an NPT state.}\\
PPT criteria is the most well-known separability criterion based on partial transposition. It is stated as follows \cite{peres}.
{\theorem (\textit{PPT criteria}) If $\rho_{AB}$ is a bipartite separable state, then it has positive partial transpose, i.e., it represents a PPT state.}\\
The contrapositive of the above theorem implies that if a state has negative partial transpose, then it is entangled. These states are called negative partial transpose entangled states (NPTES). Thus, if we find at least one negative eigenvalue of the partially transposed matrix, we can conclude that it represents an entangled state. Now the question is whether the above statement is sufficient to conclude the separability of a state. An answer to this question was given by Horodecki et al. \cite{horo1996}. It is stated as follows.
{\theorem  $\rho_{AB}$ is a separable state in $ 2 \otimes 2$ or $2 \otimes 3$ dimensional system if and only if $\rho_{AB}^{T_B} \geq 0$.}
Hence, the PPT criterion is necessary and sufficient only for two-qubit and qubit-qutrit systems. In higher dimensions, there exist entangled states with positive partial transpose. These states are called positive partial transpose entangled states (PPTES) or bound entangled states.
\subsection{Range Criterion}
The range criterion was one of the first criteria for the detection of entangled states for which the PPT criterion fails \cite{horodeckia}. It states that if a bipartite state $\rho_{AB}$ is separable, then there is a set of product vectors $|a_ib_i\rangle$ which spans the range of $\rho_{AB}$ as well as the set $|a_ib_i^*\rangle$  spans the range of $\rho_{AB}^{T_B}$. This criterion detects many entangled states that are undetected by the PPT criterion. However, it cannot be used for the detection of states that are affected by noise. This is because, in such case, the density matrix and its partial transpose will usually have full rank, hence the condition in the range criterion is automatically fulfilled.
\subsection{The Computable Cross Norm and Realignment (CCNR) Criterion} \label{sec-ccnr}
\textbf{Computable Cross Norm (CCN) Criterion:} Let $\rho_{AB}$ be a bipartite state with Schmidt decomposition $\rho_{AB} = \sum_{i=1}^{r} \lambda_i\; A_i \otimes B_i$, where $r$ is the Schmidt rank. Computable cross norm (CCN) criterion states that if $\rho_{AB}$ is separable, then $\sum_{i=1}^r\lambda_i \leq 1$ \cite{rud2005}.\\
An alternative formulation of the CCN criterion can be defined in terms of a linear map $R$ called a realignment map.
Let us start with first defining the realignment operation. Interestingly, the realignment operation can be formulated in different ways that are equivalent to each other \cite{rudolph2003}.\\
Consider a quantum state described by a density matrix $\rho_{AB}$ in $m \otimes n$ dimensional system. $\rho_{AB}$ can be written in block matrix form with $m$ number of blocks in each row and column with each block $Z_{ij}$ of size $n \times n$, i.e.,
\begin{eqnarray}
	\rho_{AB} = \begin{pmatrix}
		Z_{11} & Z_{12} & . & . & . & Z_{1n}\\
		Z_{21} & Z_{22} & . & . & . & Z_{2n} \\
		. & . & . & . & . & . \\
		. & . & . & . & . & . \\
		. & . & . & . & . & . \\
		Z_{m1} & Z_{m2} & . & . & . & Z_{mn} 	
	\end{pmatrix}
\end{eqnarray}
Then the realigned matrix of $\rho_{AB}$, denoted by $R(\rho_{AB})$, is an  $m^2 \times n^2$ matrix and it is given by \\
\begin{eqnarray}
R(\rho_{AB} )= \begin{pmatrix}
		vec(Z_{11})\\
		.\\
		.\\
		vec(Z_{1n})\\
		.\\
		.\\
		vec(Z_{n1})\\
		.\\
		.\\
		vec(Z_{mn})
	\end{pmatrix} \label{rdef}
\end{eqnarray}
where for any $m\times n$ matrix $Z_{ij} = (z_{kl}^{ij})$, the vector $vec(Z)$ is defined as
\begin{eqnarray}
	vec(Z_{ij})=	(z_{11}^{ij}, ., ., ., z_{1n}^{ij},  z_{21}^{ij}, ., ., ., z_{2n}^{ij}, ., ., ., z_{m1}^{ij}, ., ., ., z_{mn}^{ij} ) \label{vecx}
\end{eqnarray}
For example, consider the following bipartite qubit state $\rho_{AB}$ defined on $\mathbb{C}^2 \otimes \mathbb{C}^2$.
\begin{eqnarray}
	\rho_{AB}=
	\begin{pmatrix}
		\rho_{11} & \rho_{12} & \rho_{13} & \rho_{14}\\
		\rho_{12}^* & \rho_{22} & \rho_{23} & \rho_{24}\\
		\rho_{13}^* & \rho_{23}^* & \rho_{33} & \rho_{34}\\
		\rho_{14}^* & \rho_{24}^* & \rho_{34}^* & \rho_{44}\\
	\end{pmatrix} \label{rho}
\end{eqnarray}
 The matrix obtained after applying the realignment operation $R$  defined in (\ref{rdef}) on the state $\rho_{AB}$ is given as
\begin{eqnarray}
	R(\rho_{AB} )= 
	\begin{pmatrix}
		\rho_{11} & \rho_{12} & \rho_{12}^* & \rho_{22}\\
		\rho_{13} & \rho_{14} & \rho_{23} & \rho_{24}\\
		\rho_{13}^* & \rho_{23}^* & \rho_{14}^* & \rho_{24}^*\\
		\rho_{33} & \rho_{34} & \rho_{34}^* & \rho_{44}\\
	\end{pmatrix}
\label{Rrho}
\end{eqnarray}
The action of the realignment map on the tensor product of two matrices $A=\sum_{ij} a_{ij} |i\rangle \langle j|$ and $B=\sum_{kl} b_{kl} |k\rangle \langle l|$ is given as follows
\begin{eqnarray}
	R(A \otimes B)= R\left(\sum_{ijkl} a_{ij} b_{kl}  |ik\rangle \langle jl| \right) = \sum_{ijkl} a_{ij} b_{kl}  |ij\rangle \langle kl| \label{abr}
\end{eqnarray}
The matrices $A$ and $B$ can be identified as vectors lying in the tensor product ket space as a consequence of Choi-Jamiolkowski isomorphism \cite{choi, jamio}, i.e., we have
\begin{eqnarray}
	|A\rangle = \sum_{ij} a_{ij} |i\rangle |j\rangle; \quad |B\rangle = \sum_{kl} b_{kl} |k\rangle |l\rangle \label{ketab}
\end{eqnarray}
and the corresponding dual vectors in tensor product bra space are 
\begin{eqnarray}
	\langle A^*| = \sum_{ij} a^*_{ij} \langle i| \langle j|; \quad \langle B^*| = \sum_{kl} b^*_{kl} \langle k| \langle l| \label{ketab}
\end{eqnarray}
Hence we have,
\begin{eqnarray}
	R(A \otimes B) = |A\rangle \langle B^*| \label{ketbraR}
\end{eqnarray}
CCN or Realignment criterion provides a necessary condition for the separability of a quantum state. It is based on the matrix norm.  The matrix realignment criterion is stated as follows:
{\theorem(\textit{Realignment criteria}) \label{thm-ccnr} If a bipartite state $\rho_{AB}$ is separable then $||R(\rho_{AB} )||_1 \leq 1$ where $||.||_1$ defines the trace norm.
	The violation of this inequality indicates that the state $\rho_{AB}$ is entangled.}
\begin{proof}
	Let $\rho_{AB}$ be a separable state given by
	\begin{eqnarray}
		\rho_{AB} = \sum_i p_i \rho_i^A \otimes \rho_i^B \;\; \text{with}\;\; p_i \geq 0 \;\; \text{and} \;\; \sum_i p_i =1
	\end{eqnarray} 
	Using the linearity of the realignment operation, we have
	\begin{eqnarray}
		R(\rho_{AB} ) = \sum_i p_i R(\rho_i^A \otimes \rho_i^B) \;\;
	\end{eqnarray}
	By convexity of trace norm, we have
	\begin{eqnarray}
		||R(\rho_{AB} )||_1 \leq \sum_i p_i ||R(\rho_i^A \otimes \rho_i^B)||_1 \label{eqn1}
	\end{eqnarray}
	Using (\ref{ketbraR}), the realignment of a product state $\rho_i^A \otimes \rho_i^B$ is given by 
	\begin{eqnarray}
		R(\rho_i^A \otimes \rho_i^B) = |\rho_i^A\rangle \langle (\rho_i^B)^*|
	\end{eqnarray}
	By definition of trace norm (\ref{tracenorm}), we have
	\begin{eqnarray}
		|| R(\rho_i^A \otimes \rho_i^B)||_1 = Tr\left[\sqrt{|\rho_i^A\rangle \langle (\rho_i^B)^*| (\rho_i^B)^*\rangle \langle {\rho_i^A}|}\right] \leq 1 \label{eqn2}
	\end{eqnarray}
	where the inequality follows from the Hilbert Schmidt inner product $\langle X | Y \rangle = Tr[X^\dagger Y]$.  
	Using (\ref{eqn1}) and (\ref{eqn2}), we have $||R(\rho_{AB} )||_1 \leq 1$ for any bipartite separable state $\rho_{AB}$. 
\end{proof}
It has been shown that the PPT criterion and the CCNR criterion are equivalent under permutations of the density matrix's indices \cite{rudolph2004}. The generalization of the CCNR criterion was investigated in \cite{chen2}. The symmetric function of Schmidt coefficients has been used to improve the CCNR criterion in \cite{lupo,ckli2011}. Separability criteria based on the realignment of density matrices and reduced density matrices have been proposed in \cite{shen}.
In \cite{XQi}, the rank of the realigned matrix is used to obtain necessary and sufficient product criteria for quantum states.
Recently, methods for detecting bipartite entanglement based on estimating moments of the realignment matrix have been proposed \cite{tzhang}. 
\subsection{Entanglement detection via positive but not completely positive map} \label{sec-pncp}
Positive and completely positive maps are defined in (\ref{posmap}) and (\ref{cp map}). Entanglement detection criteria can be formulated from positive, but not
completely positive maps. For any separable state $\rho_{AB} \in B(\mathcal{H}_A \otimes \mathcal{H}_B)$ and any positive map $\Lambda: B(\mathcal{H}_A) \rightarrow  B(\mathcal{H}_B)$, we have 
\begin{eqnarray}
	(I_A \otimes \Lambda) (\rho_{AB}) \geq 0 \label{Itensorgamma}
\end{eqnarray}
Moreover, it has been shown that a state is separable if and only if for all positive maps $\Lambda$, the above relation holds. In this sense, the separability problem is equivalent to the classification of all positive maps. This has led to the construction of many positive, but not completely positive maps, resulting in strong separability criteria \cite{terhal2000,chruscinski2009,collins2018}. 
{\example(\textit{PPT criterion})
	The PPT criterion discussed in Sec-\ref{ptmethod} is an example of a criterion using a positive, but not completely positive map. Since matrix transpose preserves positivity, taking $\Lambda$ as the matrix transposition in (\ref{Itensorgamma}), the map $I_A \otimes \Lambda$ represents the partial transposition operation.
\example (\textit{Reduction criteria})  Another example of a positive, but not completely positive map is the reduction map, which is defined as follows \cite{mphoro1999}.
\begin{eqnarray}
	\Lambda^{Red} (X) = Tr[X] I -X \;\;\text{for any}\;\; X \in B(\mathcal{H}_A)
\end{eqnarray} 
If a bipartite state described by a density operator $\rho_{AB}$ is separable then $	(I_A \otimes \Lambda^{Red}) (\rho_{AB}) = \rho_A \otimes I_B - \rho_{AB}\geq 0$. where $\rho_A = Tr_B(\rho_{AB})$ is the reduced state defined in (\ref{reducedop}).
\remark The reduction criterion is weaker than the PPT criteria because $I_A \otimes \Lambda^{Red}$ can never detect entanglement unless $I_A \otimes T$ detects it. It provides a necessary and sufficient condition for separability
in the case of a two-qubit quantum system.}
\subsection{Majorization Criterion}
Majorization criterion relates the eigenvalues of a state with its reduced states \cite{nielson2001}. For a general  state represented by the density matrix $\rho_{AB}$, let $P = \{p_1, p_2, . . . \}$ be the set of eigenvalues of $\rho_{AB}$ arranged in decreasing order; and $Q = \{q_1, q_2, . . . \}$ be the set of eigenvalues of the reduced state $\rho_A$ arranged in decreasing order. The majorization criterion states
that if the state $\rho_{AB}$ is separable, then the following inequality holds
\begin{eqnarray}
	\sum p_i \leq \sum q_i ,\;\;\;\;\forall\; i
\end{eqnarray}
The same inequality holds when $\rho_A$ is replaced by $\rho_B$.
It was shown that the majorization
criterion follows from the reduction criterion and hence any state which can be detected by the majorization criterion, can
also be detected by the reduction criterion and, consequently, by the PPT criterion \cite{hiroshima2003}.
\subsection{Covariance matrix criterion } Another separability criterion is based on the covariance matrix (CM).
CMs constitute a well-established and powerful tool for the detection of entanglement in the infinite-dimensional setting (in particular, for Gaussian states). Strong separability criteria based on covariance matrix for finite dimensional systems have been proposed in the literature \cite{guhnerev,guhne2007}. 
{\definition \label{covmatrix}(\textit{Covariance matrix \cite{guhne2007}:}) Let $\rho$ be a given quantum state and let $\{M_1, M_2, . . ., M_n\}$ be some observables. Then the entries of the covariance matrix denoted by ${Cov}(\rho, \{M_k\}) = [c_{ij}]_{n\times n}$ are given by 
	\begin{eqnarray*}
		c_{ij} =\frac{\langle M_i M_j \rangle_{\rho} + \langle M_j M_i \rangle_{\rho} }{2} - \langle M_i \rangle_{\rho} \langle M_j \rangle_{\rho} \label{eqcm}
	\end{eqnarray*}
	\remark The covariance matrix depends on the state $\rho$ and the observables $\{M_k\}$.}

Let $\rho_{AB}$ be a bipartite state defined on the Hilbert space $\mathcal{H} = \mathcal{H}_A \otimes \mathcal{H}_B$ where $\mathcal{H}_A$ and $\mathcal{H}_B$ has dimension $d_A$ and $d_b$ respectively. We can choose $d^2_A$ observables $\{A_k\}$ in $\mathcal{H}_A$ forming an orthonormal basis of the observable space. For example, we may choose Pauli matrices along with identity for a qubit. Similarly, we take $\{B_k\}$ as a basis of observables in $\mathcal{H}_B$. Considering the total set $\{M_k\} = \{A_k \otimes I, I \otimes B_k\}$, the block matrix form of the CM is given as
\begin{eqnarray}
	{Cov}(\rho_{AB}, \{M_k\})=
	\begin{pmatrix}
		A & C \\
		C^T & B 
	\end{pmatrix} \label{cmblock}
\end{eqnarray}
where $A={Cov}(\rho_A, \{A_k\})$, $B={Cov}(\rho_B, \{B_k\})$ and $C=[C_{ij}]$ with entries $C_{ij} = \langle A_i \otimes B_j\rangle_{\rho_{AB}} - \langle A_i \rangle_{\rho_A} \langle B_j \rangle_{\rho_B}$. The separability criteria based on CM is defined as follows.\\
\textbf{Covariance matrix criterion (CMC).} Let ${Cov}(\rho_{AB}, \{M_k\})$ be a CM as defined in (\ref{cmblock}). If $\rho_{AB}$ is separable, then there exist states $|a_k\rangle \langle a_k|$ on $\mathcal{H}_A$ and $|b_k\rangle \langle b_k|$ on $\mathcal{H}_B$ such that for $\kappa_A = \sum_k p_k  Cov(|a_k\rangle \langle a_k|)$ and $\kappa_B = \sum_k p_k  Cov(|b_k\rangle \langle b_k|)$, where $\sum_k p_k =1$, the following inequality holds \cite{guhne2007}.
\begin{eqnarray}
	{Cov}(\rho_{AB}, \{M_k\}) = \kappa_A \oplus \kappa_B
\end{eqnarray}
If no such $\kappa_A$ or $\kappa_B$ exist then $\rho_{AB}$ is entangled.\\
Although the criterion is independent of the choice of the observables, a certification of violation of the criteria is simplified by choosing the Schmidt basis in operator space or by using an appropriate local filtering \cite{guhne2007}.
CMC is strong in the sense that it detects many bound entangled states, and, with the help of filtering
operations it can detect all entangled states for two qubits, just as the PPT criterion \cite{guhnerev}.
Moreover, the CCNR criterion can be shown to be a corollary of the CMC \cite{guhnerev}. 
\subsection{Entanglement detection criteria based on correlation tensor}
{\definition (\textit{Correlation matrix} :) For a bipartite state $\rho_{AB}$ defined on $\mathcal{H}=\mathcal{H}_A \otimes \mathcal{H}_B$, where $\mathcal{H}_A$ and $\mathcal{H}_B$ have dimensions $d_A$ and $d_B$, respectively, the correlation matrix is defined as \cite{sarbicki}
	\begin{equation}
		C_{a,b} = Tr[\rho_{AB} G_{a}^A \otimes G_{b}^B] \label{cor}
	\end{equation}
	where $G_{a}^A$ and $G_{b}^B$ denote arbitrary orthonormal basis in $B(\mathcal{H}_A)$ and $B(\mathcal{H}_B)$. In particular, if $(G_{a}^A)_{a\neq 0}$, $(G_{b}^B)_{b\neq 0}$ represent orthonormal traceless operators and $G_{0}^A=\frac{I_{A}}{\sqrt{d_{A}}}$, $G_{0}^B=\frac{I_{B}}{\sqrt{d_{B}}}$ then $\frac{1}{\sqrt{d_{A}}}G_{a}^A \in B(\mathcal{H}_A)$ and $\frac{1}{\sqrt{d_{B}}}G_{b}^B \in B(\mathcal{H}_B)$ forms a canonical basis.}\\
A correlation matrix may be used to develop some entanglement detection criteria. A few of them are given below.
\subsubsection{I. CCNR criterion based on correlation tensor}
If $\rho_{AB}$ is separable, then CCNR criterion reduces to the following bound \cite{rudolph2003}
\begin{equation}
	||C_{a,b}||_1  \leq 1
\end{equation}
where $C_{a,b}$ denote the correlation matrix defined in (\ref{cor}) and $||.||_1$ denotes the trace norm. 
\subsubsection{II. de Vicente (dV) Criterion}
de Vicente (dV) criterion states that if the state $\rho_{AB}$ in $d_{A}\otimes d_{B}$ dimensional system is separable then the correlation matrix $C_{a,b}$ defined in (\ref{cor}) for the state $\rho_{AB}$ satisfies the inequality \cite{dv}
\begin{equation}
	||C_{a,b}||_1 \leq \frac{\sqrt{d_A d_B (d_A -1)(d_B - 1)}}{2}
\end{equation}
\subsubsection{III. Correlation Tensor (CT) Criterion}
CT criterion \cite{sarbicki} states that if $\rho_{AB}$ is separable, then
\begin{equation}
	||D_{x}^A C^{can} D_{y}^B||_1 \leq \mathcal{N}_A (x) \mathcal{N}_B (y) \label{ct}
\end{equation}
where $D_x^A =$ diag($x, 1, 1, . . ., 1$), $D_y^B =$ diag($y, 1, 1, . . ., 1$),
$\mathcal{N}_A (x)= \sqrt{\frac{d_A - 1 + x^2}{d_A}}$, $\mathcal{N}_B (y)= \sqrt{\frac{d_B - 1 + y^2}{d_B}}$, $x,y \geq 0$, and $C^{can}$ is the correlation matrix defined by canonical basis.
\subsection{Detection of entanglement through witness operator} \label{sec-wo}
The entanglement detection methods we discussed so far require applying certain operations to a density matrix, to decide whether the state is entangled or not. However, there is a necessary and sufficient entanglement detection criterion in terms of directly measurable observables. These observables are called entanglement witnesses (or simply witnesses).
Mathematically, the entanglement witness operator is defined as
a Hermitian operator $\mathbf{W}$ that satisfies the following properties:
\begin{enumerate}
	\item[(i)] $Tr[\mathbf{W} \rho^{sep} ] \geq 0$ for all separable states $\rho^{sep}$
	\item[(ii)] $Tr[\mathbf{W} \rho^{ent} ] < 0$ for atleast one entangled state $\rho^{ent}$
\end{enumerate}
{\remark $Tr[\mathbf{W} \rho ] < 0$ implies that the state $\rho$ is entangled and detected by the witness operator $\mathbf{W}$. }

	\begin{figure}[h!]
		\begin{center}
		\includegraphics[width=0.7\textwidth]{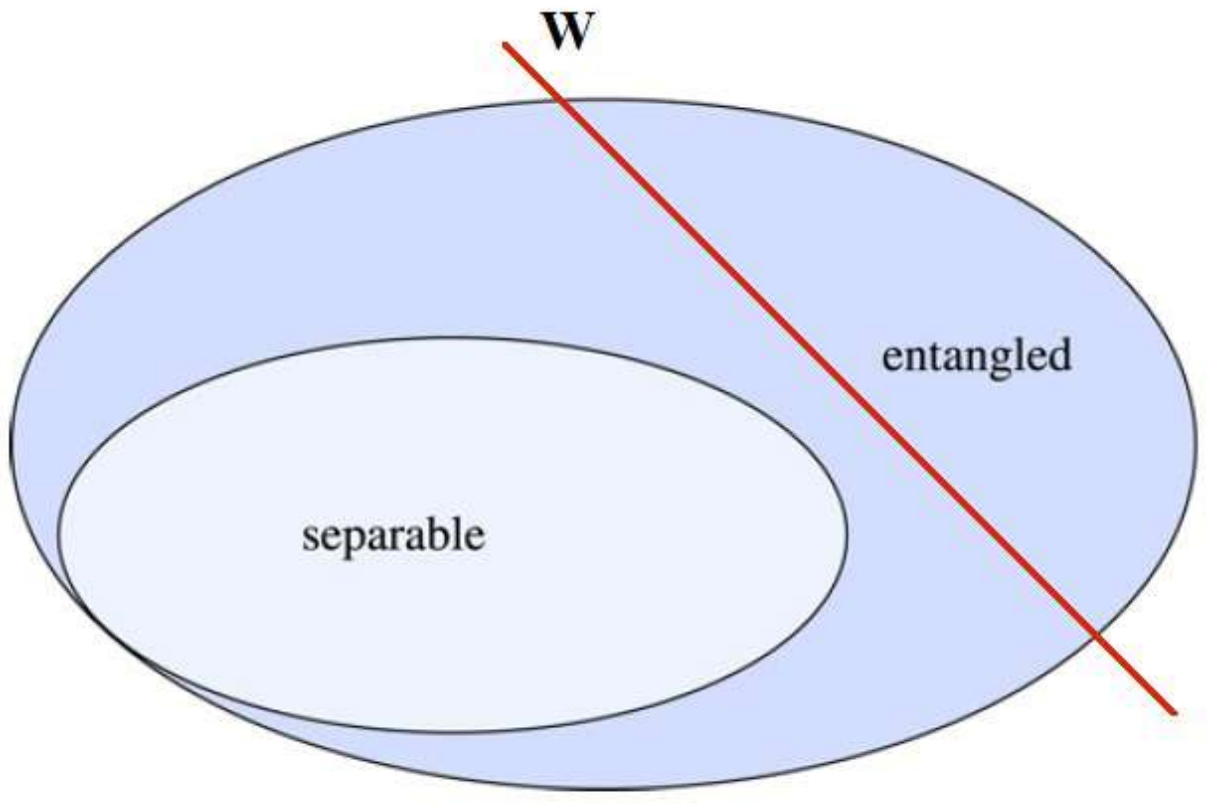}
		\caption{Schematic diagram of the set of all states; the set of separable states (lying in the light blue colour region) as a convex subset of the set of all states. The red line represents the hyperplane described by the witness operator $\mathbf{W}$. The equation of the hyperplane is given by $Tr[\mathbf{W} \rho] = 0$ where the state described by the density operator $\rho$ lies on the hyperplane.}
	\end{center}
	\end{figure}

The expectation value of an observable depends linearly on the state. Thus, geometrically the set of states where $Tr[\mathbf{W}\rho]= 0$ represents a hyperplane in the set of all states. The hyperplane divides the set of all states into two parts. A part with $Tr[\mathbf{W}\rho] \geq 0$ indicates the set of all separable states while the other part with $Tr[\mathbf{W}\rho] < 0$ represents the set of all entangled states detected by $\mathbf{W}$. \\
The existence of entanglement witnesses is a consequence of the Hahn-Banach theorem \cite{rudin} in functional analysis providing a necessary and sufficient condition to detect entanglement. A special geometric variant of the Hahn-Banach theorem states that if a set is convex and compact, then a point lying outside the set can be separated from it by a hyperplane. Since the set of separable states is convex and compact, for each entangled state, there exists an entanglement witness detecting it.
Although this theorem ensures that any entangled state can in principle be detected with an entanglement witness, the
construction of a witness operator is not an easy task. \\
Let us explore the construction of a witness operator from partial transposition for bipartite entangled states.\\
Consider an NPT entangled state denoted by $\rho_{e}$. Then by PPT criteria, we have  $\lambda_{min} (\rho_e^{T_B}) < 0$ where  $\lambda_{min} (\rho_e^{T_B})$ denotes the minimum eigenvalue of the partially transposed matrix $\rho_e^{T_B}$. Define an operator ${W}^{(PT)}$ as
 \begin{eqnarray}
 	{W}^{(PT)} = (|v\rangle \langle v|)^{T_B} \label{wpt}
 \end{eqnarray}
where $|v\rangle$ denotes the normalized eigenvector corresponding to the eigenvalue $\lambda_{min} (\rho_e^{T_B})$. 
{\theorem ${W}^{(PT)}$ is an entanglement witness operator detecting the state $\rho_e$}
\begin{proof}
	For any separable state $\rho^{sep}$, since  $(\rho^{sep})^{T_B} \geq 0$, we have 
	\begin{eqnarray}
		Tr[{W}^{(PT)}\rho^{sep}] &=& Tr[ (|v\rangle \langle v|)^{T_B}\rho^{sep} ]\nonumber \\&=& Tr[ |v\rangle \langle v|(\rho^{sep})^{T_B} ] = \langle v| (\rho^{sep})^{T_B}|v \rangle \geq 0
	\end{eqnarray}
	where the second equality follows from the identity $Tr[X^{T_B} Y] = Tr[X Y^{T_B}]$ for any matrix $X$ and $Y$.
	Now the expectation value of the witness ${W}^{(PT)}$ for the entangled state $\rho_e$ is given as
	\begin{eqnarray}
		Tr[{W}^{(PT)}\rho_{e}] = Tr[ |v\rangle \langle v|^{T_B}\rho_{e} ] = Tr[ |v\rangle \langle v|\rho_{e}^{T_B} ] = \langle v| \rho_{e}^{T_B}|v \rangle = \lambda_{min} (\rho_e^{T_B}) < 0
	\end{eqnarray}
Thus, ${W}^{(PT)}$ satisfies all the properties of a witness operator. Hence, ${W}^{(PT)}$ is a witness detecting $\rho_e$. 
\end{proof}
There are two classes of witness operators: Decomposable witness operators and indecomposable witness operators. The former detects only NPTES and the latter class of witness operators detects BES together with NPTES. Terhal \cite{terhal1} first introduced a family of indecomposable positive linear maps based on entangled quantum states using the notion of an unextendible product basis. Soon after this work, Lewenstein et al. \cite{mlewenstein} extensively studied the indecomposable witness operator and provided an algorithm to optimize them. The construction of witness operators is important in the sense that they can be used in an experimental setup to detect whether the generated state in an experiment is entangled. There are different methods of the construction of witness operator in the literature \cite{ganguly,adhikari2,sarbicki,halder,wiesniak}.
Witness operators can also be used in the detection of entangled states that act as a resource state in the teleportation protocol \cite{ganguly1,adhikari2}. Witness operator is also used in the quantification of entanglement \cite{brandao} and in the estimation of the lower bound of the concurrence of the entangled states in $d_{1}\otimes d_{2} (d_{1}\leq d_{2})$ dimensional systems \cite{mintert2,kchen}.


\subsection{Moment based methods for the detection of entangled states} \label{sec-ptmoment}
If an experiment aims at producing a specific quantum
state with few particles, entanglement witnesses or Bell
inequalities provide mature tools for entanglement detection.
However, these methods require a significant number of measurements for larger and noisy systems. Moreover,
some of the standard constructions of witnesses are not
very powerful. To overcome this, methods using locally
randomized measurements have been put forward. In
these schemes, one performs measurements on the particles 
in random bases and determines the moments from
the resulting probability distribution. It was noted earlier
that this approach allows one to detect entanglement
through the estimation of the moments of the density matrix \cite{sjvan}.
Recently, this approach has become the center of attention
and found experimental applications, for instance, it was
shown that, with these methods, entropies can be estimated
\cite{brydges}, different forms of multiparticle entanglement
can be characterized \cite{kett19}, and bound
entanglement can be detected \cite{simai}. A few moment based entanglement detection criteria discussed in the literature are given below.
\subsubsection{I. $p_3$-PPT criterion}
Elben et al. \cite{elben} proposed a method for detecting bipartite entanglement in a many-body mixed state based on estimating moments of the partially transposed density matrix. Since partial transposition operation is a positive but not completely positive map, it is not a physically realizable map and thus it may not be implemented in the experiment. But in spite of the above difficulty in realizing the partial transposition operation in an experiment, the measurement of their moments is possible. A condition to detect entanglement called $p_3$-PPT criterion, was proposed using the first three PT moments. If $\tau$ denotes the partial transposition operation, then the $k^{th}$ moment defined by $ p_k(\rho^{\tau})=Tr[(\rho^{\tau})^k]$  has the advantage that it can be estimated using shadow tomography in a more efficient way than if one had to reconstruct the state $\rho$ via full quantum state tomography \cite{elben}.
The $p_3$-PPT condition states that any PPT state $\rho$ satisfies the following inequality
\begin{equation}
	{L}_1 \equiv (p_2(\rho^{\tau}))^2 - 	p_3(\rho^{\tau}) p_1(\rho^{\tau}) \leq 0 \label{p3ppt}
\end{equation}
The violation of above inequality (\ref{p3ppt}) by any $d_1 \otimes d_2$ dimensional bipartite state $\rho$ indicates that the state is an NPT entangled state.
\subsubsection{II. Criteria based on $D_3^{(in)}$ inequality}
Neven et al. \cite{neven} proposed a set of inequalities, known as $D_k^{(in)}$ inequalities to detect bipartite NPT entangled states. Each $D_k^{(in)}$ involves the first $k$ moments of the partially transposed operator $\rho^{\tau}$. The violation of any single $D_k^{(in)}$ inequality implies $\rho$ is NPT entangled state. The inequalities $D_1^{(in)}, D_2^{(in)}, D_3^{(in)}, D_4^{(in)}$ are given below.
\begin{eqnarray}
	&D_1^{(in)}:& p_1(\rho^{\tau}) \geq 0 \\
	&D_2^{(in)}:& (p_1(\rho^{\tau}))^2 - p_2(\rho^{\tau})  \geq 0 \\
	&D_3^{(in)}:& p_3(\rho^{\tau}) + \frac{1}{2} (p_1(\rho^{\tau}))^3 - \frac{3}{2} p_1(\rho^{\tau})p_2(\rho^{\tau}) \geq 0\\
	&D_4^{(in)}:&  \left(\frac{1}{2} (p_1(\rho^{\tau}))^2 - p_2(\rho^{\tau})\right)^2 - \frac{1}{3} (p_1(\rho^{\tau}))^4 +\frac{4}{3} p_1(\rho^{\tau})p_3(\rho^{\tau}) \geq 0
\end{eqnarray}
In a similar way, the inequality involves up to $k$th moment which is also known as $D_k^{(in)}$ inequality may be constructed.
 Since, $p_1(\rho^{\tau}) = 1$ for any quantum state $\rho$, it implies that the inequality $D_1^{(in)}$ is trivially satisfied. Similarly, for any quantum state $\rho$, we have $p_2(\rho^{\tau})=p_2(\rho)$, where $p_2(\rho)$ is the purity of $\rho$, the inequality $D_2^{(in)}$ is
also trivially satisfied \cite{neven}. Therefore, one can obtain the first non-trivial condition in the form of $D_3^{(in)}$ inequality, which may be rewritten as
\begin{equation}
	{L}_2 \equiv \frac{3}{2} (p_1(\rho^{\tau}))(p_2(\rho^{\tau})) - \frac{1}{2}(p_1(\rho^{\tau}))^3 - p_3(\rho^{\tau})  \leq 0
	\label{d3}
\end{equation}
The inequality (\ref{d3}) is satisfied by all PPT states and its violation certifies the existence of NPTES, which may be expressed as 
\begin{equation}
	{L}_2  > 0.
	\label{d3v}
\end{equation}
Neven et al. \cite{neven} showed that knowing only the first three moments $p_1(\rho^{\tau})$, $p_2(\rho^{\tau})$ and $p_3(\rho^{\tau})$, the above inequality   detects more entangled states than the $p_3$-PPT criterion when the purity of $\rho^{\tau}$ is greater than or equal to $1/2$, i.e., when $1/2 \leq p_2(\rho^{\tau}) \leq 1$. In the other region, i.e., when $0 \leq p_2(\rho^{\tau}) < 1/2$, the $p_3$-PPT criterion detects more entangled states than the $D_3^{(in)}$ criterion.
\subsubsection{III. $p_3$-OPPT criterion}
Yu et. al. \cite{guhne2021} introduced an optimal entanglement detection criteria based on partial moments called $p_3$-Optimal Positive Partial Transpose ($p_3$-OPPT) criterion. This optimal separability criterion can be stated as follows. For any bipartite state $\rho_{AB} \in B(\mathcal{H}_A \otimes \mathcal{H}_B)$, the following inequality holds:
\begin{eqnarray}
	{L}_3 =	\mu x^3 + (1-\mu x)^3 - p_3 \leq 0
	\label{l3}
\end{eqnarray}
where $x= \frac{\mu + \sqrt{\mu [p_2(\mu + 1 )-1]}}{\mu(\mu+1)}$ and $\mu = \lfloor \frac{1}{p_2} \rfloor$.

\subsubsection{IV. Criteria based on moments of realignment matrix}
Zhang et al. \cite{tzhang} proposed another entanglement detection
criteria in terms of the quantities called realignment moments. To derive their entanglement detection criterion, they have defined the realignment moments for a $d \otimes d$ dimensional bipartite state  $\rho_{AB}$ as 
\begin{eqnarray}
	r_k [R(\rho_{AB} )] = Tr[R(\rho_{AB} ) (R(\rho_{AB} ))^{\dagger}]^{k/2},\; k = 1, 2, . . ., d^2 \label{eq-r_k}
\end{eqnarray}
where $d^2$ is order of the matrix $R(\rho_{AB})$.\\
The separability criterion based on realignment moments $r_2$ and $r_3$ is stated as follows. If a quantum state $\rho$ is separable, then
\begin{eqnarray}
	{L}_4 \equiv (r_2[R(\rho_{AB} )]^2 - r_3[R(\rho_{AB} )] \leq 0	\label{L_4}
\end{eqnarray}
Violation of the inequality (\ref{L_4}), i.e., ${L}_4 > 0$ implies that the state $\rho$ is entangled.\\
A stronger separability criterion based on Hankel matrices and involving higher order $r_k$ has been derived in \cite{tzhang}.
For $r=(r_0, r_1, r_2, . . ., r_{d^2})$, Hankel matrices can be constructed as  $[{H}_k (r)]_{ij} = r_{i+j}$ and $[{B}_l (r)]_{ij}=r_{i+j+1}$ for $i,j=0,1,2,. . ., k$. The criterion may be stated as follows.
If $\rho_{AB}$ is separable, then for $k = 1, 2, . . ., \lfloor \frac{d^2}{2} \rfloor$  and $l = 1, 2, . . ., \lfloor \frac{d^2 - 1}{2} \rfloor$ with $r_1[R(\rho_{AB} )] = 1$, we have 
\begin{eqnarray}
	\widehat{H}_k (r) &=& [r_{i+j}[R(\rho)]] \geq 0 \label{hk} \\ \widehat{B}_l (r)&=& [r_{i+j+1} [R(\rho)]] \geq 0
	\label{bl}
\end{eqnarray}
\begin{center}
****************
\end{center}

\chapter{Entanglement Detection via Positive but not Completely Positive Map}\label{ch2}
\vspace{1cm}
	{\small \emph{ "I cannot seriously believe in [the quantum theory] because it cannot be
		reconciled with the idea that physics should represent a reality in time
		and space, free from spooky actions at a distance."\\
		- Albert Einstein}}
\vspace{1cm}
\noindent \hrule
\noindent\emph{ In this chapter\;\footnote { This chapter is based on the published research article, ``S. Aggarwal, S. Adhikari, \emph{Search for an efficient entanglement witness operator for bound entangled states in bipartite quantum systems},  Annals of Physics, {\bf{444}}, 169043, (2022)"} we take an analytical approach to construct a family of witness operators detecting NPTES and BES in arbitrary dimensional bipartite quantum systems. To construct the family of witnesses, we first construct a linear map using partial transposition and realignment operation. Then we find some conditions on the parameters of the map for which the map is positive but not completely positive. Finally, we prove its efficiency by detecting several bipartite bound entangled states that were previously undetected by some well-known separability criteria, namely, the dV criterion, CCNR criterion, and the separability criteria based on correlation tensor (CT) proposed by Sarbicki et al and find that our witness operator detects bound entangled states undetected by these criteria.}
\noindent \hrulefill
\newpage
The entanglement detection problem is one of the important problems in quantum information theory. Gurvit \cite{gurvit, gurvit1,huang2014} showed that this problem is NP-complete and thus this may be the possible reason that only one criterion is not sufficient to detect all entangled states. In 1996, Peres designed the positive partial transpose (PPT) criterion to detect entanglement \cite{peres}. After a few years, another strong criterion for separability called the computable cross norm or realignment (CCNR) criterion was found by Rudolph \cite{rudolph2003,brussperes} and Chen et al. \cite{chenwu}. The realignment criterion is independent of the PPT criterion and turns out to be strong enough to detect bound entanglement, for which the former criterion fails.  
More criteria discussed in section \ref{sec-entdet}, such as range criterion, majorization criterion, and covariance matrix criterion have been developed for detecting the PPTES. 
Although there is a lot of progress in this line of research, till now there does not exist any necessary and sufficient condition for entanglement detection in multipartite and higher dimensional bipartite systems. This may be due to the fact that in higher dimensional systems, the eigenspectrum of partial transposition of entangled state contains positive eigenvalues. These types of states are  PPTES or BES as discussed in the sections \ref{sec-bes} and \ref{ptmethod}. 
In spite of these criteria, there exist other approaches such as the construction of a positive but not completely positive map discussed in section \ref{sec-pncp}, which may also help in the detection of entangled states. 
In this chapter, we construct a witness operator by employing the theoretical method based on the linear maps which are positive but not completely positive, also referred to as PnCP maps. Positive linear maps are useful in entanglement detection \cite{terhal2000}. The simplest example of a PnCP map is the partial transposition (PT) map. Any positive, but not completely positive map gives rise to a construction of entanglement witnesses via the Choi-Jamiolkowski isomorphism \cite{jamio}. The entanglement witness operator discussed in section \ref{sec-wo} plays a significant role in the entanglement detection problem since if we have some prior partial information about the state that is to be detected then entanglement can be detected in the experiment by the construction of witness operator \cite{guhnerev}. We also compared the detection power of our witness operator with three well-known powerful entanglement detection criteria, namely, dV criterion \cite{dv}, CCNR criterion \cite{rudolph2003}, and the separability criteria based on correlation tensor (CT) proposed by Sarbicki et. al. \cite{sarbicki} and find that our witness operator detects more entangled states than these criteria.

\section{Construction of a Linear Map}
To start with, let us consider a map $\Phi_{\alpha,\beta}: M_{n}{(\mathbb{C})} \rightarrow M_{n}{(\mathbb{C})} $, $n\geq2$ defined as
\begin{equation}
	\Phi_{\alpha,\beta}(A) = \alpha  A^{T_B} + \beta R(A)
	\label{phi}
\end{equation}
where $A \in M_n{(\mathbb{C})}$, $\alpha,\beta \in \mathbb{R^{+}}$. $R$ represents the realignment operation and $T_B$ denote the partial transposition operation with respect to subsystem $B$. The parameters $\alpha$ and $\beta$ are chosen in such a way that the map $\Phi$ represents a positive map and then we fix the parameters $\alpha$ and $\beta$.\\
Now our prime task is to search for the parameters $\alpha$ and $\beta$ for which the map $\Phi_{\alpha,\beta}$ will be positive. Once chosen the parameters $\alpha$ and $\beta$, we then move on to investigate whether the positive map will be completely positive for such $\alpha$ and $\beta$.
\subsection{Positivity of the map $\Phi_{\alpha,\beta}$}
Let $A \in M_n{(\mathbb{C})}$ be a positive semi-definite matrix. If $\Phi_{\alpha,\beta}(A)$ is positive for some $\alpha, \beta \in \mathbb{R^{+}}$ then we can say that the map $\Phi_{\alpha,\beta}$ represent a positive map. To show the positivity of the map, it is enough to show $\lambda_{min}(\Phi_{\alpha,\beta}(A)) \geq 0$ for $A\geq 0$. Therefore, $\Phi_{\alpha,\beta}$ will become a positive map if for some fixed $\alpha,\beta \in \mathbb{R^{+}}$, we have
\begin{equation}
	\lambda_{min}(\alpha  A^{T_B} + \beta  R(A)) \geq 0
	\label{mineigen}
\end{equation}
\subsubsection{Determination of the parameters $\alpha$ and $\beta$}
Let us consider $d_{1} \otimes d_{2}$ dimensional quantum system which can be described by the density operator, say, $\varrho_{AB}$. Thus, we consider the domain of the map as the set $D$ of all bipartite density matrices $\varrho_{AB} \in M_{d_{1}d_{2}}{(\mathbb{C})}$ such that $R(\varrho_{AB})$ is Hermitian. Therefore, the map $\Phi_{\alpha,\beta}$ given in (\ref{phi}) can be re-expressed as
\begin{equation}
	\Phi_{\alpha,\beta}(\varrho_{AB}) = \alpha \varrho_{AB}^{T_B} + \beta  R(\varrho_{AB})
	\label{phiredefine}
\end{equation}
The map $\Phi_{\alpha,\beta}(\varrho_{AB})$ can be shown to be positive by considering two different cases. In the first case, we assume that the state $\varrho_{AB} \in M_{d_{1}d_{2}}{(\mathbb{C})}$ is PPT and in the second case, we will consider $\varrho_{AB}$ as an NPT entangled state.\\
\textbf{Case-I:} When $\varrho_{AB}$ is a PPT state.\\
Since $\varrho_{AB}$ is a PPT state so $\lambda_{min}(\varrho_{AB}^{T_{B}})\geq 0$, where $\lambda_{min}(.)$ denote the minimum eigenvalue of $(.)$. Thus, using (\ref{phiredefine}) and Weyl's inequality given in (\ref{weyl}), we can write
\begin{equation}
	\lambda_{min}(\Phi_{\alpha,\beta}(\varrho_{AB})) \geq \alpha \lambda_{min}(\varrho_{AB}^{T_B}) + \beta \lambda_{min}(R(\varrho_{AB}))
	\label{mineig1}
\end{equation}
(i) If $\lambda_{min}(R(\varrho_{AB}))$ is non-negative then $\lambda_{min}\{\Phi_{\alpha,\beta}(\varrho_{AB})\}\geq 0$ for all $\alpha, \beta > 0$. Thus the map $\Phi_{\alpha,\beta}$ is positive $ \forall \alpha, \beta > 0$.\\
(ii) If $\lambda_{min}(R(\varrho_{AB}))$ is negative then we can choose the parameters $\alpha$ and $\beta$ in such a way so that  $\lambda_{min}\{\Phi_{\alpha,\beta}(\varrho_{AB})\}\geq 0$. Therefore, the chosen parameters $\alpha$ and $\beta$ should satisfy the inequality
\begin{equation}
	\frac{\alpha}{\beta} \geq \frac{\lambda_{min}(R(\varrho_{AB}))}{\lambda_{min}(\varrho_{AB}^{T_B})}
	\label{alphabeta1}
\end{equation}
Hence, if $\lambda_{min}(R(\varrho_{AB}))$ is negative then the map  $\Phi_{\alpha,\beta}$ is positive for some parameter $\alpha$ and $\beta$ that satisfies the inequality (\ref{alphabeta1}).\\
\textbf{Case-II:} When $\varrho_{AB}$ is an NPT entangled state.\\
If $\varrho_{AB}$ is an NPT entangled state then $\lambda_{min}(\varrho_{AB}^{T_{B}}) < 0$. Thus, from (\ref{mineig1}), we conclude the following fact:\\
(i) If $\lambda_{min}(R(\varrho_{AB}))$ is non-negative then $\lambda_{min}\{\Phi_{\alpha,\beta}(\varrho_{AB})\}\geq 0$ when the parameter $\alpha, \beta$ satisfy the inequality
\begin{equation}
	\frac{\alpha}{\beta} \leq \frac{\lambda_{min}(R(\varrho_{AB}))}{\lambda_{min}(\varrho_{AB}^{T_B})}
	\label{alphabeta2}
\end{equation}
therefore, the map $\Phi_{\alpha,\beta}$ is positive if $\alpha,\beta$ satisfies (\ref{alphabeta2}).\\
(ii) If $\lambda_{min}(R(\varrho_{AB}))$ is negative then the map  $\Phi_{\alpha,\beta}$ never be a positive map for any $\alpha, \beta > 0$.
\subsubsection{Illustration}
\noindent Here, we will analyse the positivity of the map $\Phi_{\alpha,\beta}: M_{4}{(\mathbb{C})} \rightarrow M_{4}{(\mathbb{C})}$. For this, we start with a two-qubit maximally mixed marginals state which is given by \cite{luo2008, rau2010}
\begin{equation}
	\varrho = \frac{1}{4} [I_2 \otimes I_2 + \sum_{j=1}^{3} t_j \sigma_j \otimes \sigma_j]
\end{equation}
where $I_{2}$ denote the identity matrix of order 2 and $\sigma_{i},i=1,2,3$ denote the Pauli matrices.\\
The matrix representation of $\varrho$ in the computational basis as
\begin{equation}
	\varrho=\frac{1}{4}
	\begin{pmatrix}
		1+t_3 & 0 & 0 & t_1 - t_2 \\
		0 & 1-t_3 & t_1+t_2 & 0\\
		0 & t_1+t_2 &1-t_3 & 0 \\
		t_1-t_2 & 0 & 0 & 1+t_3 \\
	\end{pmatrix}
\end{equation}
The eigenvalues of $\varrho$ are non-negative if the following inequalities holds \cite{ali2010}
\begin{eqnarray}
	(1 - t_3)^2 \geq (t_1 + t_2)^2 \quad \text{and}\quad
	(1 + t_3)^2 \geq (t_1 - t_2)^2
	\label{non-negeigenvalue}
\end{eqnarray}
The matrix representation of partial transposed state $\varrho^{T_{B}}$ is given by
\begin{equation}
	\varrho^{T_{B}}=\frac{1}{4}
	\begin{pmatrix}
		1+t_3 & 0 & 0 & t_1 + t_2 \\
		0 & 1-t_3 & t_1-t_2 & 0\\
		0 & t_1-t_2 &1-t_3 & 0 \\
		t_1+t_2 & 0 & 0 & 1+t_3 \\
	\end{pmatrix}
\end{equation}
If $\lambda_{min}(.)$ denote the minimum eigenvalue of $(.)$ then the minimum eigenvalue of $\varrho^{T_{B}}$ is given by
\begin{eqnarray}
	\lambda_{min}(\varrho^{T_{B}})= min\{\frac{1+t_{1}-t_{2}-t_{3}}{4},\frac{1-t_{1}+t_{2}-t_{3}}{4},\frac{1-t_{1}-t_{2}+t_{3}}{4},
	\frac{1+t_{1}+t_{2}+t_{3}}{4}\}
	\label{partranseig}
\end{eqnarray}
The realigned matrix $R(\varrho)$ can be expressed as
\begin{equation}
	R(\varrho)=\frac{1}{4}
	\begin{pmatrix}
		1+t_3 & 0 & 0 & 1 - t_3 \\
		0 & t_1-t_2 & t_1+t_2 & 0\\
		0 & t_1+t_2 & t_1-t_2 & 0 \\
		1 - t_3 & 0 & 0 & 1+t_3 \\
	\end{pmatrix}
\end{equation}
the minimum eigenvalue of $R(\varrho)$ is given by
\begin{eqnarray}
	\lambda_{min}(R(\varrho))= min\{\frac{1}{2},\frac{t_1}{2},\frac{t_2}{2},\frac{t_3}{2}\}
	\label{realigneig}
\end{eqnarray}
We are now in a position to determine the non-negative parameters $\alpha$ and $\beta$ for which the map $\Phi_{\alpha,\beta}$ is positive. To be specific, we discuss here only one parameter family of separable and entangled state.\\
\textbf{Case-I:} One parameter family of separable states $\varrho_{1}$ for which $\lambda_{min}(R(\varrho_1))<0$.\\
Let us take $t_{1}=t_{2}=0$ and $t_{3}=-t, 0< t< 1$. The minimum eigenvalues $\lambda_{min}(\varrho_1^{T_{B}})$ and $\lambda_{min}(R(\varrho_1))$ are given by
\begin{eqnarray}
	\lambda_{min}(\varrho_{1}^{T_{B}})= min\{\frac{1+t}{4},\frac{1+t}{4},\frac{1-t}{4},\frac{1-t}{4}\}=\frac{1-t}{4}
	\label{min1}
\end{eqnarray}
\begin{eqnarray}
	\lambda_{min}(R(\varrho_{1}))= min\{0,0,\frac{1}{2},\frac{-t}{2}\}=\frac{-t}{2}
	\label{min2}
\end{eqnarray}
Since $\lambda_{min}(\varrho_{1}^{T_{B}})>0$ so the state $\varrho_{1}$ represent a family of separable state. Using (\ref{min1}) and (\ref{min2}), we can determine $\alpha$ and $\beta$ for which the map is positive. Therefore, the map $\Phi_{\alpha,\beta}$ will be positive if
\begin{eqnarray}
	\frac{\alpha}{\beta} \geq \frac{2t}{1-t}, ~~0<t<1
	\label{posmap1}
\end{eqnarray}
\textbf{Case-II:} One parameter family of separable states $\varrho_{2}$ for which $\lambda_{min}(R(\varrho_{2}))\geq0$.\\
In this case, we take $t_{1}=t_{2}=0$ and $0 \leq t_{3} \leq 1$. In this case, $\lambda_{min}(\varrho_{2}^{T_{B}})$ is found to be same as given in
(\ref{min1}) and $\lambda_{min}(R(\varrho_{2}))$ is given by
\begin{eqnarray}
	\lambda_{min}(R(\varrho_{2}))= min\{0,0,\frac{1}{2},\frac{t_{3}}{2}\}=0
	\label{min4}
\end{eqnarray}
Since $\lambda_{min}(\varrho_{2}^{T_{B}})\geq 0$ so the state $\varrho_{2}$ represent another class of separable states. In this case, it can be easily checked that the map $\Phi_{\alpha,\beta}$ will be positive for all parameters $\alpha, \beta \geq 0$. Thus, we can choose any positive $\alpha$ and $\beta$ to construct a positive map.\\
\textbf{Case-III:} One parameter family of entangled states $\varrho_{3}$ for which $\lambda_{min}(R(\varrho_{3}))<0$.\\
In this case, we can consider $t_{1}=1$ and $t_{3}=-t_{2}=-t, -1 \leq t \leq 1$. In this case, $\lambda_{min}(\varrho_{3}^{T_{B}})$ and $\lambda_{min}(R(\varrho_{3}))$ are same and are given by
\begin{eqnarray}
	\lambda_{min}(\varrho_{3}^{T_{B}})= \lambda_{min}(R(\varrho_{3}))= min\{\frac{1}{2},\frac{1}{2},\frac{t}{2},\frac{-t}{2}\}
	\label{min5}
\end{eqnarray}
Since the parameter $t$ lying in the interval $-1 \leq t \leq 1$ so it can be easily seen that both $\lambda_{min}(\varrho_{3}^{T_{B}})$ and $\lambda_{min}(R(\varrho_{3}))$ are negative. Thus, the state $\varrho_{3}$ represent a family of entangled states for $-1 \leq t \leq 1$ and within this interval of $t$, the map $\Phi_{\alpha,\beta}$ can never be positive for any parameters $\alpha, \beta > 0$.\\
\textbf{Case-IV:} An entangled state described by the density operator $\varrho_{4}$ for which $\lambda_{min}(R(\varrho_{4}))>0$.\\
Let us consider the case when $t_{1}=t_{2}=t_{3}=-1$. $\lambda_{min}(\varrho_{4}^{T_{B}})$ and $\lambda_{min}(R(\varrho_{4}))$ are given by
\begin{eqnarray}
	\lambda_{min}(\varrho_{4}^{T_{B}})= min\{\frac{1}{2},\frac{1}{2},\frac{1}{2},\frac{-1}{2}\}=\frac{-1}{2}
	\label{min7}
\end{eqnarray}
\begin{eqnarray}
	\lambda_{min}(R(\varrho_{4}))= min\{\frac{1}{2},\frac{1}{2},\frac{1}{2},\frac{1}{2}\}=\frac{1}{2}
	\label{min8}
\end{eqnarray}
The state $\varrho_{4}$ represents an entangled state. The map $\Phi^{(2,2)}_{\alpha,\beta}$ can never be positive for any parameters $\alpha, \beta > 0$.

\subsection{Is the map $\Phi_{\alpha,\beta}$ completely positive?}
\noindent Now, we ask whether the positive map $\Phi^{(d_{1},d_{2})}_{\alpha,\beta}: M_{d_{1}d_{2}}{(\mathbb{C})} \rightarrow M_{d_{1}d_{2}}{(\mathbb{C})}$ for some positive $\alpha,\beta$ is also completely positive. To investigate this question, we consider the domain set as the set of all $2 \otimes 2$ dimensional quantum states. The map for this particular domain will reduce to $\Phi^{(2,2)}_{\alpha,\beta}: M_{4}{(\mathbb{C})} \rightarrow M_{4}{(\mathbb{C})}$. 
From (\ref{phi}), the explicit form of $\Phi_{\alpha, \beta}^{(2,2)} (A)$ is given as
\begin{eqnarray*}
	\Phi_{\alpha, \beta}^{(2,2)} (A)=
	\begin{pmatrix}
		(\alpha + \beta) a_{11} && \alpha a_{21} + \beta a_{12} && \alpha a_{13} + \beta a_{21} && \alpha a_{23} + \beta a_{22}\\
		\alpha a_{12} + \beta a_{13} && \alpha a_{22} + \beta a_{14} &&\alpha a_{14} + \beta a_{23}&& (\alpha + \beta) a_{24}\\
		(\alpha + \beta) a_{31} && \alpha a_{41} + \beta a_{32} &&\alpha a_{33} + \beta a_{41} && \alpha a_{43} + \beta a_{42}\\
		\alpha a_{32} + \beta a_{33} && \alpha a_{42} + \beta a_{34} &&\alpha a_{34} + \beta a_{43} &&(\alpha + \beta) a_{44}\\
	\end{pmatrix}
\end{eqnarray*}
where $A= [a_{ij}]_{i,j=1}^{4}$.\\
We have already shown that for this domain set, there exist  $\alpha=\alpha'$ and $\beta=\beta'$ for which the map $\Phi^{(2,2)}_{\alpha',\beta'}$ is positive. The positive map $\Phi^{(2,2)}_{\alpha',\beta'}$ may or may not be completely positive. We show here the condition under which, the positive map $\Phi^{(2,2)}_{\alpha',\beta'}$ will be completely positive. We have considered here $2 \otimes 2$ dimensional quantum states in the domain to reduce the complexity of the calculation. For the general case, that is,  when the domain of the map $\Phi^{(d_{1},d_{2})_{\alpha,\beta}}$ contain $d_{1} \otimes d_{2}$ dimensional quantum states, the method of showing the completely positiveness property of the positive map $\Phi^{(d_{1},d_{2})}_{\alpha,\beta}$ will be same as we have shown below.\\
To start our investigation, we construct the Choi matrix using the result given in (\ref{choi}). The Choi matrix $C_{\Phi^{(2,2)}_{\alpha,\beta}}$ corresponding to the map $\Phi^{(2,2)}_{\alpha,\beta}$ can be constructed as
\begin{equation}
	C_{\Phi^{(2,2)}_{\alpha,\beta}} =
	\begin{bmatrix}
		C_{11} & C_{12}\\
		C_{21} & C_{22}\\
	\end{bmatrix}
\end{equation}
where the $8\times 8$ block matrices $C_{ij},i,j=1,2$ can be expressed as
\begin{equation}
	C_{11} = \begin{pmatrix}
		\alpha + \beta & 0 & 0 & 0 & 0 & \beta & 0 & 0\\
		0 & 0 & 0 & 0 &\alpha & 0 & 0 & 0 \\
		0 & 0 & 0 & 0& 0 & 0 & 0 & 0\\
		0 & 0 & 0 & 0& 0 & 0 & 0 & 0 \\
		0 & \alpha & \beta & 0& 0 & 0 & 0 & \beta\\
		0 & 0 & 0 & 0& 0 & \alpha & 0 & 0 \\
		0 & 0 & 0 & 0& 0 & 0 & 0 & 0\\
		0 & 0 & 0 & 0& 0 & 0 & 0 & 0\\
	\end{pmatrix}
\end{equation}

\begin{equation}
	C_{12} = \begin{pmatrix}
		0 & 0 & \alpha & 0 & 0 & 0 & 0 & 0\\
		\beta & 0 & 0 & 0& 0 &\beta & \alpha & 0\\
		0 & 0 & 0 & 0& 0 & 0 & 0 & 0\\
		0 & 0 & 0 & 0& 0 & 0 & 0 & 0 \\
		0 & 0 & 0 & \alpha & 0 & 0 & 0 & 0\\
		0 & 0 & \beta & 0& 0 & 0 & 0 & \alpha+\beta \\
		0 & 0 & 0 & 0& 0 & 0 & 0 & 0\\
		0 & 0 & 0 & 0& 0 & 0 & 0 & 0\\
	\end{pmatrix}
\end{equation}
\begin{equation}
	C_{21} = \begin{pmatrix}
		0 & 0 & 0 & 0& 0 & 0 & 0 & 0 \\
		0 & 0 & 0 & 0& 0 & 0 & 0 & 0 \\
		\alpha + \beta & 0 & 0 & 0& 0 & \beta & 0 & 0 \\
		0 & 0 & 0 & 0& \alpha & 0 & 0 & 0 \\
		0 & 0 & 0 & 0& 0 & 0 & 0 & 0  \\
		0 & 0 & 0 & 0& 0 & 0 & 0 & 0  \\
		0 & \alpha & \beta & 0& 0 & 0 & 0 & \beta  \\
		0 & 0 & 0 & 0& 0 & \alpha & 0 & 0 \\
	\end{pmatrix}
\end{equation}
\begin{equation}
	C_{22} = \begin{pmatrix}
		0 & 0 & 0 & 0& 0 & 0 & 0& 0 \\
		0 & 0 & 0 & 0& 0 & 0 & 0& 0 \\
		0 & 0 & \alpha & 0& 0 & 0 & 0& 0 \\
		\beta & 0 & 0 & 0&0 & \beta & \alpha & 0 \\
		0 & 0 & 0 & 0& 0 & 0 & 0& 0 \\
		0 & 0 & 0 & 0& 0 & 0 & 0& 0 \\
		0 & 0 & 0 & \alpha& 0 & 0 & 0& 0 \\
		0 & 0 & \beta & 0& 0 & 0 & 0& \alpha+\beta
	\end{pmatrix}
\end{equation}
From the Result-\ref{res-choicp}, it is known that if the Choi matrix $C_{\Phi^{(2,2)}_{\alpha,\beta}}$ represent a positive semi-definite matrix then the map $\Phi^{(2,2)}_{\alpha,\beta}$ is completely positive. Thus, we calculate the eigenvalues $\lambda_{i},i=1,2,...16$ of the Choi matrix $C_{\Phi^{(2,2)}_{\alpha,\beta}}$ and they are given by
\begin{eqnarray}
	&&\lambda_{1}=-2\alpha,\;\;\; \lambda_{2}= 2\alpha,\;\;\; \lambda_{3}=2(\alpha+\beta),\;\;\; \lambda_{i}=0,~~i=4,5,...16
	\label{eigval}
\end{eqnarray}
(i) When $\alpha=0$ and $\beta\geq 0$, all eigenvalues of $C_{\Phi^{(2,2)}_{\alpha,\beta}}$ comes out to be positive. Thus, $C_{\Phi^{(2,2)}_{\alpha,\beta}}$ will become a positive semi-definite matrix and hence the map $\Phi^{(2,2)}_{\alpha,\beta}$ will become a completely positive map.\\
(ii) The Choi matrix $C_{\Phi^{(2,2)}_{\alpha,\beta}}$ will not represent a positive semi-definite matrix when $\alpha>0$ and $\beta > 0$. Thus, in this case, the map will not be a completely positive map.


\section{Construction of witness operator} \label{sec-2.3}
\noindent We now use the Choi matrix $C_{\Phi^{(2,2)}_{\alpha,\beta}}$ corresponding to the map $\Phi^{(2,2)}_{\alpha,\beta}: M_{4}{(\mathbb{C})} \rightarrow M_{4}{(\mathbb{C})}$ to construct a witness operator. Thus, we show the construction by taking a simple case when the domain of the map contains only $2 \otimes 2$ density matrices but one may follow the same procedure to construct the witness operator when the domain of the map contains only $d_{1} \otimes d_{2}$ density matrices. The constructed witness operator can detect PPTES as well as NPTES that may exist in any arbitrary dimension. \\
To start with the construction of the witness operator, we first define a Hermitian operator $O_{\alpha,\beta} \in M_2(M_8)$ as
\begin{equation}
	O_{\alpha,\beta} = C_{\Phi_{\alpha,\beta}} C_{\Phi_{\alpha,\beta}}^{\dagger} =
	\begin{bmatrix}
		A & B \\
		B^{\dagger} & D \\
	\end{bmatrix}
	\label{operator1}
\end{equation}
where $A, B, D$ denote the $8 \times 8$ block matrices.\\
The $8 \times 8$ block matrices $A, B, D$ can further be expressed in terms of $4 \times 4$ block matrices as
\begin{eqnarray*}
	A=\begin{pmatrix}
		\alpha^2 + \beta^2 + (\alpha + \beta)^2 & 0& 0& 0 &	0& 2\alpha\beta& 0& 0  \\
	0& 2 \alpha^2 + 2 \beta^2& 0& 0 &	0& 0& 0& 0 \\
	0& 0& 0& 0 & 	0& 0& 0& 0  \\
	0& 0& 0& 0 & 	0& 0& 0& 0 \\
0& 0& 0& 0 &		2 \alpha^2 + 2 \beta^2& 0& 0& 0 \\
2\alpha\beta& 0& 0& 0 &	0& \alpha^2 + \beta^2 + (\alpha + \beta)^2& 0& 0 \\
0& 0& 0& 0 &	0& 0& 0& 0 \\
0& 0& 0& 0 &	0& 0& 0& 0 \\
	\end{pmatrix};
\end{eqnarray*}

\begin{equation}
	B=
	\begin{pmatrix}
		0& 0& \alpha^2 + \beta^2 + (\alpha + \beta)^2& 0&0&0&0&	2\alpha\beta\\
		0& 0& 0& 2 \alpha^2 + 2 \beta^2 &	0& 0& 0& 0\\
		0& 0& 0& 0 &	0& 0& 0& 0\\
		0& 0& 0& 0&	0& 0& 0& 0\\
		0& 0& 0& 0 &		0& 0& 2 \alpha^2 + 2 \beta^2& 0 \\
	0& 0& 	2\alpha\beta& 0 &	0& 0& 0& \alpha^2 + \beta^2 + (\alpha + \beta)^2 \\
		0& 0& 0& 0 &	0& 0& 0& 0 \\
		0& 0& 0& 0 &	0& 0& 0& 0 \\
	\end{pmatrix};\nonumber\\
\end{equation}

\begin{equation}
	D=
	\begin{pmatrix}
		0& 0& 0& 0&	0& 0& 0& 0\\
		0& 0& 0& 0&	0& 0& 0& 0\\
		0& 0& \alpha^2 + \beta^2 + (\alpha + \beta)^2& 0&	0& 0& 0& 2\alpha\beta\\
		0& 0& 0& 2 \alpha^2 + 2 \beta^2&	0& 0& 0& 0\\
		0& 0& 0& 0&	0& 0& 0& 0\\
		0& 0& 0& 0&	0& 0& 0& 0\\
			0& 0& 0& 0&0& 0& 2 \alpha^2 + 2 \beta^2& 0\\
			0& 0& 2\alpha\beta& 0&0& 0& 0& \alpha^2 + \beta^2 + (\alpha + \beta)^2
	\end{pmatrix};\nonumber\\
\end{equation}
Now, we are in a position to use the Hermitian operator $O_{\alpha,\beta}$ defined in (\ref{operator1}) for the construction of a witness operator $W_{\alpha,\beta}$ in $4\otimes4$ dimensional space to detect an entangled state lying in the same space. We define an operator $W_{\alpha,\beta}$ as
\begin{equation}
	W_{\alpha,\beta}= O_{\alpha,\beta} - \gamma I_{16}
\end{equation}
where $I_{16}$ denote the identity matrix of order 16.\\
Our task is to show that $W_{\alpha,\beta} \in M_2(M_8)$ is a witness operator for some suitable $\gamma$. To do this, we need to prove the following facts:\\
\textbf{(i)} $Tr[W_{\alpha,\beta}\sigma]\geq 0$ for all separable state $\sigma \in M_2(M_8)$.\\
\textbf{(ii)} $Tr[W_{\alpha,\beta}\rho_{e}]< 0$ for at least one entangled state $\rho_{e} \in M_2(M_8)$.\\
To prove \textbf{(i)}, let us consider any arbitrary separable state $\sigma$ lying in $4\otimes4$ dimensional space. The state $\sigma$ is given in the form as
\begin{equation}
	\sigma =
	\begin{pmatrix}
		X & Y \\
		Y^{\dagger} & Z \\
	\end{pmatrix}
\end{equation}
where $X,Z\geq 0$, and $Y$ represent a $8 \times 8$ block matrices and satisfies $X\geq YZ^{-1}Y^{\dagger}$.
Since $\sigma$ represent a separable state so $\sigma^{T_{B}}=\begin{pmatrix}
	X^{T} & Y^{T} \\
	(Y^{\dagger})^{T} & Z^{T} \\
\end{pmatrix}$ represent a positive semi-definite matrix and thus the block matrices $X^{T}$, $Y^{T}$, $(Y^{\dagger})^{T}$ and $Z^{T}$ satisfies
the following inequalities \cite{sharma2017}
\begin{equation}
	X^{T}\geq 0,~~ Z^{T}\geq 0,~~ X^{T}\geq Y^{T}(Z^{T})^{-1}(Y^{\dagger})^{T}
	\label{semi-def1}
\end{equation}
The expectation value of the operator $W_{\alpha,\beta}$ with respect to the state $\sigma$ is given by
\begin{equation}
	\langle W_{\alpha,\beta}\rangle_{\sigma}=Tr[W_{\alpha,\beta} \sigma] = Tr[O_{\alpha,\beta} \sigma] - \gamma
	\label{expectation}
\end{equation}
$Tr[O_{\alpha,\beta} \sigma]$ given in (\ref{expectation}) can be calculated as
\begin{eqnarray}
	Tr[O_{\alpha,\beta} \sigma] &=& Tr[AX + BY^{\dagger}] + Tr[B^{\dagger}Y + DZ] \nonumber\\
	&=& Tr[AX] + Tr[DZ] + 2Tr[BY^{\dagger}]\nonumber\\
	&\geq& \lambda_{min}(X) Tr[A] + \lambda_{min}(Z) Tr[D] + 2 Tr[BY^{\dagger}] \nonumber\\
	&=& (\lambda_{min}(X) + \lambda_{min}(Z)) Tr[A] + 2 Tr[BY^{\dagger}] \label{expectation1}
\end{eqnarray}
The third step follows from (\ref{weyl}) given in Result-\ref{res-weyl}  and last step holds true since $Tr[A] = Tr[D]$.\\
Using (\ref{expectation1}) in the expression (\ref{expectation}), we get
\begin{eqnarray}
	Tr[W_{\alpha,\beta} \sigma] &\geq&  (\lambda_{min}(X) + \lambda_{min}(Z)) Tr[A] + 2 Tr[BY^{\dagger}] - \gamma
	\label{wit1}
\end{eqnarray}
Choosing $\gamma=2 (Tr[BY^{\dagger}]+ Tr[A] ||Y||_2)$, the inequality (\ref{wit1}) reduces to
\begin{eqnarray}
	Tr[W_{\alpha,\beta} \sigma] &\geq&  (\lambda_{min}(X) + \lambda_{min}(Z)) Tr[A] - 2Tr[A]||Y||_2
	\label{wit2sig}
\end{eqnarray}
Since $\lambda_{min}(X)$ and $\lambda_{min}(Z)$ are non-negative so we can apply $AM \geq GM$ on $\lambda_{min}(X)$ and $\lambda_{min}(Z)$ and thus the inequality (\ref{wit2sig}) further reduces to
\begin{eqnarray}
	Tr[W_{\alpha,\beta} \sigma] &\geq& 2Tr[A] ((\lambda_{min}(X)\lambda_{min}(Z))^{\frac{1}{2}}-||Y||_2)
	\label{wit3}
\end{eqnarray}
Let us now use the following separability criterion \cite{johnston} that may be stated as, in particular, if any separable state in $4 \otimes 4$ dimensional system described by the density operator $\sigma =
\begin{bmatrix}
	X & Y \\
	Y^{\dagger} & Z \\
\end{bmatrix}$, then
\begin{equation}
	||Y||_2^{2} \leq \lambda_{min}(X) \lambda_{min}(Z) \label{res1}
\end{equation}
Considering (\ref{wit3}), (\ref{res1}) and the fact that $Tr[A]\geq 0$, we can conclude that $Tr[W_{\alpha,\beta} \sigma] \geq 0$ for any separable state $\sigma \in M_2(M_8)$.\\
To prove \textbf{(ii)},\label{eg5ch2} let us consider a $4\otimes4$ dimensional BES $\rho_{p,q}$ given by \cite{akbari}.
\begin{eqnarray}
	\rho_{p,q} = p \sum_{i=1}^4 |\omega_i\rangle \langle\omega_i| +  q \sum_{i=5}^6 |\omega_i\rangle \langle\omega_i| \label{bes4by4}
\end{eqnarray} where $p$ and $q$ are non-negative real numbers and $4p + 2q = 1$. The pure states $\{|\omega_i\rangle\}_{i=1}^6$ are defined as follows:
\begin{eqnarray*}
	&&|\omega_1\rangle = \frac{1}{\sqrt{2}} (|01\rangle + |23\rangle);\;\;\quad |\omega_2\rangle = \frac{1}{\sqrt{2}} (|10\rangle + |32\rangle); \;\;\quad |\omega_3\rangle = \frac{1}{\sqrt{2}} (|11\rangle + |22\rangle)\\
	&& |\omega_4\rangle = \frac{1}{\sqrt{2}} (|00\rangle - |33\rangle);\;\;\quad
	|\omega_5\rangle = \frac{1}{2} (|03\rangle + |12\rangle) + \frac{|21\rangle }{\sqrt{2}};\;\;\quad |\omega_6\rangle = \frac{1}{2} (-|03\rangle + |12\rangle)+ \frac{|30\rangle }{\sqrt{2}}
\end{eqnarray*}
The state $\rho_{p,q}$ becomes invariant under partial transposition when $p = \frac{q}{\sqrt2}$. Therefore, the state $\rho_{p_0,q_0}$ is a PPT state, where $p_0 = \frac{\sqrt{2} - 1}{2\sqrt{2}}$ and $q_0=\frac{\sqrt{2} - 1}{2}$. Since $\|R(\rho_{p_{0},q_{0}})\|_1 = 1.08579$, which is greater than 1, so by matrix realignment criteria explained in $Theorem-\ref{thm-ccnr}$, one can say that $\rho_{p_{0},q_{0}}$ is a PPTES.\\
A simple calculation shows that $\rho_{p_{0},q_{0}}$ is detected by our witness operator $W_{\alpha,\beta}$ for all $\alpha>0, \beta > 0$. To show it explicitly, let us calculate $Tr[W_{\alpha,\beta} \rho_{p_{0},q_{0}}]$. It gives
\begin{eqnarray}
	Tr[W_{\alpha,\beta} \rho_{p_{0},q_{0}}] = -1.07107 (2\alpha^2 + \alpha\beta + 2\beta^2),  \quad \forall \; \alpha, \beta > 0
\end{eqnarray}
Since, $W_{\alpha,\beta}$ detects the BES described by the density operator $\rho_{p_{0}, q_{0}}$, so $W_{\alpha,\beta}$ satisfies both \textbf{(i)} and \textbf{(ii)} and thus it qualifies to become a witness operator for all $\alpha>0$ and $\beta > 0$.
\section{Efficiency of the constructed witness operator $W_{\alpha,\beta}$}
Now, we study the detection power of our criteria by comparing them with other powerful existing entanglement detection criteria in the literature \cite{rudolph2003,kchen,dv,sarbicki}. We take a few examples of the family of PPTES and then show that our witness operator detect most PPTES in the family in comparison to other well-known entanglement criterion such as the realignment or computable cross norm (CCNR) criterion \cite{rud2005}, the de Vicente criterion (dV) \cite{dv} and the separability criterion based on correlation tensor (CT) given in \cite{sarbicki}.
\subsection{Example-1}
Let us illustrate the efficiency of our witness operator $W_{\alpha,\beta}$ using a class of $4\otimes 4$ bound entangled state which cannot be detected by realignment criteria.\\
Consider the following $4\otimes4$ PPTES \cite{kye2020}.\\
\begin{equation}
	\rho_{z,p,r}= \frac{1}{N}
	\begin{pmatrix}
		A_{11} & A_{12} & A_{13} & A_{14}\\
		A_{12}^* & A_{22} & A_{23} & A_{24}\\
		A_{13}^* & A_{23}^* & A_{33} & A_{34}\\
		A_{14}^* & A_{24}^* & A_{34}^* & A_{44}\\
	\end{pmatrix}
\end{equation}
where $ A_{11} =
\begin{pmatrix}
	z + \overline{z} & 0 & 0 & 0\\
	0 & \frac{1}{p} & 0 & 0\\
	0 & 0 & p & 0\\
	0 & 0 & 0 & \frac{r}{p} + r\\
\end{pmatrix};\;
A_{12} =
\begin{pmatrix}
	0 & -z & 0 & 0\\
	0 & 0 & 0 & 0\\
	0 & 0 & 0 & -r\\
	0 & 0 & 0 & 0\\
\end{pmatrix};\;
A_{13} =
\begin{pmatrix}
	0 & 0 & -\overline{z} & 0\\
	0 & 0 & 0 & -rz\\
	0 & 0 & 0 & 0\\
	0 & 0 & 0 & 0\\
\end{pmatrix}\\
A_{22} =
\begin{pmatrix}
	p & 0 & 0 & 0\\
	0 & z + \overline{z} & 0 & 0\\
	0 & 0 & \frac{r}{p} + r & 0\\
	0 & 0 & 0 & \frac{1}{p}\\
\end{pmatrix};
A_{24} =
\begin{pmatrix}
	0 & 0 & -rz & 0\\
	0 & 0 & 0 & -z\\
	0 & 0 & 0 & 0\\
	0 & 0 & 0 & 0\\
\end{pmatrix};\;
A_{33} =
\begin{pmatrix}
	\frac{1}{p} & 0 & 0 & 0\\
	0 & rp + r & 0 & 0\\
	0 & 0 & z + \overline{z} & 0\\
	0 & 0 & 0 & p\\
\end{pmatrix};\\
A_{34} =
\begin{pmatrix}
	0 & -r & 0 & 0\\
	0 & 0 & 0 & 0\\
	0 & 0 & 0 & -\overline{z}\\
	0 & 0 & 0 & 0\\
\end{pmatrix};\;
A_{44} =
\begin{pmatrix}
	rp + r & 0 & 0 & 0\\
	0 & p & 0 & 0\\
	0 & 0 & \frac{1}{p} & 0\\
	0 & 0 & 0 & z + \overline{z}\\
\end{pmatrix}; A_{14} = A_{23} = [0]_{4\times4}\\$ and $N = \frac{4}{p} + 4p + 4r + \frac{2r}{p} + 2pr + 8 Re(z)$.\\
$\rho_{z,p,r}$ is PPTES when $p>0$, $0<r<1$ and $z$ is a complex number such that $|z|=1$ and $\frac{-\pi}{4}<Arg(z)<\frac{\pi}{4}$.
If we take $p=z=1$ then we find that the state $\rho_{1,1,r}$ is not detected by realignment criterion, since $||R(\rho_{1,1,r})||_1 = \frac{2}{2+r} < 1$. $R(\rho_{1,1,r})$ denoting the realigned matrix of $\rho_{1,1,r}$.  \\
The state $\rho_{1,1,r}$ can be expressed in the block matrix form as
\begin{equation}
	\rho_{1,1,r} =
	\begin{pmatrix}
		X_r & Y_r \\
		Y_r^{\dagger} & Z_r \\
	\end{pmatrix}
\end{equation}
where $X_r, Y_r$ and $Z_r$ are $8 \times 8$ block matrices given as

\begin{equation*}
	X_r = \frac{1}{16+ 8r}
	\begin{pmatrix}
		2 & 0 & 0 & 0 & 0 & -1 &0 &0 \\
		0 & 1 & 0 & 0&0 & 0 & 0 & 0\\
		0 & 0 & 1 & 0&0 & 0 & 0 & -r\\
		0 & 0 & 0 & 2r &0 & 0 & 0 & 0\\
		0 & 0 & 0 & 0&1 & 0 & 0 & 0\\
		-1 & 0 & 0 & 0&0 & 2 & 0 & 0\\
		0 & 0 & 0 & 0&0 & 0 & 2r & 0\\
		0 & 0 & -r & 0&0 & 0 & 0 & 1\\
	\end{pmatrix}
\end{equation*}

\begin{equation*}
	Y_r = \frac{1}{16+ 8r}
	\begin{pmatrix}
		0 & 0 & -1 & 0 & 0 & 0 &0 &0 \\
		0 & 0 & 0 & -r &0 & 0 & 0 & 0\\
		0 & 0 & 0 & 0&0 & 0 & 0 & 0\\
		0 & 0 & 0 & 0&0 & 0 & 0 & 0\\
		0 & 0 & 0 & 0&0 & 0 & -r & 0\\
		0 & 0 & 0 & 0&0 & 0 & 0 & -1\\
		0 & 0 & 0 & 0&0 & 0 & 0 & 0\\
		0 & 0 & 0 & 0&0 & 0 & 0 & 0\\
	\end{pmatrix}
\end{equation*}

\begin{equation*}
	Z_r = \frac{1}{16 + 8r}
	\begin{pmatrix}
		1 & 0 & 0 & 0 & 0 & -r &0 &0 \\
		0 & 2r & 0 & 0&0 & 0 & 0 & 0\\
		0 & 0 & 2 & 0&0 & 0 & 0 & -1\\
		0 & 0 & 0 & 1&0 & 0 & 0 & 0\\
		0 & 0 & 0 & 0&2r & 0 & 0 & 0\\
		-r & 0 & 0 & 0&0 & 1 & 0 & 0\\
		0 & 0 & 0 & 0&0 & 0 & 1 & 0\\
		0 & 0 & -1 & 0&0 & 0 & 0 & 2\\
	\end{pmatrix}
\end{equation*}
$\rho_{1,1,r}$ is a PPT state since eigenvalues of $\rho_{1,1,r}^{T_B}$  are non negative. Now our task is to see whether the witness operator constructed here is able to detect the state described by the density operator $\rho_{1,1,r}$. To probe this fact, we calculate the expectation value of witness operator $W_{\alpha,\beta}$ for the state $\rho_{1,1,r}$, which is given by
\begin{eqnarray}
	Tr[W_{\alpha,\beta}\rho_{1,1,r}]
	&=& \frac{1}{2+r} ((\alpha^2 + \beta^2)(3-2\sqrt{2+2r^2})+  \alpha \beta (1-\sqrt{2+2r^2})) 
\end{eqnarray}
After simple calculation, we observed the following:
\begin{itemize}
	\item[(i)] If $r \in (0,\frac{1}{2\sqrt{2}})$ then $Tr[W_{\alpha,\beta}\rho_{1,1,r}] < 0$ for $\frac{\alpha\beta}{\alpha^2 + \beta^2} > \frac{3-2\sqrt{2-2r^2}}{\sqrt{2+2r^2} - 1}$.
	\item[(ii)] If $r \in [\frac{1}{2\sqrt{2}},1)$ then $Tr[W_{\alpha,\beta}\rho_{1,1,r}] < 0$ for any $\alpha, \beta > 0$.
\end{itemize}
The above facts implies that there exist positive parameters $\alpha$ and $\beta$ so that $Tr[W_{\alpha,\beta}\rho_{1,1,r}] < 0$ in the whole range $0 < r < 1$. Thus, $\rho_{1,1,r}$ is detected by $W_{\alpha,\beta}$ for some positive values of $\alpha$ and $\beta$. Hence, our witness operator $W_{\alpha,\beta}$ detects this family of PPTES in the whole range $0 < r < 1$.\\
For comparison, let us examine the separability criteria based on the correlation tensor for the detection of entanglement of the states belonging to the family $\rho_{1,1,r}$. Since, the other criteria such as CCNR criteria and the dV criteria are special cases of the CT criteria, their detection ability depends on CT criteria. Thus, we will discuss CT criteria for the detection of the state described by the density operator $\rho_{1,1,r}$.\\
The correlation matrix $C^{can}_r$ for the state $\rho_{1,1,r}$ is given by
\begin{equation}
	C^{can}_r =
	\begin{pmatrix}
		C_{11} & C_{12} \\
		C_{21} & C_{22} \\
	\end{pmatrix}
\end{equation}
where
\begin{eqnarray*}
	C_{11} =
	\begin{pmatrix}
		\frac{1}{4} & 0 &0 &	0 & 0 &0 &	0 & 0\\
		0 & \frac{-1}{8(2+r)} &0 &	0 & 0 &0 &	\frac{-r}{8(2+r)} & 0\\
		0 & 0 &\frac{-1}{8(2+r)} &	0 & 0 &\frac{-r}{8(2+r)} &	0 & 0\\
		0 & 0 &0 &	0 & 0 &0 &	0 & 0\\
		0 & 0 &0 &	0 & 0 &0 &	0 & 0\\
		0 & 0 &\frac{-r}{8(2+r)} &	0 & 0 &\frac{-1}{8(2+r)} &	0 & 0\\
		0 & \frac{-r}{8(2+r)} &0 &	0 & 0 &0 &	\frac{-1}{8(2+r)} & 0\\
		0 & 0 &0 &	0 & 0 &0 &	0 & \frac{1}{8(2+r)}\\
	\end{pmatrix};
\end{eqnarray*}
\begin{eqnarray*}
	C_{12} =
	\begin{pmatrix}
		0 & 0 &0 &	0 & 0 &0 &	0 & 0\\
		0 & 0 &0 &	0 & 0 &0 &	0 & 0\\
		0 & 0 &0 &	0 & 0 &0 &	0 & 0\\
		0 & 0 &0 &	0 & 0 &0 &	0 & 0\\
		0 & 0 &0 &	0 & 0 &0 &	0 & 0\\
		0 & 0 &0 &	0 & 0 &0 &	0 & 0\\
		0 & 0 &0 &	0 & 0 &0 &	0 & 0\\
		0 & 0 &0 &	0 & \frac{r}{8(2+r)} &0 &	0 & 0\\
	\end{pmatrix};
\end{eqnarray*}

\begin{eqnarray*}
	C_{22} =
	\begin{pmatrix}
		\frac{1}{8(2+r)} & 0 &0 &	\frac{r}{8(2+r)} & 0 &0 &	0 & 0\\
		0 & 0 &0 &	0 & 0 &0 &	0 & 0\\
		0 & 0 &0 &	0 & 0 &0 &	0 & 0\\
		\frac{r}{8(2+r)} & 0 &0 &	\frac{1}{8(2+r)} & 0 &0 &	0 & 0\\
		0 & 0 &0 &	0 & \frac{1}{8(2+r)} &0 &	0 & 0\\
		0 & 0 &0 &	0 & 0 & \frac{1}{8(2+r)} &	\frac{-1+2r}{8\sqrt{3}(2+r)} & \frac{1-2r}{4\sqrt{6}(2+r)}\\
		0 & 0 &0 &	0 & 0 &\frac{-1+2r}{8\sqrt{3}(2+r)} &	\frac{5-4r}{48+24r} & \frac{1-2r}{12\sqrt{2}(2+r)}\\
		0 & 0 &0 &	0 &0& \frac{1-2r}{4\sqrt{6}(2+r)}  &	\frac{1-2r}{12\sqrt{2}(2+r)} & \frac{2-r}{12(2+r)}\\
	\end{pmatrix}
\end{eqnarray*} and $C_{21} = C_{12}^{\dagger}$.
Let us evaluate the expression $E_{CT}=||D_{x}^A C^{can}_{r} D_{y}^B||_1 - \mathcal{N}_A (x) \mathcal{N}_B (y)$. If $E_{CT} \leq 0$ then CT criteria fails to detect the entangled state. Thus, the expression $E_{CT}$ for the state $\rho_{1,1,r}$  is given by
\begin{eqnarray}
	E_{CT} &=& \frac{1}{4} (-1 + \frac{8}{2+r} + xy - \sqrt{(3+x^2)(3+y^2)}) \nonumber
	\\ &\leq& 0 \;\; \forall\; x,y\geq 0 \; \text{and}\; r \in (0, 1)
	\label{cteg1}
\end{eqnarray}
Hence, the state $\rho_{1,1,r}$ is not detected by CT criterion for any value of the state parameter $r \in (0,1)$ and for any non-negative constant $x$ and $y$.\\
\textbf{Few Particular cases:}\\
For $(x,y)=(0,0)$ and $(1,1)$ we have $$||D_{x}^A C^{can}_{r} D_{y}^B||_1 - \mathcal{N}_A (x) \mathcal{N}_B (y) = \frac{-4r}{2+r} <0$$
which implies dV criterion and CCNR criterion fail to detect PPT entangled states in the family $\rho_{1,1,r}$.
The efficiency of the entanglement detection criteria such as dV criteria, CCNR criteria, and CT criteria can be summarized as follows (table-\ref{rtable}):\\
\begin{table}[h!]
	\begin{center}
		\begin{tabular}{| p{4.3cm} | p{4.3cm} |}
			\hline
			\multicolumn{2}{|c|}{The state $\rho_{1,1,r},~~0<r<1$} \\
			\hline
			Criterion & Detection range \\
			\hline
			\hline
			1. dV &  Does not detect \\
			2. CCNR &  Does not detect \\
			3. CT & Does not detect\\
			4. Witness operator $W_{\alpha,\beta}$ & $0< r < 1$\\
			\hline
		\end{tabular}
		\caption{Comparing the detection efficiency of different criteria for the detection of PPTES $\rho_{1,1,r}$ in the range $0 < r < 1$}
		\label{rtable}
	\end{center}
\end{table}
\subsection{Example-2}
In this example, we construct a one parameter family of 2 ququart states with positive partial transpose. These states are obtained by mixing the bound entangled state $\rho_{p_0,q_0}$ described in (\ref{bes4by4}) with white noise:
\begin{equation}
	\rho^{\lambda}_{p_0,q_0} = \lambda\rho_{p_0,q_0} + \frac{1-\lambda}{16} I_{4} \otimes I_4,~~0\leq \lambda \leq 1
\end{equation}
We note the following facts:\\
(i) $\rho^{\lambda}_{p_0,q_0}$ is invariant under partial transposition.\\
(ii) \textbf{Realignment criteria:} By matrix realignment criteria, $\rho^{\lambda}_{p_0,q_0}$ is entangled when $\lambda \in (0.897358,1]$.\\
\textbf{Detection of $\rho^{\lambda}_{p_0,q_0}$ by our witness operator method:} We now examine our criteria to detect entanglement in the family of states $\rho^{\lambda}_{p_0,q_0}$.
The state $\rho^\lambda_{p_0,q_0}$ can be expressed in the block matrix form as
\begin{equation}
	\rho^\lambda_{p_0,q_0} =
	\begin{pmatrix}
		X_{\lambda} & Y_{\lambda} \\
		Y_{\lambda}^{\dagger} & Z_{\lambda} \\
	\end{pmatrix}
\end{equation}
where $X_{\lambda}, Y_{\lambda}$ and $Z_{\lambda}$ are $8 \times 8$ block matrices given as
\begin{eqnarray*}
	X_{\lambda}= diag\{\frac{(1+(3-2\sqrt{2})\lambda)}{16}, \frac{(1+(3-2\sqrt{2})\lambda)}{16} , \frac{1-\lambda}{16}, \frac{(1 +(-5 + 4\sqrt{2})\lambda)}{16},\\ \frac{(1+(3-2\sqrt{2})\lambda)}{16} , \frac{(1+(3-2\sqrt{2})\lambda)}{16}, \frac{(1 +(-5 + 4\sqrt{2})\lambda)}{16}, \frac{1-\lambda}{16}\}
\end{eqnarray*}
\begin{eqnarray*}
	Y_{\lambda}=
	\begin{pmatrix}
		0 & 0 & 0 &0 & 0 & 0 &0 & \frac{1}{8}(-2 + \sqrt{2})\lambda\\
	0 & 0 & 0 &\frac{1}{8}(2 - \sqrt{2})\lambda & 	0 & 0 & 0 &0\\
	0 & 0 & 0 &0 & 	0 & 0 & 0 &0\\
	0 & \frac{(-1+\sqrt2)}{4\sqrt{2}}\lambda & 0 &0 & 	-\frac{(-1+\sqrt2)}{4\sqrt{2}}\lambda & 0 &0 & 0\\
		0 & 0 & 0 & 0 & 0 & 0 &\frac{1}{8}(2 - \sqrt{2})\lambda & 0\\
	0 & 0 & \frac{1}{8}(2 - \sqrt{2})\lambda & 0 & 0 & 0 &0 & 0\\
	0 & \frac{(-1+\sqrt2)}{4\sqrt{2}}\lambda & 0 & 0 & \frac{(-1+\sqrt2)}{4\sqrt{2}}\lambda & 0 &0 & 0\\
	0 & 0 & 0 & 0 & 0 & 0 &0 & 0\\
	\end{pmatrix}
\end{eqnarray*}

\begin{eqnarray*}
	Z_{\lambda}= diag\{\frac{1-\lambda}{16}, \frac{(1 +(-5 + 4\sqrt{2})\lambda)}{16}, \frac{(1+(3-2\sqrt{2})\lambda)}{16}, \frac{(1+(3-2\sqrt{2})\lambda)}{16},\\
	\frac{(1 +(-5 + 4\sqrt{2})\lambda)}{16}, \frac{1-\lambda}{16},  \frac{(1+(3-2\sqrt{2})\lambda)}{16},\frac{(1+(3-2\sqrt{2})\lambda)}{16} \}
\end{eqnarray*}
 To apply our criteria, we need to construct witness operator for the detection of entangled states in the family represented by $\rho^{\lambda}_{p_1,q_1}$. The witness operator $W_{\alpha,\beta}$ can be constructed through the prescription given in section \ref{sec-2.3}. It is given by
\begin{eqnarray}
	W_{\alpha,\beta}= O_{\alpha,\beta} - \gamma I
\end{eqnarray}
\begin{eqnarray}
	\gamma&=& 2 (Tr[BY_{\lambda}^{\dagger}]+ Tr[A] ||Y_{\lambda}||_2) = ((-6 + 7\sqrt{2})(\alpha^2 + \beta^2) + 4(-1 + \sqrt2)\alpha\beta)\lambda
\end{eqnarray}
where
$$||Y_{\lambda}||_2 = \frac{\sqrt{3 - 2 \sqrt{2}}}{2}  \lambda; \quad
Tr[A]= 6 \alpha^2 + 6 \beta^2 + 2(\alpha+\beta)^2 ; \quad Tr[BY_{\lambda}^{\dagger}]= \frac{1}{2} (2-\sqrt2)(a^2 +b^2) \lambda$$

The expectation value of $W_{\alpha,\beta}$ with respect to the state $\rho^\lambda_{p_0,q_0}$ is given by
\begin{equation}
	Tr[W_{\alpha,\beta} \rho^\lambda_{p_0,q_0}] = \frac{1}{2} (2\alpha^2 + \alpha\beta + 2\beta^2)(1 + (11-10\sqrt{2})\lambda)
\end{equation}
Therefore, it can be easily shown that $Tr[W_{\alpha,\beta} \rho^\lambda_{p_0,q_0}] < 0$ for $\alpha,\beta > 0$ and $\lambda> \frac{1}{-11+10\sqrt{2}} \approx 0.318255$.\\
\textbf{Detection of $\rho^{\lambda}_{p_0,q_0}$ by CT criterion:}
First we calculate the correlation matrix $C^{can}_{\lambda}$ for the state $\rho^{\lambda}_{p_0,q_0}$ using generalized Gell-Mann matrix (GGM) basis consisting of six symmetric GGM $\{G_i\}_{i=1}^{6}$; six antisymmetric GGM $\{G_i\}_{i=7}^{12}$ and three diagonal GGM $\{G_{13},G_{14}, G_{15}\}$ given in section (\ref{GMmatrices}). \\ $C^{can}_{\lambda}$ is $16 \times 16$ matrix with entries $C_{a,b} = \langle G_a \otimes G_b \rangle_{\rho^{\lambda}_{p_0,q_0}} $  where
\begin{eqnarray*}
	C_{0,0}&=& \langle G_0 \otimes G_0 \rangle_{\rho^{\lambda}_{p_0,q_0}}= \frac{1}{4}\\
	C_{2,5}&=& \langle G_2 \otimes G_5 \rangle_{\rho^{\lambda}_{p_0,q_0}}= \frac{1}{4}(2-\sqrt{2})\lambda\\	
	C_{3,3}&=& \langle G_3 \otimes G_3 \rangle_{\rho^{\lambda}_{p_0,q_0}}= \frac{1}{4}(-2+\sqrt{2})\lambda\\
	C_{4,4}&=& \langle G_4 \otimes G_4 \rangle_{\rho^{\lambda}_{p_0,q_0}}= \frac{1}{4}(2-\sqrt{2})\lambda\\	
	C_{13,14}&=& \langle G_{13} \otimes G_{14} \rangle_{\rho^{\lambda}_{p_0,q_0}}=\frac{1}{4\sqrt{3}}(-1+\sqrt{2})\lambda\\	
	C_{13,15}&=& \langle G_{13} \otimes G_{15} \rangle_{\rho^{\lambda}_{p_0,q_0}}=\frac{1}{2\sqrt{6}}(1-\sqrt{2})\lambda\\	
	C_{14,14}&=& \langle G_{14} \otimes G_{14} \rangle_{\rho^{\lambda}_{p_0,q_0}}=\frac{1}{6}(3-2\sqrt{2})\lambda\\	
	C_{14,15}&=& \langle G_{14} \otimes G_{15} \rangle_{\rho^{\lambda}_{p_0,q_0}}=\frac{1}{12}(-4+3\sqrt{2})\lambda\\
	C_{15,15}&=& \langle G_{15} \otimes G_{15} \rangle_{\rho^{\lambda}_{p_0,q_0}}=\frac{1}{12}(3-2\sqrt{2})\lambda	
\end{eqnarray*}
$C^{can}_{\lambda}$ is symmetric with $C_{2,5}=C_{5,2}$, $C_{13, 14}=C_{14,13}$, $C_{13,15}=C_{15,13}$ and $C_{14,15}=C_{15,14}$. The rest of the entries of $C^{can}_{\lambda}$ are zero.\\
The LHS of the inequality (\ref{ct}) can be calculated as
\begin{eqnarray}
	||D_{x}^A C^{can}_{\lambda} D_{y}^B||_1 -\mathcal{N}_A (x) \mathcal{N}_B (y) = \frac{1}{4} ((9 - 4\sqrt{2})\lambda + xy - \sqrt{(3 + x^{2})(3+y^2)})
	\label{cr1}
\end{eqnarray}
Therefore, the inequality (\ref{ct}) is satisfied $\forall\; x,y\geq 0 \; \text{ and when}\; \lambda \leq \frac{\sqrt{(3+x^2)(3+y^2)}-xy}{9-4\sqrt{2}}$.
The inequality (\ref{ct}) is violated  for $(x,y) = (\frac{1}{16},\frac{1}{32})$, and when $\lambda \geq \frac{3681}{4096} \approx 0.898682$. It implies that $\rho^\lambda_{p_0,q_0}$ represent the family of entangled states when $\lambda \in [0.898682,1]$.\\
Also for the case when $x=y$, the LHS of the inequality (\ref{ct}) can be re-expressed as
\begin{eqnarray}
	||D_{x}^A C^{can}_{\lambda} D_{x}^B||_1 - \mathcal{N}_A (x) \mathcal{N}_B (x) = \frac{1}{4} ((9 - 4\sqrt{2})\lambda -3) 
	\label{ct3}
\end{eqnarray}
We can now observe that the expression given in (\ref{ct3}) is independent of $x$ and hence for $y$ also. Therefore, CT separability criterion is violated for any $x=y\geq 0$ and when $\lambda > \frac{3}{9- 4\sqrt{2}} \approx 0.897358$. Thus CT criterion detect all entangled states $\rho^{\lambda}_{p_0,q_0}$ where $\lambda \in (0.897358,1] $ and for any $x = y$.\\
Now since Sarbicki et. al \cite{sarbicki}  have mentioned that $(x, y) = (1, 1)$ reproduces CCNR criterion and $(x, y) = (0, 0)$ reproduces dV criterion so using above argument for $x=y$, we conclude that these three criteria detect entanglement in the range $0.897358< \lambda \leq 1$. 
Thus our witness operator detects entanglement in the range $0.318255< \lambda \leq 1$, which is better than dV, CCNR, and CT criterion. We now summarize the results in Table-\ref{tab-4by4}. It shows the range in which $\rho^{\lambda}_{p_0,q_0}$ is detected using different separability criteria.
\begin{table}[h!]
	\begin{center}
		\begin{tabular}{| p{4cm} | p{4.3cm} |}
			\hline
			\multicolumn{2}{|c|}{The family of states $\rho^{\lambda}_{p_0,q_0}$} \\
			\hline
			Criterion & Detection range  \\
			\hline
			\hline
			1. dV &  $0.897358< \lambda \leq 1$ \\
			2. CCNR &  $0.897358< \lambda \leq 1$ \\
			3. CT & $0.897358< \lambda \leq 1$ \\
			4. Witness operator $W_{\alpha,\beta}$ & $0.318255< \lambda \leq 1$\\
			\hline
		\end{tabular}
		\caption{Comparing the efficiency of different criteria for the detection of BES $\rho^{\lambda}_{p_0, q_0}$ in the range $0\leq \lambda \leq 1$}
		\label{tab-4by4}
	\end{center}
\end{table}

\section{Conclusion}
To summarize, we have constructed a witness operator to detect NPTES and PPTES. 
We have constructed a linear map which is shown to be a positive map under certain restrictions on the parameters involved in the construction of the map. To make our discussion simple, we have considered $2\otimes 2$ dimensional system and then find out the condition for which the map is positive. To investigate the completely positivity of the map, we started with the Choi matrix associated with the constructed map and showed that the Choi matrix has always at least one negative eigenvalue for the parameters with respect to which the map is positive. Thus the Choi matrix is not a positive semi-definite matrix and hence the map is not completely positive. Furthermore, we find that the Choi matrix is not Hermitian so we use the product of the Choi matrix and its conjugate transpose. The resulting product now represents a Hermitian matrix. In the next step, we take the linear combination of the obtained Hermitian matrix and the identity matrix to construct the witness operator. We have also shown that the constructed witness operator not only detects NPTES but also may be used to detect PPTES. Since the Choi matrix is generated from the linear map defined in (\ref{phi}), if we consider $d\otimes d$ dimensional system as the input of the map $\phi$ then the generated Choi matrix is of order $d^{2}$. Also since the construction of our witness operator depends on the Choi matrix so our witness operator may detect NPTES and PPTES lying in $d^2 \otimes d^2$ dimensional Hilbert space.  Interestingly, we found that our witness operator is efficient in detecting several bipartite bound entangled states that were previously undetected by some well-known separability criteria. Moreover, we also compared the detection power of our witness operator with three well-known separability criteria, namely, the dV criterion, CCNR criterion, and the separability criteria based on correlation tensor (CT) proposed by Sarbicki et al. \cite{sarbicki} and found that our witness operator detects more PPTES than the criteria mentioned above.

\begin{center}
	****************
\end{center}

\chapter{Entanglement Detection via Partial Realigned Moments}\label{ch3}
\vspace{1cm}

\noindent{\small \emph{"Quantum information is more like the information in a dream."\\
		- Charles Bennett}}\\
\vspace{1cm}
\noindent \hrule
\noindent \emph{In this chapter,\footnote{ This chapter is based on the published research article, ``S. Aggarwal, S Adhikari, A. S. Majumdar, \emph{Entanglement detection in arbitrary dimensional bipartite quantum systems through partial realigned moments}, Phys. Rev. A 109, 012404 (2024)"}, firstly, we propose a separability criterion for detecting bipartite entanglement in $m \otimes n$ ($mn \neq 4$) dimensional quantum states using partial moments of the realigned density matrix. Our approach
	enables the detection of both distillable and bound entangled states through a common framework. We illustrate the significance of our method through examples of states belonging to both the above categories, which are not detectable using other schemes relying on partial state information. Secondly, the formalism of employing partial realigned moments proposed here is further adopted to give an effective separability criterion for two-qubit systems too.} 
\noindent \hrulefill
\newpage
\section{Introduction}\label{sec3.1}
Entanglement \cite{horodecki9} is a remarkable feature of quantum systems that has no classical analogue. 
The criteria to decide whether or not a given quantum state is entangled are of high theoretical and practical interest. Even though numerous entanglement criteria have been proposed in the past years, nevertheless there exists no universally applicable method to verify whether a given quantum state is entangled or not. Historically,
Bell-type inequalities were the first operational criterion
to distinguish between entangled and separable states \cite{jsbell}. Due to the importance of entanglement in quantum information processing, there has been a steady quest for devising more and more efficient methods for detecting entanglement in quantum states \cite{guhnerev, horodecki9}.\\
For bipartite systems, there exist two famous separability
criteria: the positive partial transpose (PPT) criterion \cite{peres} and the matrix realignment criterion \cite{chenwu, rud2005}. The former criterion is able to detect the entanglement of all non-positive partial transpose (NPT) states but cannot detect any PPTES. The latter criterion is weaker than the former one over NPTES; however, it is able to detect some PPTES. Although PPT criterion detects all NPTES, it cannot be implementable physically. There exist operational tools to detect entanglement in practice. Entanglement witness based on measurement of observables provides
a method to detect and characterize entanglement \cite{bmterhalPLA,guhne2003}. There exist different schemes for the construction of witness operators in the literature \cite{mlewenstein,ganguly2013}, though all such schemes rely on certain prior information about the quantum state. Besides, entanglement can also be detected using measurement statistics in a device independent manner through an approach called self-testing \cite{selftest, selftest2, selftest3, selftest4} which again relies on certain additional assumptions.\\
In practical situations, complete information about the quantum state may not be always available, and entanglement detection based on the partial knowledge of the density matrix may be easier to implement in experiments \cite{entdet}. Recently, Elben et al.  \cite{elben} proposed a method for detecting bipartite entanglement based on estimating moments of the partially transposed density matrix. Nevel et al.  \cite{neven} proposed an ordered set of experimentally accessible conditions for detecting entanglement in mixed states. Moments have the advantage that they can be estimated using shadow tomography in a more efficient way than if one had to reconstruct the state via full quantum state tomography.
Such works discussed in section \ref{sec-ptmoment} are focussed on the detection of NPTES only \cite{elben,neven,guhne2021}.   Detection of BES using moments needs a deeper investigation. In the present chapter, we provide a separability criterion based on moments of the Hermitian matrix obtained after applying the realignment operation on quantum states in arbitrary dimensional bipartite systems.\\
The motivation of this chapter is to construct computable entanglement conditions that can detect NPTES as well as BES
through partial knowledge of the density matrix. Recently, detection of BES using a moment based criterion has been studied in \cite{tzhang, gamma}.  Our criterion based on
moments is stronger in the sense that not all but only a few realigned moments are sufficient to detect entanglement. Importantly, this approach is not only conceptually sound but also tractable from an experimental perspective. The classical shadows formalism allows for reliably estimating moments
from randomized single-qubit measurements \cite{chad,kett}.
If we make multiple copies of a state represented by an $m \times m$ density matrix $\rho$, the moments $Tr[\rho^2]$,\;.\;.\;.\;, $Tr[\rho^m]$ can be measured using cyclic shift operators \cite{horoekert2002, keyl}. It has been shown that measuring partial moments is technically possible using $m$ copies of the state and controlled swap operations \cite{ha, cai, barti}. A method based on machine learning for measuring moments of any order has also been proposed \cite{sougato}.\\
In this chapter, we introduce a novel entanglement criterion for bipartite systems based on the moments of the Hermitian matrix  $[{R(\rho_{AB} )}]^{\dagger} R(\rho_{AB} )$. Here $R(\rho_{AB})$ may be 
obtained after applying realignment operation on a quantum state $\rho_{AB}$. The moments of the Hermitian matrix  $[{R(\rho_{AB})}]^{\dagger} R(\rho_{AB} )$ are known as realigned moments (or $R$-moments). First, we  show that our criterion effectively detects BES in higher dimensional bipartite systems. Further, we  show that it performs better than the other existing criteria based on partial moments in some cases for the detection of NPTES in higher dimensional systems. The $R$-moment criterion is formulated as a simple inequality which must be fulfilled by separable states of bipartite systems, and hence, its violation by a state reveals that the state is entangled.  We illustrate the significance of $R$-moment criterion for the detection of NPT and bound entanglement by examining some examples. Finally, we devise another separability criterion for 2-qubit systems and probe some examples of 2-qubit NPTES that are undetected by other partial moments based criteria and realignment criteria. 
\section{Realigned Moments or $R-$moments}
Before presenting our separability criterion  based on realigned moments, let us first define the idea of realigned moments or $R$-moments in $m \otimes n$ dimensional systems.  To make the task simpler, consider first a $2\otimes 3$ system described by the density operator $\sigma_{12}$ which is given by
\begin{eqnarray}
	\sigma_{12}= 
	\begin{pmatrix}
		Z_{11} & Z_{12} \\
		Z_{21} & Z_{22} 
	\end{pmatrix}
\end{eqnarray}
where $Z_{11}=\begin{pmatrix}
	t_{11} & t_{12} & t_{13} \\
	t_{12}^{*} & t_{22} & t_{23} \\
	t_{13}^{*} & t_{23}^{*} & t_{33} \\
\end{pmatrix}$, $Z_{12}=\begin{pmatrix}
	t_{14} & t_{15} & t_{16} \\
	t_{24} & t_{25} & t_{26} \\
	t_{34} & t_{35} & t_{36} \\
\end{pmatrix}$, $Z_{21}=Z_{12}^{\dagger}$, $Z_{22}=\begin{pmatrix}
	t_{44} & t_{45} & t_{46} \\
	t_{45}^{*} & t_{55} & t_{56} \\
	t_{46}^{*} & t_{56}^{*} & t_{66} \\
\end{pmatrix}$.\\
The normalization condition of $\sigma_{12}$ is given by $\sum_{i=1}^{6}t_{ii}=1$. The realigned matrix of $\sigma_{12}$ is denoted by $R(\sigma_{12})$ and it is given by
\begin{eqnarray}
	R(\sigma_{12})&=& \begin{pmatrix}
		(vecZ_{11})^{T} \\
		(vecZ_{12})^{T}\\
		(vecZ_{21})^{T}\\
		(vecZ_{22})^{T} 
	\end{pmatrix}=\begin{pmatrix}
		t_{11} & t_{12} & t_{13} & t_{12}^{*} & t_{22} & t_{23} & t_{13}^{*} & t_{23}^{*} & t_{33}\\
		t_{14} & t_{15} & t_{16} & t_{24} & t_{25} & t_{26} & t_{34} & t_{35} & t_{36}\\
		t_{14}^{*} & t_{24}^{*} & t_{34}^{*} & t_{15}^{*} & t_{25}^{*} & t_{35}^{*} & t_{16}^{*} & t_{26}^{*} & t_{36}^{*} \\
		t_{44} & t_{45} & t_{46} & t_{45}^{*} & t_{55} & t_{56} & t_{46}^{*} & t_{56}^{*} & t_{66} 
	\end{pmatrix}
\end{eqnarray}
where for any $n \times n$ matrix $X_{ij}$ with entries ${x_{ij}}$, $vecX_{ij}$ is defined as
\begin{eqnarray}
	vecX_{ij} = [x_{11}, . . ., x_{1n}, x_{21}, . . ., x_{2n},. . .,x_{n1}, . . . x_{nn}]^T
\end{eqnarray}
Note that $(R(\sigma_{12}))^{\dagger}R(\sigma_{12})$ is a matrix of order $9\times 9$. 
Also, the number of non-zero singular values of $R(\sigma_{12})$ is equal to the rank of $R(\sigma_{12})$ and hence it will be at most four. 
We can now generalize this fact for $m\otimes n$ systems. 
Let $\rho_{AB}$ be a density matrix representing a $m \otimes n$ dimensional state and it can be written as a block matrix with $m$ number of blocks in each row and column with each block being a $n \times n$ matrix.
The realigned matrix $R(\rho_{AB})$ obtained after applying the realignment operation has dimension $m^2 \times n^2$. 
The first step to obtain the $R$-moment criterion is to find the characteristic equation of the $n^2 \times n^2$ Hermitian operator $(R(\rho_{AB}))^{\dagger}R(\rho_{AB})$. It is given by
\begin{eqnarray}
	&&det\left([R(\rho_{AB})]^{\dagger}R(\rho_{AB})-\lambda I\right) = 0\nonumber\\&&
	\implies \prod_{i=1}^{n^2} \left(\lambda_{i}([R(\rho_{AB})]^{\dagger}R(\rho_{AB})) - \lambda\right)=0\nonumber\\&&
	\implies  \lambda^{n^2} + D_1 \lambda^{n^2-1} + D_2\lambda^{n^2-2} + . \; .\;. \;. + D_{n^2} = 0  \label{charr}
\end{eqnarray}
where $\lambda_{i}([R(\rho_{AB})]^{\dagger}R(\rho_{AB}))$ given in the second step denotes the roots of the characteristic polynomial (\ref{charr}). 
Using well-known results related to Newton polynomials and the Faddeev-LeVerrier algorithm for the characteristic polynomial and traces of powers of a matrix \cite{verrier,faddeev,zeil},
the coefficients $\{D_i\}_{i=1}^{n^2}$ given in the third step can be described in terms of moments of $[R(\rho_{AB})]^{\dagger}R(\rho_{AB})$, i.e., 
\begin{eqnarray}
	D_i = (-1)^i \frac{1}{i!} 
	\begin{vmatrix}
		T_1 & T_2 & T_3 &.&.&.&...&T_m \\
		1 & T_1 & T_2 & T_3 &.&.&...&T_{m-1} \\
		0 & 2 & T_1 & T_2 & T_3 &.&...&T_{m-2} \\
		0 & 0 & 3 & T_1 & T_2 & T_3 &...&T_{m-3} \\
		. & . & . & . & . & . &...&. \\
		. & . & . & . & . & . &...&. \\
		. & . & . & . & . & . &\ddots&. \\
		0 & 0 & 0 & 0 & 0 & 0 &...& T_1 \\
	\end{vmatrix} \label{dm}
\end{eqnarray}
for $i=1$ to $n^2$, where $D_{n^2} = det([R(\rho_{AB})]^{\dagger}R(\rho_{AB}))$. $T_k = Tr[\left([R(\rho_{AB})]^{\dagger}R(\rho_{AB})\right)^k]$ denotes the $k^{th}$ realigned moment of $[R(\rho_{AB})]^{\dagger}R(\rho_{AB})$.

Let us arrange the singular values of $R(\rho_{AB})$ in descending order, i.e., $\sigma_1(R(\rho_{AB})) \geq . \; . \; . \geq \sigma_r(R(\rho_{AB})) \geq 0$, where $r$ denotes the total number of singular values of $R(\rho_{AB})$. 
Thus for $i=1$ to $r$, we have 
$$ \lambda_{i}([R(\rho_{AB})]^{\dagger}R(\rho_{AB})) = \sigma_i^2(R(\rho_{AB}))$$  Therefore, the relation between the coefficients $D_{i}$ and the singular values  $\sigma_i(R(\rho_{AB}))$ is given by
\begin{eqnarray}
	\sum_{i=1}^r \sigma_{i}^2 (R(\rho_{AB})) &=& -D_1 \label{re1} \\
	\sum_{i < j} \sigma_{i}^2 (R(\rho_{AB})) \sigma_{j}^2 (R(\rho_{AB})) &=& D_2 \label{e2}\\
	\sum_{i < j < k} \sigma_{i}^2 (R(\rho_{AB})) \sigma_{j}^2 (R(\rho_{AB})) \sigma_{k}^2 (R(\rho_{AB}))&=& -D_3 \label{D3} \\
	....&&.... \nonumber\\
	....&&....\nonumber\\
	....&&.... \nonumber\\
	\prod_{i=1}^{r}	\sigma_{i}^2 (R(\rho_{AB}))  &=& D_r \label{e4}
\end{eqnarray}
Using (\ref{dm}), the coefficients $D_1, D_2$ and $D_3$ in terms of first moment $T_{1}$, second moment $T_{2}$ and third moment $T_{3}$ of $[R(\rho_{AB})]^{\dagger}R(\rho_{AB})$ can be expressed as
\begin{eqnarray}
	&&D_1 = -T_1 \label{d1}\label{rd_1} \\&&
	D_2 = \frac{1}{2} ({T_1}^2 - T_2)\label{d_2}\\&&
	D_3 = -\frac{1}{6} ({T_1}^3 - 3T_1 T_2 + 2 T_3) \label{d_3}
\end{eqnarray}
In general, the coefficient $D_{n^2}$ may be expressed as $D_{n^2} = det([R(\rho_{AB})]^{\dagger}R(\rho_{AB}))$.
\section{Separability Criterion Based on $R$-moments in $m \otimes n$ Systems} \label{sec-rmoment}
Although, the separability criteria based on partial moments (discussed in section \ref{sec-ptmoment}) involves upto $3^{rd}$ order moments \cite{elben,neven}, it fails to detect several NPTES in higher dimensional systems, and is also  not applicable for the detection of BES. This gives us a strong motivation to investigate the concept of realigned moments in entanglement detection in higher dimensional systems. 

\subsection{Separability criterion}
We are now ready to present our separability criterion based on realigned moments  for  $m \otimes n$ ($mn\neq 4$) dimensional systems. It is formulated in the form of an inequality that involves the $k^{th}$ order moments where $k$ is the rank of the matrix $R(\rho_{AB})$.
{\theorem \label{thm4.1} Let $\rho_{AB}$ be any $m \otimes n$ ($mn\neq 4$) dimensional bipartite state. Consider the $k$ non-zero singular values of the realigned matrix $R(\rho_{AB})$ that may be denoted as $\sigma_1, \sigma_2, \ldots \sigma_k$ with $1\leq k \leq min\{m^2,n^2\}$. If $\rho_{AB}$ is separable then the following inequality holds:
\begin{equation}
	{R}_1 \equiv	k(k-1) {D_k}^{1/k} + T_1 - 1 \leq 0 \label{thm4}
\end{equation}
where $D_k = \prod_{i=1}^{k} \sigma_i^2(R(\rho_{AB}))$ and $T_1 = Tr[[R(\rho_{AB})]^{\dagger} R(\rho_{AB})]$.}
\begin{proof}
Let $\rho_{AB}$ be any separable state in $m \otimes n$ system and assume that the number of non-zero eigenvalues of $[R(\rho_{AB})]^{\dagger} R(\rho_{AB})$ are $k$.  That is, there exist a number $k$  $(1\leq k \leq min\{m^2,n^2\})$ depending upon the number of non-zero singular values of $R(\rho_{AB})$ for which $D_{k}=\prod_{i=1}^{k} \sigma_i^2(R(\rho_{AB}))\neq 0$.
The degenerated characteristic equation of $[R(\rho_{AB})]^{\dagger} R(\rho_{AB})$ is given by
\begin{eqnarray}
	\lambda^k + \sum_{i=1}^{k} D_i \lambda^{k-i} = 0
\end{eqnarray}
where $D_i$ is defined in (\ref{dm}) for $i=1$ to $k$ and $D_i = 0$ for $i>k$.
Using (\ref{re1}) and (\ref{rd_1}), the first realigned moment $T_1$ can be expressed  in terms of the singular values of $ R(\rho_{AB})$ as
\begin{eqnarray}
	T_1 &=&\sum_{i=1}^{k} \sigma_{i}^2 ( R(\rho_{AB})) \nonumber\\ &=& \left(\sum_{i=1}^{k} \sigma_{i} ( R(\rho_{AB}))\right)^2 - 2\sum_{i < j} \sigma_{i} ( R(\rho_{AB})) \sigma_{j} ( R(\rho_{AB}))  \label{e5}
\end{eqnarray}
(\ref{e5}) can be re-expressed as
\begin{eqnarray}
	\sum_{i < j} \sigma_{i} ( R(\rho_{AB})) \sigma_{j} ( R(\rho_{AB}))= \frac{1}{2} \left(\left(\sum_{i=1}^{k} \sigma_{i} ( R(\rho_{AB}))\right)^2 - T_1\right) \label{e6}
\end{eqnarray}
It is elementary to note that the arithmetic mean of a list of non-negative real numbers is greater than or equal to their geometric mean.
Since $\sigma_i$'s for $i=1$ to $k$ are non-negative real numbers, using Result \ref{res-amgmineq}, we have 
\begin{eqnarray}
	\sum_{i < j} \sigma_{i} ( R(\rho_{AB})) \sigma_{j} ( R(\rho_{AB})) \geq \frac{k(k-1)}{2}\left(\prod_{i=1}^{k} {\sigma_i}(R(\rho_{AB}))\right)^{2/k} \label{e8}
\end{eqnarray}
Using (\ref{e6}) and (\ref{e8}), we get
\begin{eqnarray}
	\frac{1}{2} (||R(\rho_{AB})||_1^2 - T_1) \geq \frac{k(k-1)}{2} D_k^{1/k} \label{e9}
\end{eqnarray}
where $||R(\rho_{AB})||_1 =\sum_{i=1}^{k} \sigma_{i} (R(\rho_{AB}))$ and  $D_k= \prod_{i=1}^{k} \sigma_i^2(R(\rho_{AB}))$.
Thus, we obtain 
\begin{equation}
	||R(\rho_{AB})||_1^2 \geq k(k-1) D_k^{1/k} +T_1 \label{e10}
\end{equation} 
Since $\rho_{AB}$ is any arbitrary separable state, so using the realignment criterion in the above inequality (\ref{e10}), we have $k(k-1) D_k^{1/k} +T_1 \leq 1$, which proves (\ref{thm4}).
\end{proof}
{\corollary \label{cor4.1} Let $\rho_{AB}$ be any bipartite state in $m \otimes n$ ($mn\neq 4$)  dimensional system. Let $\sigma_1, \sigma_2, \ldots \sigma_k$ with $1\leq k \leq min\{m^2,n^2\}$ be $k$ non-zero singular values of the realigned matrix $R(\rho_{AB})$. If any state $\rho_{AB}$ violates (\ref{thm4}), then it is an entangled state.}\\
It is to be noted that the $R$-moment based separability criterion we have developed here is more fruitful for those density matrices $\rho_{AB}$ in $m \otimes n$  ($mn\neq 4$) system for which $det([R(\rho_{AB})]^{\dagger} R(\rho_{AB})) = 0$.  Therefore, the $R$-moment criterion works well when $R(\rho_{AB})$ is non-full rank. To test the separability criteria based on $R$-moments, consideration of the non-full rank state is advantageous in the sense that it does not require all the $R$-moments, and hence, our criterion holds good even when we do not have full information of the state.   In $m \otimes n$ ($mn\neq 4$) system, the condition $det([R(\rho_{AB})]^{\dagger} R(\rho_{AB})) = 0$ is valid when  (i) the number of non-zero eigenvalues of the matrix$[R(\rho_{AB})]^{\dagger} R(\rho_{AB})$  are less than $n^{2}$ where $m \geq n$, or when (ii) the number of non-zero eigenvalues of the matrix $[R(\rho_{AB})]^{\dagger} R(\rho_{AB})$ are less than $m^{2}$ where $m \leq n$.\\ 
It may be noted that there is no universal entanglement detection criterion that could outperform all other criteria. Any chosen criterion could work better for a given class of states, and vice-versa \cite{ravi2016}. 
We now discuss examples of PPTES and NPTES in $3\otimes 3$ and $4 \otimes 4$ systems which are detected by our $R$-moment criterion, but not by certain other criteria discussed in the literature.
\subsection{Examples}
We have now considered few examples of bipartite two-qutrit, two-ququart and a 2-parameter family of $2\otimes n$ quantum system to verify the criteria given in $Theorem-\ref{thm4.1}$. 
\subsubsection{I. A family of NPTES and BES in  4 $\otimes$ 4 dimensional system}
Let us consider the family of entangled states in $4\otimes4$ dimensional system described by the density operator $\rho_{p,q}$, which is defined in (\ref{bes4by4}). \\
\textbf{A few important properties of the state $\rho_{p,q}$: }\\
\textbf{P1.}
For this state, $det([R(\rho_{p,q})]^{\dagger} R(\rho_{p,q})) = 0$, which implies that the matrix $R(\rho_{p,q})$ is not of full rank.\\
\textbf{P2.} 
$\rho_{p_0,q_0}$ is a PPT state for $q_0 = \frac{\sqrt{2} - 1}{2}$ and $p_0 = \frac{1 - 2q_0}{4}$. \\
\textbf{P3.} It may be noted that $\|R(\rho_{p_0,q_0})\|_1 = 1.08579$, which is greater than one. Thus, by the matrix realignment criterion, one can say that $\rho_{p_0,q_0}$ is a PPTES. \\
Moreover, we find that $R(\rho_{p_0,q_0})$ has $8$ non-zero singular values. Therefore, we have $k=8$. The degenerated characteristic equation is $ \lambda^8 + \sum_{i=1}^8 D_i \lambda^{8-i} = 0$ where $D_i's$ are defined in (\ref{dm}) and $D_8 = \prod_{i=1}^8 \sigma_i^2 (R(\rho_{p_0,q_0}))$. Thus, we find that in this example, the left hand side of the inequality (\ref{thm4}) is given by
${R}_1 \equiv 56 {D_8}^{1/8} + T_1 -1 = 0.02082 > 0$. So inequality (\ref{thm4}) is violated and $Corollary-\ref{cor4.1}$ implies that the state $\rho_{p_0,q_0}$ is a PPTES. We further observe that the criterion \cite{tzhang} given in (\ref{L_4}) does not detect the BES described by the density operator $\rho_{p_0,q_0}$ belonging to the $\rho_{p,q}$ family.
On the other hand, when $(p, q) \neq (p_0, q_0)$, $\rho_{p,q}$ represents an NPTES for which the detection range by employing our $R$-moment
criterion is given by  ($0.00659601 < q < 0.153105 $) and ($ 0.26477 < q \leq 1/2$), which is comparatively larger than the range ($0.425035 < q \leq 1/2 $)  detected by employing the $D_3^{(in)}$ criterion (\ref{d3}).
\subsubsection{II. A family of NPTES in 3 $\otimes$ 3-dimensional system} 
Next, consider the class of NPT entangled states in $3\otimes 3$ dimensional system, which is defined as
\cite{garg}
\begin{eqnarray}
	\rho_a =
	\begin{pmatrix}
		\frac{1-a}{2} & 0 & 0&0&0&0&0&0 & \frac{-11}{50}\\
		0 & 0&0&0&0&0&0 &0&0\\
		0 & 0&0&0&0&0&0 &0&0\\
		0 & 0&0&0&0&0&0 &0&0\\
		0 & 0&0&0&\frac{1}{2} - a&  \frac{-11}{50} &0 &0&0\\
		0 & 0&0&0&  \frac{-11}{50} &a&0 &0&0\\
		0 & 0&0&0&0&0&0 &0&0\\
		0 & 0&0&0&0&0&0 &0&0\\
		\frac{-11}{50} & 0&0&0&0&0&0 &0& \frac{a}{2}\\
	\end{pmatrix};\;\; \frac{1}{50} (25 - \sqrt{141}) \leq a \leq \frac{1}{100}(25 + \sqrt{141})
	\label{rhoa_npt}
\end{eqnarray} 
Since $R(\rho_a)$ forms a matrix of rank $5$, so we have $k=5$. Thus, to detect whether $\rho_a$ is entangled or not, we need only 5 moments of $[R(\rho_{a})]^{\dagger} R(\rho_a)$. Further we note that $T_1 = \frac{867}{1250} - \frac{3a}{2} + \frac{5 a^2}{2}$. 
By calculating all the five moments, we find that the inequality in (\ref{thm4}) is violated in the whole range of $a$. Thus, applying $Corollary-\ref{cor4.1}$, we can say that $\rho_a$ is entangled for $\frac{1}{50} (25 - \sqrt{141}) \leq a \leq \frac{1}{100}(25 + \sqrt{141}) $.
Figure-\ref{rmo-fig1}(i) shows the violation of the inequality (\ref{thm4}), i.e., $ {R}_1 \equiv 20 {D_5}^{1/5} + T_1 - 1 > 0$ for all $0.262513 \leq a \leq 0.368743 $.
Fig-\ref{rmo-fig1}(i) and (ii) shows the comparison of $R$-moment criterion with the $p_3$-PPT, $D_3^{(in)}$  and  $p_3$-OPPT criteria defined in (\ref{p3ppt}),  (\ref{d3}), and (\ref{l3}), respectively. Since  ${L}_1 < 0$, $ {L}_2 < 0$ and $ {L}_3 < 0$, we find that the state $\rho_a$ is not detected by any of the above partial moment based  criteria for any value of the parameter $a$ in the given range. This is illustrated in Fig-\ref{rmo-fig1}(ii).
\begin{figure}[h!]
	\begin{center}
		\includegraphics[width=0.45\textwidth]{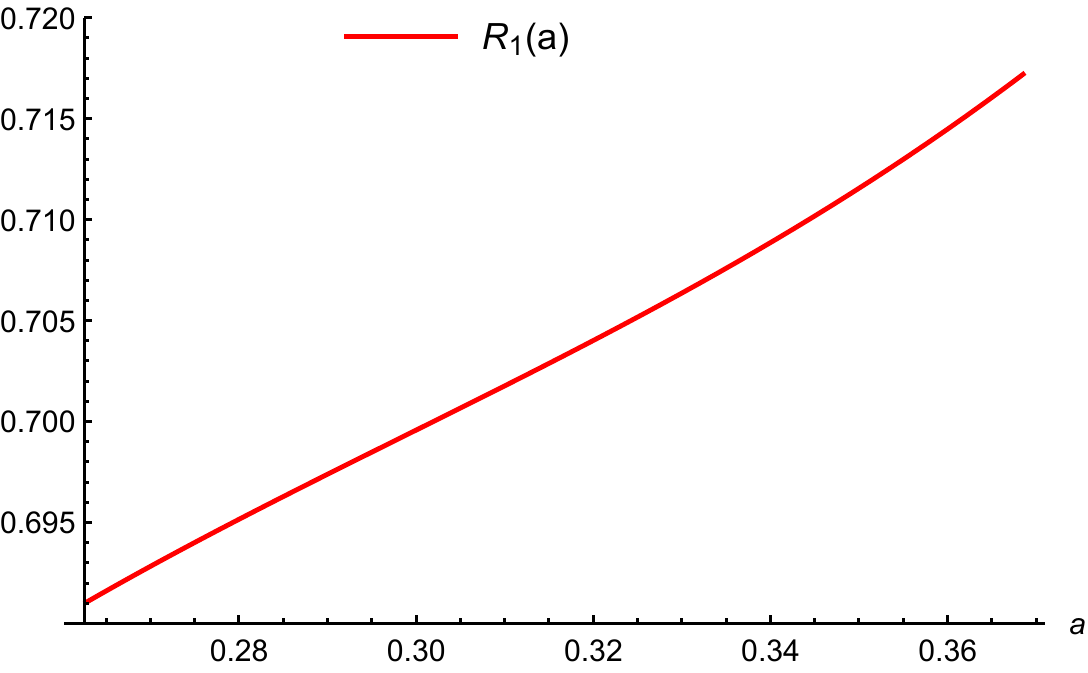} (i) \includegraphics[width=0.45\textwidth]{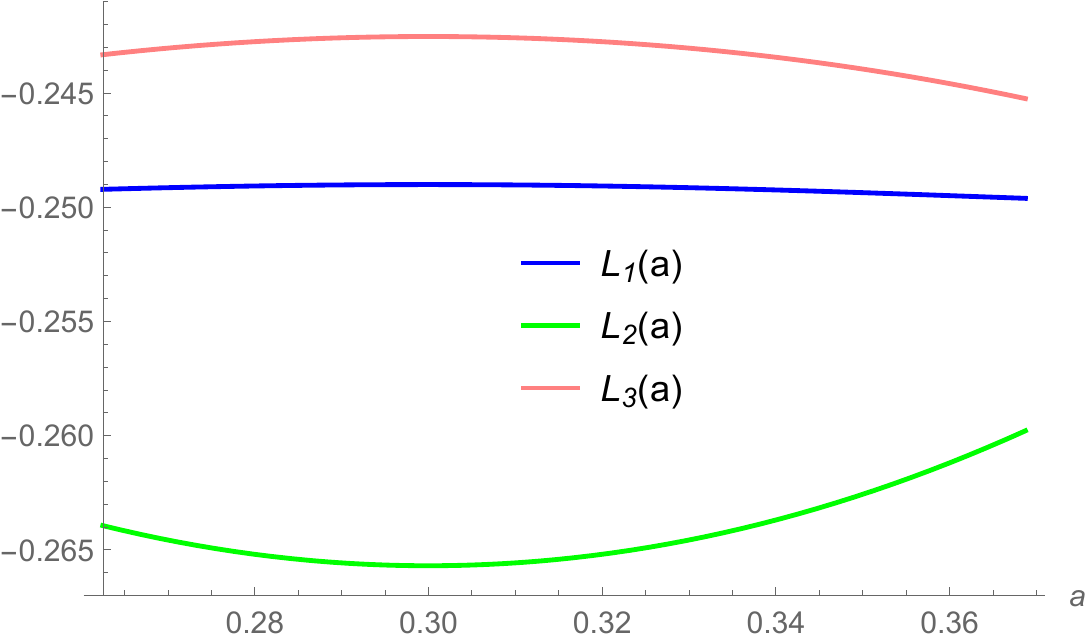} (ii)
		\caption{(i) The red curve represents our $R$-moments criterion for the state $\rho_a$ given
			by Eq.(\ref{rhoa_npt}). ${R}_1(a) > 0$ certifies the detection of entanglement in $\rho_a$. Here, $x$-axis represents the state parameter $a$. (ii) For the state $\rho_a$ (\ref{rhoa_npt}), the blue curve represents the $p_3$-PPT criterion (\ref{p3ppt}), the green curve represents the $D_3^{(in)}$ citerion (\ref{d3}), and the pink curve represents the $p_3$-OPPT criterion (\ref{l3}). The graph is plotted with respect to the state parameter $a$. $L_1, L_2, L_3 < 0$ show that $p_3$-PPT, $D_3^{(in)}$ and $p_3$-OPPT criteria fail to detect this state in the whole range.}
		\label{rmo-fig1}
	\end{center}
	\end{figure}

\subsubsection{III. Two-parameter class of states in
	$2 \otimes n$ quantum systems}
Consider the two parameter class of state defined in $2 \otimes n$ quantum systems \cite{chi}:
\begin{eqnarray}
	\rho^{(n)}_{\alpha,\gamma} = \alpha \sum_{i=0}^{1}\sum_{j=2}^{n-1} |ij\rangle \langle ij| + \beta (|\phi^+\rangle \langle \phi^+| + |\phi^-\rangle \langle \phi^-|  + |\psi^+\rangle \langle \psi^+|) + \gamma(|\psi^-\rangle \langle \psi^-|) \label{rhon}
\end{eqnarray}
where $0 \leq \alpha \leq \frac{1}{2(n-1)}$; $0 \leq \gamma \leq 1$; and ${|ij\rangle; i=0,1; j=0, 1, . . ., n-1}$ forms an orthonormal basis for $2 \otimes n$ quantum systems,
\begin{eqnarray}
	|\phi^{\pm}\rangle = \frac{1}{\sqrt{2}} (|00\rangle \pm |11\rangle) ;\quad
	|\psi^{\pm}\rangle = \frac{1}{\sqrt{2}} (|01\rangle \pm |10\rangle)
\end{eqnarray}
The normalization condition of $	\rho^{(n)}_{\alpha,\gamma}$ connect the parameters $\alpha$, $\beta$ and $\gamma$ as follows
\begin{eqnarray}
	\beta = \frac{1 - 2(n-2)\alpha - \gamma}{3}
\end{eqnarray} 
Now we compare the detection power of $R$-moment criterion with the moment based criterion given in section \ref{sec-ptmoment}. Let us consider the realignment moment based criterion given by Zhang et al \cite{tzhang} mentioned in (\ref{L_4}), (\ref{hk}), and (\ref{bl}). Since in $2 \otimes n$ systems, $rank(\rho^{(n)}_{\alpha,\gamma}) \leq 4$, the criterion given in (\ref{L_4}) is equivalent to the separability criterion based on Hankel matrices given in (\ref{hk}) and (\ref{bl}).\\
In $2 \otimes 3$ systems, $\rho^{(3)}_{\alpha,\gamma}$ is entangled when $0 \leq \alpha \leq \frac{1}{4}$ and $\frac{1-2\alpha}{2} \leq \gamma \leq 1-2\alpha$. The inequality in (\ref{L_4}) is violated, i.e., $L_4 (\alpha,\gamma)>0$ for $(\alpha, \gamma)$ lying in the blue shaded region in Fig-\ref{alphagamma}(i). Hence the entanglement is detected in this region by Zhang's criteria.
Now, applying the $R$-moment criterion on $\rho^{(3)}_{\alpha,\gamma}$, the inequality in (\ref{thm4}) is violated for the states lying in the blue as well as yellow shaded regions. It can be thus seen that the $R$-moment 
criterion performs better for such systems. Similarly, it is also possible to
show (see  Fig-\ref{alphagamma}(ii)) that for $n=4$, again a larger set of states is detected by the $R$-moment
criterion compared to the Zhang's realignment moment criterion.
\begin{figure}[h!]
	\begin{center}
		\includegraphics[width=0.4\textwidth]{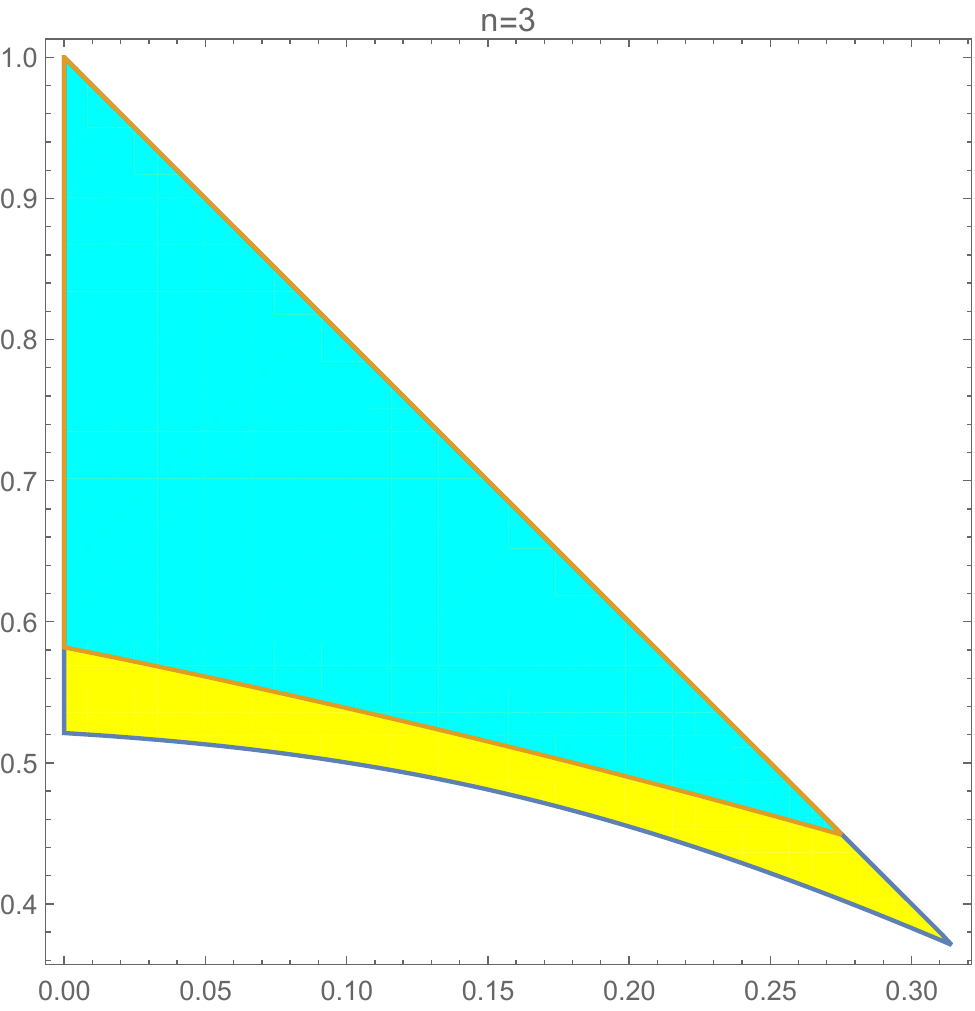}(i)
		\includegraphics[width=0.4\textwidth]{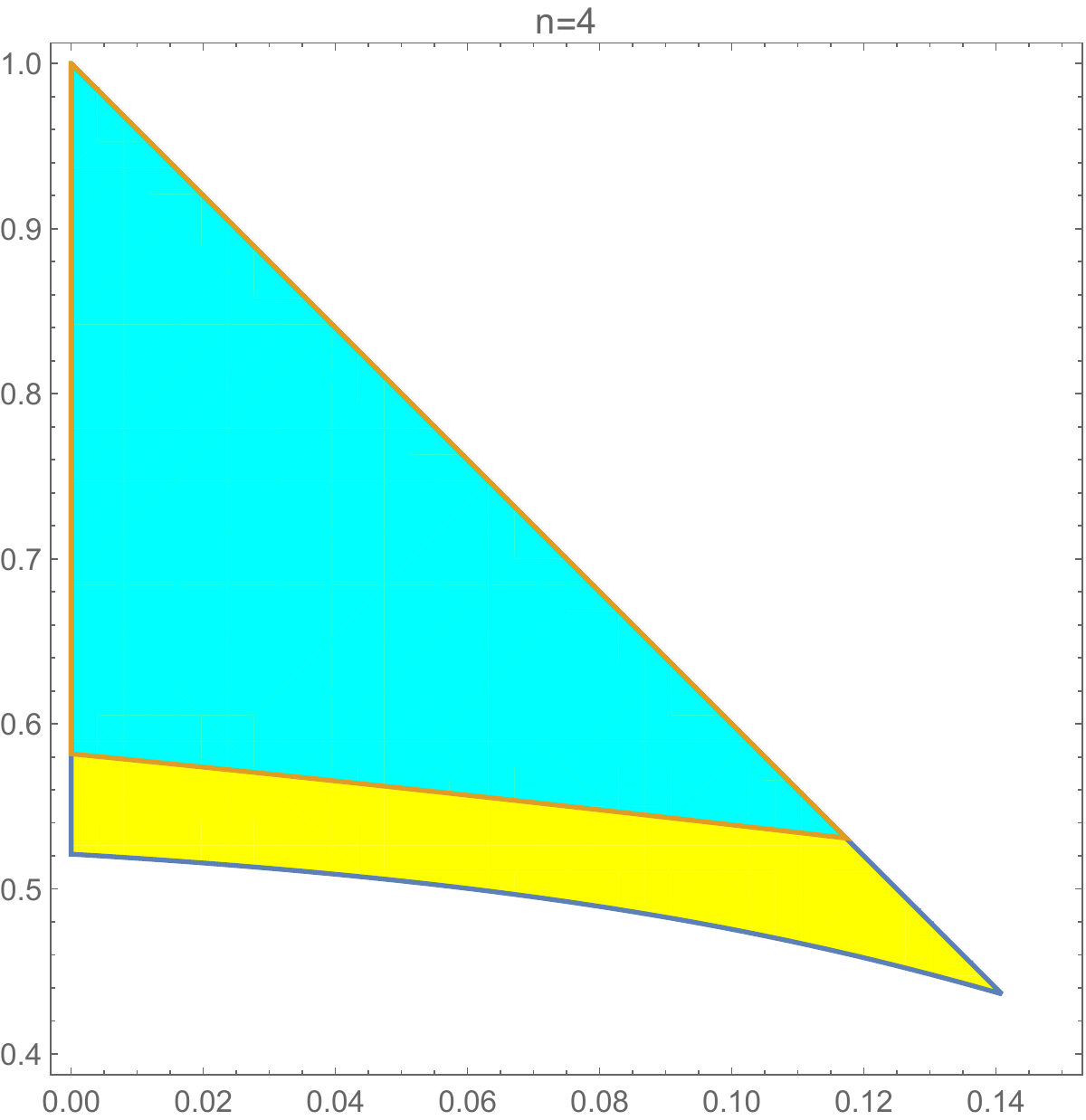}(ii)
		\caption{The above region plot in (i) and (ii), respectively, shows the detection region of the state $\rho^{(n)}_{\alpha,\gamma}$, for $n = 3$ and $n=4$ by $R$-moment criteria and the Zhang's realignment moment based criteria given in (\ref{L_4}).
			The blue region represents the entangled states detected by both the above mentioned criteria. The yellow region shows the state detected by $R$-moment criteria but undetected by Zhang's criteria. Here, $x$-axis represents the state parameter $\alpha$ and $y$-axis represents the state parameter $\gamma$.}
		\label{alphagamma}
	\end{center}
\end{figure}
\section{Non-existence of Non-full Rank Realigned Matrix of Two-qubit Entangled States} \label{rmo-AppxA}
\noindent 
We now show that the coefficient $D_{4}$ will not take value zero for any entangled two-qubit state. Let us consider an arbitrary two-qubit state that can be transformed by local filtering operation into either the Bell-diagonal state $\rho^{(BD)}=\sum_{i=1}^{4}p_{i}|\phi_{i}\rangle\langle \phi_{i}|$, where $\sum_i p_i =1$; $|\phi_{i}\rangle's$ denote the Bell states, or the states described by the density operator $\rho^{(1)}$ which is of the form \cite{ishizaka} 
\begin{eqnarray}
	\rho^{(1)} = \frac{1}{2} 
	\begin{pmatrix}
		1+c & 0 & 0 & d\\
		0 & 0 & 0 & 0\\
		0 & 0 & b-c & 0\\
		d & 0 & 0 & 1-b 
	\end{pmatrix}
\end{eqnarray}
where the state parameters $b$, $c$, $d$ satisfies any one of the following:
\begin{enumerate}
	\item[(C1)] $-1 \leq b < 1$, $~c=-1$, $~d=0$ 
	\item[(C2)] $b = 1$,$~-1< c \leq -1$, $~d=0$
	\item[(C3)] $-1 \leq b < 1$, $~-1< c \leq b$, $~d\leq |\sqrt{(1-b)(1+c)}|$ 
\end{enumerate}
\textbf{Case-I:} Let us consider the case when after the application of filtering operation on any two qubit state, the state is transformed as an entangled state $\rho^{(1)}$ if the following condition holds:
\begin{enumerate}
	\item[(E1)] $b,d \in \mathbb{R}$, $~c \leq b$, $~d \neq 0$
\end{enumerate}
The realigned matrix of $\rho^{(1)}$ is denoted by $R(\rho^{(1)})$ and it is given by
\begin{eqnarray}
	R(\rho^{(1)}) = \frac{1}{2} 
	\begin{pmatrix}
		1+c & 0 & 0 & 0\\
		0 & d & 0 & 0\\
		0 & 0 & d & 0\\
		b-c & 0 & 0 & 1-b 
	\end{pmatrix}
\end{eqnarray}
The determinant of $R(\rho^{(1)})$ is given by
\begin{equation}
	det(R(\rho^{(1)}))=\frac{(1-b)(1+c)d^{2}}{16}
	\label{det}
\end{equation} 
Using the conditions $(C1)$, $(C2)$, $(C3)$, and $(E1)$ in the determinant $det(R(\rho^{(1)}))$, it follows that the determinant must not be equal to zero. Thus, the Hermitian matrix $[R(\rho^{(1)})]^{\dagger}R(\rho^{(1)})$ is a full rank matrix. Therefore, $D_{4}\neq 0$ for any two-qubit entangled state $\rho^{(1)}$.\\
\textbf{Case-II:} If the state is transformed as a Bell diagonal state $\rho^{(BD)}$, then also it can be shown that the entanglement condition and $det(R(\rho^{(BD)}))=0$ does not hold simultaneously. This implies that in this case too  $D_{4}\neq 0$ for any two-qubit entangled state $\rho^{(BD)}$.
Thus, combining the above two cases, we can say that $D_{4}\neq 0$ for any two-qubit entangled state.

\section{Separability Criterion based on Realigned Moments ($R$-moments) in $2\otimes 2$ systems}
For $2\otimes 2$ dimensional systems, the necessary and sufficient condition
for entanglement is provided by the Peres-Horodecki PPT criterion 
\cite{peres, horo1996}. It may be noted though that partial transposition
is positive but not completely positive, and hence, it is difficult 
to be realized directly in an experiment. On the other hand, quantitative
measures of entanglement such as the entanglement of formation (EoF)
have been proposed \cite{wootters2,wootters}, which have important applications
in using entanglement as a resource for implementing tasks such as teleportation and dense coding \cite{woot3}. Bounds on the EoF have further been proposed \cite{fei1, fei2}, which are in principle, measurable, though precise measurement schemes for EoF are yet to be developed. Note further, that certain other ingenious measurement schemes for detecting two-qubit entanglement have been proposed, such as those based on employing weak measurements \cite{entdet}. However, only pure states can be detected by local operations based on the weak measurement scheme.\\
Since we have already shown that $D_4$ is non-zero for all 2-qubit entangled states (see section \ref{rmo-AppxA}), we have to use full information of the state to show the violation of the inequality (\ref{thm4}) for 2-qubit states. 
This motivates us to develop a criterion that detects entanglement using less number of moments in $2 \otimes 2$ systems. \\
Now, we introduce a different separability criterion based on $R$-moments for $2 \otimes 2$ systems. The criterion is formulated in the form of an inequality that involves up to $3^{rd}$ order moments for detecting entanglement in a 2-qubit system. In general, any separability criterion for $2 \otimes 2$ system requires up to $4^{th}$ order moments, except for a few separability criteria \cite{elben,neven,guhne2021,lin} that require up to $3^{rd}$ order moments.  
{\lemma \label{lemm4.1} If $\sigma_{i}$'s denote the singular values of the realigned matrix $R(\rho)$ arranged in the descending order as $\sigma_1(R(\rho)) \geq \sigma_2(R(\rho)) \geq \sigma_3(R(\rho)) \geq \sigma_4(R(\rho))$, then the following inequality holds:
\begin{eqnarray}
	\prod_{j=2}^{4} \left(\sigma_{1} (R(\rho)) + \sigma_{j} (R(\rho))\right) 
	&\geq& \lambda_{max}^{lb} ||R(\rho)||_1 + \sqrt{|D_3|}
	\label{lemma3.1}
\end{eqnarray}
where $\lambda_{max}^{lb}\equiv \lambda^{lb}_{max} ([R(\rho)]^{\dagger}R(\rho))$ denotes the lower bound of the maximal eigenvalue of $[R(\rho)]^{\dagger}R(\rho)$, i.e., $\lambda_{max}([R(\rho)]^{\dagger}R(\rho))\geq \lambda^{lb}_{max} $.}
\begin{proof}
	The LHS of inequality (\ref{lemma3.1}) can be expressed as
	\begin{eqnarray}
		\prod_{j=2}^{4} (\sigma_{1} (R(\rho)) + \sigma_{j} (R(\rho))) &=& \sigma_{max}^2(R(\rho))||R(\rho)||_1 + \sum_{i<j<k} \sigma_{i}(R(\rho))\sigma{j}(R(\rho))\sigma_k(R(\rho))\nonumber\\
		&\geq& \lambda_{max}^{lb} ||R(\rho)||_1 + \sqrt{\sum_{i<j<k} \sigma_{i}^2(R(\rho))\sigma_{j}^2(R(\rho))\sigma_k^2(R(\rho))} \label{e63} \nonumber\\
		&=& \lambda_{max}^{lb} ||R(\rho)||_1 + \sqrt{|D_3|}
		\label{eq50}
	\end{eqnarray}
	where the inequality (\ref{e63}) follows from the fact that  $(\sum_{i=1}^n x_i)^2 \geq \sum_{i=1}^n x_i^2$ holds for positive integers $x_i$, $i=1, 2, . . ., n$.
\end{proof}
{\lemma \label{lemm4.2} If $\rho$ describes a $2 \otimes 2$ dimensional quantum state and $R(\rho)$ denotes the realigned matrix of $\rho$, then we have
\begin{eqnarray}
	||R(\rho)||_1^2 &\geq& 2 \sqrt{D_2} + T_1 \label{lemma2}
\end{eqnarray}}
\begin{proof}
	The equation (\ref{e5}) can be re-written as
	\begin{eqnarray}
		||R(\rho)||_1^2 &=&	2\sum_{i < j} \sigma_{i} (R(\rho)) \sigma_{j} (R(\rho)) + T_1  \label{e65}
	\end{eqnarray}
	Applying the inequality $(\sum_{i,j=1}^n x_{ij})^2 \geq \sum_{i,j=1}^n x_{ij}^2$ for $i<j$ in (\ref{e65}) where $x_{ij} = \sigma_{i} (R(\rho)) \sigma_{j} (R(\rho))$,  we have
	\begin{eqnarray}
		\sum_{i < j} \sigma_{i} (R(\rho)) \sigma_{j} (R(\rho))  \geq \sqrt{\sum_{i < j} \sigma_{i}^2 (R(\rho)) \sigma_{j}^2 (R(\rho))} \label{e66}
	\end{eqnarray}
	Using the inequality (\ref{e66}), the equation (\ref{e65}) reduces to
	\begin{eqnarray}
		||R(\rho)||_1^2 &\geq& 2 \sqrt{\sum_{i < j} \sigma_{i}^2 (R(\rho)) \sigma_{j}^2 (R(\rho))} + T_1 \nonumber\\ &=& 2 \sqrt{D_2} + T_1 \label{eq51}
	\end{eqnarray}
\end{proof}
{\lemma \label{lemm4.3} If $\sigma_{i}$'s denote the singular values of the realigned matrix $R(\rho)$ arranged in the descending order as $\sigma_1(R(\rho)) \geq \sigma_2(R(\rho)) \geq \sigma_3(R(\rho)) \geq \sigma_4(R(\rho))$ and if $D_2$ and $T_1$ have their usual meaning, then we have the following:
\begin{eqnarray}
	\sum_{1< i < j} \sigma_{i}^2 (R(\rho)) \sigma_{j}^2 (R(\rho))
	= D_2 - \sigma_{1}^2 (R(\rho))(T_1- \sigma_{1}^2 (R(\rho)) ) \label{lemma3}
\end{eqnarray}
}
\begin{proof}
	The LHS of (\ref{lemma3}) can be expressed as
	\begin{eqnarray}
		\sum_{1< i < j} \sigma_{i}^2 (R(\rho)) \sigma_{j}^2 (R(\rho)) &=& \sigma_{2}^2 (R(\rho)) \sigma_{3}^2 (R(\rho)) + \sigma_{2}^2 (R(\rho)) \sigma_{4}^2 (R(\rho)) + \sigma_{3}^2 (R(\rho)) \sigma_{4}^2 (R(\rho))
		\nonumber\\
		&=& D_2 - \sigma_{1}^2 (R(\rho))(\sum_{i=2}^4 \sigma_{i}^2 (R(\rho)))\nonumber\\
		&=& D_2 - \sigma_{1}^2 (R(\rho))(T_1- \sigma_{1}^2 (R(\rho)) )
	\end{eqnarray}
Hence proved.
\end{proof}
Now we are ready to discuss our separability criteria based on realigned moments for $2 \otimes 2$ systems. Let $\rho$ be a density matrix representing a $2 \otimes 2$ dimensional state and $R(\rho)$ denote the realigned matrix obtained after applying the realignment operation. Let $\sigma_{max}(R(\rho))$ be the maximum singular value of $R(\rho)$, and  $\lambda^{lb}_{max} ([R(\rho)]^{\dagger}R(\rho)) $ and $\lambda^{ub}_{max}([R(\rho)]^{\dagger}R(\rho))$ denote respectively, the lower and upper bound of the maximal eigenvalue of $[R(\rho)]^{\dagger}R(\rho)$, where  $\lambda_{max}([R(\rho)]^{\dagger}R(\rho))\leq \lambda^{ub}_{max} ([R(\rho)]^{\dagger}R(\rho))$.\\
Using (\ref{maxlb}) and (\ref{gtb}) and since $\lambda_{i}([R(\rho)]^{\dagger}R(\rho))= \sigma_{i}^2 (R(\rho))$, we have
\begin{eqnarray}
	\lambda_{max}^{lb}([R(\rho)]^{\dagger}R(\rho)) \leq \sigma_{max}^2 (R(\rho)) \leq \lambda_{max}^{ub} ([R(\rho)]^{\dagger}R(\rho))
\end{eqnarray} 
where $\lambda_{max}^{lb} ([R(\rho)]^{\dagger}R(\rho)) = f(T_1,T_2,T_3)$ and $\lambda_{max}^{ub} ([R(\rho)]^{\dagger}R(\rho))= g(T_1,T_2,T_3)$.\\ 
We are now in a position to state the following theorem on the separability condition based on $R$-moments.
{\theorem \label{thm4.2} Let $\rho$  be a positive trace class linear operator acting on the Hilbert space $\mathcal{H}^2_A \otimes \mathcal{H}^2_B$. The realigned matrix $R(\rho)$ has singular values arranged in the order as $\sigma_1(R(\rho)) \geq \sigma_2(R(\rho)) \geq \sigma_3(R(\rho)) \geq \sigma_4(R(\rho))$. If $\rho_s$ denotes a separable state, then the following inequality holds true:
\begin{equation}
	{R}_2 \equiv \sqrt{3 X^{2/3} + 2 Y - 2T_1} - 1 \leq 0 \label{xy}
\end{equation}
where $X$ and $Y$ are the functions of the first three realigned moments, given by
\begin{eqnarray}
X= \lambda_{max}^{lb} \sqrt{2\sqrt{D_2} + T_1} + \sqrt{|D_3|}\quad \text{and}\quad Y= T_1 - \lambda_{max}^{ub} + \sqrt{D_2 - \lambda_{max}^{ub} T_1 + (\lambda_{max}^{lb})^2}
\label{eq-xy}
\end{eqnarray}

}
\begin{proof}
	 Let us start with the first realigned moment $T_1$. Using (\ref{e5}), it can be expressed as
\begin{eqnarray}
	T_1 &=& (\sum_{i=1}^{4} \sigma_{i} (R(\rho_s))^2 - 2\sum_{i < j} \sigma_{i} (R(\rho_s)) \sigma_{j} (R(\rho_s))  \label{e51}
\end{eqnarray}
The second term of (\ref{e51}) can be expressed as
\begin{eqnarray}
	2\sum_{i < j} \sigma_{i} (R(\rho_s)) \sigma_{j} (R(\rho_s)) 
	&=&\sum_{i < j} \left(\sigma_{i} (R(\rho_s)) + \sigma_{j} (R(\rho_s))\right)^2 - \sum_{i < j} (\sigma_{i}^2 (R(\rho_s))+ \sigma_{j}^2 (R(\rho_s))) \nonumber\\ 
	&=& \sum_{i < j} (\sigma_{i} (R(\rho_s)) + \sigma_{j} (R(\rho_s)))^2 -3T_1 \nonumber\\
	&=& \sum_{j=2}^4 (\sigma_{1} (R(\rho_s)) + \sigma_{j} (R(\rho_s)))^2 +  \sum_{1< i < j} (\sigma_{i} (R(\rho_s)) + \sigma_{j} (R(\rho_s)))^2 -3T_1  \label{sum}\nonumber\\
\end{eqnarray}
Since, arithmetic mean of a list of non-negative real numbers is greater than or equal to their geometric mean and $\sigma_i$'s for $i=1$ to $4$ are non-negative real numbers, we have 
\begin{eqnarray}
	\sum_{j=2}^4 (\sigma_{1} (R(\rho_s)) + \sigma_{j} (R(\rho_s)))^2 
	\geq 3\left(\prod_{j=2}^{4} (\sigma_{1} (R(\rho_s)) + \sigma_{j} (R(\rho_s)))\right)^{2/3} \label{amgm}
\end{eqnarray}
Using Lemma-\ref{lemm4.1}, the RHS of the inequality (\ref{amgm}) may be simplified to
\begin{eqnarray}
3\left(\prod_{j=2}^{4} (\sigma_{1} (R(\rho_s)) + \sigma_{j} (R(\rho_s)))\right)^{2/3}
	\geq 3\left(\lambda_{max}^{lb} ||R(\rho_s)||_1 + \sqrt{|D_3|}\right)^{2/3}
	\label{eq28}
\end{eqnarray}
Using (\ref{eq28}) and (\ref{lemma2}) in (\ref{amgm}), we obtain
\begin{eqnarray}
	\sum_{j=2}^4 (\sigma_{1} (R(\rho_s)) + \sigma_{j} (R(\rho_s)))^2
	&\geq& 3\left( \lambda_{max}^{lb} \sqrt{ 2 \sqrt{D_2} + T_1} + \sqrt{|D_3|} \right)^{2/3}\nonumber\\ &=& 3 X^{2/3} \label{X}
\end{eqnarray}
The second term of (\ref{sum}) can be expanded using Lemma-\ref{lemm4.3} as
\begin{eqnarray}
	\sum_{1<i<j} (\sigma_{i} (R(\rho_s)) + \sigma_{j} (R(\rho_s)))^2 &=& 2\sum_{i=2}^4 \sigma_{i}^2 (R(\rho_s))  + 2\sum_{1< i < j} \sigma_{i} (R(\rho_s)) \sigma_{j} (R(\rho_s)) \nonumber\\
	&\geq& 2(T_1 - \sigma_{max}^2(R(\rho_s))) + 2\sqrt{ \sum_{1< i < j} \sigma_{i}^2 (R(\rho_s)) \sigma_{j}^2 (R(\rho_s))}\nonumber\\
	&=& 2(T_1 - \lambda_{max}^{ub}) + 2\sqrt{D_2 - \sigma_{max}^2 (R(\rho_s))(T_1 - \sigma_{max}^2 (R(\rho_s)) )} \nonumber \\
	&\geq&  2(T_1 - \lambda_{max}^{ub}) + 2\sqrt{D_2 - \lambda_{max}^{ub} T_1 + (\lambda_{max}^{lb})^2 } \nonumber\\
	&=& 2Y \label{Y}
\end{eqnarray}
Using (\ref{X}) and (\ref{Y}) in (\ref{sum}), we get
\begin{eqnarray}
	||R(\rho_s)||_1 &\geq& \sqrt{3X^{2/3} + 2Y -2 T_1} \label{eq61}
\end{eqnarray}
Now we will use $||R(\rho_s)||_1\leq 1$ in (\ref{eq61}) to obtain the desired result, i.e., if $\rho_s$ is separable, then 
\begin{eqnarray}
	\sqrt{3X^{2/3} + 2Y -2 T_1} \leq 1
\end{eqnarray}
Hence proved.
\end{proof} 
{\corollary
If any two-qubit state $\rho$ violates the inequality (\ref{xy}), i.e., if ${R}_2 > 0$, then the state is an entangled state.}
\subsection{Examples}
{\example
The two-qubit isotropic state is given by \cite{zhao2010}
\begin{eqnarray}
	\rho_{f} = \frac{1-f}{3} I_2 \otimes I_2 + \frac{4f-1}{3} |\psi^+\rangle \langle \psi^+|,  0\leq f \leq 1
\end{eqnarray}
with 	$|\psi^+\rangle = \frac{1}{\sqrt{2}} (|00\rangle + |11\rangle)$ and $I_2$ denotes the $2\times2$ identity matrix.
One can use the PPT and matrix realignment criteria to verify that the isotropic state $\rho_f$ is entangled for $\frac{1}{2} < f \leq 1$. The $D_3^{(in)}$ criterion given in (\ref{d3}) ensures that the state $\rho_f$ is entangled in the range $0.625< f \leq 1$. In order to apply $R$-moment criterion on $\rho_f$, our task is to probe whether the inequality in (\ref{xy}) holds for $\rho_f$.
After simple calculations, we get 
\begin{eqnarray}
	&&T_1 = Tr[[R(\rho_f)]^{\dagger} R(\rho_f)]= \frac{1}{3} (1 - 2 f + 4 f^2)\nonumber
\end{eqnarray}
Thus, for the $2 \otimes 2$ isotropic state, the inequality (\ref{xy}) reads
\begin{equation}
	{R}_2 \equiv \sqrt{3X_f^{2/3} + 2Y_f -2 T_1} - 1 \leq 0 \label{iso}
\end{equation}
where $X_f$ and $Y_f$ are the functions of $T_1$, $T_2$, and $T_3$ for the state $\rho_f$, defined in (\ref{eq-xy}).
The above inequality $R_2 \leq 0$ is violated for $0.608594<f\leq 1$ and this implies that the state $\rho_f$ is entangled in this range which is better than that provided by the $D_3^{(in)}$ criterion.
\example
Consider the two-parameter family of $2 \otimes 2$ states represented by the density matrix \cite{rudolph2003}
\begin{eqnarray}
	\rho_{s,t} = 
	\begin{pmatrix}
		\frac{5}{8} & 0 & 0 & \frac{t}{2}\\
		0 & 0 & 0 & 0\\
		0 & 0 & \frac{1}{2}(s-\frac{1}{4}) & 0\\
		\frac{t}{2} & 0 & 0 & \frac{1-s}{2} \\
	\end{pmatrix}, t\neq 0, \frac{1}{4} < s \leq 1
\end{eqnarray} 
The state $\rho_{s,t}$ has non-positive partial transpose for all values of the state parameters $t$ and $s$ for which the state is defined. Therefore, by the PPT criterion, it is entangled for any non-zero value of the parameter $t$ and $\frac{1}{4} < s \leq 1$.\\
Let us now apply the $R$-moment criterion for the detection of the entangled state belonging to the family of the states described by the density operator $\rho_{s,t}$. The first moment of the Hermitian operator $[R(\rho_{s,t})]^{\dagger} R(\rho_{s,t})$ can be calculated as $T_1 = Tr[[R(\rho_{s,t})]^{\dagger} R(\rho_{s,t})]= \frac{1}{32}\left(21 -20s + 16( s^2 +|t|^2)\right)$. Similarly, calculating $T_2$ and $T_3$ and putting these values in (\ref{xy}), we find that the state $\rho_{s,t}$ satisfies the inequality ${R}_2 > 0$ for different ranges of the state parameter $s$ and $t$ as shown in Fig-\ref{rud1}. 
This implies that the state $\rho_{s,t}$ is entangled in this region and detected by $R$-moment criteria.\\
Further, to show the significance of $R$-moment criteria, we compare our criterion with partial transpose moments based criteria such as $p_3$-PPT and $D_{3}^{(in)}$ criterion given in (\ref{p3ppt}) and (\ref{d3}). It can be shown that the $D_{3}^{(in)}$ criterion performs better than the $p_3$-PPT criterion for detecting entangled states in $\rho_{s,t}$ family.
$D_{3}^{(in)}$ criterion gives the following inequality for the family of states described by $\rho_{s,t}$
\begin{eqnarray}
	{L}_2 > 0, \; \text{for}\;\; |t|> \sqrt{\frac{ 20s^2-25s +5 }{16s-26}}
\end{eqnarray}
This shows that $D_{3}^{(in)}$ inequality is violated in the region $ |t|> \sqrt{\frac{ 20s^2-25s +5 }{16s-26}}$ and hence the entangled states are detected by $D_{3}^{(in)}$  inequality.\\
Also, it is important to note that the $R$-moment criterion detects those entangled states in $\rho_{s,t}$ family which are neither detected by partial transpose moments based criteria, such as $D_3^{(in)}$ and $p_3$-PPT criteria nor by the realignment criterion. This is illustrated in Fig-\ref{rud1}. 
Similarly, we can get identical region $S=\{t  \;| -0.25\leq t \leq -0.2\}$  of entangled states detected by $R$-moment criterion but not by partial moment based criteria.\\
\begin{figure}[h!]
	\begin{center}
			\includegraphics[width=0.48\textwidth]{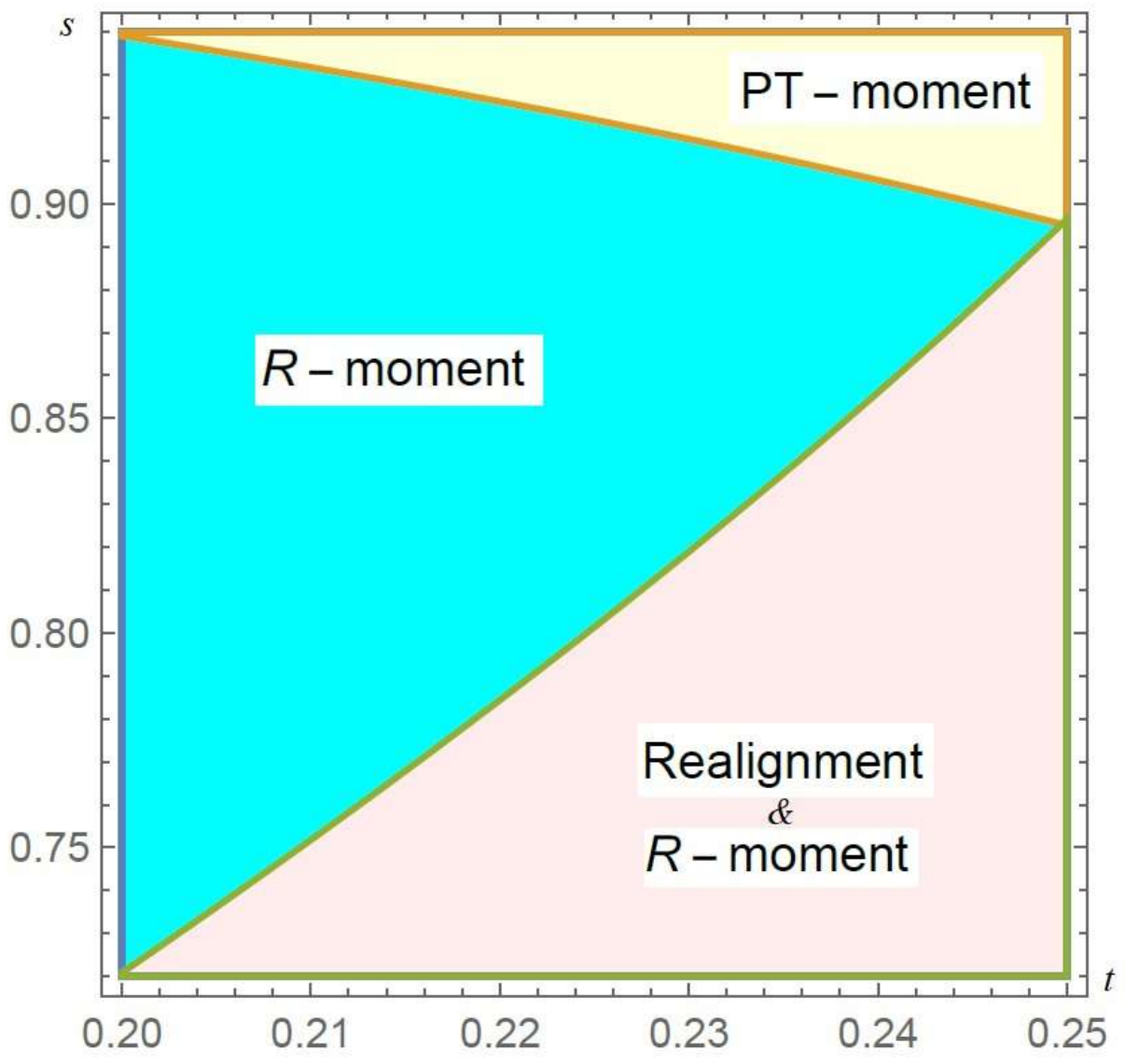}
		\caption{This figure shows the detection region of entangled states belonging to $\rho_{s,t}$ family of states. Here $x$-axis and $y$-axis denote the state parameters $t$ and $s$, respectively. The states lying in the blue region are  detected by the $R$-moment criterion, but not by either the 
			criterion using partial transpose moments or by the realignment criterion. The states lying in the pink region are detected by $R$-moment as well as realignment criterion. Yellow region depicts the states detected by $p_3$-PPT as well as $D_3^{(in)}$ criterion. } 
		\label{rud1}
	\end{center}
\end{figure}
}

\section{Conclusion}
In this chapter, we have introduced a separability criterion for detecting the entanglement of arbitrary dimensional bipartite states based on partial information of the density matrix by employing realigned moments. Our proposed approach enables the detection of both PPT and NPT entangled states within the same framework using a few moments of the realigned matrix.
The formalism presented here is thus advantageous compared to the recently formulated entanglement detection schemes using partial transpose moments \cite{elben,neven,guhne2021} which fails to detect BES.\\
We have demonstrated the significance of our separability criterion with the help of several examples of higher dimensional states.  We have studied the effectiveness of our criterion by comparing it with other partial transpose moment based criteria. Our $R$-moment criterion for $m \otimes n$ systems is further significant since it is able to detect certain NPTES that are undetected by the criteria based on partial transpose moments. Moreover, our separability criterion is able to detect
certain BES that are not detected by another recently
proposed criterion \cite{tzhang} using realigned moments.\\
Additionally, for two-qubit systems, we have shown that our approach can be
slightly modified to yield another separability criterion that can detect
entanglement in $2 \otimes 2$ systems without requiring complete information
about the quantum state. Interestingly, we have found that our realigned moment based criterion detects some two-qubit entangled states that are neither detected by partial moments based criteria \cite{elben,neven}, nor by the matrix realignment criteria \cite{chenwu, rud2005}. In \cite{liu}, authors have proposed a method based on permutation moments for the detection of multipartite entangled state. They have shown that their criterion reduces to partial transposition and realignment criterion in particular cases. In their criterion, odd order moments are inaccessible and their proposed criterion needs $2n$ copies of the entangled states to calculate $n$th order moment. On the other hand, odd and even moments are accessible in our criterion and it needs only $n$ copies of the state to calculate $n$th order moment.

\begin{center}
	****************
\end{center}
\chapter{Physical Realization of Realignment Criteria Using Structural Physical Approximation}\label{ch4}
\vspace{1cm}
{\small \emph{"Those who are not shocked when they first come across quantum theory cannot possibly have understood it."\\
- Niels Bohr}}
\vspace{1.5cm}
\noindent \hrule
\noindent \emph{ In this chapter\;\footnote { This chapter is based on the published research paper ``S. Aggarwal, A. Kumari, S. Adhikari, \emph{Physical realization of realignment criteria using structural physical approximation}, {Physical Review A} \textbf{108}, 012422, (2023)"},
	we propose the method of structural physical approximation (SPA) of realignment map. 
	Firstly, we have approximated the realignment map to a positive map using the method of SPA and then we have shown that the structural physical approximation of the realignment map (SPA-R) is completely positive. The positivity of the constructed map is characterized using moments that can be physically measured. Next, we develop a separability criterion based on our SPA-R map in the form of an inequality and have shown that the developed criterion not only detects NPTES but also PPTES. Further, we have shown that for a special class of states called Schmidt symmetric states, the SPA-R separability criteria reduce to the original form of realignment criteria. We have provided some examples to support the results obtained. Moreover, we have analyzed the error that may occur because of approximating the realignment map.} 
\noindent \hrulefill
\newpage
\section{Introduction}\label{sec4.1}
Positive maps are strong detectors of entanglement but they may not be realized in an experiment. Thus, in the case of describing a quantum channel or the reduced dynamics of an open system, a stronger positivity condition is required \cite{kraus}. On the other hand, completely positive maps play an important role in quantum information theory as they make a positive map physical. Completely positive maps were introduced by Stinespring in the study of dilation problems for operators \cite{stein}. Positive and completely positive maps are defined in section \ref{sec-laresults}. 
Choi described the operator sum representation of completely positive maps in \cite{choi}. In \cite{poon}, Kwong and Poon studied the necessary and sufficient conditions for the existence of completely positive maps. \\
A physical way by which a positive map can be approximated by a completely positive map is called structural physical approximation (SPA) \cite{horoekert2002,jaro, augusiak11,hakye,augusiak14,jbae}. The idea of SPA is to mix a positive map $\Phi$ with a maximally mixed state, making the mixture completely positive \cite{horoekert2002}. If a $d\otimes d$ dimensional system described by the density operator $\rho$ then the SPA to the positive map $\Phi$ may be defined as 
\begin{equation}
	\tilde{\Phi}(\rho) = \frac{p^*}{d^2} I_{d^2} + (1-p^*)\Phi(\rho)
\end{equation}
where $I_{d^2}$ denotes the identity matrix of order $d^2$ and $p^*$ is the minimum value of the probability $p$ for which the approximated map $\tilde{\Phi}$ is completely positive \cite{adhikari2018}.\\
The resulting map $\tilde{\Phi}$ can then be physically realized in a laboratory and its action characterizes entanglement of the states detected by $\Phi$. Geometrically, the map $\tilde{\Phi}$ retains the direction of the generalized Bloch vector of the output state described by $\Phi(\rho)$ and it only changes the length of the vector by some factor \cite{horo2003}.\\
Although PPT criteria is one of the most important and widely used criteria, it suffers from a few serious drawbacks. One of the major drawbacks is that it is based on the negative eigenvalues of the partially transposed matrix and is thus used to detect NPTES only. Another drawback is that the partial transposition map is positive but not a completely positive map and hence, may not be implementable in an experiment. In order to make it experimentally implementable, partial transposition maps have been approximated to a completely positive map using the method of SPA \cite{horoekert2002}. A lot of work has been done on the structural physical approximation of partial transposition (SPA-PT) \cite{horoekert2002,adhikari2018, hlim1, hlim2, hlim3, kumari2019}. The SPA-PT has been used to detect and quantify entanglement \cite{horoekert2002,kumari2022} but till now it can only be used to detect and quantify NPTES. \\
Realignment criteria is a powerful criterion in the sense that it may be used to detect NPTES as well as PPTES. Although it is one of the best for the detection of PPTES the problem with this criteria is that it may not be used to detect entanglement practically. It is due to the fact that the realignment map corresponds to a non-positive map and it is known that the non-positive maps are not experimentally implementable. The defect that the realignment map may not be realized in an experiment may be overcome by approximating the non-positive realignment map to a completely positive map  \cite{korbicz}. The work included in this chapter is significant because although there has been considerable progress in entanglement detection using the SPA of partial transposition map, the idea of SPA of realignment operation is still unexplored.\\
In this chapter, we approximate the non-positive realignment map to a completely positive map. To achieve this goal, we first approximate the non-positive realignment map with a positive map and then we show that the obtained positive map is also completely positive. We estimate the eigenvalues of the realignment matrix using moments that may be used physically in an experiment \cite{brun,tanaka,sougato,guhne2021}. Further, we formulate a separability criterion that we call SPA-R criteria, using our approximated map that not only detects NPTES but also PPTES.  Next, we have shown that the SPA-R criteria reduces to the original formulation of realignment criteria for a class of states called Schmidt-symmetric states. 
\section{Structural Physical Approximation of Realignment Map: Positivity and Completely positivity}
\noindent We employ the method of structural physical approximation to approximate the realignment map. To proceed toward our aim, let us first recall the depolarizing map which may be defined in the following way: A map $\Psi: M_n \longrightarrow M_n$ is said to be depolarizing if 
\begin{eqnarray}
	\Psi(A)= \frac{Tr[A]}{n} I_{n}
	\label{dep}
\end{eqnarray}
In the method of structural physical approximation of the realignment map, we mix an appropriate proportion of the realignment map with a depolarizing map in such a way that the resulting map will be positive. This may happen because the lowest negative eigenvalues generated by the realignment map can be offset by the eigenvalues of the maximally mixed state generated by the depolarizing map.\\
Consider any quantum state $\rho \in \mathcal{D} \subset B(\mathcal{H}_A \otimes \mathcal{H}_B)$, where the Hilbert spaces $ \mathcal{H}_A$ and $ \mathcal{H}_B$ have dimension $d$ and $\mathcal{D}$ denote the set of all states $\rho$ whose realignment matrix $R(\rho)$ have real eigenvalues and positive trace. The structural physical approximation of the realignment map may be defined as $\widetilde{R}: M_{d^2} (\mathbb{C}) \longrightarrow M_{d^2}(\mathbb{C})$ such that
\begin{eqnarray}
	\widetilde{R}(\rho)=\frac{p}{d^2}I_{d\otimes d}+\frac{(1-p)}{Tr[R(\rho)]}{R(\rho)},~~0 \leq p \leq 1
	\label{sparealign}
\end{eqnarray}
\subsection{Positivity of structural physical approximation of realignment map}
It is known that $R(\rho)$ forms an indefinite matrix, its eigenvalues may be negative or positive. Let us first consider the case when all the eigenvalues of $R(\rho)$ are non-negative. So, by the definition of $\widetilde{R}$ given in (\ref{sparealign}), $\widetilde{R}(\rho)$ is positive for all $ p \in [0,1]$ and hence $\widetilde{R}$ defines a positive map. On the other hand, if $R(\rho)$ has negative eigenvalues, then $\widetilde{R} (\rho)$ may be positive under some conditions.
However, since the realignment operation is not physically realizable, it is not feasible to compute the eigenvalues of $R(\rho)$. To overcome this challenge, we use the lower bound of the minimum eigenvalue of $R(\rho)$ which is denoted by $\lambda_{min}^{lb}[R(\rho)]$. We find that the map $\widetilde{R}(\rho)$ is positive for some range of $p$ that can be expressed in terms of $\lambda_{min}^{lb}[R(\rho)]$ defined in (\ref{lb}). Further, it can be observed that $\lambda_{min}^{lb}[R(\rho)]$ can be expressed in terms of $Tr[R(\rho)]$ and $Tr[(R(\rho))^2]$, which may be measured experimentally that is shown in the later section. We can analyze the positivity of $\widetilde{R}(\rho)$ when we know the nature of the eigenvalues of $R(\rho)$. Therefore, the question reduces to the following: how to determine the sign of the real eigenvalues of $R(\rho)$ experimentally without directly computing its eigenvalues? To investigate the asked question, we adopt a method to determine the sign of real eigenvalues of $R(\rho)$ that may be implemented in the experiment.
\subsubsection{Method for determining the sign of real eigenvalues of $R(\rho)$}
Let $\rho \in \mathcal{D}$ be a $d \otimes d$ dimensional state such that $R(\rho)$ has real eigenvalues $\lambda_1, \lambda_2, . . ., \lambda_{d^2}$. The characteristic polynomial of $R(\rho)$ is given as 
\begin{eqnarray}
	f(x) = \prod_{i=1}^{d^2} (x - \lambda_i) = \sum_{k=0}^{d^2} (-1)^k a_k x^{d^2-k}
\end{eqnarray} 
where $a_0 = 1$ and $\{a_k\}_{k=1}^{d^2}$ are the functions of eigenvalues of $R(\rho)$.\\
Let us now consider the polynomial $f(-x)$, which
effectively replaces the positive eigenvalues of $R(\rho)$ with negative ones and vice versa. 
For a polynomial with real roots, Descartes' rule of
sign states that the maximum number of positive roots is given by the number of sign changes between consecutive elements in the ordered list of its nonzero coefficients \cite{descartes}. The matrix $R(\rho)$ is positive semi-definite if and only if the number of sign changes in the ordered list of non-zero coefficients of $f(x)$ is equal to the degree of the polynomial $f(x)$. These non-zero coefficients can be determined in terms of moments of the matrix $R(\rho)$. The coefficients $a_i$'s are related to the moments of $R(\rho)$ by the recursive formula \cite{zeil}.
\begin{eqnarray}
	a_k = \frac{1}{k} \sum_{i=0}^k (-1)^{i-1} a_{k-i}  m_i (R(\rho)) \label{rec}
\end{eqnarray}
where $ m_i (R(\rho)) = Tr[(R(\rho))^i]$ denotes the $ith$ order moment of the matrix $R(\rho)$. For convenience, we write $ m_i (R(\rho))$ as $m_i$. The $ith$ order moment can be explicitly expressed as 
\begin{eqnarray}
	m_i  = (-1)^{i-1} i a_i + \sum_{k=1}^{i-1} (-1)^{i-1+k} a_{i-k} m_k
\end{eqnarray}
Using (\ref{rec}), we get
\begin{eqnarray}
	a_1 &=& m_1\\
	a_2 &=& \frac{1}{2} (m_1^2 - m_2) \\
	a_3 &=& \frac{1}{6} ( m_1^3 - 3 m_1 m_2 + 2m_3)
\end{eqnarray} and so on.\\
Therefore, the matrix $R(\rho)$ is positive semi-definite iff $a_i \geq 0$ for all $i = 1, ..., d^2$. 
\subsubsection{Positivity of $\widetilde{R}(\rho)$:}
We now derive the condition for which the approximated map $\tilde{R}(\rho)$ will be positive when (i) $R(\rho)$ is positive; and when (ii) $R(\rho)$ is indefinite. The obtained conditions are stated in the following theorem.\\
{\theorem Let $\rho$ be a $d \otimes d$ dimensional state such that its realignment matrix  $R(\rho)$ has real eigenvalues. The structural physical approximation of realignment map $\widetilde{R}(\rho)$ is a positive operator for $p \in [l, 1]$, where $l$ is given by
\begin{eqnarray}
	l= \left\{
	\begin{array}{lrr}
		0 & when & \lambda_{min} [R(\rho)] \geq 0 \\
		\frac{d^2k}{Tr[R(\rho)]+d^2k} \leq p \leq 1 & when & \lambda_{min} [R(\rho)] < 0
	\end{array}\right\}
	\label{thm3.1eq}
\end{eqnarray}
where $k=max[0,-\lambda_{min}^{lb}[R(\rho)]]$
and $\lambda_{min}^{lb}[R(\rho)])$ denotes the lower bound of the minimum eigenvalue of $R(\rho)$ defined in (\ref{lb}). \label{thm3.1}}
\begin{proof}
	Recalling the definition (\ref{sparealign}) of the SPA of the realignment map, the minimum eigenvalue of $\widetilde{R}(\rho)$ is given by
	\begin{eqnarray}
		\lambda_{min}[\widetilde{R}(\rho)]=\lambda_{min}[\frac{p}{d^2}I_{d\otimes d}+\frac{(1-p)}{Tr[R(\rho)]}{R(\rho)}]
		\label{minr}
	\end{eqnarray}
	where $\lambda_{min}(.)$ denote the minimum eigenvalue of $[.]$. Using Weyl's inequality given in (\ref{weyl}) on RHS of (\ref{minr}), it reduces to
	\begin{eqnarray}
		\lambda_{min}[\widetilde{R}(\rho)]&\geq& \lambda_{min}[\frac{p}{d^2}I_{d\otimes d}]+\lambda_{min}[\frac{(1-p)}{Tr[R(\rho)]}{R(\rho)}]\nonumber\\&=& \frac{p}{d^2}+\frac{(1-p)}{Tr[R(\rho)]}\lambda_{min}{[R(\rho )]}
		\label{lambmin}
	\end{eqnarray}
	\noindent Now our task is to find the range of $p$ for which $\widetilde{R}$ defines a positive map. 
	Based on the sign of $\lambda_{min}[R(\rho )]$, we consider the following two cases.\\
	\textbf{Case-I:} When $\lambda_{min}[R(\rho )]\geq 0$, the RHS of the inequality in (\ref{lambmin}) is positive for every $0 \leq p\leq 1$ and hence $\widetilde{R}(\rho)$ represent a positive map for all $p \in [0,1]$.\\
	\textbf{Case-II:} If $\lambda_{min}[R(\rho)]<0$, then (\ref{lambmin}) may be rewritten as 
	\begin{eqnarray}
		\lambda_{min}[\widetilde{R}(\rho)] &\geq& \frac{p}{d^2}+\frac{(1-p)}{Tr[R(\rho)]}\lambda_{min}^{lb}[R(\rho )]
		\label{lambminneg}
	\end{eqnarray}
	where $\lambda_{min}^{lb}[R(\rho)]$ is given in (\ref{lb}) and may be re-expressed in terms of moments as
	\begin{eqnarray}
		\lambda_{min}^{lb}[R(\rho)]= \frac{m_1}{d^2} - \sqrt{(d^{2}-1)\left(\frac{m_2}{d^2}- \left(\frac{m_1}{d^2}\right)^2\right)} \label{lbr}
	\end{eqnarray}
	where $m_1 = Tr[R(\rho)]$ and $m_2 = Tr[(R(\rho))^2] $.\\
	Taking $\lambda_{min}^{lb}[R(\rho)]=-k,~~k(>0)\in \mathbb{R}$, (\ref{lambminneg}) reduces to
	\begin{eqnarray}
		\lambda_{min}[\widetilde{R}(\rho)] &\geq&  \frac{p}{d^2}-k\frac{(1-p)}{Tr[R(\rho)]}
		\label{lambminneg1}
	\end{eqnarray}
	Now, if we impose the condition on the parameter $p$ as $p\geq \frac{d^2k}{Tr[R(\rho)]+d^2k}=l$ then $\lambda_{min}[\widetilde{R}(\rho)]\geq 0$. 
	Thus combining the above discussed two cases, we can say that the approximated map $\widetilde{R}(\rho)$ represent a positive map when (\ref{thm3.1eq}) holds. Hence the theorem is proved.
\end{proof}
\subsection{Completely positivity of structural physical approximation of realignment map} \label{sec-cp}
In order to show that the approximated map $\widetilde{R}(\rho)$ defined in (\ref{sparealign}) may be realized in an experiment, it is not enough to show that $\widetilde{R}(\rho)$ is positive but also we need to show that it is completely positive.\\ 
When $l \leq p \leq 1$ there exist non-negative real numbers $\gamma_1$ and $\gamma_2$ such that the following conditions hold 
\begin{eqnarray}
	\lambda_{min} [\widetilde{R}(\rho)] &\geq& \gamma_1 \lambda_{min}[\rho] \label{26}\\
	\lambda_{max} [\widetilde{R}(\rho)] &\leq& \gamma_2 \lambda_{max}[\rho] \label{27}
\end{eqnarray}
Hence, using Result \ref{rescp}, $\widetilde{R}(\rho)$ is a completely positive operator for $p \in [l,1]$.

\section{Realization of Moments of the Realigned  Matrix} \label{sec-expmoments}
Now we formulate a procedure to show how the moments of the realigned matrix may be realized in an experiment.
Recently, it has been shown that the measurement of the moments of partially transposed density matrices is practically possible \cite{ha, cai, barti,sougato}. We consider here the moments of the realigned matrix for its possible realization in an experiment. To achieve our aim, we adopt the idea presented in \cite{barti, sougato}, where it has been shown that the $k$th partial moment can be measured using SWAP operators \cite{kwek2002} on $k$ copies of the state. The technique involved is to express the matrix power as the expectation of the permutation operator. We adopt this approach applied to the realigned matrix in order to show how the measurement of realigned moments could be accomplished. \\
For a $d\otimes d$ dimensional state $\rho$, the $k$ copies are given by $\otimes_{c=1}^k \rho_c$. Let $m_k$ denote the $k$th moment of the realigned matrix $R(\rho)$, i.e., 
\begin{eqnarray}
	m_k = Tr[(R(\rho))^k] \label{mk}
\end{eqnarray}
The $k$th moment of $R(\rho)$ can be expressed in terms of expectation value of the permutation operator as 
\begin{eqnarray}
	m_k &=& Tr[(\otimes_{c=1}^k R(\rho_c)) P^k] \label{mkc}
\end{eqnarray}
where $P$ is the normalized permutation operator defined as 
$P=\frac{1}{d} \sum_{i,j=0}^{d-1} |ij\rangle \langle ji |$. It is also known as the SWAP operator.\\
Since the separability criteria presented in $Theorem-\ref{thm4.1}$ and $Theorem-\ref{thm4.2}$ are based on the moments of the realigned matrix,  we need to  estimate the  moments of the matrix $[R(\rho)]^\dagger R(\rho)$. In particular, we  show here the procedure of determining the  first and second moments of $[R(\rho)]^\dagger R(\rho)$, which are denoted by $T_{1}$ and $T_{2}$, respectively.
The first moment of the realigned matrix is given by
\begin{eqnarray}
	m_1=	Tr[R(\rho)] 
	\label{m100}
\end{eqnarray}
Applying the results \ref{res-trnorm} and \ref{res-n1n2}, we get 
\begin{eqnarray}
	m_1=	Tr[R(\rho)] \leq \sqrt{k} ||R(\rho)||_2 = \sqrt{k \sum_{i=1}^k \sigma_{i}^2(R(\rho))} =\sqrt{kT_1}\nonumber\\
\end{eqnarray}
Hence, we have
\begin{eqnarray}
	T_1 \geq \frac{m_1^2}{k} \label{lbt1}
\end{eqnarray}
As the derived separability criterion in $Theorem-\ref{thm4.1}$ depends on the singular values of the realigned matrix,  we further need to estimate the singular values of $R(\rho)$. Here we use the result \cite{merikoski}
\begin{eqnarray}
	\sigma_{j}(R(\rho)) \leq \frac{m_1}{d^2} + \sqrt{\frac{d^2-j}{j}\left(T_1 - \frac{m_1^2}{d^2}\right)} \label{sigub}
\end{eqnarray} 
where $1\leq j \leq k$.
The above inequality (\ref{sigub}) holds when 
\begin{eqnarray}
	T_1 \leq \frac{m_1^2}{j} \label{ubt1}
\end{eqnarray}
Now, combining (\ref{lbt1}) and (\ref{ubt1}), the first moment of $[R(\rho)]^\dagger R(\rho)$ is bounded from above and below by
\begin{eqnarray}
	\frac{m_1^2}{k} \leq T_1 \leq \frac{m_1^2}{j},~~1\leq j \leq k 
	\label{lbub}
\end{eqnarray}
It may be observed that the optimal value of $T_1$ may be obtained by substituting $j=k$ in (\ref{lbub}), which turns out to be
\begin{eqnarray}
	T_1^{opt} = \frac{m_1^2}{k}
\end{eqnarray}
Next, we obtain the estimate of $T_{2}$ using the Result-\ref{res-yang}, where we have used $A=B=[R(\rho)]^\dagger R(\rho)$ and $q=1$.
Hence, the inequality (\ref{ineq2}) reduces to
\begin{eqnarray}
	T_{2}=Tr[\left([R(\rho)]^\dagger R(\rho)\right)^2] \leq Tr[[R(\rho)]^\dagger R(\rho)]^2=(T_{1})^{2}
	\label{ineq3}
\end{eqnarray}
On the other hand, the lower bound turns out to be \cite{merikoski}
\begin{eqnarray}
	T_{2}\geq \frac{T_{1}^2}{d^2}
	\label{ineq4}
\end{eqnarray}
Combining (\ref{ineq3}) and (\ref{ineq4}), the bounds on $T_{2}$ expressed in terms of $m_{1}$ is given by
\begin{eqnarray}
	\frac{m_{1}^4}{d^2k^2} \leq T_{2} \leq \frac{(m_{1})^{4}}{j^2},~~~1\leq j \leq k
	\label{ineq5}
\end{eqnarray}
The optimal range of $T_2$ may be obtained by putting $j=k$, which is given by
\begin{eqnarray}
	\frac{m_{1}^4}{d^2k^2} \leq T_{2}^{opt} \leq \frac{(m_{1})^{4}}{k^2}
	\label{ineqopt}
\end{eqnarray}
Therefore,  $T_{1}^{opt}$ and $T_{2}^{opt}$  may be experimentally determined
by measuring the first moment $m_1$ of $R(\rho)$.\\
Since $R(\rho)$ is not physically realizable, we need to first express it in terms of a physically realizable operator using the approximated map $\widetilde{R}$ defined in (\ref{sparealign}). \\
It may be observed from (\ref{sparealign}) that the realigned matrix described by the density operator $R(\rho)$ is proportional to its SPA operator $\widetilde{R}(\rho)$, {\it i. e.},
\begin{eqnarray}
	R(\rho) \propto \widetilde{R}(\rho)- \frac{p}{d^2} I_{d \otimes d}
\end{eqnarray}
Hence, the $k$th moment of $R(\rho)$ may be estimated as
\begin{eqnarray}
	m_k &\simeq& Tr	\left[\left(\otimes_{c=1}^k \left(\widetilde{R}(\rho_c)- \frac{p}{d^2} I_{d \otimes d}\right) \right)P^k\right]\\
	&=& Tr	\left[\left(\otimes_{c=1}^k \widetilde{R}(\rho_c) \right)P^k - \frac{p}{d^2} P^k \right]\\
	&=& Tr	\left[\left(\otimes_{c=1}^k \widetilde{R}(\rho_c) \right)P^k\right] - \frac{p}{d^2} Tr[P^k] \label{mk} 
\end{eqnarray}
where $p \in [\frac{ d^2l}{m_1+ d^2l}, 1]$  and $l=max[0,-\lambda_{min}^{lb}[R(\rho_{AB})]]$. \\
Since $ \widetilde{R}(\rho_c)$ is a Hermitian, positive semi-definite operator with unit trace, we introduce a quantity $s_k$ which is given by $s_k := Tr	\left[\left(\otimes_{c=1}^k \widetilde{R}(\rho_c) \right)P^k\right]$ . The quantity $s_k$ can be measured using controlled swap operations \cite{horoekert2002, kwek2002}.\\
In particular, we show how the first moment $m_1$ may be 
determined. 
Equation (\ref{mk}) may be re-expressed for $k=1$ as
\begin{eqnarray}
	m_1 &\simeq& Tr	[\widetilde{R}(\rho)P] - \frac{p}{d^2} Tr[P]\\
	&=& Tr	[\widetilde{R}(\rho)P] - \frac{p}{d^2}\\
	&\leq& Tr	[\widetilde{R}(\rho)P] - \frac{l}{m_1+ d^2l} \label{rmo-m1}
\end{eqnarray}
In the last line, we have used $p \geq \frac{ d^2l}{m_1+ d^2l} $ and $l=max[0,-\lambda_{min}^{lb}[R(\rho_{AB})]]$, which is defined in $Theorem-\ref{thm3.1}$.
We have used (\ref{thm3.1eq}) in the last step.
Therefore, the inequality (\ref{rmo-m1}) may be re-written as
\begin{eqnarray}
	m_1 + \frac{ l}{m_1+ d^2l} \leq  Tr	[\widetilde{R}(\rho)P] := s_1 \label{rmo-s1}
\end{eqnarray}
Simplifying (\ref{rmo-s1}), we get 
\begin{eqnarray}
	m_1^2 + m_1(d^2l -s_1) + l(1-d^2s_1) \leq 0 \label{quad1}
\end{eqnarray}
Solving the above quadratic equation for $m_1$, we have
\begin{eqnarray}
	\frac{ -(d^2l - s_1) - \sqrt{(d^2l-s_1)^2 - 4l(1-d^2s_1)}}{2}	\leq m_1  \leq \frac{ -(d^2l - s_1) + \sqrt{(d^2l-s_1)^2 - 4l(1-d^2s_1)}}{2} \nonumber
\end{eqnarray}
For $m_1$ to be real, we have 
\begin{eqnarray}
	(d^2l-s_1)^2 - 4l(1-d^2s_1) \geq 0 \quad
	\Rightarrow \quad d^4l^2 + 2l(d^2s_1 -2) + s_1^2 \geq 0  \label{rmo-quad2}
\end{eqnarray}
Inequality (\ref{rmo-quad2}) holds when either $l \geq \frac{2 -d^2s_1+ 2\sqrt{1-d^2s_1}}{d^4}$ or $l \leq \frac{2 -d^2s_1 - 2\sqrt{1-d^2s_1}}{d^4}$\\
\textbf{Case 1:} When $ 2-d^2s_1+ 2\sqrt{1-d^2s_1} \leq d^4l \leq d^4 $
\begin{eqnarray}
	f_l(s_1)\leq	m_1 \leq f_u(s_1) \label{rmo-case1}
\end{eqnarray}
\textbf{Case 2:} When $0 \leq d^  4l \leq 2 -d^2s_1 - 2\sqrt{1-d^2s_1}$
\begin{eqnarray}
	g_l(s_1)\leq	m_1 \leq g_u(s_1) \label{rmo-case2}
\end{eqnarray}
where $f_l$, $f_u$, $g_l$, $g_u$ are functions of $d$ and $s_1$ given as follows:
\begin{eqnarray}
	f_l(s_1) &=&  \frac{1}{2}
	(-d^2 + s_1) -\frac{1}{2d^2} \sqrt{d^8 + 2d^6s_1 + 4d^2s_1 + d^4s_1^2 -8(1+\sqrt{x})} \ \nonumber\\
	f_u(s_1) &=& \frac{-1}{d^2}(x + \sqrt{x}) +  \frac{1}{2d^2}\left(\sqrt{d^8 + 2d^6s_1 + 4d^2s_1 +d^4s_1^2 -8 (1+ \sqrt{x})}\right) \nonumber\\
	g_l(s_1) &=& \frac{1}{d^2}\left(  -x + \sqrt{x} - \sqrt{1+x-2\sqrt{x}} \right) \nonumber\\
	g_u(s_1) &=& \frac{s_1}{2} + \frac{1}{d^2} \sqrt{1+x-2\sqrt{x}}
\end{eqnarray}
where $x := 1 - d^2s_1$.
Hence, the first moment $m_1$ of $R(\rho)$ may be estimated in terms of $s_1$ using the relations (\ref{rmo-case1}) and (\ref{rmo-case2}).  Since these relations are expressed in terms of $s_1= Tr[\widetilde{R}(\rho)P]$, the first moment of realigned matrix $R(\rho)$ can be estimated experimentally. 
Similarly, it can be shown that the second moment of the realigned matrix can also be practically estimated. 
Hence, the $k$th realigment moment may measured using (\ref{mk}).
Thus, this scheme can be generalized to higher dimensional systems as well. 
Hence, the measurement of the moments $m_k$ of the realigned matrix may be practically possible.

\section{Detection using the Experimental Implementable Form of Realignment Criteria}
\noindent Now, we derive a separability condition for the detection of NPTES and PPTES that may be implemented in the laboratory. The separability condition obtained depends on the structural physical approximation of the realignment map and thus the condition may be termed as SPA-R criterion. We further identify a class of states known as Schmidt-symmetric state for which the SPA-R criterion is equivalent to the original form of realignment criterion \cite{rudolph2004,chen2} and weak form of realignment criterion \cite{hertz}. 
\subsection{SPA-R Criterion}
We are now in a position to derive the laboratory-friendly (for clarification, see section \ref{sec-expmoments}) separability criterion that may detect the NPTES and PPTES. The proposed entanglement detection criterion is based on the structural physical approximation of realignment criterion and it may be stated in the following theorem.
{\theorem \label{thm3.2} If any quantum system described by a density operator $\rho_{sep}$ in $d\otimes d$ system is separable then 
	\begin{eqnarray}
		||\widetilde{R}(\rho_{sep})||_1 \leq \frac{p[Tr[R(\rho_{sep})]-1]+1}{Tr[R(\rho_{sep})]}=[\widetilde{R}(\rho_{sep})]_{UB} \label{thm3.2eq}
\end{eqnarray}}
\begin{proof}
	Let us consider a two-qudit bipartite separable state described by the density matrix $\rho_{sep}$, then after the application of the approximated realignment map (\ref{sparealign}) on $\rho_{sep}$, we have
	\begin{eqnarray}
		\widetilde{R}(\rho_{sep}) = \frac{p}{d^2}I_{d\otimes d}+\frac{1-p}{Tr[R(\rho_{sep})]}R(\rho_{sep})
		\label{spa1}
	\end{eqnarray}
	Taking trace norm on both sides of (\ref{spa1}) and using triangular inequality on the norm, it reduces to 
	\begin{eqnarray}
		||\widetilde{R}(\rho_{sep})||_1 &\leq& 
		p+\frac{1-p}{Tr[R(\rho_{sep})]}||R(\rho_{sep})||_1
		\label{spa3}
	\end{eqnarray}
	Since $\rho_{sep}$ denote a separable state, using realignment criteria, we have $||R(\rho_{sep})||_1\leq 1$ \cite{chenwu,rud2005}. Therefore, (\ref{spa3}) further reduces to
	\begin{eqnarray}
		||\widetilde{R}(\rho_{sep})||_1 &\leq& p+\frac{1-p}{Tr[R(\rho_{sep})]}\nonumber\\
		&=& \frac{p[Tr[R(\rho_{sep})]-1]+1}{Tr[R(\rho_{sep})]} = [\widetilde{R}(\rho_{sep})]_{UB}
	\end{eqnarray}
	Hence proved.
\end{proof} 
{\corollary If for any two-qudit bipartite state $\rho$, the inequality
	\begin{eqnarray}
		||\widetilde{R}(\rho)||_1 > [\widetilde{R}(\rho)]_{UB} \label{cor1}
	\end{eqnarray}   
	holds then the state $\rho$ is an entangled state.}\\
We should note an important fact that $[\widetilde{R}(\rho_{sep})]_{UB}$ given in (\ref{thm3.2eq}) depends on $Tr[R(\rho)]$, which may be measured in an experiment (See section \ref{sec-expmoments}).
	
	\subsection{Schmidt-symmetric states}
	Let us consider a class of states known as Schmidt-symmetric states which may be defined as \cite{hertz} 
	\begin{eqnarray}
		\rho_{sc}=\sum_{i} \lambda_{i} A_{i}\otimes A_{i}^{*}
		\label{sc}
	\end{eqnarray}
	where $A_{i}$ represent the orthonormal bases of the operator space and $\lambda_{i}$ denote non-negative real numbers known as Schmidt coefficients.\\ 
	We are considering this particular class of states to show that the separability criteria using the SPA-R map become equivalent to the original form of realignment criteria for such a class of states. Hertz et al. \cite{hertz} studied the Schmidt-symmetric states and proved that a bipartite state $\rho_{sc}$ is Schmidt-symmetric if and only if 
	\begin{eqnarray}
		||R(\rho_{sc})||_1 = Tr[R(\rho_{sc})] \label{sch}
	\end{eqnarray}
	For any Schmidt-symmetric state described by the density operator $\rho_{sc}$, the realignment matrix $R(\rho_{sc})$ defines a positive semi-definite matrix. Hence, using $Theorem-\ref{thm3.1}$, $\widetilde{R}(\rho_{sc})$ is positive $ \forall \; p\in [0,1]$. Also, using (\ref{26}) and (\ref{27}), $\widetilde{R}(\rho_{sc})$ can be shown as a completely positive. To achieve the motivation, let us start with the following lemma.
	{\lemma For any Schmidt-symmetric state $\rho_{sc}$,
		\begin{eqnarray}
			||\widetilde{R}(\rho_{sc})||_1 = 1 
			\label{lemma4.1}
	\end{eqnarray}}
	\begin{proof} Let us recall (\ref{sparealign}), which may provide the structural physical approximation of the realignment of the Schmidt-symmetric state. Therefore, $\widetilde{R}(\rho_{sc})$  is given by
		\begin{eqnarray}
			\widetilde{R}(\rho_{sc}) = \frac{p}{d^2}I_{d\otimes d}+\frac{(1-p)}{Tr[R(\rho_{sc})]}{R(\rho_{sc})}
		\end{eqnarray}
		Taking trace norm on both sides and using triangle inequality we have,
		\begin{eqnarray}
			||\widetilde{R}(\rho_{sc})||_1 \leq p + \frac{(1-p)}{Tr[R(\rho_{sc})]} ||R(\rho_{sc}) ||_1
			\label{m1}
		\end{eqnarray}
		Using (\ref{sch}), the inequality (\ref{m1}) reduces to
		\begin{eqnarray}
			||\widetilde{R}(\rho_{sc})||_1 \leq 1
			\label{r1}
		\end{eqnarray}
		Again, using (\ref{sparealign}), $	Tr[\widetilde{R}(\rho_{sc})] $ is given by
		\begin{eqnarray}
			Tr[\widetilde{R}(\rho_{sc})] = Tr[\frac{p}{d^2}I_{d\otimes d}+\frac{(1-p)}{Tr[R(\rho_{sc})]}{R(\rho_{sc})}] = 1 
			\label{trace}
		\end{eqnarray}
		Applying Result-\ref{res-trnorm} on $\widetilde{R}(\rho_{sc})$, we get
		\begin{eqnarray}
			Tr[\widetilde{R}(\rho_{sc})] \leq ||\widetilde{R}(\rho_{sc})||_1
			\label{spr11}
		\end{eqnarray}
		Using (\ref{trace}), the inequality (\ref{spr11}) reduces to
		\begin{eqnarray}
			||\widetilde{R}(\rho_{sc})||_1 \geq 1
			\label{spineq1}
		\end{eqnarray}
		Both (\ref{r1}) and (\ref{spineq1}) holds only when
		\begin{eqnarray}
			||\widetilde{R}(\rho_{sc})||_1 = 1
			\label{r2}
		\end{eqnarray}
		Hence proved.
	\end{proof}
	We are now in a position to show that SPA-R criteria may reduce to the original form of realignment criteria for Schmidt-symmetric states. It may be expressed in the following theorem.
	{\theorem For Schmidt-symmetric state, SPA-R separability criterion reduces to the original form of realignment criterion.}
	\begin{proof} Let $\rho_{sc}^{sep}$ be any separable Schmidt-symmetric state. The SPA-R separability criterion for $\rho_{sc}^{sep}$ is given by
		\begin{eqnarray}
			||\widetilde{R}(\rho_{sc}^{sep})||_1 \leq [\widetilde{R}(\rho_{sc}^{sep})]_{UB}
			\label{f1} 
		\end{eqnarray}
		Using (\ref{lemma1}), the inequality (\ref{f1}) reduces to 
		\begin{eqnarray}
			&&[\widetilde{R}(\rho_{sc}^{sep})]_{UB} = \frac{p[Tr[R(\rho_{sc}^{sep})]-1]+1}{Tr[R(\rho_{sc}^{sep})]} \geq 1 \nonumber\\
			&\implies& \;\; p[Tr[R(\rho_{sc}^{sep})]-1] + 1 \geq Tr[R(\rho_{sc}^{sep})]\nonumber\\
			&\implies&\;\; Tr[R(\rho_{sc}^{sep})] (p-1) \geq (p-1)\nonumber\\
			&\implies&\;\; Tr[R(\rho_{sc}^{sep})] \leq 1 \nonumber\\
			&\implies&\;\; ||R(\rho_{sc}^{sep})||_1 \leq 1
		\end{eqnarray}
		The last step follows from (\ref{sch}). Hence proved.
	\end{proof}
	\section{Illustrations}
	Let us illustrate the SPA-R separability criteria with the help of a few examples and show that our criteria not only detect NPTES but also PPTES.
	{\example \label{eg3.1} Consider the family of two-qubit states $\rho(r,s,t)$ discussed in \cite{rudolph2003}. For $r=\frac{1}{4}$ and $s=\frac{1}{2}$, the family is represented by
		\begin{eqnarray}
			\rho_t=\frac{1}{2}
			\begin{pmatrix}
				\frac{5}{4} & 0 & 0 & t\\
				0 & 0 & 0 & 0\\
				0 & 0 & \frac{1}{4} & 0\\
				t & 0 & 0 & \frac{1}{2}
			\end{pmatrix} \label{rhot} ,~~  |t| \leq \frac{\sqrt{5}}{2\sqrt{2}} 
		\end{eqnarray}
		By PPT criterion, $\rho_t$ is entangled when $t \in [-\frac{\sqrt{5}}{2\sqrt{2}},\frac{\sqrt{5}}{2\sqrt{2}}] - \{0\}$. Realignment criteria detect the entangled states for $|t| > 0.116117$.\\
		Using the prescription given in (\ref{sparealign}), we construct the SPA-R map $\widetilde{R}: M_4 (\mathbb{C}) \longrightarrow M_4 (\mathbb{C})$ as
		\begin{eqnarray}
			\widetilde{R}(\rho_t) = \frac{p}{4} I_4 + \frac{(1-p)}{Tr[R(\rho_t)]} R(\rho_t); \quad 0 \leq p \leq 1
		\end{eqnarray}
		Our task is now to calculate the eigenvalues of $R(\rho_t)$ and thus we find the characteristic polynomial of the matrix $R(\rho_t)$. It can be expressed as
		\begin{eqnarray}
			f_1(x) = x^4 - a_1(t) x^3 + a_2(t) x^2 - a_3(t) x + a_4(t)
		\end{eqnarray}
		Using (\ref{rec}), we get
		\begin{eqnarray}
			&&a_1(t) = m_1 = t + \frac{7}{8}
			\label{a1t}\\
			&&a_2(t) =\frac{1}{2} (m_1^2 - m_2) = \frac{1}{32}(8t^2 +28 t + 5)\\
			&&a_3(t) = \frac{1}{6} ( m_1^3 - 3 m_1 m_2 + 2m_3)= \frac{1}{32}(7t^2 + 5t)\\
			&&a_4(t) = \frac{1}{24} (m_1^4 - 6 m_1^2 m_2 + 8 m_1 m_3 + 3 m_2^2 - 6m_4) = \frac{5}{128}t^2
		\end{eqnarray}
		where $m_k=Tr[(R(\rho_t))^k]$, $k=1,2,3,4$.\\
		The nature of the sign of the eigenvalues of $R(\rho_{t})$ may be determined by Descarte's rule of sign. To investigate, we divide the whole range of the state parameter $t$ into two parts (i) $t \in [0,0.790569]$ and (ii) $t \in [-0.790569,0]$.\\ 
		\textbf{Case 1:} If $t \in [0,0.790569]$ then $a_{i}\geq 0,\; i=1,2,3,4$, i.e. all the coefficients of the characteristic polynomial $f_1(-x)$ are positive. Therefore, there is no sign change in the ordered list of coefficients of $f_1(-x)$. Thus, $R(\rho_t)$ has no negative eigenvalue for $t\geq 0$. Hence $R(\rho_t)$ is a positive semi-definite operator for $t \in [0,0.790569]$. \\
		\textbf{Case 2:} In this case, we find that (i) if $t \in [-0.790569,-0.188751]$ then  $a_2(t) < 0$ and (ii) if $t \in [-0.714286, 0]$ then $a_3(t) < 0$. Thus, for every $t \in [-0.790569,0]$, atleast one coefficient of $f_1(x)$ is negative. Hence $R(\rho_t)$ has at least one negative eigenvalue, i.e.,  $R(\rho_t)$ is not positive semi-definite (PSD) for $t<0$.\\
		From the above analysis, it may be concluded that\\
		\textbf{ (i)} if $t \geq 0$, then the SPA-R map $\widetilde{R}(\rho_{t})$ defines a positive map $\forall p\in [0,1]$ and \\
		\textbf{(ii)} if $t < 0$ then according to the $Theorem-\ref{table3.1}$, SPA-R map $\widetilde{R}(\rho_{t})$ will be positive when the lower bound $l$ of the proportion $p$ is given by
		\begin{eqnarray}
			l = \frac{4k}{Tr[R(\rho_t)] + 4k}  =\frac{2(13-24t+8t^2)-\sqrt{3(67-112t+64t^2)}}{(-5+4t)^2} \equiv p_1 \label{p1}
		\end{eqnarray}
		Applying $Theorem-\ref{thm3.1}$, it can be shown that the approximated map $\widetilde{R}(\rho_t)$ is positive as well as completely positive for $p_1 \leq p \leq 1$.
		
		Thus, the SPA-R map $\widetilde{R}(\rho_t)$, which is a completely positive map may be suitable for detecting the entanglement in the family of states described by the density operator $\rho_{t}$.
		Now we apply our separability criterion discussed in $Theorem-\ref{thm3.2}$ which involves the comparison of  $||\widetilde{R}(\rho_t)||_1$ and the upper bound $[\widetilde{R}(\rho_t)]_{UB}$.
		After a few steps of the calculation, we obtain
		\begin{eqnarray}
			||\widetilde{R}(\rho_t)||_1 > [\widetilde{R}(\rho_t)]_{UB};\;\text{for}\; \begin{array}{lrl}
				t \in (-0.790569, -0.665506] &\text{when}& p_1 \leq p < p_2 \\
				t \in (0.116117, 0.125] &\text{when}& 0 \leq p < p_3 \\
				t \in (0.125, 0.790569] &\text{when}& 0 \leq p\leq 1
			\end{array}\label{rhoteq} 
		\end{eqnarray}
		where 
		\begin{eqnarray}
			p_2 &=& \frac{(-91 - 48 t - 64 t^2) -
				\sqrt{8673 + 9632 t - 8832 t^2 - 6144 t^3 + 4096 t^4}}{2
				(-7 + 48 t)^2}\\
			p_3 &=&  \frac{(14 - 128 t + 64 t^2)}{(7 - 80 t + 128 t^2)}
		\end{eqnarray}
		Thus, the inequality (\ref{thm3.2eq}) is violated for $t \in [-0.790569,-0.665506) \cup (0.116117, 0.790569]$ which implies that in this range of $t$, the state $\rho_{t}$ is entangled.\\
		The comparison of $||\widetilde{R}(\rho_t)||_1$ and $[\widetilde{R}(\rho_t)]_{UB}$ for the two-qubit state $\rho_t$ has been studied in Fig-\ref{fig3.1} for different range of $t$ given in (\ref{rhoteq}). 
		\begin{figure}[h!]
			\begin{center}
				\includegraphics[scale=.4]{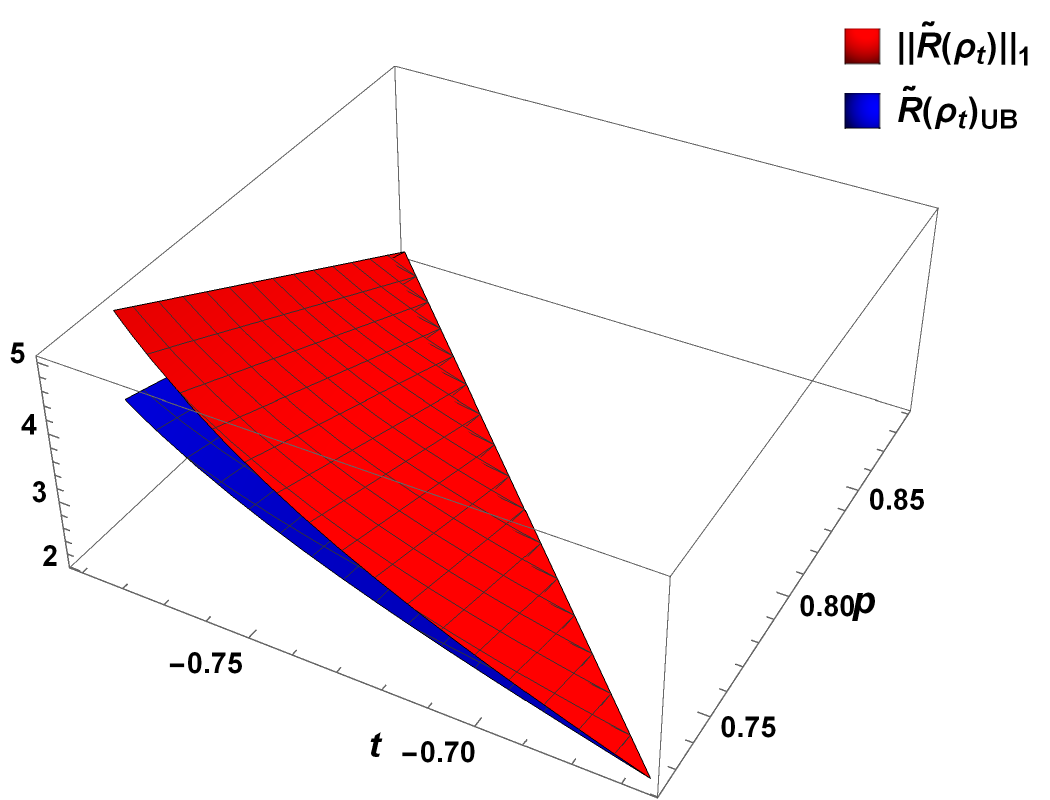}
				4.1(i) 
				\includegraphics[scale=.3]{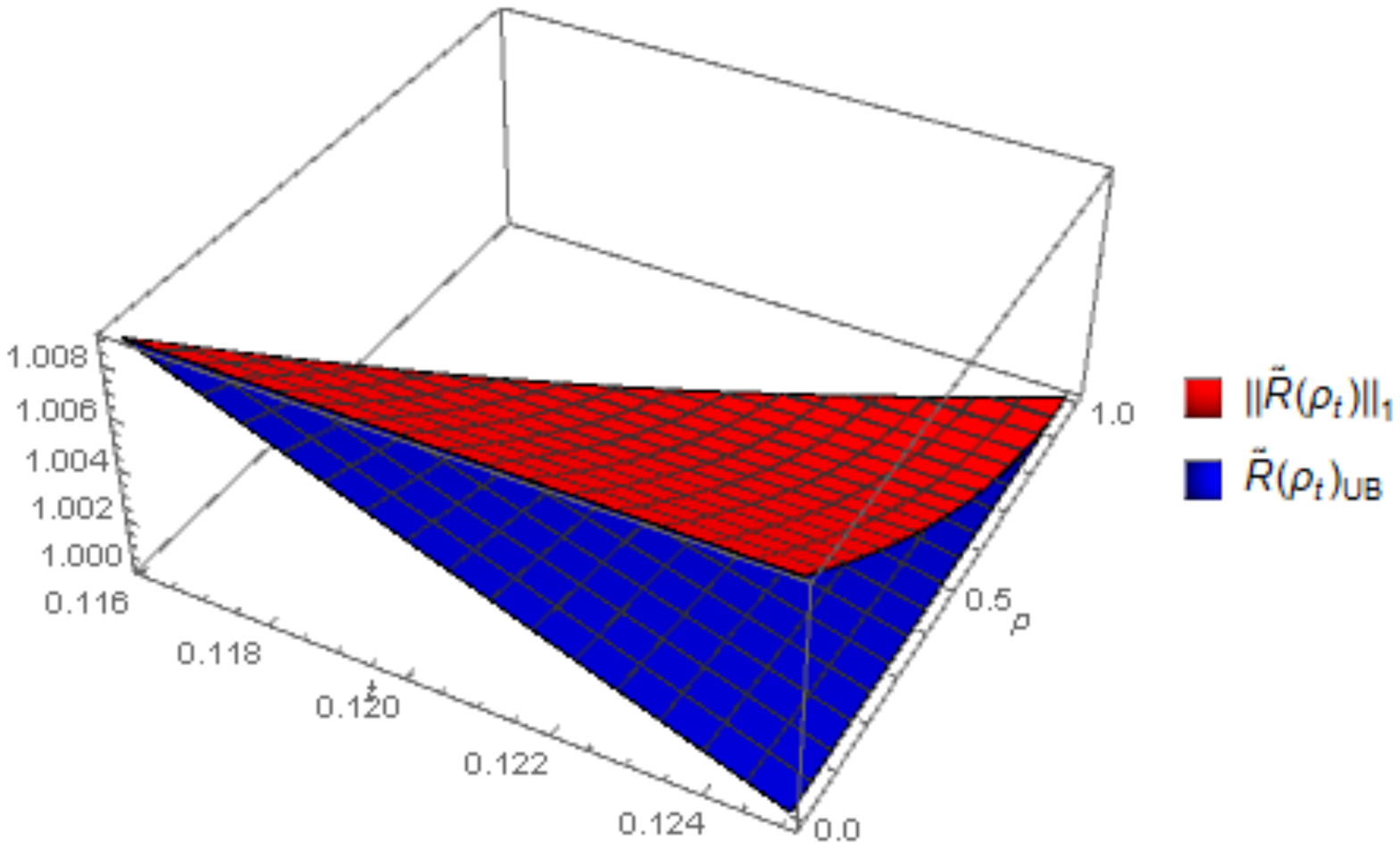}
				4.1(ii)
				\includegraphics[scale=.3]{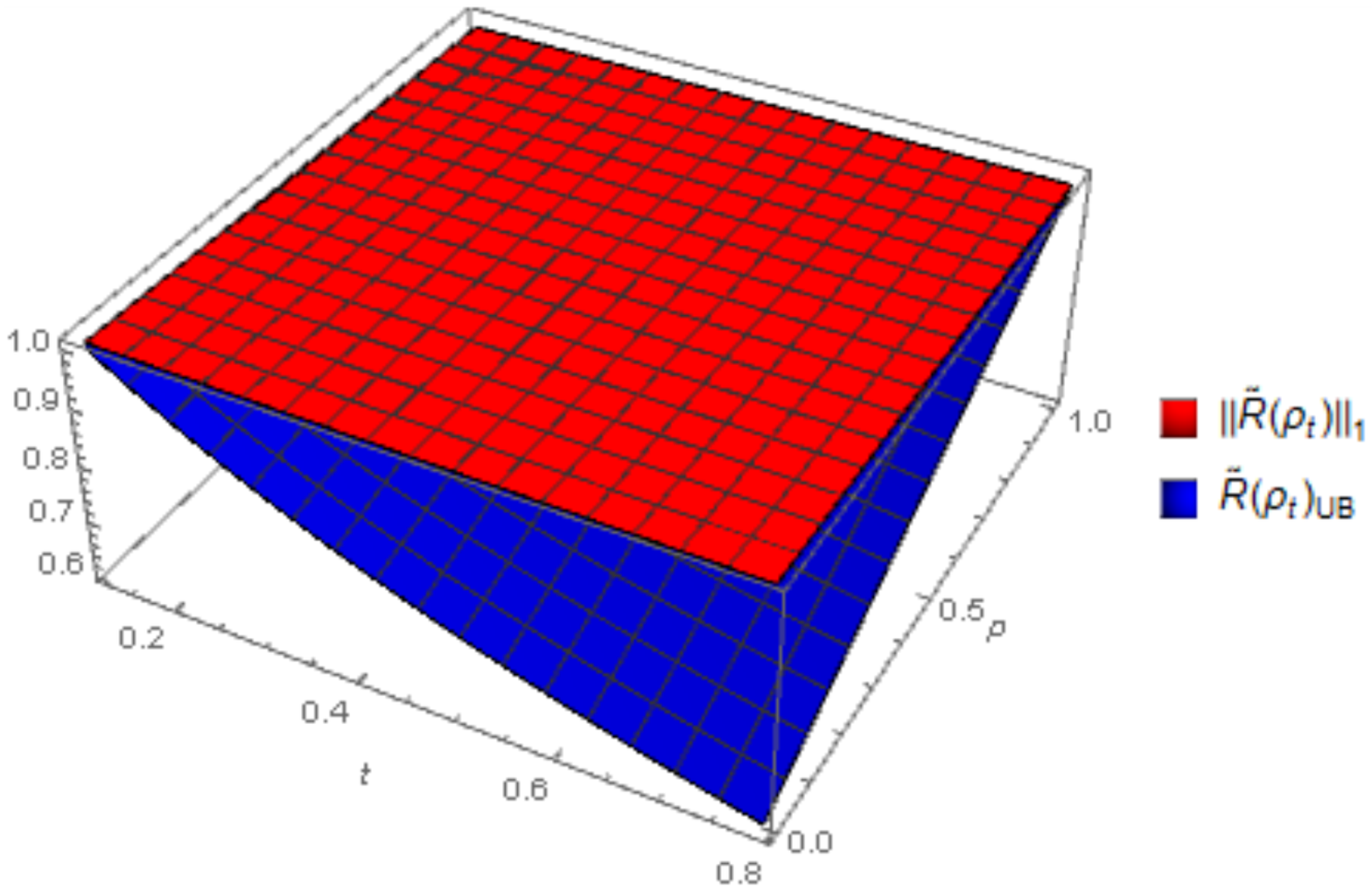}
				4.1(iii)
				\caption{The comparison between the $||\widetilde{R}(\rho_t)||_1$ and $[\widetilde{R}(\rho_t)]_{UB}$ for the two-qubit state $\rho_{t}$ has been displayed. In Fig 4.1(i), one can observe that the inequality (\ref{thm3.2eq}) obtained in $Theorem-\ref{thm3.2}$ is violated when $-0.790569 \leq t < -0.665506$ for $p \in [p_1,p_2]$ whereas in Fig 4.1(ii) the inequality is violated when $0.116117< t \leq 0.125$ and $p$ lies in the interval $[0,p_3)$. Fig. 4.1(iii) shows the violation of inequality (\ref{thm3.2eq}) when $t > 0.125$ and $0 \leq p \leq 1$.} \label{fig3.1}
			\end{center}
		\end{figure}
		\example \label{eg3.2} Consider a two-qutrit state defined in \cite{swapan}, which is described by the density operator
		\begin{eqnarray}
			\rho_{\mu} = \frac{1}{5+2{\mu}^2}\sum_{i=1}^{3}{|\psi_i\rangle \langle \psi_i|},~~\frac{1}{\sqrt{2}}\leq \mu \leq 1
			\label{eg2}
		\end{eqnarray}
		where $|\psi_i\rangle=|0i\rangle-\mu|i0\rangle$, for $i=\{1,2\}$ and $|\psi_3\rangle=\sum_{i=0}^{2}{|ii\rangle}$.\\
		The state described by the density operator $\rho_{\mu}$ is NPTES \cite{swapan}.
		Using the prescription given in (\ref{sparealign}), we construct the SPA-R map $\widetilde{R}: M_9 (\mathbb{C}) \longrightarrow M_9 (\mathbb{C})$ as
		\begin{eqnarray}
			\widetilde{R}(\rho_{\mu}) = \frac{p}{9} I_9 + \frac{(1-p)}{Tr[R(\rho_{\mu})]}{R(\rho_{\mu})},~ 0 \leq p \leq 1
		\end{eqnarray}
		The eigenvalues of $R(\rho_{\mu})$ can be calculated by finding the characteristic equation of the matrix $R(\rho_{\mu})$. The characteristic polynomial can be expressed as
		\begin{eqnarray}
			f_2(x) = x^9 - a_1(\mu) x^8 +a_2(\mu) x^7 - a_3(\mu) x^6 + a_4(\mu) x^5  - a_5(\mu) x^4  + a_6(\mu) x^3 - a_7(\mu) x^2 + a_8(\mu) x - a_{9}(\mu)\nonumber
		\end{eqnarray}
		where the coefficients $a_i(\mu)$, $i=1,2,...,9$ are given by
		\begin{eqnarray}
			a_1({\mu})&=&\frac{9}{5+2{\mu}^2},\quad a_2({\mu}) = - \frac{4(-9+{\mu}^2)}{(5+2{\mu}^2)^2},\quad  a_3({\mu}) = -\frac{28(-3+{\mu}^2)}{(5+2{\mu}^2)^3},\nonumber\\a_4({\mu}) &=& \frac{126-84{\mu}^2+5{\mu}^4}{(5+2{\mu}^2)^4}, \quad a_5({\mu}) = \frac{126-140{\mu}^2+25{\mu}^4}{(5+2{\mu}^2)^5},  \; \; a_6({\mu}) = -\frac{2(-42+70{\mu}^2-25{\mu}^4+{\mu}^6)}{(5+2{\mu}^2)^6}, \nonumber\\
			a_7({\mu}) &=&-\frac{2(-18+42{\mu}^2-25{\mu}^4+3{\mu}^6)}{(5+2{\mu}^2)^7},\quad a_8({\mu}) =- \frac{-9+28{\mu}^2-25{\mu}^4+6{\mu}^6}{(5+2{\mu}^2)^8}\nonumber\\
			a_9({\mu}) &=&-\frac{(-1+{\mu}^2)^2(-1+2{\mu}^2)}{(5+2{\mu}^2)^9}
		\end{eqnarray}
		From the coefficients of $f_2(x)$, it can be observed that atleast one coefficient of $f_2(x)$ is negative. This means $R(\rho_{\mu})$ has atleast one negative eigenvalue. Using Descarte's rule of sign, we find that $R(\rho_{\mu})$ is not a positive semi-definite operator. \\
		Using $Theorem-\ref{thm3.1}$, the approximated map $\widetilde{R}(\rho_{\mu})$ is positive as well as completely positive when the lower bound $l$ of the proportion $p$ is given as
		\begin{eqnarray}
			l = \frac{-1+15\sqrt{2}w+6\sqrt{2}{\mu}^2w}{3\sqrt{2}(5+2{\mu}^2)w},~w=\sqrt{\frac{1}{56+9{\mu}^2(5+{\mu}^2)}} \label{spa-eg2-l}
		\end{eqnarray}
		Since the SPA-R map $\widetilde{R}(\rho_{\mu})$ is positive as well as completely positive for $p \in [l,1]$ where $l$ is given in (\ref{spa-eg2-l}), so it is suitable for detecting the entanglement in the state $\rho_{\mu}$ experimentally.\\
		Now we apply our separability criterion discussed in $Theorem-\ref{thm3.2}$ which involves the comparison of  $||\widetilde{R}(\rho_{\mu})||_1$ and the upper bound $[\widetilde{R}(\rho_{\mu})]_{UB}$ defined in (\ref{thm3.2eq}).
		For $\frac{1}{\sqrt{2}}\leq {\mu} \leq 1$, we find that
		\begin{eqnarray}
			||\widetilde{R}(\rho_{\mu})||_1 > [\widetilde{R}(\rho_{\mu})]_{UB}
		\end{eqnarray}
		The comparison of $||\widetilde{R}(\rho_{\mu})||_1$ and $[\widetilde{R}(\rho_a)]_{UB}$ for the two-qutrit state $\rho_{\mu}$ has been studied in Fig-\ref{fig3.2}.
		From Fig-\ref{fig3.2}, it is evident that the inequality (\ref{thm3.2eq}) obtained in $Theorem-\ref{thm3.2}$ is violated. Thus, the state described by the density operator $\rho_{\mu}$ is an entangled state.\\
		\begin{figure}[h!]
			\centering
			\includegraphics[scale=.4]{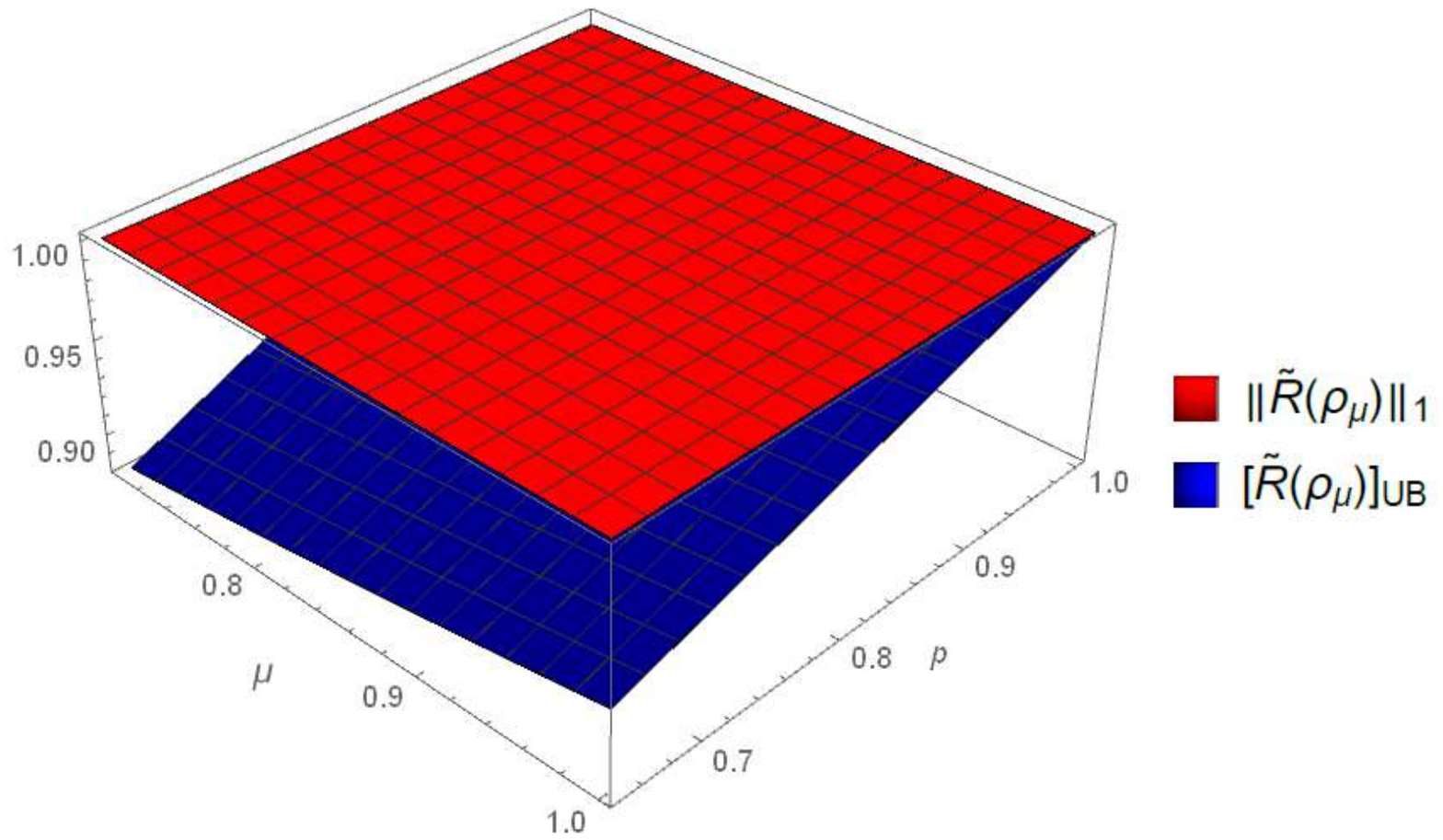}
			\caption{The comparison between the $||\widetilde{R}(\rho_{\mu})||_1$ and $[\widetilde{R}(\rho_{\mu})]_{UB}$ for the two-qutrit state $\rho_{\mu}$ has been displayed. It has been observed that the inequality (\ref{thm3.2eq}) is violated for $\rho_{\mu}$ in the whole range of $\mu$ and for $p \in [l, 1]$} \label{fig3.2}
		\end{figure}
		
		\example \label{eg3.3} Let us consider a two-qutrit isotropic state described by the density operator $\rho_{\beta}$ \cite{bert2008}
		\begin{eqnarray}
			\rho_{\beta}=\beta|\phi_{+}\rangle \langle \phi_{+}|+\frac{1-\beta}{9}I_9, ~~-\frac{1}{8}\leq \beta \leq 1
			\label{betastate}
		\end{eqnarray}
		where $I_9$ denotes the identity matrix of order 9 and the state $|\phi_{+}\rangle$ represents a Bell state in a two-qutrit system and may be expressed as
		\begin{eqnarray}
			|\phi_{+}\rangle=\frac{1}{\sqrt{3}}(|11\rangle +|22\rangle +|33\rangle)
			\label{bellstate}
		\end{eqnarray}
		Using realignment criteria, the state $\rho_{\beta}$ is an entangled state for $\frac{1}{3}< \beta \leq 1$.\\
		Let us calculate the eigenvalues of $R(\rho_\beta)$.  The characteristic polynomial of $R(\rho_\beta)$ is given by
		\begin{eqnarray}
			f_3(x) = x^9 - a_1(\beta) x^8 +a_2(\beta) x^7 - a_3(\beta) x^6 + a_4(\beta) x^5  - a_5(\beta) x^4  + a_6(\beta) x^3 - a_7(\beta) x^2 + a_8(\beta) x - a_{9}(\beta)\nonumber
		\end{eqnarray}
		where the coefficients $a_i(\beta)$, $i=1,2,...,9$, may be expressed as
		\begin{eqnarray}
			a_1(\beta) = \frac{1}{3}(1+8\beta),  \; \; a_2(\beta) = \frac{4}{9}\beta(2+7\beta), \; \; a_3(\beta) = \frac{28}{27}\beta^2(1+2\beta),  \; \; \nonumber\\ a_4(\beta) = \frac{14}{81}\beta^3(4+5\beta),\;\; a_5(\beta) = \frac{14}{243}\beta^4(5+4\beta),  \; \; a_6(\beta) = \frac{28}{729}\beta^5(2+\beta), \nonumber\\
			a_7(\beta) = \frac{4}{2187}\beta^6(7+2\beta),  \; \; a_8(\beta) = \frac{1}{6561}\beta^7(8+\beta),\;\;
			a_9(\beta) =\frac{\beta^8}{19623}
		\end{eqnarray}
		Since all the coefficients $a_i(\beta),~~i=1~ \text{to} ~9$, of $f_3(-x)$ are positive, realignment matrix $R(\rho_{\beta})$ is positive semi-definite.  
		Using Descarte's rule of sign, we find that the realignment matrix $R(\rho_{\beta})$ is positive semi-definite for $0 \leq p \leq 1$.\\
		The comparison between $||\widetilde{R}(\rho_{\beta})||_1$ and $[\widetilde{R}(\rho_{\beta})]_{UB}$ has been studied in Fig-\ref{fig3.3}. 
		\begin{figure}[h!]
			\centering
			\includegraphics[scale=.33]{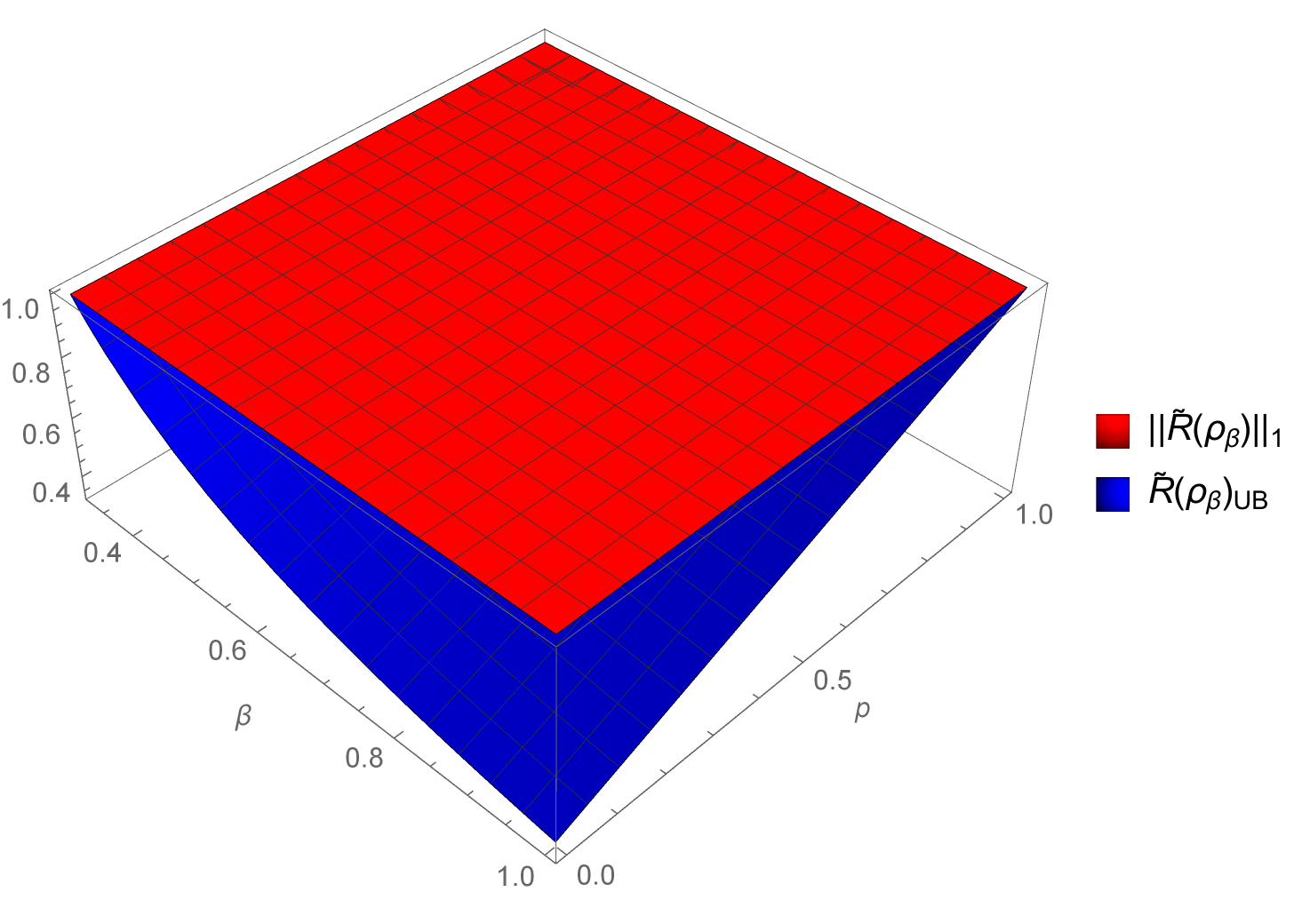}
			\caption{The comparison between the $||\widetilde{R}(\rho_{\beta})||_1$ and $[\widetilde{R}(\rho_{\beta})]_{UB}$ for the two-qutrit state $\rho_{\beta}$ has been displayed. It has been observed that the inequality (\ref{thm3.2eq}) is violated for all values of $\beta \in (1/3, 1]$ and for any $p \in [0,1]$.} \label{fig3.3}
		\end{figure}
		From Fig-\ref{fig3.3}, it can be observed that the inequality (\ref{thm3.2eq}) is violated for $\frac{1}{3}< \beta \leq 1$ and $p\in [0,1]$. Thus, using $Theorem-\ref{thm3.2}$, the state described by the density operator $\rho_{\beta}$ is an entangled state.\\
		\example \label{eg3.4} Consider the state described by the density operator $\rho_{a}$, for $0\leq a \leq 1$  defined in (\ref{astate-matrix}). It has been shown that this state is PPTES for $0<a<1$ \cite{horodeckia}. \\
		The eigenvalues of $R(\rho_{a})$ can be calculated by finding the characteristic equation of the matrix $R(\rho_{a})$. The characteristic polynomial can be expressed as
		\begin{eqnarray}
			f_4(x) = x^9 - a_1(a) x^8 +a_2(a) x^7 - a_3(a) x^6 + a_4(a) x^5 - a_5(a) x^4  + a_6(a) x^3 - a_7(a) x^2 + a_8(a) x - a_{9}(a)\nonumber
		\end{eqnarray}
		where the coefficients $a_i(a)$, $i=1$ to $9$ are given as
		\begin{eqnarray}
			a_1(a) = \frac{1 + 17a}{2(1 + 8a)},\quad a_2(a) = \frac{a(7+ 59a)}{2(1+8a)^2},\quad a_3(a) = \frac{a^2(21 + 109a)}{2(1+8a)^3},\quad a_4(a) = \frac{5a^3(7+ 23a)}{2(1+8a)^4}, \nonumber\\ a_5(a) = \frac{a^4(35 + 67a)}{2(1+8a)^5},\quad a_6(a) = \frac{a^5(21 + 17a)}{2(1+8a)^6},\quad
			a_7(a) = \frac{\alpha^6(7 - a)}{2(1+8a)^7},\quad a_8(a) = \frac{a^7(1-a)}{2(1+8a)^8},\quad a_9(a) =0\nonumber
		\end{eqnarray}
		Now since $a_i(a) \geq 0$ for $i=1$ to $9$, using Descarte's rule of sign, $R(\rho_{a})$ is PSD. Hence, by $Theorem-\ref{thm3.1}$, $\widetilde{R}(\rho_{a})$ defines a positive map for $0 \leq p \leq 1$ and for all $a \in (0,1)$.	 Further, using the section \ref{sec-cp}, it can be easily shown that the SPA-R map $\widetilde{R}(\rho_{a})$ is completely positive for any $p \in[0,1]$. \\
		It has been observed that the inequality (\ref{thm3.2eq}) is violated for different ranges of $p$ and for some values of $a$, which is shown in table \ref{table3.1}.
		\begin{table}[h!]
			\begin{center}
				\begin{tabular}{ p{2.0cm} p{4cm}  p{2.5cm} }
					\hline
					$a$   & $Range~of~p$   & $Theorem-\ref{thm3.2}$  \\ 
					\hline
					$0.1$  &$0\leq p \leq 0.019383$  & $Violated$ \\
					$0.2$  &$0\leq p \leq 0.022143$  & $Violated$\\
					$0.3$  &$0\leq p \leq 0.021903$ & $Violated$\\
					$0.4$  &$0\leq p \leq 0.020444$ & $Violated$\\
					$0.5$  &$0\leq p \leq 0.018284$ & $Violated$\\
					$0.6$  &$0\leq p \leq 0.015611$ & $Violated$\\
					$0.7$  &$0\leq p \leq  0.012488$  & $Violated$\\
					$0.8$  &$0\leq p \leq 0.008904$  & $Violated$\\
					$0.9$  &$0\leq p \leq  0.004791$  & $Violated$\\
					\hline   
				\end{tabular}
			\end{center}
			\caption{The table shows the range of the probability $p$ for which the inequality (\ref{thm3.2eq}) is violated for different values of the state parameter $a$}.
			\label{table3.1}
		\end{table}
		Thus, the inequality given in $Theorem-\ref{thm3.2}$ is violated by $\rho_{a}$, and hence our criterion detects the bound entangled state given by (\ref{astate-matrix}).
		\section{Efficiency of SPA-R Criterion}
		We now show how the SPA-R criterion is efficient in comparison to other entanglement detection criteria. In particular, we are considering three entanglement detection criteria such as (a) Zhang's separability criterion based on realigned moment \cite{tzhang} and (b) $R$-moment criterion discussed in section \ref{sec-rmoment} for comparing the efficiency of SPA-R criterion.
		\subsection{Comparing SPA-R and Zhang's realignment moment based criterion}
		We employ Example-\ref{eg3.1} and Example-\ref{eg3.4} to compare the SPA-R criterion with Zhang's realignment moment based criterion. \\
		\textbf{(i)} Let us recall Example-\ref{eg3.1}, where the family of states is described by the density operator $\rho_t$. Interestingly, for this family of states when $t>0$, our SPA-R criteria detects entanglement in the region $t \in (0.116117, 0.790569]$. But Zhang's realignment moment based criteria given in (\ref{L_4}) detect the entangled state in the range $t\in (0.370992,  0.790569]$. Clearly, SPA-R criteria detects the NPTES $\rho_t$ for $t>0$ in a better range than Zhang's criteria.\\
		\textbf{(ii)} Let us consider the BES studied in Example-\ref{eg3.4}, which is described by the density operator $\rho_{a}$, $0< a <1$. As shown in table \ref{table3.1}, the state $\rho_a$ is detected by SPA-R criteria. \\
		Using (\ref{eq-r_k}), Zhang's realignment moment for a bipartite state $\rho_{a}$ may be defined as  
		\begin{eqnarray}
			r_k (R(\rho_{a})) = Tr[R(\rho_{a}) (R(\rho_{a}))^{\dagger}]^{k/2},\; k = 1, 2, 3, . . ., 9 \label{zhangdef}
		\end{eqnarray}
		As given in (\ref{L_4}), the separability criterion based on realignment moments $r_2$ and $r_3$ may be stated as: If a quantum state $\rho_{a}$ is separable, then
		\begin{eqnarray}
			Q_1 = (r_2(R(\rho_{a}))^2 - r_3(R(\rho_{a}) \leq 0	\label{rzhang}
		\end{eqnarray}
		$Q_1 > 0$ certifies that the given state is entangled.\\
		Fig-\ref{q2al} shows that the inequality (\ref{rzhang}) is not violated for the BES $\rho_{a}$ in the whole range $0 < a < 1$. Hence the BES $\rho_{a}$ is undetected by Zhang's realignment moment based criteria.
		\begin{figure}[h!]
			\centering
			\includegraphics[scale=.48]{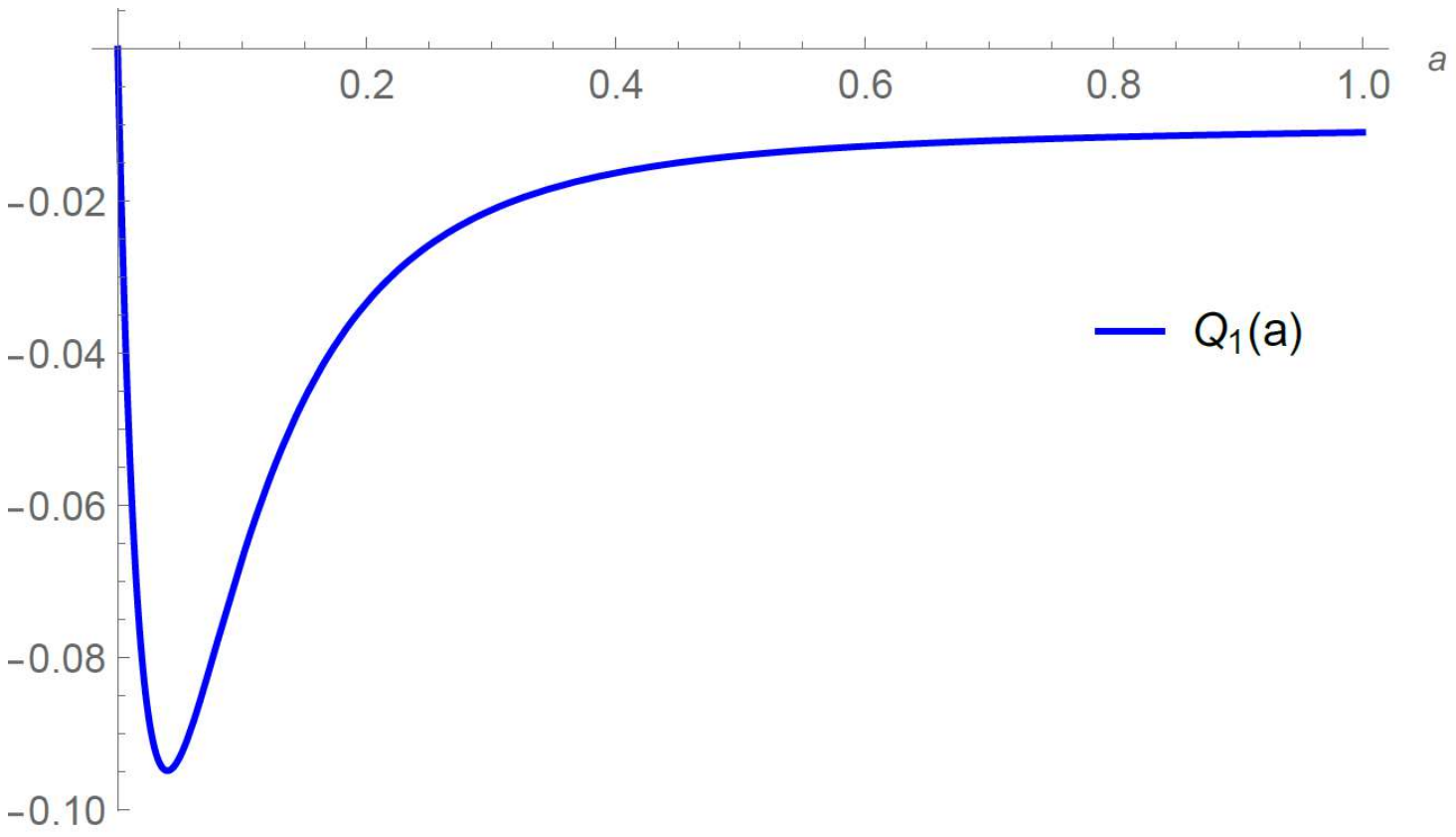}
			\caption{The blue curve represents $Q_1$ for the state $\rho_{a}$ and \textit{x}-axis depicts the state parameter $a$.}.
			\label{q2al}
		\end{figure} 
	}
	\subsection{Comparing SPA-R and $R$-moment criterion}
	\noindent Let us again recall Example-\ref{eg3.1} and Example-\ref{eg3.4} to compare the SPA-R criterion with the $R$-moment criterion.\\ 
	\textbf{(i)} In the Example-\ref{eg3.1}, the family of states described by the density operator $\rho_t$ and it is detected by SPA-R criteria in the range $(0.116117, 0.790569]$. By $R$-moment criterion given in section \ref{sec-rmoment}, $\rho_t$ is detected when $t \in (0.214312, 0.790569] \subset (0.116117, 0.790569] $. Therefore, SPA-R criteria detect' more entangled states than the $R$-moment criterion.\\
	\textbf{(ii) }Let us now consider the BES studied in Example-\ref{eg3.4} and find that the state is detected by SPA-R criteria. Applying $R$-moment criterion on the BES described by the density operator $\rho_{a}$, $0<a<1$, we get
	\begin{eqnarray}
		Q_2 \equiv 	56 D_8^{1/8} + T_1 - 1 \leq 0 \;\; \forall \;  a \in (0,1) \label{ineqal}
	\end{eqnarray}
	where $D_8 = \prod_{i=1}^8 \sigma_i^2(\rho_{a})$ and $T_1 = Tr[R(\rho_{a})]$. Here $\sigma_i(\rho_{a})$ represents the $ith$ singular value of $\rho_{a}$.
	Since the above inequality is not violated for any $a \in (0,1)$, the BES $\rho_{a}$ is undetected by  $R$-moment based criteria. This is shown in Fig-\ref{q1al}.
	\begin{figure}[h!]
		\centering
		\includegraphics[scale=.48]{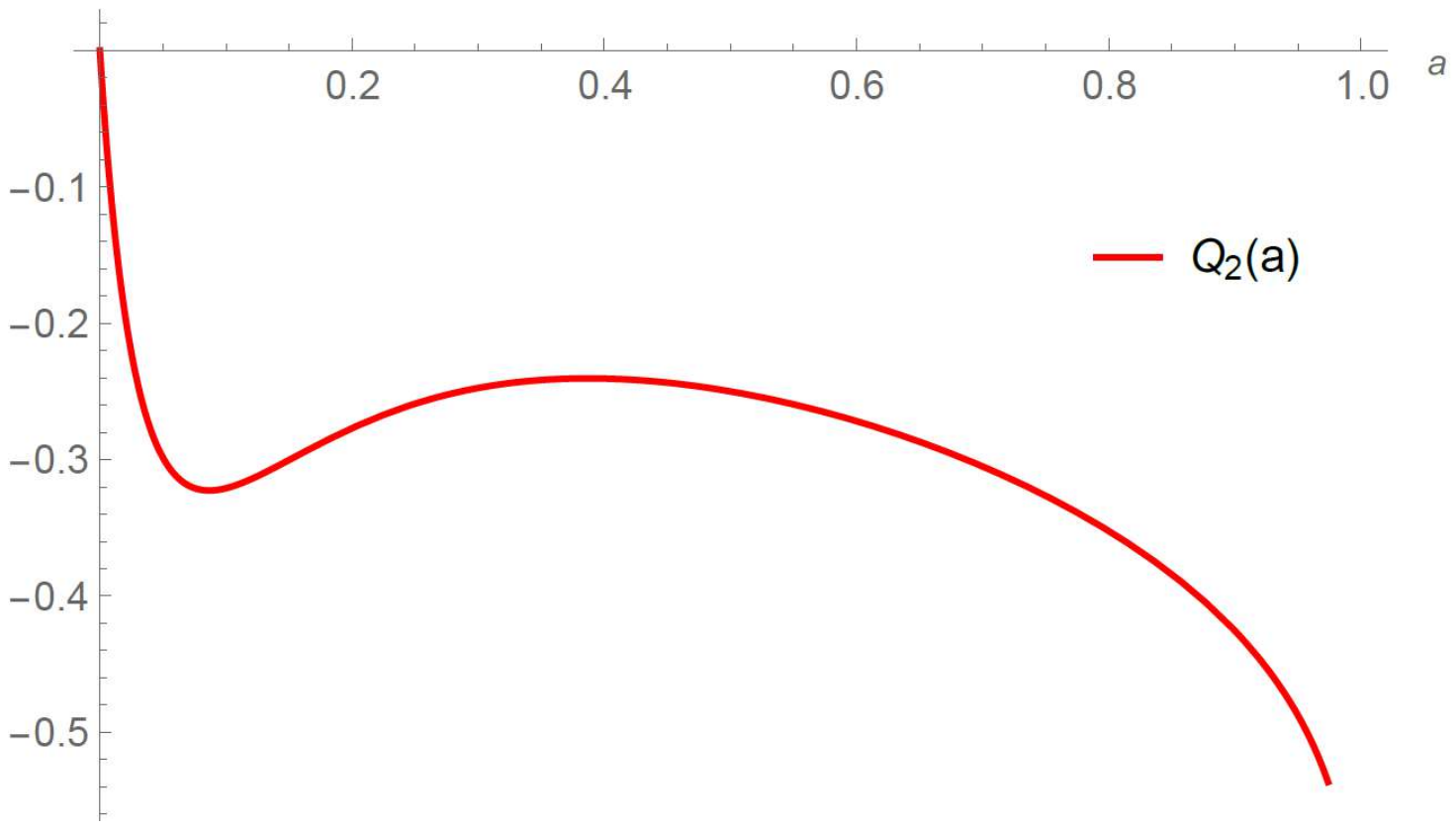}
		\caption{The red curve represents $Q_2$ for the state $\rho_{a}$ and $x$-axis depicts the state parameter $a$.}.
		\label{q1al}
	\end{figure} 
	
	\section{Error in the Approximated Map}
	\noindent We now study and analyze the error generated when $R(\rho)$ is approximated by its SPA. 
	In the approximated map, we have added an appropriate proportion of maximally mixed state such that the approximated map has no negative eigenvalue. The error between the approximated map $\widetilde{R}(\rho)$ and the realignment map $R(\rho)$ may be calculated as:
	\begin{eqnarray}
		||\widetilde{R}(\rho)-R(\rho)||_1&=&||\frac{p}{d^2}I_{d\otimes d}+\frac{(1-p)}{Tr[R(\rho)]}R(\rho)-R(\rho)||_1\nonumber\\
		&=&||\frac{p}{d^2}I_{d\otimes d}+[\frac{(1-p)}{Tr[R(\rho)]}-1]R(\rho)||_1
		\label{rrho1}
	\end{eqnarray}
	Using triangular inequality for trace norm, (\ref{rrho1}) reduces to
	\begin{eqnarray}
		||\widetilde{R}(\rho)-R(\rho)||_1\leq p+{\frac{1-p-Tr[R(\rho)]}{Tr[R(\rho)]}}||R(\rho)||_1
		\label{error}
	\end{eqnarray}
	The inequality (\ref{error}) may be termed as error inequality. The error inequality holds for any two-qudit bipartite state.
	{\remark The error inequality (\ref{error}) is not defined for Schmidt symmetric entangled states.}\\
	\textbf{Proposition 1:} The equality relation   
	\begin{eqnarray}
		||\widetilde{R}(\rho_{sep})-R(\rho_{sep})||_1= \frac{(1-p)(1-Tr[R(\rho_{sep})])}{Tr[R(\rho_{sep})]}
		\label{errorsep1}
	\end{eqnarray}
	holds for separable state described by the density operator $\rho_{sep}$ such that $||R(\rho_{sep})||_1=1$.
	\begin{proof} 
		Equality in (\ref{error}) holds if and only if 
		\begin{eqnarray}
			\frac{p}{d^2}I_{d\otimes d}=[\frac{(1-p)}{Tr[R(\rho)]}-1]R(\rho)
		\end{eqnarray}
		i.e. equality in (\ref{error}) holds when the realigned matrix takes the form
		\begin{eqnarray}
			R(\rho)=\frac{pTr[R(\rho)]}{1-p-Tr[R(\rho)]}\frac{I}{d^2},~~0\leq p\leq 1
			\label{eq22}
		\end{eqnarray} 
		Taking trace norm, (\ref{eq22}) reduces to
		\begin{eqnarray}
			||R(\rho)||_1=\frac{pTr[R(\rho)]}{1-p-Tr[R(\rho)]},~~0\leq p\leq 1
			\label{eq23}
		\end{eqnarray}
		For separable state $\rho_{sep}$, (\ref{eq23}) reduces to
		\begin{eqnarray}
			1=\frac{pTr[R(\rho_{sep})]}{1-p-Tr[R(\rho_{sep})]},~~0\leq p\leq 1
			\label{eq24}
		\end{eqnarray}
		Simplifying (\ref{eq24}), the value of $p$ and $1-p$ may be expressed as
		\begin{eqnarray}
			p=\frac{1-Tr[R(\rho_{sep})]}{1+Tr[R(\rho_{sep})]},~~1-p=\frac{2Tr[R(\rho_{sep})]}{1+Tr[R(\rho_{sep})]}
		\end{eqnarray}
		Substituting values of $p$ and $1-p$ in (\ref{eq22}), the realigned matrix for separable state, $R(\rho_{sep})$ takes the form
		\begin{eqnarray}
			R(\rho_{sep})=\frac{1}{d^2}I
			\label{mes}
		\end{eqnarray}
		Therefore, (\ref{mes}) holds only for separable states. This means that there exists a separable state $\rho_{sep}$ such that $||\rho_{sep}||_{1}=1$ for which the equality condition in the error inequality (\ref{error}) holds.
	\end{proof}
	{\result If any quantum system described by a density operator $\rho$ in $d\otimes d$ system is separable then the error inequality is given by
		\begin{eqnarray}
			||\widetilde{R}(\rho)-R(\rho)||_1&\leq&  \frac{(1-p)[1-Tr[R(\rho)]]}{Tr[R(\rho)]}
			\label{errorsep1}
	\end{eqnarray}}
	\begin{proof} Let us consider a $d \otimes d$ dimensional separable state $\rho_{sep}$. Using realignment criteria in $Theorem-\ref{thm-ccnr}$, we have $||R(\rho_{sep})||_{1}\leq 1$. Therefore, error inequality (\ref{error}) reduces to
		\begin{eqnarray}
			||\widetilde{R}(\rho_{sep})-R(\rho_{sep})||_1&\leq& p+{\frac{1-p-Tr[R(\rho_{sep})]}{Tr[R(\rho_{sep})]}}\nonumber\\
			&=& \frac{(1-p)[1-Tr[R(\rho_{sep})]]}{Tr[R(\rho_{sep})]}
			\label{errorsep}
		\end{eqnarray}
		Hence proved.
	\end{proof}
	{\corollary If inequality (\ref{errorsep1}) is violated by any bipartite $d \otimes d$ dimensional quantum state, then the state under investigation is entangled.}
	
	\section{Conclusion}
	\noindent To summarize, we have developed a separability criterion by approximating realignment operation via structural physical approximation (SPA). Since the partial transposition (PT) operation is limited to detecting only NPTES, we have studied here the realignment operation, which may detect both NPTES and PPTES. However, since realignment map is not a positive map and thus it does not represent a completely positive map, it is difficult to implement it in a laboratory. Therefore, in order to make the realignment map completely positive, firstly, we have approximated it to a positive map using the method of SPA and then we have shown that this approximated map is also completely positive. We have shown that the positivity of the SPA-R map can be verified in an experiment because the lower bound of the fraction $p$ can be expressed in terms of the first and second moments of the realignment matrix. Interestingly, we have shown that the derived separability criterion using the approximated (SPA-R) map detects bipartite NPTES and PPTES.  Some examples are cited to support our obtained results. Although there are other PPT criteria that may detect NPTES and PPTES, our result is interesting in the sense that it may be realized in an experiment. Our obtained results may be realized in an experiment but to achieve this aim, we pay a price in terms of the short range detection. This fact can be observed in Example-\ref{eg3.1} where the range of the state parameter for the detection of entangled state is smaller than the range obtained by usual realignment operation (without approximation). We also have analyzed the error that occurred during the structural physical approximation of the realignment map and it is described by an inequality known as error inequality. Continuing with the error inequality, we have obtained another inequality that is satisfied by all bipartite $d\otimes d$ dimensional separable states, and the violation of this inequality guarantees the fact that the state under probe is entangled. Interestingly, the SPA-R criteria coincide with the original realignment criteria for Schmidt-symmetric states.\\
	Moreover, we have presented a scheme for the measurement of the moments of the realignment matrix in order for our entanglement detection criterion to be realized in practice. We conclude by noting that our proposed entanglement detection approach should be experimentally implementable through the combination of techniques associated with structural physical approximation for realignment \cite{horoekert2002}, and SWAP operations for measuring density matrix moments \cite{barti, sougato}.
	
	\begin{center}
		****************
	\end{center}

\chapter{Detection and  Quantification of Entanglement Through Witness Operator}\label{ch5}
\vspace{1cm}
\noindent{\small \emph{If you think you understand quantum mechanics, you don't understand quantum mechanics.\\
		-Richard P. Feynman}}
\vspace{2cm}
\noindent \hrule
\noindent \emph{ In this chapter\;\footnote {This chapter is based on the published research article, ``S. Aggarwal, S. Adhikari, \emph{Witness operator provides better estimate of the lower bound of concurrence of bipartite bound entangled states in $d_1 \otimes d_2$ dimensional system}, {Quantum Information Processing} \textbf{20}, 83, (2021)"}
	we take an analytical approach to construct a family of witness operators detecting NPTES and BES in arbitrary dimensional bipartite quantum systems. The introduced family of witness operators is then used to estimate the lower bound of concurrence of the detected mixed bipartite entangled states. Next, we show that our lower bound estimate concurrence is better as compared to the lower bound of the concurrence given by Chen et al. (Phys. Rev. Lett. 95, 040504, 2005).}
\noindent\hrulefill
\newpage
\section{Introduction}\label{sec5.1}
Since bound entangled states (BES) are very weak entangled states, they behave like separable states and thus it is very difficult to separate BES from the set of separable states. Hence, in particular, the detection of BES is an important problem to consider. Also, researchers have found many applications of bipartite BES in quantum cryptography \cite{phorodecki1},  metrology \cite{gt}, and non-locality \cite{tn}.  There are some powerful entanglement detection criteria such as the partial transposition criterion, and realignment criterion but it may not be possible to implement them successfully in the experiment. This situation can be avoided if the entanglement is detected through the construction of a witness operator. Entanglement witness operator plays a significant role in the entanglement detection problem since if we have some prior partial information about the state which is to be detected then entanglement can be detected in the experiment by the construction of witness operator \cite{guhnerev}. Thus, by realizing its importance, we have constructed a family of witness operators, denoted by $W_{(n)}$, for the detection of entangled states, particularly, BES.\\
Secondly, it is known that for higher dimensional systems, we do not have any closed formula for concurrence, just like we have for a two-qubit system. Thus, the quantification of entanglement by estimating the exact value of the concurrence is a formidable task. Despite these, few attempts have been made to obtain a lower bound of the concurrence for a qubit-qudit system \cite{gerjuoy,lozinski} and to derive a purely algebraic lower bound of the concurrence \cite{fmintert}. Later, Chen et.al. \cite{kchen} derived the lower bound of the concurrence for arbitrary $d_{1}\otimes d_{2}$ ($d_{1} \leq d_{2}$) dimensional system and it is given by
\begin{eqnarray}
	C(\rho_{AB})\geq \sqrt{\frac{2}{d_{1}(d_{1}-1)}}\left(max(\|\rho_{AB}^{T_{A}}\|_1,\|R(\rho_{AB})\|_1)-1\right)
	\label{albebound}
\end{eqnarray}
where $C(\rho_{AB})$ denotes the concurrence of a mixed bipartite quantum state $\rho_{AB}$ and other notations were defined earlier.
We have modified the above lower bound of concurrence and obtained the modified lower bound by using the constructed witness operator.
We take a few steps forward in this direction of research by obtaining a new improved lower bound of concurrence using the constructed witness operator $W_{(n)}$.
\indent 
\section{Preliminary Results}\label{sec5.2}
{\result \label{res5.1} If a state described by the density operator $\rho_{AB}$ in $d_{1}\otimes d_{2}$ dimensional system represents a PPT state then the following inequalities hold \cite{lin2016,pzhang2019}
	\begin{eqnarray}
		det(I_{d_{2}}+Tr_{A}(\rho_{AB}))\leq det(I_{d_{1}d_{2}}+\rho_{AB})
		\label{resp3}
	\end{eqnarray}
	\begin{eqnarray}
		det(I_{d_{1}}+Tr_{B}(\rho_{AB}))\leq det(I_{d_{1}d_{2}}+\rho_{AB})
		\label{resp3i}
	\end{eqnarray}
	where $Tr_{A}(\rho_{AB})$ and $Tr_{B}(\rho_{AB})$ represent the partial traces of the state $\rho_{AB}$ with respect to the subsystems $A$ and $B$, respectively, while $I_d$ denotes the $d\times d$ identity matrix and $det$ denotes the matrix determinant operation.}
{\result \label{res5.2}
	If $W$ represents the witness operator that detects the entangled quantum state described by the density operator $\rho_{AB}$ and
	$C(\rho_{AB})$ denotes the concurrence of the state $\rho_{AB}$ then the lower bound of concurrence is given by \cite{mintert2}
	\begin{eqnarray}
		C(\rho_{AB})\geq -Tr[W\rho_{AB}]
		\label{resultP4}
\end{eqnarray}}
\subsection{New Theorems and Results}
Let us define an operator $A$ of the form
\begin{equation}
	A = \frac{1}{\sqrt{rank(R(\rho_{AB}))}\;  \|R(\rho_{AB})\|_{2}} \; R(\rho_{AB})
\end{equation}
Using Result \ref{res-n1n2}, it can be shown that $\|A\|_{1}\leq 1$.
{\theorem \label{thm5.1} If the bipartite state described by the density operator $\rho_{AB}$ in $d_{1} \otimes d_{2}$ dimensional system, is separable then
	\begin{equation}
		\|A\|_{1} \leq \frac{1}{\sqrt{rank(R(\rho_{AB}))}\;  \|R(\rho_{AB})\|_{2}}
		\label{res3}
\end{equation}}
\begin{proof}
	Using the fact that if the state $\rho_{AB}$ is separable then $\|R(\rho_{AB})\|_{1}\leq 1$, one can prove that the $Theorem-\ref{thm5.1}$ is indeed true.
\end{proof}
{\corollary If the inequality (\ref{res3}) is violated by a quantum state $\rho_{AB}$ then the state $\rho_{AB}$ must be entangled, i.e., if the
	state $\rho_{AB}$ satisfies
	\begin{equation}
		\frac{1}{\sqrt{rank(R(\rho_{AB}))}\;  \|R(\rho_{AB})\|_{2}} < \|A\|_{1}
		\label{res31}
	\end{equation}
	then the state $\rho_{AB}$ is entangled.}
{\theorem \label{norm1smax} Let $R(\rho_{AB})$ be the realigned matrix of the bipartite state described by the density operator $\rho_{AB}$ in $d_{1} \otimes d_{2}$ dimensional system. If the state $\rho_{AB}$ is separable then,
	\begin{equation}
		{\| \rho_{AB}^{T_B} R(\rho_{AB}) \|}_{1} \leq \sigma_{max}(\rho_{AB}^{T_B}) \label{res4}
	\end{equation}
	where $T_{B}$ denotes the partial transposition with respect to the system $B$ and $\sigma_{max}(\rho_{AB}^{T_B})$ denotes the maximum singular value of $\rho_{AB}^{T_B}$}.
\begin{proof}
	Let us consider the product of two matrices $\rho_{AB}^{T_{B}}$ and $R(\rho_{AB})$ and further suppose that $\sigma_i(\rho_{AB}^{T_B}R(\rho_{AB}))$ denoting the $i^{th}$ singular value of the product $\rho_{AB}^{T_B} R(\rho_{AB})$. Therefore, the upper bound of trace norm of $\rho_{AB}^{T_B}R(\rho_{AB})$ is given by
	\begin{eqnarray}
		\| \rho_{AB}^{T_B} R(\rho_{AB}) \|_1 = \sum \sigma_i (\rho_{AB}^{T_B} R(\rho_{AB})) &\leq& \sum \sigma_i (\rho_{AB}^{T_B})\: \sigma_i (R(\rho_{AB})) \nonumber\\
		&\leq& \sigma_{max}(\rho_{AB}^{T_B}) \sum \sigma_i (R(\rho_{AB}))\nonumber\\
		&=& \sigma_{max}(\rho_{AB}^{T_B}) \|R(\rho_{AB}) \|_{1}
		\label{productub}
	\end{eqnarray}
	where the first inequality follows from \cite{horn}.\\
	If the state $\rho_{AB}$ is separable, then by $Theorem-\ref{thm-ccnr}$, we have $ \|R(\rho_{AB})\|_{1} \leq 1$.
	Thus, the inequality (\ref{productub}) for the separable state reduces to
	\begin{eqnarray}
		\| \rho_{AB}^{T_B} R(\rho_{AB}) \|_1 \leq \sigma_{max}(\rho_{AB}^{T_B})
		\label{productsepcond}
	\end{eqnarray}
	Hence proved.
\end{proof} 
{\corollary If the state $\rho_{AB}$ is separable then,
	\begin{equation}
		|Tr[(R(\rho_{AB}))^{T_B} \rho_{AB}]|  \leq \sigma_{max}(\rho_{AB}^{T_B}) \label{cor2}
\end{equation}}
\begin{proof}
	Let us start with $|Tr[(R(\rho_{AB}))^{T_B} \rho_{AB}]|$. It is given by
	\begin{eqnarray}
		|Tr[(R(\rho_{AB}))^{T_B} \rho_{AB}]| = |Tr[R(\rho_{AB}) \rho_{AB}^{T_B}]|   
		\leq {\| \rho_{AB}^{T_B} R(\rho_{AB}) \|}_{1} \leq \sigma_{max}(\rho_{AB}^{T_B})
	\end{eqnarray}
	The second last inequality follows from (\ref{modnorm1}) and the last inequality follows from $Theorem-\ref{norm1smax}$.
\end{proof}
\section{Construction of Witness Operator}
We now construct different types of witness operators that can detect (i) only NPTES and (ii) Both NPTES and PPTES.
\subsection{Witness operator detecting only NPTES}
\noindent It is well known that partial transposition operation can detect NPTES but the problem lies in the fact that it is not a physical operation and thus not possible to implement it in real experiments. To resolve this issue, we take the approach of constructing a witness operator that does not contain the partial transposition map for the detection of NPTES.\\
Let us consider a $d_{1}\otimes d_{2}$ dimensional quantum state described by the density operator $\rho_{AB}$. Our task is to determine whether the state described by the density operator $\rho_{AB}$ is NPTES.
{\theorem A $d_{1}\otimes d_{2}$ dimensional quantum state $\rho_{AB}$ is NPTES if there exist a witness operator $\widetilde{W}$ such that
	\begin{eqnarray}
		Tr[\widetilde{W}\rho_{AB}]<0
		\label{theorem1}
	\end{eqnarray}
	where $\widetilde{W}$ is given by
	\begin{eqnarray}
		\widetilde{W}= \frac{det(I_{d_{1}d_{2}}+(.))}{\sigma}|\psi\rangle\langle\psi| - (det(I_{d_{2}}+Tr_{A}(.)))I_{d_{1}d_{2}}
		\label{witnptes1}
	\end{eqnarray}
	$(.)$ means a $d_{1}\otimes d_{2}$ dimensional bipartite state which is under investigation, $Tr_{A}(.)$ represents the partial trace with respect to the subsystem $A$ of the state under investigation, $I_{d_{1}d_{2}}$ denoting the identity matrix in $d_{1}\otimes d_{2}$ dimensional Hilbert space and $|\psi\rangle$ be the normalized eigenvector corresponding to any non-zero eigenvalue $\lambda$ of $\rho_{AB}$.}
\begin{proof}
	Let us consider any separable state $\rho^{sep}_{AB}$ in $d_{1}\otimes d_{2}$ dimensional system. The trace of the operator $\widetilde{W}$ over a separable state $\rho^{sep}_{AB}$ is given by
	\begin{eqnarray}
		Tr[\widetilde{W}\rho^{sep}_{AB}]&=& Tr[(\frac{det(I_{d_{1}d_{2}}+\rho^{sep}_{AB})}{\lambda}|\psi\rangle\langle\psi|- det(I_{d_{2}}+Tr_{A}(\rho^{sep}_{AB}))I_{d_{1}d_{2}})\rho^{sep}_{AB}]\nonumber\\&=& det(I_{d_{1}d_{2}}+\rho^{sep}_{AB})-det(I_{d_{2}}+Tr_{A}(\rho^{sep}_{AB}))
		\geq 0
		\label{cond1}
	\end{eqnarray}
	The last step follows from (\ref{resp3}). Therefore, $Tr[\widetilde{W}\rho^{sep}_{AB}]\geq 0$ for
	any separable state $\rho^{sep}_{AB}$ .\\
	Next, let us consider a state described by the density operator $\rho_{12}$ defined as
	\begin{eqnarray}
		\rho_{12}=
		\begin{pmatrix}
			\frac{13}{30} & 0 & 0 & \frac{11}{30} \\
			0 & \frac{1}{15} & 0 & 0 \\
			0 & 0 & \frac{1}{15} & 0 \\
			\frac{11}{30} & 0 & 0 & \frac{13}{30}
		\end{pmatrix}
		\label{state1}
	\end{eqnarray}
	It can be easily shown that the state $\rho_{12}$ is indeed entangled.\\
	The trace value of $\widetilde{W}$ with respect to the state $\rho_{12}$ is given by
	\begin{eqnarray}
		Tr[\widetilde{W}\rho_{12}] = -\frac{491}{7500}<0
		\label{trdet}
	\end{eqnarray}
	Thus, the operator $\widetilde{W}$ is a witness operator.
\end{proof}
We note that if $\rho_{AB}$ denotes the PPTES in $d_1 \otimes d_2$ ($d_1, d_2 \geq 3$) dimensional system, then it can be easily shown that $Tr[\widetilde{W}\rho_{AB}]\geq 0$. This happens because (\ref{resp3}) holds for any PPTES also. Thus, it is not possible to detect any PPTES using the witness operator $\widetilde{W}$. Therefore, the witness operator $\widetilde{W}$ detect only NPTES.\\
Let us now consider a family of $3\otimes 3$ dimensional isotropic state \cite{zhao2010}, which is defined by
\begin{eqnarray}
	\rho_{iso}(f)=\frac{1-f}{8}I_{9}+\frac{9f-1}{8}|\psi^{+}\rangle\langle\psi^{+}|, 0\leq f\leq 1
	\label{isotropic}
\end{eqnarray}
where $|\psi^{+}\rangle=\frac{1}{\sqrt{3}}(|00\rangle+|11\rangle+|22\rangle)$ and $f=\langle\psi^{+}|\rho_{iso}(f)|\psi^{+}\rangle$.\\
The state $\rho_{iso}(f)$ is separable when $f\leq\frac{1}{3}$ and NPTES when $f>\frac{1}{3}$.\\
We now calculate $Tr[\widetilde{W}\rho_{iso}(f)]$ to determine how efficiently $\widetilde{W}$ detects NPTES.
$Tr[\widetilde{W}\rho_{iso}(f)]$ is given by
\begin{eqnarray}
	Tr[\widetilde{W}\rho_{iso}(f)]&=& \frac{(-9+f)^{8}(1+f)}{16777216}-\frac{64}{27}
	\nonumber\\&<&0,~~~~~0.591634<f\leq 1
	\label{traceisotropicdet}
\end{eqnarray}
Thus, the witness operator $\widetilde{W}$ fails to detect a few members in the family of isotropic NPTES. If the parameter $f$ lies in the region $\frac{1}{3}<f\leq 0.591634$, then the family of entangled states are not detected by $\widetilde{W}$. Hence, we can say that the witness operator $\widetilde{W}$ is not as efficient in comparison to the other witnesses in the literature.\\
Here one may argue about the utility of constructing $\widetilde{W}$ to detect NPTES for which we already have PPT criterion. It is known that the partial transposition is not a completely positive map and thus it would be very difficult to implement it in the laboratory. One approach to overcome this complication is given in \cite{kumari2019} where the method of structural physical approximation of a partial transposition (SPA-PT) is adopted to detect NPTES. We provide a complementary approach to address this problem by constructing the witness operator $\widetilde{W}$, which is independent of the PT operation. Our criterion  is constructive and applicable to detect NPTES even in the higher dimensional bipartite systems where the SPA method can be strenuous to implement.\\
Now, our task is to construct another witness operator that can be as efficient as $\widetilde{W}$. To achieve this, let us start with $R(\rho_{AB})$, which denotes the realigned matrix of the state under investigation. Then the operator $W^{o}$ can be defined as
\begin{eqnarray}
	W^{o} = \left(1 + \frac{1 - \|R(\rho_{AB})\|_{1}}{\sqrt{rank(R(\rho_{AB}))}\;  \|R(\rho_{AB})\|_{2}}\right)I_{d^{2}} -\frac{(R(\rho_{AB}))^{T_B}}{\sigma_{max}(\rho_{AB}^{T_B})}
	\label{wit2}
\end{eqnarray}
where $T_{B}$ is the partial transpose with respect to the second subsystem $B$ and $\sigma_{max}(\rho_{AB}^{T_B})$ denotes the maximum singular value
of $\rho_{AB}^{T_B}$.\\
{\theorem The operator $W^o$ is an entanglement witness operator.}
\begin{proof}
	Let us consider any $d_{1}\otimes d_{2}$ dimensional bipartite separable state $\rho^{sep}_{AB}$. Therefore, $Tr[W^{o} \rho^{sep}_{AB}]$ is given by
	\begin{eqnarray}
		Tr[W^{o} \rho^{sep}_{AB}]= 1 - \frac{Tr[(\rho^{sep}_{AB})^{T_B} R(\rho^{sep}_{AB})}{\sigma_{max}((\rho^{sep}_{AB})^{T_B})}  +  \frac{1 - \|R(\rho^{sep}_{AB})\|_{1}}{\sqrt{rank(R(\rho^{sep}_{AB}))}\|R(\rho^{sep}_{AB})\|_2}
		\label{tracevalueWo}
	\end{eqnarray}
	From (\ref{res3}) and (\ref{cor2}), it follows that $Tr[W^{o}\rho^{sep}_{AB}]\geq 0$ for all separable states $\rho^{sep}_{AB}$.\\
	Now it remains to show that there exists at least one entangled state $\rho_{AB}$ for which $Tr[W^{o} \rho_{AB}]< 0$. For this, let us
	consider a state of the form
	\begin{eqnarray}
		\varrho_{12}=
		\begin{pmatrix}
			\frac{11}{30} & 0 & 0 & \frac{7}{30} \\
			0 & \frac{2}{15} & 0 & 0 \\
			0 & 0 & \frac{2}{15} & 0 \\
			\frac{7}{30} & 0 & 0 & \frac{11}{30}
		\end{pmatrix}
		\label{state2}
	\end{eqnarray}
	It can be easily verified that the state $\varrho_{12}$ is an entangled state.\\
	The quantity $Tr[W^{o}\varrho_{12}]$ is given by
	\begin{eqnarray}
		Tr[W^{o}\varrho_{12}]= -0.0585731 < 0
		\label{tr3}
	\end{eqnarray}
	Thus, the operator $W^{o}$ is indeed an entanglement witness operator.
\end{proof}
Let us now recall again the family of $3\otimes3$ isotropic states defined in (\ref{isotropic}) and investigate whether
the witness operator $W^{o}$ detect more members of the family of isotropic states than $\widetilde{W}$. To probe this, let us calculate the following: \\
\begin{eqnarray*}
	Tr[{\rho^{T_B}_{iso}} \; R(\rho_{iso}(f))] &=& \frac{1}{96} ( -1 + 42f - 9f^2)\\
	\sigma_{max}(\rho^{T_B}_{iso}(f)) &=&
	\begin{cases}
		\frac{1 - 3f}{6} &  0 \leq f \leq \frac{1}{9} \\
		\frac{1 + 3f}{12} &  \frac{1}{9} \leq f \leq 1 \\
	\end{cases}\\
	rank(R(\rho_{iso}(f)))&=&
	\begin{cases}
		1 &  f = \frac{1}{9}\\
		9 & f \neq \frac{1}{9} \\
	\end{cases}\\
	\|R(\rho_{iso}(f))\|_{1}&=&
	\begin{cases}
		\frac{2}{3} - 3f &  0 \leq f \leq \frac{1}{9} \\
		3f &  \frac{1}{9} \leq f \leq 1 \\
	\end{cases}\\
	\|R(\rho_{iso}(f))\|_{2}&=& \sqrt{\frac{1 - 2f + 9f^2}{8}}\\
\end{eqnarray*}
Using (\ref{tracevalueWo}), we get
\begin{eqnarray}
	Tr[W^{o}\rho_{iso}(f)]=
	\begin{cases}
		\frac{17 - 90f + 9f^2}{16 - 48f} + \frac{2\sqrt{2}(1 + 9f)}{9\sqrt{1 - 2f + 9f^2}},& 0 \leq f < \frac{1}{9}\\
		\frac{8}{3}, & f = \frac{1}{9}\\
		\frac{1}{3}\left(\frac{27 (-1 + f)^2}{8 + 24f} - \frac{2\sqrt{2} (-1 + 3f)}{\sqrt{1 - 2f + 9f^2}}\right), & \frac{1}{9} < f \leq 1
	\end{cases}
	\label{traceWo_iso}
\end{eqnarray}
Here, $	Tr[W^{o}\rho_{iso}(f)] < 0$ for $0.413285 < f \leq 1$, which improves the detection range obtained in (\ref{traceisotropicdet}). Thus, the witness operator $W^{o}$ can be considered as more efficient than the witness operator $\widetilde{W}$. We can now observe the following facts:\\
(i) $W^{o}$ may detect bound entangled states also.\\
(ii) $R(\rho_{AB})$ and $({R(\rho_{AB})})^{T_B}$ are both non-Hermitian matrices. But the real eigenvalues of $({R(\rho_{AB})})^{T_B}$ makes our witness operator CPT  symmetric \cite{bender2002, bender2005, pati2009} and capable of detecting entanglement. In most of the cases, we find that the eigenvalues of $({R(\rho_{AB})})^{T_B}$ are real.
\subsection{Witness operator detecting both NPTES and PPTES}
Now our task is to construct a witness operator that is efficient in detecting both NPTES and PPTES.\\
Let us now start with the operator $W_{(n)}$ defined in $d_{1} \otimes d_{2}$ $(d_{1}\leq d_{2})$ dimensional space as follows:
\begin{eqnarray}
	W_{(n)} = \frac{d_{1}}{d_{1}-1} \left[	
	(k_{\rho_{AB}})^{n} \left(I_{d_{1}d_{2}} - \frac{(R(\rho_{AB}))^{T_B}}{\sigma_{max}(\rho_{AB}^{T_B})}\right) + \left(\frac{1 - \|R(\rho_{AB})\|_{1}}{\sqrt{rank(R(\rho_{AB}))}\;  \|R(\rho_{AB})\|_{2}}\right)I_{d_{1}d_{2}}\right] \label{witn1}
\end{eqnarray}
where $n \in \mathbb{N}$, the set of natural numbers; and $k_{\rho_{AB}} = det(I_{d_{1}d_{2}}+\rho_{AB}) - det(I_{d_{2}} + Tr_{A}(\rho_{AB}))$.\\
{\theorem The operator $W_{(n)}$ is a witness operator that can detect PPTES.}
\begin{proof}
	Let us consider a bipartite $d_{1} \otimes d_{2}$ $(d_{1}\leq d_{2})$ dimensional separable state $\rho^{sep}_{AB}$. Using (\ref{resp3}), we find that $k_{\rho^{sep}_{AB}} \geq 0$ for any separable state $\rho^{sep}_{AB}$.
	\begin{eqnarray}
		Tr[W_{(n)} \rho^{sep}_{AB}]= \frac{d_{1}}{d_{1}-1}\left[
		(k_{\rho^{sep}_{AB}})^{n}\left(1 - \frac{Tr[(\rho^{sep}_{AB})^{T_B}\; R(\rho^{sep}_{AB})]}{\sigma_{max}((\rho^{sep}_{AB})^{T_B})}\right) +  \frac{1 - \|R(\rho^{sep}_{AB})\|_{1}}{\sqrt{rank(R(\rho^{sep}_{AB}))}\|R(\rho^{sep}_{AB})\|_2}	\right]
		\label{tracevalueWn}
	\end{eqnarray}
	Using (\ref{tracevalueWo}), it follows that $Tr[W_{(n)}\rho^{sep}_{AB}]\geq 0$ for all bipartite $d_{1} \otimes d_{2}$ dimensional separable state $\rho^{sep}_{AB}$.\\
	Let us now consider the BES given in \cite{bihalan}.
	\begin{equation}
		\rho_{BE} =
		\begin{pmatrix}
			a&0&0&0&b&0&0&0&b\\
			0&c&0&0&0&0&0&0&0\\
			0&0&a&0&0&0&0&0&0\\
			0&0&0&a&0&0&0&0&0\\
			b&0&0&0&a&0&0&0&0\\
			0&0&0&0&0&c&0&b&0\\
			0&0&0&0&0&0&c&0&0\\
			0&0&0&0&0&b&0&a&0\\
			b&0&0&0&0&0&0&0&a\\
		\end{pmatrix};\;\text{where}\;\; a=\frac{1+\sqrt{5}}{3+9\sqrt5},\;\; b=\frac{-2}{3+9\sqrt5},\;\; c=\frac{-1+\sqrt{5}}{3+9\sqrt5}
	\end{equation}
	The values of the parameters involved in the witness operator $W_{(n)}$ to detect the state $\rho_{BE}$ are given below.
	\begin{eqnarray}
		&&k_{\rho_{BE}}=0.149 > 0;\quad Tr[\rho_{BE}^{T_B} R(\rho_{BE})] = \frac{1}{363} (21 - 8\sqrt{5});\quad
		\sigma_{max} (\rho_{BE}^{T_B})= \frac{1}{33} \sqrt{29 + 12\sqrt{5}},\nonumber\\&& rank(R(\rho_{BE})) = 9;\quad 
		\|R(\rho_{BE})\|_1 = 1.025;\quad \|R(\rho_{BE})\|_2 = 0.413
	\end{eqnarray}
	The expectation value of $W_{(n)}$ with respect to the state $\rho_{BE}$ is given by
	\begin{eqnarray}
		Tr[W_{(n)} \rho_{BE}] &=& 1.5 (0.0203459 - 0.962145 \times 0.149599^n) \nonumber\\ &<&  0~~
		\text{for}~~ \; n \geq 3
	\end{eqnarray}
	Since the operator $W_{(n)}$ detects the PPTES described by the density operator $\rho_{BE}$ for each $n \geq 3$ so $W_{(n)}$ is a witness
	operator. Hence proved.
\end{proof}
\section{Estimating the Concurrence for any Arbitrary Dimensional Bipartite System}
We now derive a new lower bound of concurrence of a bipartite quantum state $\rho_{AB}$ in $d_{1} \otimes d_{2}$ dimensional system and show that our bound is better in most cases when it is compared to the lower bound of the concurrence given by \cite{kchen}. We note that the lower bound given in (\ref{albebound}) is not normalized but can be normalized to unity. If $C_{min}(\rho_{AB})$ denotes the normalized value of this bound for the state $\rho_{AB}$, then we have
\begin{eqnarray}
	C(\rho_{AB})&\geq& C_{min}(\rho_{AB})\nonumber\\ &=& \frac{1}{(d_{1}-1)}(max(\|\rho_{AB}^{T_{A}}\|_{1},\|R(\rho_{AB})\|_{1})-1)
	\label{nalbebound12}
\end{eqnarray}
We are now in a position to use the witness operator $W_{(n)}$ defined in (\ref{witn1}) in the Result \ref{res5.2} by Mintert \cite{mintert2} for getting the improvement of the lower bound of the concurrence of an arbitrary bipartite $d_{1}\otimes d_{2}$ dimensional system. It may be noted that not all witness operators improve the lower bound of the concurrence given by (\ref{albebound}). 
{\theorem Let $\rho_{AB}$ be an entangled state in $d_{1}\otimes d_{2}$ $(d_{1}\leq d_{2})$ dimensional system detected by the witness operator $W_{(n)}$ defined in (\ref{witn1}). Then there exist $n_{1}\in \mathbb{N}$ such that the lower bound of concurrence of the state $\rho_{AB}$ is given by
	\begin{eqnarray}
		C(\rho_{AB})\geq \Phi_{W_{(n)}} (\rho_{AB}),~~ \forall n\geq n_{1}
		\label{result6}
	\end{eqnarray}
	where
	\begin{eqnarray}
		\Phi_{W_{(n)}} (\rho_{AB}) &=& - Tr[W_{(n)} \rho_{AB}] \nonumber\\&=&
		\frac{d_{1}}{d_{1}-1} \left[(k_{\rho_{AB}})^{n}\left(\frac{Tr[\rho_{AB}^{T_B} R(\rho_{AB})]}{\sigma_{max}(\rho_{AB}^{T_B})} - 1\right) + \frac{\|R(\rho_{AB})\|_{1} - 1}{\sqrt{rank(R(\rho_{AB}))} \|R(\rho_{AB})\|_2}\right]
		\label{ourbound}
\end{eqnarray}}
\begin{proof}
	Let us first recall the witness operator $W_{(n)}$ defined in (\ref{witn1}). Then the theorem follows by using the witness operator $W_{(n)}$ in the result given in (\ref{resultP4}). Hence proved.
\end{proof}
\begin{lemma}
	For any bipartite state $\rho_{AB}$ in $d_1 \otimes d_2$ dimensional system, we have
	\begin{eqnarray}
		|k_{\rho_{AB}}| < 1
	\end{eqnarray}
	where $k_{\rho_{AB}} = det(I_{d_{1}d_{2}}+\rho_{AB}) - det(I_{d_{2}} + Tr_{A}(\rho_{AB}))$.
\end{lemma}
\begin{proof}
	Let us start with the expression of $Tr[I_{d_1 d_2} + \rho_{AB}]$ which is given by
	\begin{eqnarray}
		Tr[I_{d_1 d_2} + \rho_{AB}] = Tr[I_{d_1 d_2}] + Tr[\rho_{AB}] = d_1d_2+1
		\label{l1}
	\end{eqnarray}
	Moreover, the inequality $det(I_{d_{1}d_{2}}+\rho_{AB}) \leq \left(\frac{1 + d_1d_2}{d_1d_2}\right)^{d_1d_2}$ can be derived as
	\begin{eqnarray}
		1 + d_1d_2 &=& Tr[I_{d_1 d_2} + \rho_{AB}] \nonumber\\
		&=&  \sum_{i=1}^{d_1d_2} \lambda_i(I_{d_1 d_2} + \rho_{AB}) \nonumber\\
		&\geq&  d_1d_2 \left(\prod_{i=1}^{d_1d_2} \lambda_i(I_{d_1 d_2} + \rho_{AB})\right)^{\frac{1}{d_1d_2}}  \nonumber\\
		&=&  d_1d_2 \left(det(I_{d_{1}d_{2}}+\rho_{AB})\right)^{\frac{1}{d_1d_2}}
	\end{eqnarray}
	i.e.,
	\begin{equation}
		det(I_{d_{1}d_{2}}+\rho_{AB}) \leq \left(\frac{1 + d_1d_2}{d_1d_2}\right)^{d_1d_2}
		\label{expr1}
	\end{equation}
	It can be seen that R.H.S of (\ref{expr1}) tends toward \textit{Euler's} number \textit{e} as $d_1,d_2$ tends to $\infty$.\\
	Therefore, for arbitrary large value of $d_{1}$ and $d_{2}$, we have
	\begin{equation}
		det(I_{d_{1}d_{2}}+\rho_{AB}) \leq \textit{e}
		\label{expr11}
	\end{equation}
	Let us first calculate the bound of $det(I_{d_{2}} + Tr_{A}[\rho_{AB}])$ for $d_{2}=2$ and then generalize the result to arbitrary dimension $d_{2}$.
	The quantum state in a 2-dimensional system, i.e., a qubit is described by the density operator
	\begin{equation}
		Tr_A[\varrho_{AB}^{(2)}]= \frac{I_2 + \vec{r}.\vec{\sigma}}{2}
		\label{expr11}
	\end{equation}
	where $\vec{r} \in \mathbb{R}^3$ with $|\vec{r}|^2 \leq 1$ is the \textit{Bloch vector} for the state $Tr_A[\varrho_{AB}^{(2)}]$; $I_2$ is $2\times2$ identity matrix; and $\vec{\sigma} = (\sigma_x,\sigma_y,\sigma_z)$ where $\sigma_x,\sigma_y$ and $\sigma_z$ are \textit{Pauli matrices} defined in (\ref{pauli}).\\
	After carrying out a simple calculations, we arrive at the result given by
	\begin{equation}
		det(I_2 + Tr_A[\varrho_{AB}^{(2)}]) \geq 2
		\label{expr12}
	\end{equation}
	The equality holds in (\ref{expr12}) for pure states.\\
	Since pure states are rank one projectors, we have $det(I_{d_{2}} + Tr_A[\varrho_{AB}^{(d_2)}]) = 2$ for any pure state $Tr_A[\varrho_{AB}^{(d_2)}]$ in $d_2$ dimensional system. Thus, we can generalize (\ref{expr12}) to an arbitrary qudit described by the density operator $Tr_A[\varrho_{AB}^{(d_2)}]$, we obtain the following
	\begin{equation}
		det(I_{d_{2}} + Tr_A[\varrho_{AB}^{(d_2)}]) \geq 2 \label{quditcase}
	\end{equation}
	Using the results (\ref{expr1}) and (\ref{quditcase}) in $k_{\rho_{AB}}$, we get
	\begin{eqnarray}
		k_{\rho_{AB}} &=& det(I_{d_{1}d_{2}}+\rho_{AB}) - det(I_{d_{2}} + Tr_{A}[\rho_{AB}]) \nonumber\\
		&\leq& \textit{e} - 2 < 1
	\end{eqnarray}
	Similarly, we can show that $ k_{\rho_{AB}} > -1$. Thus, we have $|k_{\rho_{AB}}|<1$. Hence proved.
\end{proof}
{\remark For PPT states, we have $0 \leq k_{\rho_{AB}} < 1$.}
{\corollary \label{corlimvalue} For large value of $n$, i.e., as $n\rightarrow \infty$, the lower bound of concurrence is given by
	\begin{eqnarray}
		C(\rho_{AB})\geq \phi(\rho_{AB})\quad \text{where} \quad \phi(\rho_{AB}) = \frac{d_{1}}{d_{1}-1}  \left(\frac{\|R(\rho_{AB})\|_{1} - 1}{\sqrt{rank(R(\rho_{AB}))} \|R(\rho_{AB})\|_2}\right)
		\label{result7}
\end{eqnarray}}
\begin{proof}
	Since $|k_{\rho_{AB}}|<1$ so $(k_{\rho_{AB}})^{n}\rightarrow 0$, as $n\rightarrow \infty$. Thus, we have
	\begin{eqnarray}
		\lim_{n\rightarrow\infty}\Phi_{W_{(n)}} (\rho_{AB}) = \phi(\rho_{AB}) \label{limvalue}
	\end{eqnarray} 
	Hence proved.
\end{proof}
{\remark The lower bound of concurrence given in  (\ref{result7}) is better than that given in (\ref{nalbebound12}), when $n\rightarrow \infty$.}
\section{Efficiency of the Witness Operator}
Now we discuss some examples to illustrate the utility of the witness operator $W_{(n)}$ in the estimation of concurrence of NPTES and PPTES. By using the witness operator $W_{(n)}$, we notice an improvement in the lower bound of the concurrence of the given NPTES and PPTES. Moreover, we find examples of the state $\rho_{AB}$ to demonstrate the following relations $S_1$ and $S_2$: 
\begin{eqnarray}
	&S_1:&C(\rho_{AB})\geq \Phi_{W_{(n)}} (\rho_{AB})\geq C_{min}(\rho_{AB}) ,~~ \forall\; n\geq n_{1}\\
	&S_2:&C(\rho_{AB})\geq \phi(\rho_{AB})\geq C_{min}(\rho_{AB})
	\label{result60}
\end{eqnarray}
\subsection{Examples of Estimation of Lower Bound of Concurrence of NPTES}
{\example Let us again recall the $3\otimes3$ isotropic states described by the density operator $\rho_{iso}(f)$ defined in (\ref{isotropic}). The witness operator $W_{(n)}$ detects $3\otimes3$ isotropic states for some range of the parameters which has been shown in Table-\ref{table_iso} given below.
	\begin{table}[h!]
		\begin{tabular}{| p{1cm} | p{6.5cm} | p{6.5cm} |}
			\hline
			\multicolumn{3}{|c|}{$3\otimes 3$ Isotropic States described by $\rho_{iso}(f)$} \\
			\hline
			\textit{n} & The range of the parameter $f$ for which $Tr[W_{(n)} \rho_{iso}(f)]<0$ &$\rho_{iso}(f)$ detected/not detected by the witness operator $W_{(n)}$ \\
			\hline
			1 &  $0.35 < f \leq 1$  &  Detected by $W_{(1)}$ \\
			\hline
			2 &  $0.336 < f \leq 1$  &  Detected by $W_{(2)}$\\
			\hline
			3 &  $0.3338 < f \leq 1$  &  Detected by $W_{(3)}$\\
			\hline
			4 &  $0.3334 < f \leq 1$  & Detected by $W_{(4)}$\\
			\hline
			5 &  $0.33334 < f \leq 1$  &  Detected by $W_{(5)}$\\
			\hline
		\end{tabular}
		\caption{Detection of isotropic state (NPTES) using $W_{(n)}$ in the range $\frac{1}{3} < f \leq 1$}
		\label{table_iso}
	\end{table}\\
	We now use the witness operator $W_{(n)}$ to estimate the lower bound of the concurrence of $\rho_{iso}(f)$. With an increase in $n$, one can easily find the improvement in the lower bound of concurrence estimated by the witness operator $W_{(n)}$,  $\forall \; n \in \mathbb{N}$, when compared to the lower bound of the concurrence given in (\ref{nalbebound12}). If we take sufficiently large value of $n$ then from $Corollary-\ref{corlimvalue}$, we have
	\begin{eqnarray}
		C(\rho_{iso}(f)) \geq \phi(\rho_{iso}(f)) = \frac{\sqrt{2} (-1 + 3f)}{\sqrt{1 - 2f + 9f^2}},\quad
		\frac{1}{3}<f\leq 1
		\label{lbcex-1}
	\end{eqnarray}
	where $C(\rho_{iso}(f))$ denotes the concurrence of the isotropic state.\\
	In Fig.-\ref{fig1ch1}, we have compared the lower bound $\phi(\rho_{iso}(f))$ given in (\ref{lbcex-1}) with the lower bound $C_{min}(\rho_{iso}(f))$ given in (\ref{nalbebound12}).
	\begin{figure}[h!]
		\begin{center}
			\includegraphics[width=0.6\textwidth]{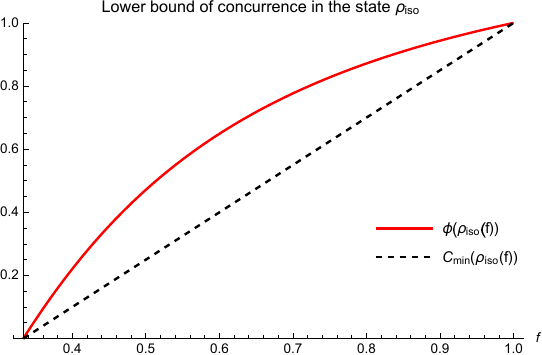}
			\caption{Isotropic states $\rho_{iso}(f)$: The dotted curve represents $C_{min}(\rho_{iso}(f))$ and the solid red curve represents the limiting value of our bound, i.e., $\phi(\rho_{iso}(f))$. Clearly, in the entangled region $\frac{1}{3}< f \leq 1$, $\phi(\rho_{iso}(f))$ gives a better estimate of the lower bound of concurrence as compared to $C_{min}(\rho_{iso}(f))$.} \label{fig1ch1}
		\end{center}
	\end{figure}
	
	\example Let us consider a class of bipartite quantum state in $3\otimes 3$ dimensional system, which is defined as \cite{pmrhoro}
	\begin{eqnarray}
		\rho_{\alpha}=\frac{2}{7}|\psi^{+}\rangle\langle\psi^{+}|+\frac{\alpha}{7}\sigma_{+}+\frac{5-\alpha}{7}\sigma_{-},~~~2\leq \alpha \leq 5
		\label{pptes2}
	\end{eqnarray}
	where $|\psi^{+}\rangle=\frac{1}{\sqrt{3}}(|00\rangle+|11\rangle+|22\rangle)$ and
	\begin{eqnarray}
		\sigma_{+}=\frac{1}{3}(|01\rangle\langle01|+|12\rangle\langle12|+|20\rangle\langle20|);
		\;\;
		\sigma_{-}=\frac{1}{3}(|10\rangle\langle10|+|21\rangle\langle21|+|02\rangle\langle02|)
		\label{sigma+-}
	\end{eqnarray}
	The state $\rho_{\alpha}$ can be characterized with respect to the parameter $\alpha$ in the interval $[2,5]$  as:\\
	(i) $\rho_{\alpha}$ is a separable state when $2\leq \alpha \leq 3$.\\
	(ii) $\rho_{\alpha}$ represents PPTES when $3<\alpha \leq 4$.\\
	(iii) $\rho_{\alpha}$ is NPTES when $4< \alpha \leq 5$. \\
	In this example, we will consider the state $\rho_{\alpha}$ for $4< \alpha \leq 5$.\\
	It can be easily seen that for each $n$, the witness operator $W_{(n)}$ detects all the NPTES belonging to the family of states described by the density operator $\rho_{\alpha},~~4< \alpha\leq 5$. Further, we can use the witness operator $W_{(n)}$ to improve the estimation of the lower bound of concurrence of the state $\rho_{\alpha},~~4< \alpha\leq 5$. In this case, we find that except for $n=1$, the witness operator $W_{(n)}$ improves the lower bound of the concurrence compared to the lower bound given in (\ref{nalbebound12}) in the whole range of the parameter $\alpha$, i.e., $4< \alpha \leq 5$. For $n=1$, the witness operator $W_{(1)}$ improves the lower bound in the interval $4.15 < \alpha \leq 5$. Figure-\ref{fig2} describes a comparison between the limiting value of our bound, i.e., $\phi(\rho_{\alpha})$ with $C_{min}(\rho_{\alpha})$.
	\begin{figure}[h!]
		\begin{center}
			\includegraphics[width=0.6\textwidth]{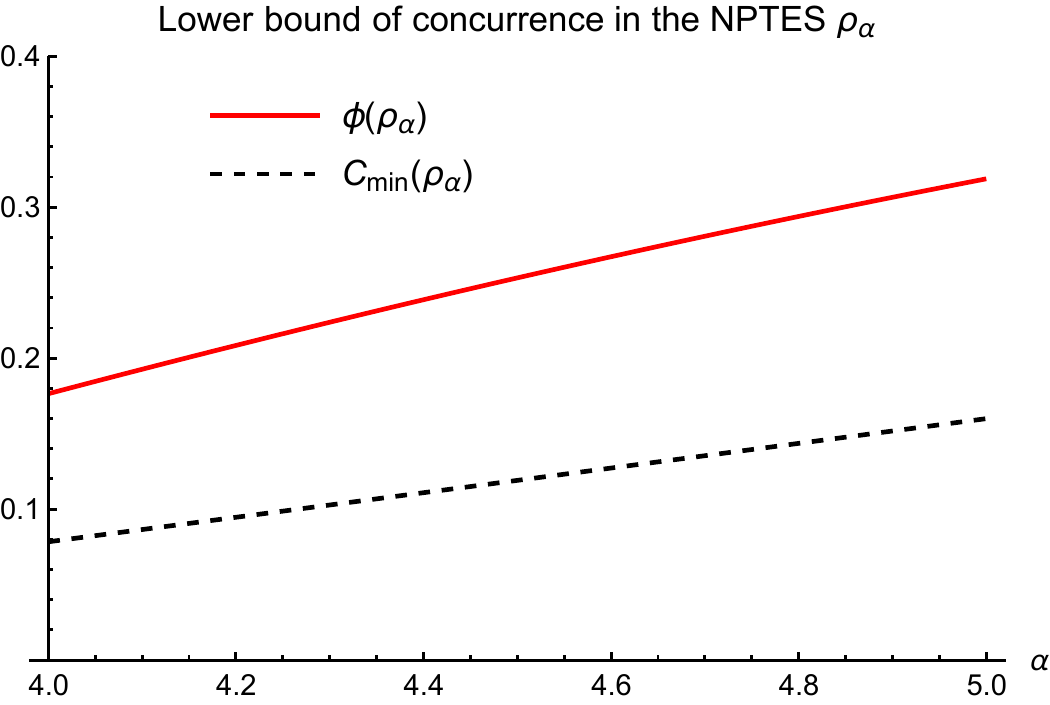}
			\caption{Horodecki alpha states $(\rho_{\alpha})$: The dotted curve represents $C_{min}(\rho_{\alpha})$ and the solid red curve represents the limiting value of our bound, i.e., $\phi(\rho_{\alpha})$. Clearly, in the NPT entangled region $4 < \alpha \leq 5$, $\phi(\rho_{\alpha})$ gives a better estimate of the lower bound of concurrence as compared to $C_{min}(\rho_{\alpha})$.} \label{fig2}
		\end{center}
	\end{figure}
}
\subsection{Examples of Estimation of Lower Bound of Concurrence of PPTES}
\noindent We will now consider a few examples of PPTES detected by the witness operator $W_{(n)}$.

\begin{example} A $3\otimes 3$ PPTES constructed from the unextendible product basis (UPB) is given by \cite{cbennett}
	\begin{eqnarray}
		\rho_{UPB} = \frac{1}{4}[I_{9}-\sum_{i=1}^{5}|\psi_{i}\rangle\langle\psi_{i}|]
		\label{bennetstate}
	\end{eqnarray}
	where the states $\{|\psi_{i}\rangle\}_{i=1}^5$ form the UPB and are given by
	\begin{eqnarray}
		&&|\psi_{1}\rangle=\frac{1}{\sqrt{2}}|0\rangle\otimes(|0\rangle-|1\rangle); \quad
		|\psi_{2}\rangle=\frac{1}{\sqrt{2}}(|0\rangle-|1\rangle)\otimes|2\rangle\nonumber; \quad
		|\psi_{3}\rangle=\frac{1}{\sqrt{2}}|2\rangle\otimes(|1\rangle-|2\rangle)\nonumber\\&&
		|\psi_{4}\rangle=\frac{1}{\sqrt{2}}(|1\rangle-|2\rangle)\otimes|0\rangle; \quad
		|\psi_{5}\rangle=\frac{1}{3}(|0\rangle+|1\rangle+|2\rangle)\otimes(|0\rangle+|1\rangle+|2\rangle)
		\label{upb}
	\end{eqnarray}
	To start with, let us calculate the following quantities for the state $\rho_{UPB}$:
	\begin{eqnarray}
		&&k_{\rho_{UPB}} = \frac{71}{768}, Tr[\rho_{UPB}^{T_B}R(\rho_{UPB})] = \frac{1}{16},
		\sigma_{max}(\rho_{UPB}^{T_B}) = \frac{1}{4},\nonumber\\&& rank(R(\rho_{UPB})) = 6,
		\|R(\rho_{UPB})\|_{1} = 1.08741, \|R(\rho_{UPB})\|_{2} = 0.5
		\label{i1}
	\end{eqnarray}
	Using the data given in (\ref{i1}), we can construct the witness operator $W_{(n)}$ and calculate its expectation value with respect to the state $\rho_{UPB}$ as
	\begin{eqnarray}
		Tr[W_{(n)} \rho_{UPB}] &=& \frac{3}{2} \left[
		(k_{\rho_{UPB}})^{n}\left(1 - \frac{Tr[\rho_{UPB}^{T_B} R(\rho_{UPB})]}{\sigma_{max}(\rho_{UPB}^{T_B})}\right)  +  \frac{1 - \|R(\rho_{UPB})\|_{1}}{\sqrt{rank(R(\rho_{UPB}))}\|R(\rho_{UPB})\|_2}\right]\nonumber\\
		&=& \frac{3}{2} \left(\frac{3}{4}\left(\frac{71}{768}\right)^n - 0.071372\right) <  0 \quad \forall \; n \in \mathbb{N}
		\label{tracevalueWn_bennet}
	\end{eqnarray}
	Thus, $W_{(n)}$ detect the PPTES described by the density operator $\rho_{UPB}$ for all $n \in \mathbb{N}$.\\
	The lower bound of the concurrence of the state $\rho_{UPB}$ is given by
	\begin{eqnarray}
		C(\rho_{UPB}) \geq \Phi_{W_{(n)}}(\rho_{UPB}) = - Tr[W_{(n)} \rho_{UPB}], \quad\forall \; n \in \mathbb{N}
	\end{eqnarray}
	In Table-\ref{Bennet_table}, we compare the lower bound of concurrence $\Phi_{W_{(n)}}(\rho_{UPB})$ with $C_{min}(\rho_{UPB})$. It shows that for $n > 1$, the function $\Phi_{W_{(n)}}$ gives a better estimate of the lower bound of concurrence as compared to the one given in (\ref{nalbebound12}), i.e.,
	\begin{eqnarray}
		\Phi_{W_{(n)}}(\rho_{UPB}) > C_{min}(\rho_{UPB}) = 0.04 \quad\forall \; n > 1
	\end{eqnarray}
	\begin{table}[h]
		\begin{center}
			\begin{tabular}{| p{1cm} | p{6cm} | p{6cm} | }
				\hline
				\multicolumn{3}{|c|}{Lower bound of concurrence for the state $\rho_{UPB}$} \\
				\hline
				\textit{n} & $\Phi_{W_{(n)}} (\rho_{UPB})$ & $C_{min}(\rho_{UPB})$  \\
				\hline
				1 &  0.00305406 &  0.04 \\
				\hline
				2 &  0.097443  & 0.04 \\
				\hline
				3 &  0.106169  &  0.04  \\
				\hline
				4 &  0.106976  & 0.04 \\
				\hline
				5 &  0.10705  &  0.04  \\
				\hline
			\end{tabular}
			\caption{Lower bound is compared with $C_{min}(\rho_{UPB})$}
			\label{Bennet_table}
		\end{center}
	\end{table}
	Also, it can be observed that as we increase the value of $n$, the value of the lower bound of concurrence is also improved. So, it would be interesting to find out the value of the lower bound of concurrence for indefinite large $n$. We calculate the lower bound of concurrence of $\rho_{UPB}$ for large $n$ and using $Corollary-\ref{corlimvalue}$, it can be estimated as
	\begin{eqnarray}
		C(\rho_{UPB})\geq \phi(\rho_{UPB})=0.107058
		\label{cor3ex1}
	\end{eqnarray}
\end{example}

\begin{example}
	Let us consider the two qutrit, non-full rank PPTES given by \cite{bandyo}
	\begin{eqnarray}
		\rho_i(\gamma) = \gamma |\psi_i\rangle \langle\psi_i| + (1-\gamma)\rho_{UPB},~~1\leq i \leq 5
		\label{ex-2b}
	\end{eqnarray}
	where $\rho_{UPB}$, $|\psi_i\rangle$ are defined in (\ref{bennetstate}) and (\ref{upb}) and $\gamma \in [0,1]$.
	The state $\rho_i(\gamma)$ satisfy the range criteria. For any $i(1\leq i \leq 5)$, the PPT states $\rho_i(\gamma)$ are entangled if and only if $0\leq \gamma < \frac{1}{5}$. $\rho_i(\gamma)$ represent separable states for $\frac{1}{5} \leq \gamma \leq 1 $.\\
	After simple calculations, we find that our witness operator $W_{(n)}$ identify the states $\rho_i(\gamma),~~1\leq i \leq 5$ given in (\ref{ex-2b}) as PPTES in the region $0\leq \gamma \leq 0.0635994$, for any $n \in \mathbb{N}$. Furthermore, for a sufficiently large value of $n$, but in the same range of $\gamma$, the witness operator $W_{(n)}$ detects the state described by the density operator $\rho_i(\gamma),~~1\leq i \leq 5$ as the matrix realignment criteria.\\
	When $0\leq \gamma \leq 0.0635994$, our derived lower bound of the concurrence of the state $\rho_i(\gamma),~~1\leq i \leq 5$ gives a better lower bound in comparison to the lower bound of concurrence given by Albeverio et.al. This has been shown in Figure-\ref{lambdaimage}.\\
	\begin{figure}[h!]
		\begin{center}
			\includegraphics[width=0.6\textwidth]{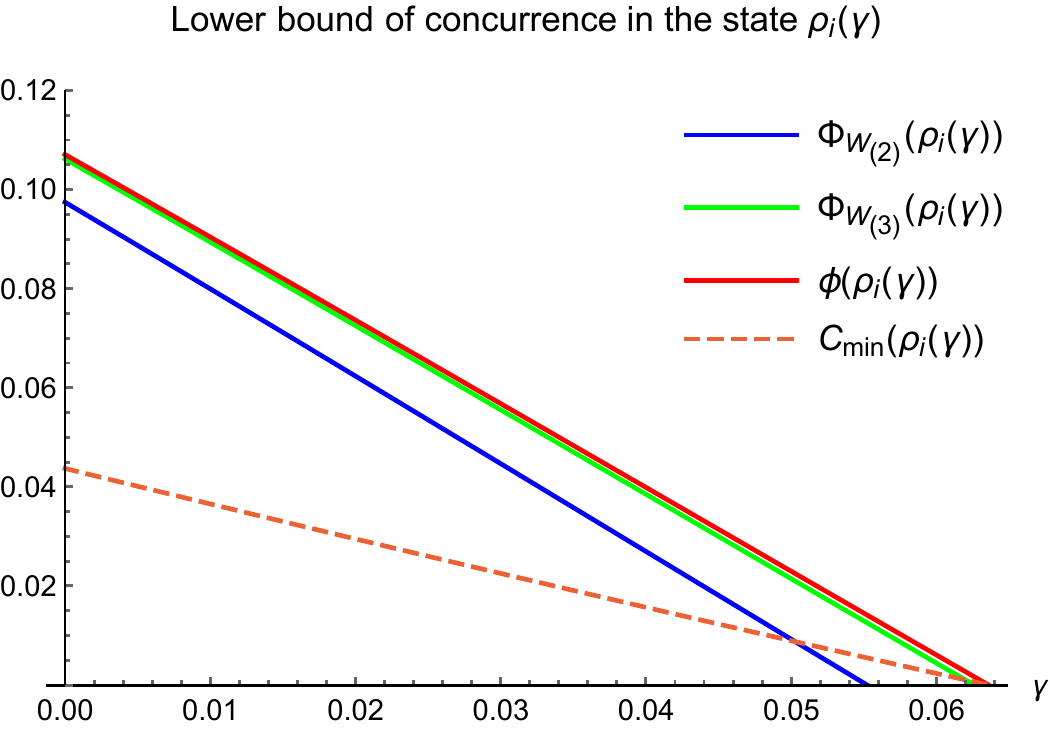}
			\caption{The lower bound $\Phi_{W_{(n)}}(\rho_i(\gamma))$, is represented by the solid curves, in blue ($n=2$), green($n=3$). One must observe that for $n>3$, the bound is slightly less than $\phi(\rho_{i}(\gamma))$ and finally converges to it, i.e., to the red curve, as n is increased. The dotted curve represents $C_{min}(\rho_i(\gamma))$, i.e., the normalized lower bound of concurrence of the state $\rho_i(\gamma)$ given by (\ref{nalbebound12}).}
			\label{lambdaimage}
		\end{center}
	\end{figure}\\
\end{example}

\begin{example}
	Let us again recall a class of bipartite quantum state $\rho_{\alpha}$ for $3 < \alpha \leq 4$, which is given by \cite{pmrhoro}
	\begin{eqnarray}
		\rho_{\alpha}=\frac{2}{7}|\psi^{+}\rangle\langle\psi^{+}|+\frac{\alpha}{7}\sigma_{+}+\frac{5-\alpha}{7}\sigma_{-},~~~3 < \alpha \leq 4
		\label{pptes2}
	\end{eqnarray}
	The state $\rho_{\alpha}$ represents PPTES when $3<\alpha \leq 4$. One can check that the marginals of $\rho_{\alpha}$ are maximally mixed and $k_{\rho_{\alpha}} > 0$ for $3< \alpha \leq 4$. We note that the PPTES of this family is not detected by the witness operators $W^o$, which is defined in (\ref{wit2}).\\
	To detect the state $\rho_{\alpha}$, let us construct the witness operator $W_{(n)}$ with the data given by
	\begin{eqnarray}
		&&k_{\rho_{\alpha}} = -\frac{64}{27} + \frac{9}{7}\left(1 + \frac{5 - \alpha}{21}\right)^3 \left(1 + \frac{\alpha}{21}\right)^3;\quad
		Tr[\rho_{\alpha}^{T_B} R(\rho_{\alpha})] = \frac{2}{21}; \quad rank(R(\rho_{\alpha})) = 9;\nonumber\\&&
		\sigma_{max}(\rho_{\alpha}^{T_B}) = \frac{1}{21} \sqrt{\frac{33}{2} - 5\alpha + 5\alpha^2 + \frac{5}{2} \sqrt{41 - 20\alpha + 4 \alpha^2}};
		\\&&  \|R(\rho_{\alpha})\|_{1} = \frac{1}{21}(19 + 2\sqrt{19 -15\alpha + 3\alpha^{2}});\quad
		\|R(\rho_{\alpha})\|_{2} = \sqrt{\frac{73}{441} + \frac{1}{882} (76 - 60\alpha + 12 \alpha^2)} \nonumber
	\end{eqnarray}
	The range of the parameter $\alpha$ for which $Tr[W_{(n)} \rho_{\alpha}] < 0,~~n=1~~\textrm{to}~~5$ is given in Table \ref{alphatable}. It shows that as the value of $n$ increases, more and more PPTES are detected by the witness operator $W_{(n)}$. \\
	\begin{table}[h!]
		\begin{center}
			\begin{tabular}{| p{1cm} | p{6cm} |p{6cm} |}
				\hline
				\multicolumn{3}{|c|}{The state described by $\rho_{\alpha}$ for $3 < \alpha \leq 4$} \\
				\hline
				$n$ & The range of $\alpha$ for which $Tr[W_{(n)} \rho_{\alpha}]<0$ & $\rho_{\alpha}$ detected/not detected by the witness operator $W_{(n)}$ \\
				\hline
				1 &  $3.7 < \alpha \leq 4$  & PPTES detected by $W_{(1)}$   \\
				\hline
				2 &  $3.11 < \alpha \leq 4$ & PPTES detected by $W_{(2)}$  \\
				\hline
				3 &  $3.01 < \alpha \leq 4$  & PPTES detected by $W_{(3)}$ \\
				\hline
				4 &  $3.0025 < \alpha \leq 4$ & PPTES detected by $W_{(4)}$  \\
				\hline
				5 &  $3.0004 < \alpha \leq 4$  & PPTES detected by $W_{(5)}$ \\
				\hline
			\end{tabular}
			\caption{Detection of PPTES with the witness operator $W_{(n)}$ for different $n$ and in the range $3 < \alpha \leq 4$}
			\label{alphatable}
		\end{center}
	\end{table} \\
	The witness operator $W_{(n)}$ not only detect the PPTES $\rho_{\alpha}$ but also estimate the lower bound $\Phi_{n}(\rho_{\alpha})$ of the concurrence of $\rho_{\alpha}$ when $3 < \alpha \leq 4$. It can be easily shown that the value of $\Phi_{n}(\rho_{\alpha})$ improves as we increase the value of $n$. Thus, for large $n$, the lower bound of the concurrence $\phi(\rho_{\alpha})$ is given by
	\begin{eqnarray}
		\phi(\rho_{\alpha})=\frac{-1 + \sqrt{19 - 15\alpha + 3\alpha^2}}{\sqrt{111 - 30\alpha + 6\alpha^2}}
		\label{lbcex-3}
	\end{eqnarray}
	Again, the normalized lower bound $C_{min}(\rho_{\alpha})$ of the concurrence of $\rho_{\alpha}$ can be calculated by the prescription given in \cite{kchen}
	\begin{eqnarray}
		C_{min}(\rho_{\alpha})=
		\frac{1}{21}(\sqrt{3\alpha^{2}-15\alpha+19}-1)
		\label{maxeig}
	\end{eqnarray}
	We then compare the value of $\phi(\rho_{\alpha})$ given in (\ref{lbcex-3}) with the lower bound given in (\ref{maxeig}). The comparison is shown in the Figure-\ref{alphaimage} and it can be concluded that $C(\rho_{\alpha})\geq \phi(\rho_{\alpha}) \geq C_{min}(\rho_{\alpha})$  where $3 < \alpha \leq 4$.

	\begin{figure}[h!]
		\begin{center}
			\includegraphics[width=0.6\textwidth]{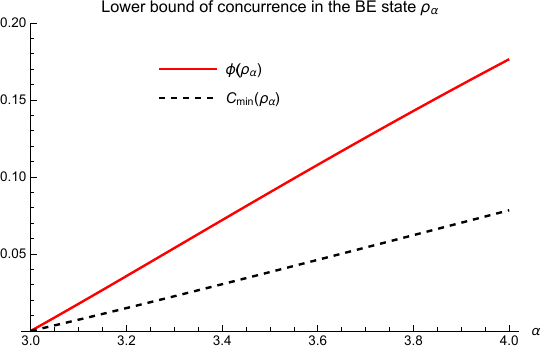}
			\caption{The dotted curve represents $C_{min}(\rho_{\alpha})$ and the solid red curve represents the limiting value of our bound, i.e., $\phi(\rho_{\alpha})$ in the PPT entangled region $3 < \alpha \leq 4$. Our bound varies from 0 (at $\alpha = 3$) to 0.176 (at $\alpha = 4$) which gives a better estimate of the lower bound of concurrence as compared to $C_{min}(\rho_{\alpha})$ which varies from 0 (at $\alpha = 3$) to 0.078 (at $\alpha = 4$).}
			\label{alphaimage}
		\end{center}
	\end{figure}
\end{example}

\begin{example} \label{egch2-astate}
	Let us consider another family of PPTES described by the density operator $\rho_{a}$, given in (\ref{astate-matrix}). The state $\rho_a$ is separable when $a = 0$ and 1.
	It is known that the density matrix $\rho_a$ represents a family of PPTES for $0 < a < 1$ \cite{horodeckia}. Further, it may be noted that the eigenvalues of the realigned matrix $R(\rho_a)$ and that of the partial transpose of the realigned matrix of $\rho_a$, i.e.,$(R(\rho_a))^{T_B}$, are real and non-negative.\\
	In Table-\ref{atable}, we have shown that there exist different witness operators from the family of witness operators $W_{(n)}$ that can detect PPTES from the family of states described by the density operator $\rho_{a}$ for different ranges of the parameter $a$.
	\begin{table}[h!]
		\begin{center}
			\begin{tabular}{| p{1cm} | p{5.6cm} | p{6.6cm} |}
				\hline
				\multicolumn{3}{|c|}{The state $\rho_{a},~~0<a<1$} \\
				\hline
				\textit{n} & The range of the parameter $a$ for which $Tr[W_{(n)}\rho_{a}]<0$ &$\rho_{a}$ detected/not detected by the witness operator $W_{(n)}$ \\
				\hline
				1 &  Does not exist  &  $\rho_{a}$ is not detected by $W_{(1)}$ \\
				\hline
				2 &  $0 < a < 0.016$  &  PPTES detected by $W_{(2)}$\\
				\hline
				3 &  $0 < a < 0.62$  &  PPTES detected by $W_{(3)}$\\
				\hline
				4 &  $0 < a < 0.951$  &  PPTES detected by $W_{(4)}$\\
				\hline
				5 &  $0 < a < 0.9932$  &  PPTES detected by $W_{(5)}$\\
				\hline
				6 &  $0 < a < 0.999$  &  PPTES detected by $W_{(6)}$\\
				\hline
				7 &  $0 < a < 0.99987$  &  PPTES detected by $W_{(7)}$\\
				\hline
				8 &  $0 < a < 0.99998$  & PPTES detected by $W_{(8)}$\\
				\hline
			\end{tabular}
			\caption{Detection of PPTES in the range $0 < a < 1$}
			\label{atable}
		\end{center}
	\end{table} \\
	Next, our task is to estimate the lower bound of the concurrence of the state $\rho_{a},~~0<a<1$. We first calculate the lower bound $\Phi_{W_{(n)}}(\rho_{a})$ of the concurrence and then compare it with the lower bound $C_{min}(\rho_{a})$ given in (\ref{nalbebound12}). The comparison is shown in Figure-\ref{aimage1} and we find that there exist a critical value of $n$, say $n_{1}$ such that for $n\geq n_{1}$, the quantity $\Phi_{W_{(n)}}(\rho_{a})$ gives better lower bound of the concurrence than $C_{min}(\rho_{a})$. Also, from $Corollary-\ref{corlimvalue}$, we know that $\Phi_{W_{(n)}}(\rho_{a}) \rightarrow \phi(\rho_{a}),~~\textrm{as}~~n\rightarrow \infty$, so we obtained the inequality for the state $\rho_{a},~~0<a<1$
	\begin{eqnarray}
		C(\rho_{a})\geq \phi(\rho_{a}) \geq C_{min}(\rho_{a})
	\end{eqnarray} 
	\begin{figure}[h!]
		\begin{center}
			\includegraphics[width=0.8\textwidth]{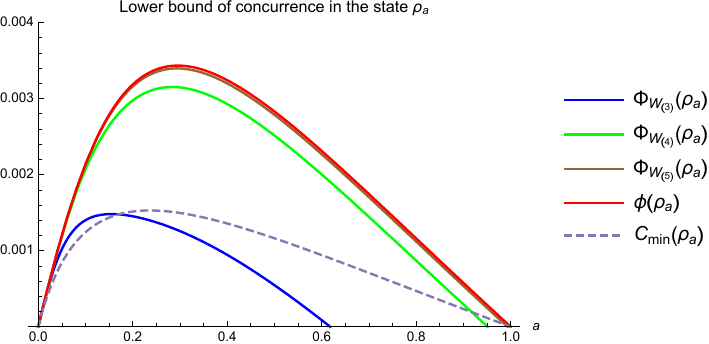}
			\caption{The lower bound $\Phi_{W_{(n)}}(\rho_{a})$, is represented by different solid curves for different values of $n$. The solid curve is given in blue for $n=3$, green for $n=4$, and brown for $n=5$. One may note here that for $n>5$, the bound converges to the red curve which represents $\phi(\rho_{a})$. The dotted curve represents $C_{min}(\rho_{a})$, i.e., the normalized lower bound concurrence of $\rho_a$ given by (\ref{nalbebound12}). }
			\label{aimage1}
		\end{center}
	\end{figure}
\end{example}
\begin{example}  Let us consider a $4\otimes4$ dimensional PPTES $\rho_{p,q}$ defined in (\ref{bes4by4}).
	The state $\rho_{p,q}$ is a PPT state for $q = \frac{\sqrt{2} - 1}{2}\equiv q_{0}$ and $p = \frac{1 - 2q}{4}\equiv p_{0}$. Since $\|R(\rho_{p_{0},q_{0}})\|_1 = 1.08579$, which is greater than one, by matrix realignment criteria one can say that $\rho_{p_{0},q_{0}}$ is a PPTES. Note that for this PPTES, the realigned matrix $R(\rho_{p_{0},q_{0}})$ is Hermitian and so is $(R(\rho_{p_{0},q_{0}}))^{T_B}$. The witness operator $W^o$ defined in (\ref{wit2}) fails to detect this state.
	A simple calculation shows that $\rho_{p_{0},q_{0}}$ is detected by our witness operator $W_{(n)}$ for all $n$, i.e.,
	\begin{eqnarray}
		Tr[W_{(n)} \rho_{p_{0},q_{0}}] &=& \frac{4}{3} \left(\left(\frac{3}{2} - \frac{1}{\sqrt{2}}\right)\left(\frac{7}{128} (-17 + 13\sqrt{2})\right)^n  + \frac{1}{8}(-2 + \sqrt{2}) \right)\nonumber\\&<& 0, \quad \forall \; n
	\end{eqnarray}
	Let us see now how efficiently, we estimate the lower bound of concurrence through the witness operator $W_{(n)}$. The lower bound $\Phi_{W_{(n)}} (\rho_{p_{0},q_{0}})$ of the concurrence of the state $\rho_{p_{0},q_{0}}$ is estimated in Table-\ref{4by4_table}.\\
	\begin{table}[h!]
		\begin{center}
			\begin{tabular}{| p{1cm} | p{3.5cm} | p{3.5cm} |}
				\hline
				\multicolumn{3}{|c|}{$4\otimes4$ bound entangled state $\rho_{p_{0},q_{0}}$} \\
				\hline
				\textit{n }& $\Phi_{W_{(n)}} (\rho_{p_{0},q_{0}})$ & $C_{min}(\rho_{p_{0},q_{0}})$ \\
				\hline
				1 &  0.01757 & 0.0285955 \\
				2 &  0.0915681  &  0.0285955 \\
				3 &  0.0971719  &  0.0285955 \\
				4 &  0.0975963  &  0.0285955 \\
				5 &  0.0976284  &  0.0285955\\
				\hline
			\end{tabular}
			\caption{Comparison of $C_{min}$ with our lower bound of concurrence $\Phi_{W_{(n)}}$, for $n = 1$ to 5.}
			\label{4by4_table}
		\end{center}
	\end{table}\\
	We then verified the following relation, for $n \geq 1$
	\begin{eqnarray}
		&&C(\rho_{p_{0},q_{0}}) \geq \Phi_{W_{(n)}} (\rho_{p_{0},q_{0}}) \geq C_{min}(\rho_{p_{0},q_{0}})
	\end{eqnarray}
	For sufficiently large value of $n$, the lower bound of the concurrence of the state $\rho_{p_{0},q_{0}}$ is $0.0976311$, which is again much better than $C_{min}(\rho_{p_{0},q_{0}})$.
\end{example}
\section{Conclusion}
To summarize, firstly we have constructed different witness operators to detect NPTES and PPTES. Our main contribution in this chapter is that the constructed witness operator is used to improve the lower bound of the concurrence for any arbitrary entangled state in $d_{1} \otimes d_{2} (d_{1}\leq d_{2})$ dimensional system. In particular, our result will be useful to estimate the lower bound of the concurrence of the BES in higher dimensional systems. This study is important because the exact expression of the concurrence for higher dimensional entangled states is not known and thus one has to depend on the lower bound of it. Also, the witness operator $W_{(n)}$ defined in (\ref{witn1}) is proved to be very useful in detecting not only the NPTES but also in the detection of many BES. Here we have illustrated the above facts with a few examples but one may use $W_{(n)}$ to detect other BES and their lower bound of the concurrence can also be estimated.\\
\begin{center}
	****************
\end{center}

\chapter{Detection of Bipartite and Genuine Three-qubit Entanglement Through Realignment Operation}\label{ch6}
{\small \emph {In an honest search for knowledge, you quite often have to abide by ignorance for an indefinite period. \\
		-Erwin Schrodinger}}
\vspace{0.5cm}
\noindent \hrule
\noindent \emph{ In this chapter\;\footnote { This chapter is based on the published research article, ``S. Aggarwal, S. Adhikari, \emph{Theoretical proposal for the experimental realization of realignment operation}, arXiv:2307.07952  (2023)"},  we have developed techniques in the detection of bipartite and tripartite genuine entangled states that may be realized in the current technology. We address the problem of experimental realization of realignment operation and to achieve this aim, we propose a theoretical proposal for the same. We first show that realignment operation on a bipartite state can be expressed in terms of the partial transposition operation along with column interchange operations. We observed that these column interchange operations form a permutation matrix which can be implemented via SWAP operator acting on the density matrix. This mathematical framework is used to exactly determine the first moment of the realignment matrix experimentally. This has been done by showing that the first moment of the realignment matrix can be expressed as the expectation value of a SWAP operator which indicates the possibility of its measurement. Further, we have provided an estimation of the higher order realigned moments in terms of the first realigned moment and thus pave the way to estimate the higher order moments experimentally. Next, we develop moments based entanglement detection criteria that detect PPTES as well as NPTES. Moreover, we define a new matrix realignment operation for three-qubit states and have devised entanglement criteria that detect three-qubit genuine entangled states.} 
\noindent 
\newpage
\section{Introduction}\label{sec6.1}
Entanglement is a fundamental property of quantum systems, promising to power the near future of quantum information processing and quantum computing \cite{niel,guhnerev,horodecki9}. Entangled states reveal some peculiar properties of the quantum world that may be used to do many interesting quantum information processing tasks which may be done better than that with the separable states alone \cite{pshor, bennett1n,harrow}.
Hence, it is important to devise implementable and efficient methods to detect bipartite as well as multipartite entanglement. PPT criterion and realignment criterion can be considered as the most important entanglement detection criterion \cite{arul2013}.
However, realignment criteria is a more powerful and operational criterion in the sense that it can detect BES as well as NPTES \cite{rud2005, rudolph2000, rudolph2004}, while PPT criterion detects only NPTES. 
Both the operations viz. partial transposition and realignment operation are unphysical and are based on the permutation of the elements of density matrix \cite{horo2006}. On one hand, the explicit form of partial transposition has been studied extensively due to which this operation has been exploited to a great extent for detection, quantification, and characterization of entanglement but on the other hand, the physical interpretation of the realignment operation is still not known which makes the study of this matrix operation more intriguing. Thus, this may be the possible reason that the permutation involved in the realignment operation has not been studied much. Some investigation shows that the realignment operation is a global operation acting on a bipartite state, unlike the partial transposition that acts on a subsystem. This makes the physical implementation of realignment operation challenging and could be another possible reason for it to be a comparatively less explored density matrix operation.\\
In the past decades, it has been shown that the moments of the density matrix contain much information about the state and measurement of these moments is relatively easy \cite{sjvan,johnston}. To avoid state tomography, several moment based methods have been proposed in the literature \cite{elben,neven,guhne2021,liu}. Moments act as practical tools in estimating some crucial properties of quantum systems. A method for experimentally measuring the moments by measuring the purity of state is given in \cite{elben19, brydges}.  A scheme based on the random unitary evolution and local measurements on single-copy quantum states given in \cite{yzhou} is shown to be more practical compared with former methods based on collective measurements on many copies of the identical state. A method to measure the $k$ partial moments using $k$ copies of the state and SWAP operators has been discussed in \cite{kwek2002,sougato, ha,cai}. However, the measurement of moments of the realignment matrix may be considered an open problem. This has provided a strong motivation to develop an experimentally feasible method to estimate the moments of the realigned matrix.\\
In this chapter, our aim is to study the following: (i) the physical nature of permutations involved in the realignment operation and (ii) to devise physically implementable entanglement detection criteria based on realignment operation. \\
Researchers have shown that the moments of the realigned matrix can be used to detect bound entangled states and hence estimation of these moments is an important task \cite{tzhang, liu}. However, it is not known yet how to estimate the realigned moments so we have considered this problem and shown that these are expressible in terms of the expectation value of partial transposition of a permutation operator concerning any $d\otimes d$ dimensional bipartite quantum state. We have made an effort to explore a method that enables the experimental realization of the realignment operator. We have achieved this by constructing the realignment matrix with the help of SWAP operator and partial transposition operation. We have shown that the first moment of the realigned matrix may be determined exactly in an experiment. Also, it may be used in an experiment to detect and estimate the amount of entanglement in a $d\otimes d$ dimensional system. We also provided a method that estimates the $k$th moment of the hermitian matrix $[R(\rho_{AB})]^{\dagger} R(\rho_{AB})$ using its first moment. This helps in deriving the entanglement detection criterion in $d_1\otimes d_2$ dimensional system that may be experimentally feasible. Moreover, we have generalized the matrix realignment operation in a three-qubit system. We will show that the newly defined realignment operator along with the SPA-PT map may be used to develop an entanglement detection criteria that has the potential to detect three-qubit genuine entangled states.
\section{Estimation of First Moment of Realigned Matrix}
\noindent It is known that measuring the moments of the density matrix is practically possible using $m$ copies of the state and controlled swap operations \cite{ha}. Exploiting the similar approach, Gray et al. have shown that the moments of the partially transposed density matrix can also be measured using swap operators \cite{sougato}. For a state $\rho_{AB}$, it has been shown that the $k$th order moment of the partial transposed matrix $\rho_{AB}^{T_{B}}$  can be measured using the individual constituents of the $k$ copies of the state, i.e., $\rho_{AB}^{\otimes k}= \otimes_{c=1}^k \rho_{A_cB_c}$ \cite{sougato}. The idea is to write the matrix power as an expectation of a partial transposition of a permutation operator for the state $\rho_{AB}$. This may be considered as an important step since partial transposition operation is not a physical operation but its $k$th order moment has been shown to be physically realizable. However, no such protocol has been discussed in the literature till now for measurement of the moments of the realignment matrix.\\		
We adopted the approach of Gray et al \cite{sougato} to show that the first moment of $R(\rho_{AB})$, i.e., $Tr[R(\rho_{AB})]$ can be expressed as the expectation value of the partial transposition of a permutation operator with respect to the state $\rho_{AB}$ and thus can be experimentally measurable. Later, we derive a few results based on the first moment of $R(\rho_{AB})$. Finally, we present separability criteria and show that it is equivalent to the original realignment criteria for the class of Schmidt-symmetric states. 
\subsection{Realigned matrix as the product of swap operator and partial transposition operation}
In the literature, it has been found that realignment operation may be defined as a way of arranging the elements of the matrix in block matrix form, which is different from the partial transposition operation. We will show here that the realigned matrix of any $d\otimes d$ dimensional bipartite state $\rho_{AB}$ may be obtained by the following two actions: (i) interchange the two columns of the density matrix $\rho_{AB}$ and (ii) apply the partial transposition operation on a subsystem. This may be illustrated by the following figure (Fig.1).
\begin{figure}[h!]
	\begin{center}
		\begin{tikzpicture}[node distance=1.5cm]
			\tikzstyle{rho} =[]
			\node (in0) [rho] {$\rho_{AB}$};
			\tikzstyle{P} = [rectangle, rounded corners, minimum width=1cm, minimum height=1cm,text centered,  draw=black, fill=blue!8]
			\node (start) [P, xshift=1.4cm] {$P$};
			
			\tikzstyle{Ptb} = [rectangle, rounded corners, minimum width=1cm, minimum height=1cm,text centered, draw=black, fill=red!10]
			\node (in1) [Ptb, xshift=3.7cm] {$T_B$};
			
			\tikzstyle{P2} = [rectangle, rounded corners, minimum width=1cm, minimum height=1cm,text centered, draw=black, fill=blue!8]
			\node (in2) [P2, xshift=6.5cm] {$P$};
			
			\tikzstyle{rhor} =[]
			\node (in3) [rhor, xshift=7.9cm] {$R(\rho_{AB})$};
			
			\tikzstyle{mid1} =[]
			\node (in4) [mid1, xshift=2.6cm, yshift=0.3cm] {$\rho_{AB}P$};
			
			\tikzstyle{mid2} =[]
			\node (in5) [mid2, xshift=5.1cm, yshift=0.3cm] {$(\rho_{AB}P)^{T_B}$};
			
			\tikzstyle{arrow} = [thick,->,>=stealth]
			\draw [arrow] (in0) -- (start);
			\draw [arrow] (start) -- (in1);
			\draw [arrow] (in1) -- (in2);
			\draw [arrow] (in2) -- (in3);
			<TikZ code>
		\end{tikzpicture}	\label{pic1}
		\caption{Schematic diagram representing the realignment operation in terms of partial transposition and permutation operator}
	\end{center}
\end{figure}
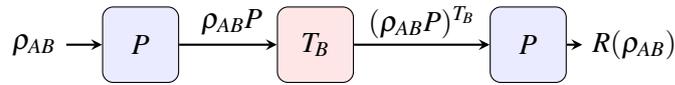
To understand these actions (i) and (ii), let us recall a two-qubit bipartite state described by the density operator $\rho_{AB}$ given in (\ref{rho}). After performing the first action (i), i.e., interchanging the second and third columns of $\rho_{AB}$, it reduces to
\begin{eqnarray}
	\rho_{AB}^{(12)}=
	\begin{pmatrix}
		\rho_{11} & \rho_{13} & \rho_{12} & \rho_{14}\\
		\rho_{12}^* & \rho_{23} & \rho_{22} & \rho_{24}\\
		\rho_{13}^* & \rho_{33} & \rho_{23}^* & \rho_{34}\\
		\rho_{14}^* & \rho_{34}^* & \rho_{24}^* & \rho_{44}\\
	\end{pmatrix}
	\label{rho12}
\end{eqnarray}
Applying the second action (ii) on $\rho_{AB}^{(12)}$, i.e., performing the partial transposition operation on qubit $B$, gives
\begin{eqnarray}
	(\rho_{AB}^{(12)})^{T_{B}}=
	\begin{pmatrix}
		\rho_{11} & \rho_{12}^* & \rho_{12} & \rho_{22}\\
		\rho_{13} & \rho_{23} & \rho_{14} & \rho_{24}\\
		\rho_{13}^* & \rho_{14}^* & \rho_{23}^* & \rho_{24}^*\\
		\rho_{33} & \rho_{34}^* & \rho_{34} & \rho_{44}\\
	\end{pmatrix}
	\label{rhopt}
\end{eqnarray}
Again applying the first action (interchanging 2nd and 3rd columns) on $(\rho_{AB}^{(12)})^{T_{B}}$,
it can be verified that the resulting transformed matrix is nothing but the realigned matrix $R(\rho_{AB})$ given in (\ref{Rrho}).  Therefore, we have shown that realignment operation, which is a non-physical operation may be expressed in terms of another non-physical operation (i.e., partial transposition operation) together with the permutation operation. It may be observed that the interchange of two columns may be realized in an experiment using a swap operator. Thus, if $\rho_{AB}$ denote a $d \otimes d$ dimensional state and $R(\rho_{AB})$ denotes the matrix obtained after applying realignment operation to $\rho_{AB}$, then the realigned matrix $R(\rho_{AB})$ can be expressed as
\begin{equation}
	R(\rho_{AB}) = (\rho_{AB} P)^{T_B}P  \label{rhorp}
\end{equation}
where $T_B$ denotes the partial transposition operation with respect to the subsystem $B$ and $P$ denotes the SWAP operator in the computational basis, which is given by
\begin{eqnarray}
	P	= \sum_{i,j=0}^{d-1} |ij\rangle \langle ji| \label{swap}
\end{eqnarray} The SWAP operator $P$ has the following properties:
\begin{enumerate}
	\item $P$ is a unitary operator with $P^2 = I_{d^2}$, where $I_{d^2}$  denotes the identity operator in $d \otimes d$ dimensional system.
	\item $Tr[P] = Tr[P^{T_B}] = d$
	\item  $PP^{T_B} =  P^{T_B}$ where 
	$P^{T_B} = \sum_{i,j=0}^{d-1} |ii\rangle \langle jj|$
	denote the partial transposition of the SWAP operator.
\end{enumerate}

\subsection{Estimation of the first moment of $R(\rho_{AB})$}
\noindent We are now in a position to estimate the first moment of the realigned matrix $R(\rho_{AB})$, where $\rho_{AB}$ denote the $d\otimes d$ dimensional bipartite system.
{\theorem \label{thm-6.1} Let $\rho_{AB}$ be a $d \otimes d$ dimensional bipartite state. If $R(\rho_{AB})$ denote its realigned matrix then the first moment of $R(\rho_{AB})$ is given by
	\begin{eqnarray}
		t_1 :=	Tr[R(\rho_{AB})] = Tr[\rho_{AB} P^{T_B}] \label{thm1}
	\end{eqnarray}
	where $P$ is the SWAP operator defined in (\ref{swap}).}
\begin{proof}
	The first moment of $R(\rho_{AB})$ may be defined as
	\begin{eqnarray}
		t_1 :=	Tr[R(\rho_{AB})] 
		\label{def1}
	\end{eqnarray}
	Using (\ref{def1}) and the expression of $R(\rho_{AB})$ given in (\ref{rhorp}), we get 
	\begin{eqnarray}
		t_1 :=	Tr[R(\rho_{AB})] &=& Tr[(\rho_{AB} P)^{T_B}P] \nonumber\\
		&=& Tr[(\rho_{AB}P) P^{T_B}] \label{e1}\nonumber\\
		&=& Tr[\rho_{AB} (P P^{T_B})]\nonumber\\
		&=& Tr[\rho_{AB} P^{T_B}] 
		\label{e2}
	\end{eqnarray}
	The equality in the second line follows from the relation $Tr[X Y^{T_B}] = Tr[X^{T_B} Y]$ and the last step follows from the fact that $P  P^{T_B} =  P^{T_B}$.
\end{proof}
Therefore, we may conclude that the first moment of $R(\rho_{AB})$, i.e., $t_1 = Tr[R(\rho_{AB})]$ can be measured experimentally since it can be expressed as the expectation value of the partial transposition of the permutation operator with respect to any $d \otimes d$ dimensional bipartite state $\rho_{AB}$.
\section{Usefulness of the First Moment of Realigned Matrix}	
\noindent We now show that the first moment of the realigned matrix may be used in the detection and quantification of entanglement.
\subsection{Detection of $d \otimes d$ dimensional entangled state}   	
\noindent The realignment criterion can be considered as a strong criterion as it detects both PPTES and NPTES. However, it is not known whether this criterion may be implemented in the laboratory or not. To address this question, we show that the realignment criterion may be re-expressed in terms of the first moment of the realigned matrix and thus it may be realized in an experiment.
{\theorem \label{thm-6.2} If a bipartite state $\rho_{AB}$ is separable, then 
	\begin{eqnarray}
		t_{1}:=Tr[\rho_{AB} P^{T_B}] \leq 1
\end{eqnarray}}
where the SWAP operator $P$ is given in (\ref{swap}).
\begin{proof} Let us start with the statement of the realignment criterion. It state that if a bipartite state $\rho_{AB}$ is separable then 
	\begin{eqnarray}
		||R(\rho_{AB})||_1 \leq 1
		\label{realigncr}
	\end{eqnarray}
	where $||.||_1$ defines the trace norm.\\
	Using Result \ref{res-trnorm} and (\ref{realigncr}), we get
	\begin{eqnarray}
		Tr[R(\rho_{AB})] \leq ||R(\rho_{AB})||_1 \leq 1
		\label{ineq301}
	\end{eqnarray} 
	From Theorem \ref{thm-6.1}, we have $Tr[\rho P^{T_B}] = Tr[R(\rho_{AB})]$.
	Therefore, the theorem is proved using (\ref{ineq301}) 
\end{proof}
{\remark For the class of Schmidt-symmetric states described by the density operator $\sigma_{AB}$ defined in \cite{hertz}, we have $Tr[R(\sigma_{AB})] = ||R(\sigma_{AB})||_1$. Therefore, the criteria given in $Theorem-\ref{thm-6.2}$ will become equivalent to the original realignment criteria. Thus, the result given in $Theorem-\ref{thm-6.2}$ will be strong when considering the class of Schmidt-symmetric states.}
\subsection{Quantification of $d \otimes d$ dimensional entangled state}
In the case of a bipartite system, concurrence is a well-known measure of entanglement. An elegant formula of concurrence for two-qubit states has been given by Wootters \cite{wootters}. However, the closed formula of concurrence in a higher dimensional quantum system is not known yet. To overcome this situation, researchers have proposed several analytical lower bounds of concurrence for mixed bipartite states in higher dimensions \cite{kchen,mintert2}. Also, the lower bound of the concurrence has been studied through the witness operator in the literature \cite{mintert2}. Witness operator may be implementable in the experiment but it is not known that the lower bound of the concurrence obtained through witness operator always gives better results than the bound obtained by Chen et al. \cite{kchen}. Thus, there may exist a $d \otimes d$ dimensional bipartite entangled state for which the Chen et al. bound is better than the bound obtained through the witness operator. Hence it is important to review the lower bound of the concurrence and express it in terms of the first moment of the realigned matrix. If we are able to do this then the lower bound of the concurrence given by Chen et al. \cite{kchen} may be realized in the experiment. To achieve our aim, let us rewrite the lower bound of the concurrence for arbitrary $d \otimes d$-dimensional system as \cite{kchen}
\begin{eqnarray}
	C(\rho_{AB}) \geq \sqrt{\frac{2}{d(d - 1)}}\left(max(||\rho_{AB}^{T_B}||_1, ||R(\rho_{AB})||_1)-1\right) \label{albe}
\end{eqnarray}
Using the property of trace norm, we have the following inequality 
\begin{eqnarray}
	Tr[\rho_{AB}^{T_B}P] \leq ||\rho_{AB}^{T_B}P||_1 \leq ||\rho_{AB}^{T_B}||_1 ||P||_1 = d^2 ||\rho_{AB}^{T_B}||_1
	\label{clb1}
\end{eqnarray}
The inequality (\ref{clb1}) may be re-expressed in the simplified form as
\begin{eqnarray}
	||\rho_{AB}^{T_B}||_1 \geq \frac{Tr[\rho_{AB}^{T_B}P]}{d^2}=\frac{Tr[\rho_{AB}P^{T_{B}}]}{d^2}
	\label{clb2}
\end{eqnarray}
Thus, we have $ max(||\rho_{AB}^{T_B}||_1, ||R(\rho_{AB})||_1) \geq max(\frac{Tr[\rho_{AB}P^{T_{B}}]}{d^2},Tr[\rho_{AB}P^{T_{B}}]) = Tr[\rho_{AB}P^{T_{B}}]$. Hence, the lower bound of concurrence can be derived as 
\begin{eqnarray}
	C(\rho_{AB}) \geq \sqrt{\frac{2}{d(d - 1)}}(Tr[\rho_{AB}P^{T_B}] - 1)
	\label{lbcon}
\end{eqnarray}
From (\ref{lbcon}), it is clear that the lower bound of the concurrence of any $d \otimes d$ dimensional entangled system $\rho_{AB}$ may be estimated through the first moment of the realigned matrix of $\rho_{AB}$. Since the value of  $Tr[R(\rho_{AB})]=Tr[\rho_{AB}P^{T_B}]$ may be estimated practically so the estimation of the lower bound of the concurrence may be possible in an experiment.  
\section{Estimation of the Moments of the Hermitian Matrix $[R(\rho_{AB})]^\dagger R(\rho_{AB})$}
We may observe that the results obtained in the previous section may be used to detect $d\otimes d$ dimensional entangled states. But to deal with the more general case, i.e., for $d_1\otimes d_2$ dimensional system, let us first define the $k$th moment of Hermitian positive semi-definite operator $[R(\rho_{AB})]^\dagger R(\rho_{AB})$, which may be defined as
\begin{eqnarray}
	T_k = Tr[([R(\rho_{AB})]^\dagger R(\rho_{AB}))^k]
\end{eqnarray} 
As $R(\rho_{AB})$ is not Hermitian, so the moments of the operator $[R(\rho_{AB})]^\dagger R(\rho_{AB})$ has been considered. Furthermore, since the physical implementation of the moments of the realignment matrix is not known, we have studied the representation of the moments of the matrix  $[R(\rho_{AB})]^\dagger R(\rho_{AB})$ that may be implementable in the experiment. \\
To achieve the above task, we derive a lower bound of the first moment of $[R(\rho_{AB})]^\dagger R(\rho_{AB})$ in terms of the first moment of $R(\rho_{AB})$ that has been shown to be experimentally measurable. Then we estimate the $k$th moment of $[R(\rho_{AB})]^\dagger R(\rho_{AB})$ that may also be expressed as the product of the maximum singular value of $R(\rho_{AB})$ and the first moment of $[R(\rho_{AB})]^\dagger R(\rho_{AB})$. Finally, we present entanglement detection criteria based on these moments and provide a few examples to illustrate that the derived criteria can detect BES as well as NPTES.
\subsection{Estimation of the first moment and kth moment of $[R(\rho_{AB})]^\dagger R(\rho_{AB})$ }
Let us consider a $d_1\otimes d_2$ dimensional quantum state described by the density operator $\rho_{AB}$ and the realigned matrix of it is denoted by $R(\rho_{AB})$. Assume that $R(\rho_{AB})$ has $k$ non-zero singular values that may be arranged in an ascending order as $\sigma_1 \geq \sigma_2 \geq . . . \geq \sigma_k$, where $1\leq k \leq min\{d_1^2,d_2^2\}$.\\ 
{\lemma \label{lemma-6.1} The first moment $T_{1}$ of the Hermitian positive semi-definite matrix $[R(\rho_{AB})]^\dagger R(\rho_{AB})$ may be bounded below by the following bound:
	\begin{eqnarray}
		T_1 \geq \frac{(Tr[\rho_{AB} P^{T_B}])^2}{k} \label{T1lb}
	\end{eqnarray}
	where $k$ denotes the number of non-zero singular values of $R(\rho_{AB})$.}
\begin{proof}
	Applying Result-\ref{res-trnorm} and using $Tr[R(\rho_{AB})]=Tr[\rho_{AB}^{T_{B}}P]$, we have
	\begin{eqnarray}
		Tr[\rho^{T_B} P] = Tr[R(\rho_{AB})] \leq ||R(\rho_{AB})||_1 \label{r11}
	\end{eqnarray}
	Further, using Result-\ref{res-n1n2} with $R(\rho_{AB})$, we have
	\begin{eqnarray}
		||R(\rho_{AB})||_1 \leq \sqrt{k}||R(\rho_{AB})||_2 = \sqrt{k \sum_{i=1}^k \sigma_i^2 (R(\rho_{AB}))}= \sqrt{kT_1} \label{r12}
	\end{eqnarray}
	where $\sigma_i$ denotes the $ith$ singular value of $R(\rho_{AB})$ and $T_{1}=\sum_{i=1}^{k}\sigma_{i}^{2}(R(\rho))$.\\
	From (\ref{r11}) and (\ref{r12}), we have
	\begin{eqnarray}
		Tr[\rho^{T_B} P] \leq \sqrt{kT_1}
		\label{ineq42}
	\end{eqnarray}
	Simplifying (\ref{ineq42}), we prove the required lemma.
\end{proof}
Now, we are in the position to derive the lower and upper bound of the kth moment of the hermitian matrix $[R(\rho_{AB})]^\dagger R(\rho_{AB})$ in terms of $T_{1}$ and it may be expressed by the following theorem.
{\theorem \label{thm-6.3} The lower and upper bound of the $k$th moment of  $[R(\rho_{AB})]^\dagger R(\rho_{AB})$ may be given by
	\begin{eqnarray}
		(\sigma^2_{min}(R(\rho_{AB})))^{k-1} T_1 \leq T_k \leq 	(\sigma^2_{max}(R(\rho_{AB})))^{k-1} T_1
\end{eqnarray}}
\begin{proof}
	Since $[R(\rho_{AB})]^\dagger R(\rho_{AB})$ is Hermitian so we can apply Result-\ref{res-lmintr} on $[R(\rho_{AB})]^\dagger R(\rho_{AB})$. Therefore, using (\ref{jb}), the lower and upper bound of the second moment of the matrix $[R(\rho_{AB})]^\dagger R(\rho_{AB})$ may be obtained as follows.
	\begin{eqnarray}
		&&	\lambda_{min}([R(\rho_{AB})]^\dagger R(\rho_{AB})) Tr[[R(\rho_{AB})]^\dagger R(\rho_{AB})]  \leq Tr[([R(\rho_{AB})]^\dagger R(\rho_{AB}))^2] \nonumber\\&&\leq  \lambda_{max}([R(\rho_{AB})]^\dagger R(\rho_{AB})) Tr[[R(\rho_{AB})]^\dagger R(\rho_{AB})]
	\end{eqnarray}
	Since, $\lambda_{i} ([R(\rho_{AB})]^\dagger R(\rho_{AB})) = \sigma^2_i (R(\rho_{AB}))$, the above inequalities can be re-expressed as 
	\begin{eqnarray}
		\sigma_{min}^2(R(\rho_{AB})) T_1 \leq T_2 \leq  \sigma_{max}^2(R(\rho_{AB}))  T_1 
	\end{eqnarray}
	where $T_{1}=Tr[[R(\rho_{AB})]^\dagger R(\rho_{AB})]$.\\
	Following the same procedure for the $k$th moment, we deduce that the following inequality holds for all $1 < k \leq n$.
	\begin{eqnarray}
		\sigma_{min}^2(R(\rho_{AB})) T_{k-1} \leq T_k \leq  \sigma_{max}^2(R(\rho_{AB}))  T_{k-1} 
	\end{eqnarray}
	By the principle of mathematical induction, we have
	\begin{eqnarray}
		(\sigma^2_{min}(R(\rho_{AB})))^{k-1}  T_1 &\leq&	\sigma^2_{min}(R(\rho_{AB})) T_{k-1} \leq T_k \nonumber\\ &\leq&  \sigma^2_{max}(R(\rho_{AB}))  T_{k-1} \leq (\sigma^2_{max}(R(\rho_{AB})))^{k-1}  T_1 
	\end{eqnarray}
	Hence proved.
\end{proof}
\subsection{Entanglement detection criteria for $d_1 \otimes d_2$ dimensional bipartite system}
We now derive an entanglement detection criterion for $d_1 \otimes d_2$ dimensional bipartite system. The required entanglement detection criterion is derived in terms of the inequality that involves the $k$th moment $T_{k}$. The kth moment $T_{k}$ may be estimated in terms of the function of the first moment $t_{1}$ of the realigned matrix. Since $t_{1}=Tr[\rho_{AB} P^{T_B}]$ is an experimentally realizable quantity, the $d_1 \otimes d_2$ dimensional entangled state may be detected in an experiment.
{\theorem \label{thm-6.4} Let $\rho_{AB}$ be any bipartite state in $d_1 \otimes d_2$ dimensional Hilbert space. Consider the $k$ non-zero singular values of the realigned matrix $R(\rho_{AB})$ that may be denoted as $\sigma_1, \sigma_2, \ldots \sigma_k$ with $1\leq k \leq min\{d_1^2,d_2^2\}$. If $\rho_{AB}$ is separable then the following inequality holds:
	\begin{equation}
		\frac{t_{1}^2}{k} \leq	1 - k(k-1) {D_k}^{1/k} \label{eq-thm6.4}
	\end{equation}
	where $D_k = \prod_{i=1}^{k} \sigma_i^2(R(\rho_{AB}))$.}
\begin{proof} Let $\rho_{AB}$ be any arbitrary separable state in $d_1 \otimes d_2$ dimensional system. Using the $R$-moment criterion given in $Theorem-\ref{thm4.1}$, we have
	\begin{eqnarray}
		k(k - 1) D_k^{1/k} + T_1 \leq 1 \label{rmoment1}
	\end{eqnarray}
	where the quantity $D_k$ can be estimated in terms of moments using the bounds given in $Lemma-\ref{lemma-6.1}$ and $Theorem-\ref{thm-6.3}$. The inequality given in (\ref{rmoment1}) can be rewritten as
	\begin{eqnarray}
		T_1 \leq 1 - k(k-1) D_k^{1/k} \label{rmoment}
	\end{eqnarray}
	Thus, the inequality in (\ref{eq-thm6.4}) follows using $Lemma-\ref{lemma-6.1}$. 
\end{proof}

\subsection{Examples}
Now we illustrate the efficiency of the criteria given in $Theorem-\ref{thm-6.4}$ using the following examples.
{\example Let us consider the following class of $3 \otimes 3$ PPTES \cite{hakye}.
	\begin{eqnarray}
		\rho_{\epsilon} = \frac{1}{N}
		\begin{pmatrix}
			1 & 0 & 0&0&1&0&0&0 & 1\\
			0 & 1/\epsilon^2 &0&1&0&0&0 &0&0\\
			0 & 0&\epsilon^2&0&0&0&1 &0&0\\
			0 & 1&0&\epsilon^2&0&0&0 &0&0\\
			1 & 0&0&0&1&0 &0 &0&1\\
			0 & 0&0&0&0 &1/\epsilon^2&0 &1&0\\
			0 & 0&1&0&0&0&1/ \epsilon^2 &0&0\\
			0 & 0&0&0&0&1&0 &\epsilon^2&0\\
			1 & 0&0&0&1&0&0 &0& 1\\
		\end{pmatrix}
		\label{rhoep}
	\end{eqnarray} 
	where $\epsilon>0$, $\epsilon \neq 1$ and $N=3(1+\epsilon^2+\frac{1}{\epsilon^2})$ is the normalization constant.\\
	Matrix rank of $R(\rho_{\epsilon})$ is 8. 
	Inequality in (\ref{eq-thm6.4}) is violated for $\epsilon \in [0.622496, 0.780349] \cup [1.281481, 1.606435]$ and hence the PPTES lying in this region are detected by the criteria given in $Theorem-\ref{thm-6.4}$.
	\example  Consider the class of NPTES in $3\otimes 3$ dimensional system, which is defined in (\ref{rhoa_npt})
	\begin{eqnarray}
		\rho_a =
		\begin{pmatrix}
			\frac{1-a}{2} & 0 & 0&0&0&0&0&0 & \frac{-11}{50}\\
			0 & 0&0&0&0&0&0 &0&0\\
			0 & 0&0&0&0&0&0 &0&0\\
			0 & 0&0&0&0&0&0 &0&0\\
			0 & 0&0&0&\frac{1}{2} - a&  \frac{-11}{50} &0 &0&0\\
			0 & 0&0&0&  \frac{-11}{50} &a&0 &0&0\\
			0 & 0&0&0&0&0&0 &0&0\\
			0 & 0&0&0&0&0&0 &0&0\\
			\frac{-11}{50} & 0&0&0&0&0&0 &0& \frac{a}{2}
		\end{pmatrix};\; \frac{1}{50} (25 - \sqrt{141}) \leq a \leq \frac{1}{100}(25 + \sqrt{141})
	\end{eqnarray} 
	The inequality in (\ref{eq-thm6.4}) is violated for all $a \in [\frac{1}{50} (25 - \sqrt{141}), \frac{1}{100}(25 + \sqrt{141})]$. Hence, the above state is detected by the entanglement detection criteria given in $Theorem-\ref{thm-6.4}$. }
\section{Three-qubit System}
Let us start with any three-qubit state described by the density operator $\rho_{ABC}$, which may be expressed as
\begin{eqnarray}
	\rho_{ABC} = \frac{1}{8}[
	I \otimes I \otimes I + \vec{l}. \vec{\sigma} \otimes I \otimes I + I \otimes \vec{m}.\vec{\sigma} \otimes I  + I \otimes I \otimes \vec{n}.\vec{\sigma} + \vec{u}.\vec{\sigma}  \otimes \vec{v}.\vec{\sigma} \otimes I \nonumber\\ + \vec{u}.\vec{\sigma}  \otimes I \otimes \vec{w}.\vec{\sigma} + I \otimes \vec{v}.\vec{\sigma}  \otimes \vec{w}.\vec{\sigma} + \sum_{i,j,k = x,y,z}t_{ijk} \sigma_i \sigma_j \sigma_k]
\end{eqnarray}
with 
\begin{eqnarray}
	&&	l_i = Tr[\rho_{ABC} (\sigma_i \otimes I \otimes I)], \;\; 	m_i = Tr[\rho_{ABC} (I \otimes  \sigma_i \otimes I)];\;\; 	n_i = Tr[\rho_{ABC} ( I \otimes I \otimes\sigma_i )]\nonumber\\
	&&	u_iv_i =Tr[\rho_{ABC} (\sigma_i \otimes \sigma_i \otimes I)], \;\; v_iw_i =Tr[\rho_{ABC} (I \otimes  \sigma_i \otimes \sigma_i)];\;\;  u_iw_i =Tr[\rho_{ABC} ( \sigma_i \otimes I  \otimes \sigma_i)]\; \text{for}\; i=x,y,z\nonumber\\
	&&	t_{ijk} = Tr[\rho_{ABC} (\sigma_i \otimes \sigma_j\otimes \sigma_k)]\;\; \text{for}\; (i,j,k=x,y,z)
\end{eqnarray}
A density matrix $\rho_{ABC}$ on the Hilbert space $\mathcal{H}_A\otimes \mathcal{H}_B \otimes \mathcal{H}_C$ ($\mathcal{H}_i,i=A,B,C$, denotes the Hilbert spaces of dimension 2) is fully separable if it can be written as 
\begin{eqnarray}
	\rho_{ABC} = \sum_i p_i \rho_i^{A} \otimes \rho_i^{B} \otimes \rho_i^{C} \;\text{with}\; p_i \geq 0 \; \text{and}\;\sum p_i = 1 \nonumber
\end{eqnarray}
If a state is not fully separable then the state either belongs to biseparable class or the class of genuine entangled states. A three-qubit mixed state is said to be a PPTES or BES if it is genuinely entangled and its partial transposition on any qubit is positive. As we know the partial transposition operation is not completely positive so the SPA-PT map for the three-qubit system has been studied to classify its different SLOCC inequivalent classes \cite{kumari1}. Here, we apply the SPA-PT map on the three-qubit system for the detection of genuine entangled states.\\ 
Let $\rho_{ABC}$ be a three-qubit state and  $\rho_{ABC}^{T_A}$, $\rho_{ABC}^{T_B}$,  $\rho_{ABC}^{T_C}$ denote the partial transposition of $\rho_{ABC}$ with respect to the subsystem $A$, $B$ and $C$ respectively. We now use the structural physical approximation of partial transposition (SPA-PT) of a single qubit in a three-qubit system. The SPA-PT map on the qubit $X$ ($X=A,B,C)$ of
the three-qubit state $\rho_{ABC}$, may be defined as 
\begin{eqnarray}
	\widetilde{\rho^{T_X}_{ABC}} = \frac{1}{10} (I_2 \otimes I_2 \otimes I_2) + \frac{1}{5}\rho^{T_X}_{ABC},~~X=A,B,C \label{spapt}
\end{eqnarray}
where $I_2$ denotes the $2 \times 2$ identity matrix.
The SPA-PT map defined above is completely positive and hence may be implemented in an experiment \cite{kumari1}.
The statement of the necessary conditions for the separability and biseparability of a three-qubit state is given by \cite{kumari1}.
\begin{enumerate}
	\item[(a)]  If a tripartite state $\rho_{ABC}$ is separable or biseparable in the $A|BC$ cut, then 
	\begin{eqnarray}
		\lambda_{min} (\widetilde{\rho^{T_A}_{ABC}}) \geq \frac{1}{10} \label{Anu1}
	\end{eqnarray}
	\item [(b)] If $\rho_{ABC}$ is separable or biseparable in the $B|AC$ cut, then 
	\begin{eqnarray}
		\lambda_{min} (\widetilde{\rho^{T_B}_{ABC}}) \geq \frac{1}{10} \label{anu2}
	\end{eqnarray}
	\item [(c)] If $\rho_{ABC}$ is separable or biseparable in the $C|AB$ cut, then 
	\begin{eqnarray}
		\lambda_{min} (\widetilde{\rho^{T_C}_{ABC}}) \geq \frac{1}{10} \label{anu3}
	\end{eqnarray}
\end{enumerate}	
It may be noted that the violation of the inequalities (\ref{Anu1}), (\ref{anu2}), (\ref{anu3}) implies that the state $\rho_{ABC}$ is genuinely entangled.
\subsection{Realignment map on three-qubit system}
Here, we generalize the concept of a realignment map to three party system and thus define a realignment operation on $2 \otimes 2 \otimes 2$ system. If $\rho_{ABC}$ denotes a quantum state in $2 \otimes 2 \otimes 2$ dimensional Hilbert space, then $\rho_{ABC}$ can be expressed in block matrix form as
\begin{eqnarray}
	\rho_{ABC} =
	\begin{pmatrix}
		A &B&C&D\\
		B^*&E&F&G\\
		C^*&F^*&H&I\\
		D^*&G^*&I^*&J
	\end{pmatrix}
\end{eqnarray}
where $A, B, C, D, E, F, G, H, I$, and $J$ represent a $2 \times 2$ block matrices.\\
Then the realignment matrix of $\rho_{ABC}$ may be defined as
\begin{eqnarray}
	\mathcal{R}(\rho_{ABC}) =
	\begin{pmatrix}
		vec(A) & vec(B)\\
		vec(C) & vec(D)\\
		vec(B^*) & vec(E)\\
		vec(F) & vec(G)\\
		vec(C^*) & vec(F^*)\\
		vec(H) & vec(I)\\
		vec(D^*) & vec(G^*)\\
		vec(I^*) & vec(J)
	\end{pmatrix} 
	\label{rhortri}
\end{eqnarray}
For any matrix $Z=(z_{ij}) \in C^{m\times n}$,  $vec(Z)$ is defined as
\begin{eqnarray}
	vec(Z) = (z_{11}, . . . z_{m1}, z_{12}, . . .,z_{m2}, z_{1n}, . . .,z_{mn})
\end{eqnarray}
Let $Q$ be a permutation operator defined by $Q(e_i) = \pi{(e_i)}$ where $\pi$ is a permutation on $S_8$ given by $\pi(1,2,3,4,5,6,7,8) = (1,3,5,7,2,4,6,8)$ and $e_i$ denotes a $8 \times 1$ column vector with $i$th entry equal to 1 and other entries equal to zero. Thus, the permutation operator $Q$ may be expressed in matrix form as
\begin{eqnarray}
	Q=
	\begin{pmatrix}
		1 & 0 & 0 & 0 & 0 & 0 & 0 & 0\\
		0 & 0 & 1 & 0 & 0 & 0 & 0 &0\\
		0 & 0 & 0 & 0 & 1 & 0 & 0&0\\
		0 & 0 & 0 & 0 & 0 & 0 & 1&0\\
		0 & 1 & 0 & 0 & 0 & 0 & 0 &0\\
		0 & 0 & 0 & 0 & 0 & 1 & 0 & 0\\
		0 & 0 & 0 & 0 & 0 & 0 & 0& 1
	\end{pmatrix}
\end{eqnarray}
Let $X =[X_{ij}]_{i,j=1}^4$ be any $8 \times 8$ matrix written in block matrix form with each block matrix $X_{ij}$ of size 2. Then $\tau$ is the matrix operation defined as $X^{\tau} = [X_{ij}^T]_{i,j=1}^4$ where $T$ denotes the usual transpose operation. Then the transformed matrix $\mathcal{R}(\rho_{ABC})$ obtained after applying the realignment operation on the state $\rho_{ABC}$ can be expressed as
\begin{eqnarray}
	\mathcal{R}(\rho_{ABC}) = (\rho_{ABC} Q)^{\tau} \label{rhorq}
\end{eqnarray}
\subsection{Detection of genuine entanglement in three-qubit system}
Now we use the realignment operation defined in (\ref{rhortri}) to derive a genuine entanglement detection criteria for three-qubit states. To develop this criterion, we use the SPA-PT map defined in (\ref{spapt}). This criterion involves the computation of the minimum eigenvalue of the Hermitian matrix $[\mathcal{R}(\rho_{ABC})]^{\dagger}\mathcal{R}(\rho_{ABC})$, where $\mathcal{R}$ given in (\ref{rhorq}) is the realignment operation defined on the three-qubit system. Later, we will show that our criteria have the potential to detect three-qubit entangled states.
{\theorem \label{thm-6.5} Let $\rho_{ABC}$ be a tripartite state in $2 \otimes 2 \otimes 2$ dimensional Hilbert space. If the state $\rho_{ABC}$ is biseparable in $A|BC$ cut or $B|AC$ cut or $C|AB$ cut respectively or if it is fully separable then 
	\begin{eqnarray}
		\lambda_{min}([\mathcal{R}(\rho_{ABC})]^{\dagger}\mathcal{R}(\rho_{ABC}) + \widetilde{\rho_{ABC}^{T_{X}}}) \geq \lambda_{min}([\mathcal{R}(\rho_{ABC})]^{\dagger}\mathcal{R}(\rho_{ABC})) + \frac{1}{10} \nonumber\\ \label{s1}
	\end{eqnarray}
	where $\widetilde{\rho_{ABC}^{T_X}}$ denotes the SPA-PT with respect to the subsystem $X=A,B,C$.}
\begin{proof}
	Let us assume that the three-qubit state $\rho_{ABC}$ represents either a separable or a biseparable state in the $A|BC$ cut. Since $[\mathcal{R}(\rho_{ABC})]^{\dagger}\mathcal{R}(\rho_{ABC})$ and $\rho_{ABC}^{T_A}$ are Hermitian, using (\ref{weyl}), we have
	\begin{eqnarray}
		\lambda_{min}([\mathcal{R}(\rho_{ABC})]^{\dagger}\mathcal{R}(\rho_{ABC}) + \widetilde{\rho_{ABC}^{T_A}}) &\geq& \lambda_{min}([\mathcal{R}(\rho_{ABC})]^{\dagger}\mathcal{R}(\rho_{ABC})) + \lambda_{min}( \widetilde{\rho_{ABC}^{T_A}}) \nonumber \\ &\geq& \lambda_{min}([\mathcal{R}(\rho_{ABC})]^{\dagger}\mathcal{R}(\rho_{ABC}))+ \frac{1}{10}
	\end{eqnarray}
	where the last inequality follows from (\ref{Anu1}). This proves $Theorem-\ref{thm-6.5}$ for $X=A$. Similarly, the proof follows for $X=B$ and $C$.
\end{proof}
{\corollary \label{cor6.1} A tripartite state $\rho_{ABC}$ is genuine entangled if the following inequality holds with respect to every subsystem $X = A, B, C$
	\begin{eqnarray}
		\lambda_{min}([\mathcal{R}(\rho_{ABC})]^{\dagger}\mathcal{R}(\rho_{ABC}) + \widetilde{\rho^{T_X}}) < \lambda_{min}([\mathcal{R}(\rho_{ABC})]^{\dagger}\mathcal{R}(\rho_{ABC})) + \frac{1}{10} \label{s4}
\end{eqnarray}}
The result stated in $Corollary-\ref{cor6.1}$ is strong in the sense that it can detect three-qubit genuine entangled states. 
\subsection{Example}
Consider the following family of entangled three-qubit states constructed from mutually unbiased basis as \cite{jaferi}
\begin{eqnarray}
	\rho_{ABC}(p_1,p_2,p_3) &=& \frac{1}{8} [I_2 \otimes I_2 \otimes I_2 + r_1 \sigma_z \otimes \sigma_z \otimes I_2 + r_2 \sigma_z \otimes I_2 \otimes \sigma_z + r_3 I_2 \otimes \sigma_z \otimes \sigma_z + r_4 \sigma_x \otimes \sigma_x \otimes \sigma_x \nonumber \\ &&+  r_5 \sigma_x \otimes \sigma_y \otimes \sigma_y  + r_6 \sigma_y \otimes \sigma_x \otimes \sigma_y + r_7 \sigma_y \otimes \sigma_y \otimes \sigma_x] \label{eq-pstate}
\end{eqnarray}
where $r_1 = r_2 = r_3 = p_1 + p_2 - p_3$, $r_4 = p_1 - p_2 + 3p_3$,  $r_5 = r_6 = r_7 = -p_1 + p_2 + p_3$ with $p_1 + p_2 + 3 p_3 = 1$ and $0\leq p_i \leq 1$ for $i=1,2,3$.
Using the separability criteria given in $Theorem- \ref{thm-6.5}$, we find that the inequality given in (\ref{s1}) is violated for each subsystem $X=A,B,C$ for some values of $(p_1, p_3)$, which is shown in Fig-\ref{pstate}.
Thus, by $Corollary-\ref{cor6.1}$, we conclude that $\rho_{ABC}(p_1,p_2,p_3)$ is a genuine three-qubit entangled state. 
\begin{figure}[h!]
	\begin{center}
		\includegraphics[width=0.33\textwidth]{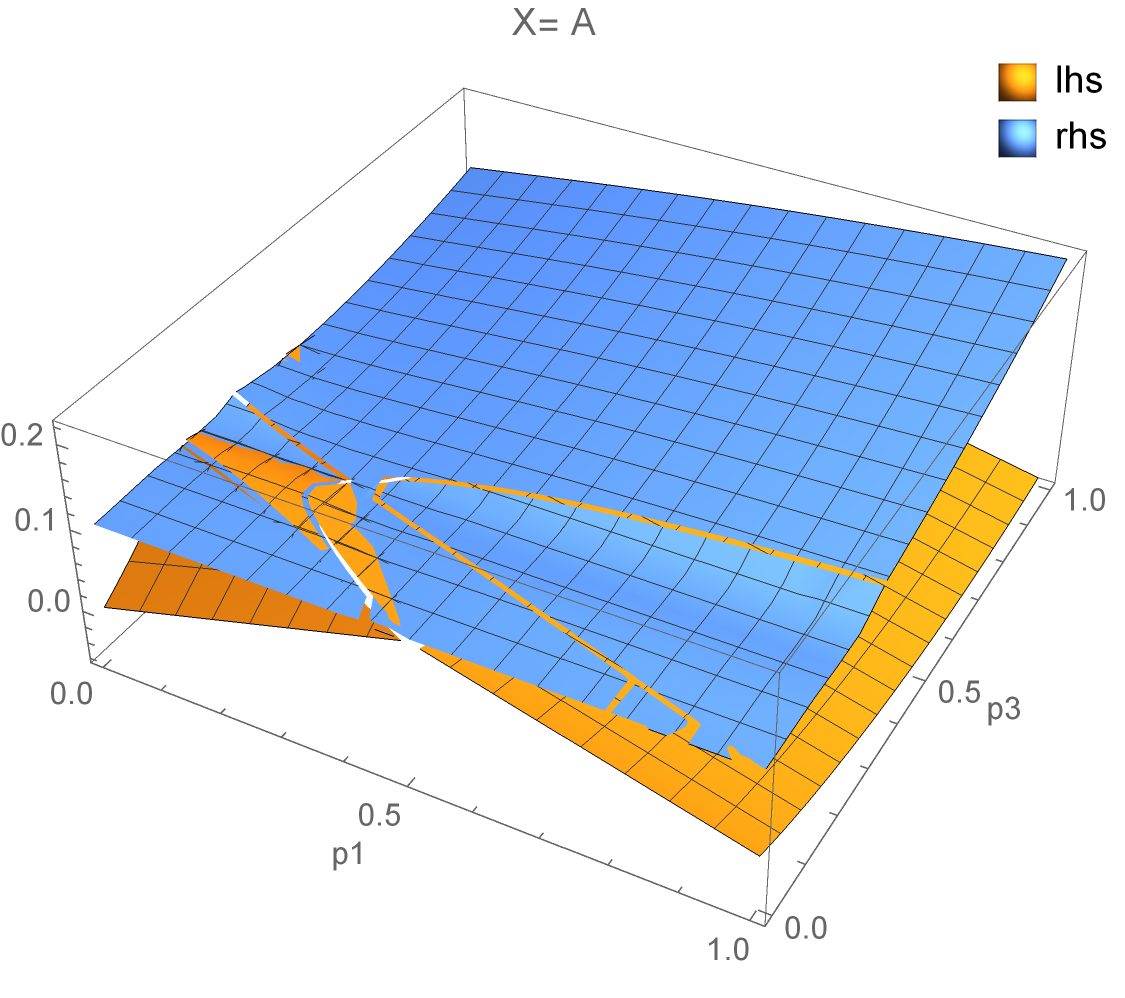}
		\includegraphics[width=0.33\textwidth]{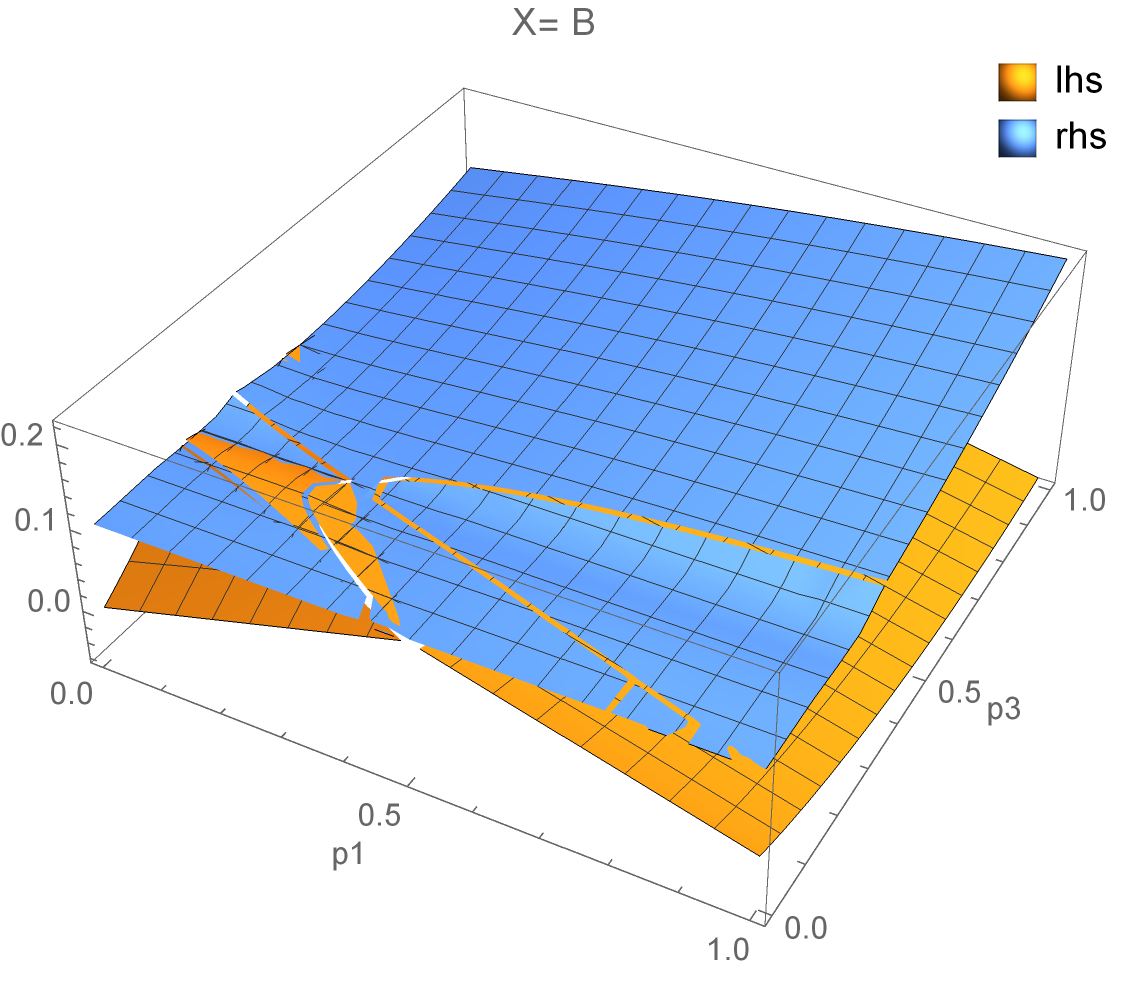}
		\includegraphics[width=0.33\textwidth]{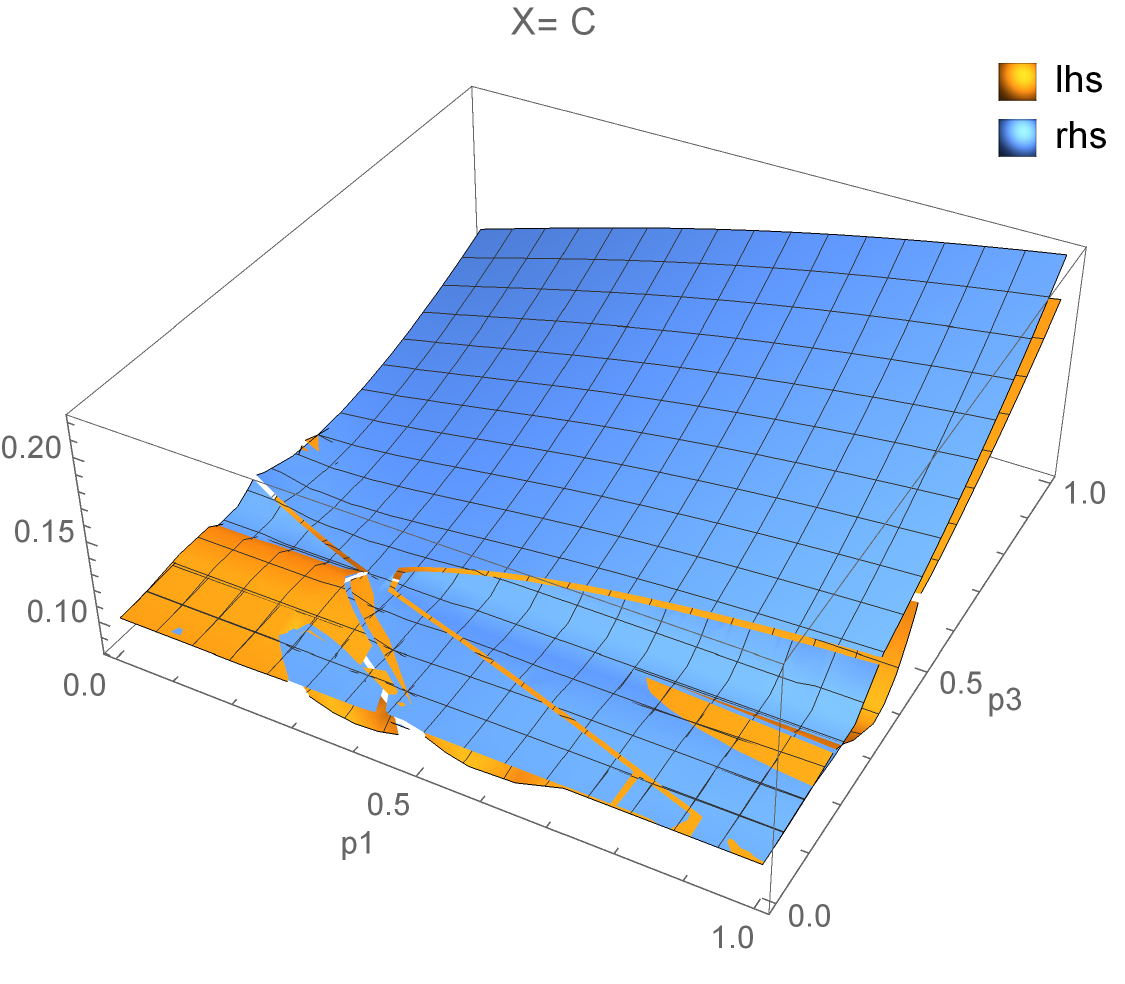}
		\caption{For the state $\rho_{ABC}{(p_1,p_2,p_3)}$ given in (\ref{eq-pstate}), the yellow and blue curve represents the left hand side and the right hand side of the inequality (\ref{s4}) for $X=A, B$, and $C$. The graph is plotted with respect to the state parameter $p_1$ and $p_3$. The figure shows that the inequality (\ref{s4}) holds for a large number of states and hence those entangled states are detected by the method developed in this chapter.}
		\label{pstate}
	\end{center}
\end{figure}

\section{Conclusion}
To summarize, we have shown that the realignment operation may be expressed in terms of partial transposition and permutation operator (SWAP operator). By doing this, we have demonstrated that $Tr[R(\rho_{AB})]$ can be expressed in terms of the expectation value of SWAP operator which indicates the possibility of its physical measurement. Next, this physically realizable quantity, $Tr[R(\rho_{AB})]$ has been used to obtain separability criteria. This criterion is weaker in terms of detection power, but it is important because it reduces to the original realignment criteria for Schmidt-symmetric states. Since $R(\rho_{AB})$ is not Hermitian, so we have derived the $k$th moments of the Hermitian operator $[R(\rho_{AB})]^{\dagger} R(\rho_{AB})$ in terms of its first moment which reduces the heavy calculations required to compute higher order moments. Using these moments, we propose entanglement detection criteria that are shown to detect NPTES as well as BES. Moreover, we have defined a new realignment operation for a three-qubit system and have shown that the three-qubit realigned matrix can be expressed using a permutation operator. We derive separability criteria to detect three-qubit genuine entangled states using SPA-PT and the proposed realignment operation. 
\begin{center}
	****************
\end{center}

\addcontentsline{toc}{chapter}{Conclusion and Future Scope}
\chapter*{Conclusion and Future Scope}
\section*{Conclusion}
In the present thesis, we have extensively studied and used realignment criteria to characterize entanglement in bipartite and multipartite systems. We have established a few separability criteria that successfully detect NPTES as well as PPTES. Although the topic of detection of entanglement has been extensively studied in the literature through many approaches, the majority of these criteria are not physically realizable. This means that they are well accepted in the mathematical language but cannot be implemented in a laboratory setting. In this thesis, we have proposed some theoretical ideas to realize these entanglement detection criteria experimentally.\\
In \noindent { $Chapter-1$}, we have introduced basic definitions and concepts of linear algebra and quantum mechanics, along with some important mathematical results obtained in the literature. We then provide a brief review on the theory of bipartite and multipartite entanglement. We also discuss here a few existing entanglement detection criteria. This chapter provides a foundation for understanding the upcoming chapters of the thesis. \\
In $Chapter-2$, we have constructed a family of witness operators that is efficient in detecting BES and NPTES. For this, we first constructed a linear map using the combination of partial transposition and
realignment operation. Then we found some conditions on the
parameters of the map for which the map represented a positive
map. Further, we have constructed a Choi matrix corresponding
to the map and have shown that it is not completely positive.
We then constructed an entanglement witness operator, which is
based on the linear combination of the function of the Choi matrix and the identity matrix and has shown that it can detect both NPTES and PPTES. Finally, we prove its efficiency by detecting several bipartite BES that were previously undetected by some well-known separability criteria. We also compared the
detection power of our witness operator with three well-known
powerful entanglement detection criteria, namely, dV criterion \cite{dv}, CCNR criterion \cite{rudolph2003} and the separability criteria based on correlation tensor (CT) proposed by Sarbicki et al. \cite{sarbicki} and find that our witness operator detects more entangled states than these criteria. \\
\noindent  In $Chapter-3$, we have introduced a novel entanglement criterion for bipartite systems based on the moments of the Hermitian matrix  $[R(\rho_{AB})]^{\dagger} R(\rho_{AB})$, where $R(\rho_{AB})$ denotes the matrix obtained after applying realignment operation on a quantum state $\rho_{AB}$, which we call here realigned moments (or $R$-moments). The motivation of this chapter is to construct experimentally realizable entanglement detection criteria that can detect NPTES as well as BES through partial knowledge of the density matrix. It has been shown that measuring partial moments is technically possible using $m$ copies of the state and controlled swap operations \cite{ha}. In this chapter, firstly, we have proposed a separability criterion (called $R$-moment criterion), which is based on moments and has been formulated in the form of an inequality that must be fulfilled by all bipartite separable states. Thus, the violation of the derived inequality reveals that the state is entangled. The criterion is efficient in the sense that it detects BES for $m \otimes n$, $mn \neq 4$ dimensional systems. Furthermore, we have shown that it performs better than the other existing criteria based on partial moments in some cases, for detection of NPTES in higher dimensional systems.   We have illustrated the efficiency of the $R$-moment criterion for the detection of NPTES and BES by examining some examples. Since the above discussed criterion requires all four moments for $2\otimes 2$ dimensional system, we have devised another moment based separability criterion for two-qubit systems that need only three moments. We also have chosen some examples of two-qubit NPTES that are undetected by other criteria based on partial moments and realignment operation. \\
\noindent In $Chapter-4$, the physical realization of a realignment map using the method of structural physical approximation has been proposed. Although the realignment criterion is one of the best for the detection of PPTES, it may not be used to detect
entanglement practically. It happens because the realignment map corresponds to a non-positive map and it is known that the non-positive maps are not experimentally implementable.  In this chapter, we have started with approximating the realignment map to a positive map using the method of structural physical approximation (SPA) and then we have shown that the structural physical approximation of realignment map (SPA-R) is completely positive. Next, we developed a separability criterion based on SPA-R map in the form of an inequality and have shown that the developed criterion not only detects NPTES but also PPTES. We estimated the eigenvalues of the
realignment matrix using moments that may be used physically
in an experiment \cite{brun,tanaka,sougato,guhne2021}. We have provided some examples to support the results obtained. Moreover, we discuss the accuracy of our approximated realignment (SPA-R) map by calculating the error of the approximation in trace norm. We have also introduced an error inequality which holds for all separable states. 
This chapter is significant because the idea of SPA of
realignment operation has been discussed here which has not been explored yet.\\
\noindent In $Chapter-5$, we have constructed a family of witness operators to detect and quantify NPTES and PPTES. Our main contribution in this work is that the constructed witness operator is used to improve the earlier obtained lower bound of the concurrence of any arbitrary entangled state in $d_{1} \otimes d_{2}\;(d_{1}\leq d_{2})$ dimensional system. In particular, our result may be used to estimate the lower bound of the concurrence of the bound entangled states in higher dimensional systems. This study is important because the exact expression of the concurrence for higher dimensional mixed bipartite entangled state is not known and thus has to depend on the lower bound. The defined witness operator has been proven to be useful in the detection and quantification of not only the NPTES but also BES. We have illustrated the above facts with a few examples.\\
\noindent In $Chapter-6$, the problem of the physical implementation of the realignment criterion has been discussed. In this chapter, we have shown that the matrix obtained after applying the realignment operation on a bipartite state can be expressed in terms of the partial transposition operation along with column interchange operations. These column interchange operations form a permutation matrix which can be implemented through swap operator acting on the state.  Using this mathematical framework, we have determined the exact expression for the first moment of the realignment matrix experimentally. This has been accomplished by showing that the first moment of the realignment matrix can be expressed as the expectation value of a permutation operator which indicates the possibility of its measurement. Further, we have provided an estimation of the higher order moments of the Hermitian matrix $[R(\rho_{AB})]^{\dagger} R(\rho_{AB})$ 
in terms of its first realigned moment. Moreover, we have defined a new realignment operation for three-qubit states. We have shown that the three-qubit realigned matrix can be expressed using a permutation operator and derived separability criteria to detect three-qubit genuine entangled states using SPA-PT and the proposed realignment operation. It has been illustrated using an example that the criteria has the potential to detect BES.

\section*{Future Scope}
In the literature, various schemes exist for the detection and classification of entangled states
in higher dimensional bipartite and multipartite systems but most of them may not be implemented experimentally. In the present thesis, we have focussed on overcoming this problem
and have presented theoretical solutions to a few of them, but still, there is a scope to improve this situation. A few ideas for future work are discussed below:\\
\textbf{(i)} \textbf{Quantification of bound entanglement via SPA-R map}\\
There exist various entanglement measures such as concurrence, negativity, relative entropy of entanglement, and geometric measure of entanglement that can quantify the amount of entanglement in a two-qubit as well as higher dimensional bipartite pure and mixed state. Now, the question arises of whether the entanglement measures can quantify the amount of entanglement for any arbitrary dimensional bipartite system physically. This situation is complicated in higher dimensional systems since only a handful of measures are available to quantify higher dimensional bipartite mixed states. For instance, concurrence is a well-defined measure but we do not have a closed formula for concurrence in higher dimensional systems. Secondly, we have another easily computable measure of entanglement, namely negativity, which may be used to quantify the amount of entanglement in higher dimensional bipartite states but the problem with this measure is that it depends on the negative eigenvalues of the non-physical partial transposition operation. Thus, negativity does not correspond to a completely positive map and hence, is difficult to implement in the laboratory. Also, negativity cannot be used for the quantification of BES. In the present thesis, we provide experimentally realizable entanglement detection criteria based on structural physical approximation of realignment operation. This work can be further extended to the quantification of bound entanglement in arbitrary dimensional bipartite and multipartite quantum systems using the method of SPA.\\
\textbf{(ii)} \textbf{Physical realization of higher order moments of realignment matrix}\\
In recent technology, moments become practical tools in estimating properties of quantum systems, including quantum entanglement. Several moment based methods have been proposed in the literature. However, the measurement of moments of the realigned matrix is still considered an open problem. In this thesis, we have revealed the exact expression for the first moment of the realignment matrix in terms of the expectation value of a permutation operator which makes its measurement possible. For higher ordered moments, we have obtained lower bounds and upper bounds in terms of the first order moment. This work can be extended to obtain the exact expression for the higher order moments of the realignment matrix in terms of an experimentally realizable quantity. \\
\textbf{(iii) Detection and classification of bound entangled states in multipartite systems}\\
Classification of multi-qubit quantum state is another important problem in quantum information theory that gets more complicated when the number of qubits is greater than or equal to three. We have developed a new realignment operation for three-qubit states and shown that it can be used to detect genuine entanglement. There is further scope to develop methods for the classification of multipartite entanglement using the proposed realignment operation.
\begin{center}
	****************
\end{center}


\renewcommand{\bibname}{Bibliography}

\begin{center}
	****************
\end{center}

\newpage
\addcontentsline{toc}{chapter}{List of Publications}
\chapter*{List of Publications}
\begin{enumerate}
\item 	\textbf{Shruti Aggarwal} and Satyabrata Adhikari, \textit{"Witness operator provides better estimate of the lower bound of concurrence of bipartite bound entangled states in $d_1 \otimes d_2$ dimensional system"}, Quantum Information Processing, 20, 83 (2021). \textbf{SCIE}, Impact Factor: 2.349.
\item 	\textbf{Shruti Aggarwal} and Satyabrata Adhikari, \textit{"Search for an efficient entanglement witness operator for bound entangled states in bipartite quantum systems"}, Annals of Physics, 444, 169043 (2022). \textbf{SCI}, Impact Factor: 3.036.
\item 	\textbf{Shruti Aggarwal}, Anu Kumari and Satyabrata Adhikari, \textit{"Physical realization of realignment criteria using structural physical approximation"}, Physical Review A, 108, 012422 (2023) \textbf{ SCI}, Impact Factor: 2.9
\item	\textbf{Shruti Aggarwal}, Satyabrata Adhikari and A. S. Majumdar, \textit{"Entanglement detection in arbitrary dimensional bipartite quantum systems through partial realigned moments"}, Physical Review A, 109, 012404 (2024), \textbf{SCI}
\item 	\textbf{Shruti Aggarwal} and Satyabrata Adhikari, \textit{"Theoretical proposal for the experimental realization of realignment operation"}, arXiv:2307.07952v1\\ Status: Communicated
\end{enumerate}

\begin{center}
	****************
\end{center}

\newpage
\section*{Conferences attended and Paper presented}
\begin{enumerate}
\item	 Paper entitled \emph{"Improved Lower bound of Concurrence for Arbitrary Bipartite Entangled States"} presented in 5th International conference \textbf{Recent Advances in Mathematical Sciences and its Applications (RAMSA-2021)} organized by Jaypee Institute of Technology, Noida from 2nd-4th December 2021.
\medskip
\item Presented a paper entitled \emph{"Detection and quantification of entanglement using witness operator in mixed bipartite quantum states"} in the \textbf{International symposium in honor of great mathematician Srinivasa Ramanujan} on National Mathematics Day held on 22 December, 2021 at Dr. Harisingh Gour Vishwavidyalaya, Sagar, M.P. India.
\medskip
\item 		Presented a paper entitled \emph{"An entanglement witness operator for bound entangled states in bipartite quantum systems"} at \textbf{International Conference on
	Quantum Computing and Communications [QCC-2023]} held during February 8th-11th 2023, at Baba Farid College, Bhatinda, Punjab, India.
\item Presented a poster entitled \emph{"Search for an efficient entanglement witness operator for bound entangled states in bipartite quantum systems"} at \textbf{Young Quantum 2023
	(You-Qu-23)} held during February 15th-18th 2023, at Harish-Chandra Research Institute (HRI), Prayagraj
\item Presented the work entitled \emph{"Entanglement detection via realigned moments in arbitrary dimensional bipartite quantum systems"} in the researcher's meet at fourth international conference on \textbf{Quantum Information and Quantum Technology (QIQT-2023)} organized during June 11th-15th 2023 (in online mode) by Indian Institute of Science Education and Research (IISER), Kolkata.
\end{enumerate}
\noindent\hrulefill

\end{document}